\documentclass[twoside]{article}

\usepackage{amsmath,amssymb,algorithm,amsthm,,cite,cases,bm,float,subfigure,graphicx,url}
\usepackage{pbox}
\numberwithin{equation}{section}
\newtheorem{defn}{Definition}[section]
\newtheorem{fact}{Fact}[section]

\newtheorem{propo}{Proposition}[section]
\newtheorem{conj}{Conjecture}[section]
\newtheorem{thm}{Theorem}[section]
\newtheorem{lem}{Lemma}[section]
\newtheorem{form}{Formula}
\newtheorem{cor}{Corollary}[section]

\newtheorem{assume}{Assumption}[section]

\newcommand{\eqname}[1]{ \textcolor{red}{ #1 }}

\usepackage{amsmath,amssymb,algorithm,amsthm,,cite,subfigure,caption,cases,bm,float,subfigure,graphicx,url,color}
\usepackage[hidelinks]{hyperref}
\usepackage{comment}
\usepackage{thm-restate, enumerate }

\definecolor{darkred}{RGB}{150,0,0}
\definecolor{darkgreen}{RGB}{0,150,0}
\definecolor{darkblue}{RGB}{0,0,200}
\hypersetup{colorlinks=true, linkcolor=darkred, citecolor=darkgreen, urlcolor=darkblue}

\newcommand{\Keywords}[1]{\par\noindent
{\small{\em \indent{\bf{Keywords}}\/}: #1}}

\usepackage[sc]{mathpazo} % Use the Palatino font
\usepackage[T1]{fontenc} % Use 8-bit encoding that has 256 glyphs
\usepackage{microtype} % Slightly tweak font spacing for aesthetics

\usepackage[hmarginratio=1:1,top=32mm,columnsep=20pt]{geometry} % Document margins
\usepackage{multicol} % Used for the two-column layout of the document
\usepackage{hyperref} % For hyperlinks in the PDF

\usepackage{lettrine} % The lettrine is the first enlarged letter at the beginning of the text
\usepackage{paralist} % Used for the compactitem environment which makes bullet points with less space between them

\usepackage{abstract} % Allows abstract customization
\usepackage{titlesec} % Allows customization of titles
\titleformat{\section}[block]{\large\scshape\centering}{\thesection.}{1em}{} % Change the look of the section titles
\titleformat{\subsection}[block]{\large}{\thesubsection.}{1em}{} % Change the look of the section titles

\usepackage{fancyhdr} % Headers and footers
\title{\vspace{-15mm}\fontsize{24pt}{10pt}\selectfont\textbf{The Squared-Error of Generalized LASSO:\\A Precise Analysis}} % Article title

\author{
\large
\textsc{Samet Oymak, Christos Thrampoulidis and Babak Hassibi}\thanks{ 
This work was supported in part by the National Science Foundation under grants CCF-0729203, CNS-0932428 and CIF-1018927, by the Office of Naval Research under the MURI grant N00014-08-1-0747, and by a grant from Qualcomm Inc.}\\[2mm] % Your name
\normalsize Department of Electrical Engineering \\ \normalsize Caltech, Pasadena -- 91125
 \\ % Your institution
\normalsize \href{mailto:soymak@caltech.edu, cthrampo@caltech.edu, hassibi@caltech.edu}{soymak@caltech.edu, cthrampo@caltech.edu, hassibi@caltech.edu} % Your email address
\vspace{-5mm}
}
\date{}

\begin{document}

\maketitle % Insert title

\thispagestyle{fancy} % All pages have headers and footers

%Equations
\newcommand{\beq}{\begin{equation}}
\newcommand{\eeq}{\end{equation}}
\newcommand{\bea}{\begin{align}}
\newcommand{\eea}{{\end{align}}}

\newcommand{\vp}{\vspace{4pt}}
%Bold upper
\newcommand{\M}{\mathbf{M}}
\newcommand{\tC}{\tilde{\mathbf{C}}}
\newcommand{\tD}{\tilde{\mathbf{D}}}
\newcommand{\Fh}{\mathbf{\hat{F}}}
\newcommand{\ft}{\mathbf{\tilde{f}}}
\newcommand{\Z}{\mathbf{Z}}
\newcommand{\W}{\mathbf{W}}
\newcommand{\U}{\mathbf{U}}
\newcommand{\V}{\mathbf{V}}
\newcommand{\F}{\mathbf{F}}
\newcommand{\Delxf}{{\mathbf{D}}_f(\x_0,{\mathbb{R}}^+)}
\newcommand{\Pelxf}{{\mathbf{P}}_f(\x_0,{\mathbb{R}}^+)}
\newcommand{\Celxf}{{\mathbf{C}}_f(\x_0,{\mathbb{R}}^+)}
\newcommand{\Dlf}{{\mathbf{D}}_f(\x_0,\la )}
\newcommand{\order}[1]{\mathcal{O}\left(#1\right)}

\newcommand{\Dmin}{{\mathbf{D}}^*_f}
\newcommand{\Plf}{{\mathbf{P}}_f(\x_0,\la )}
\newcommand{\Df}{{\mathbf{D}}_f(\x_0,}
\newcommand{\Pf}{{\mathbf{P}}_f(\x_0,}
\newcommand{\Clf}{{\mathbf{C}}_f(\x_0,\la )}
\newcommand{\Cf}{{\mathbf{C}}_f(\x_0,}

\newcommand{\DC}{{\mathbf{D}}(\Cc)}
\newcommand{\CC}{{\mathbf{C}}(\Cc)}
\newcommand{\PC}{{\mathbf{P}}(\Cc)}

\newcommand{\lac}{\lambda_{\text{crit}}}
\newcommand{\labb}{\lambda_{\text{best}}}
\newcommand{\taub}{\tau_{\text{best}}}
\newcommand{\lam}{\lambda_{{\max}}}

\newcommand{\DelxC}{{\mathbf{\Delta}}_\Cc(\x_0)}
\newcommand{\Gb}{\mathbf{G}}
\newcommand{\Sb}{\mathbf{S}}
\newcommand{\X}{\mathbf{X}}
\newcommand{\A}{\mathbf{A}}
\newcommand{\Lb}{\mathbf{L}}
\newcommand{\Y}{\mathbf{Y}}

%Bold lower
\newcommand{\bu}{\text{Proj}}
\newcommand{\corr}{\text{corr}}
\newcommand{\bi}{{\Pi}}
\newcommand{\w}{\mathbf{w}}
\newcommand{\heta}{\hat{E}}
\newcommand{\geta}{\hat{\zeta}}
\newcommand{\feta}{{E^*}}
\newcommand{\x}{\mathbf{x}}
\newcommand{\ub}{\mathbf{u}}
\newcommand{\g}{\mathbf{g}}
\newcommand{\vb}{\mathbf{v}}
\newcommand{\q}{\mathbf{q}}
\newcommand{\p}{\mathbf{p}}
\newcommand{\bb}{\mathbf{b}}
\newcommand{\e}{\mathbf{e}}
\newcommand{\fh}{{\hat{f}}}
\newcommand{\tb}{\mathbf{t}}
\newcommand{\y}{\mathbf{y}}
\newcommand{\s}{\mathbf{s}}
\newcommand{\z}{\mathbf{z}}
\newcommand{\bseta}{{\boldsymbol{\eta}}}
\newcommand{\Iden}{\mathbf{I}}
\newcommand{\cb}{\mathbf{c}}  
\newcommand{\ab}{\mathbf{a}}
\newcommand{\oneb}{\mathbf{1}}
\newcommand{\h}{\mathbf{h}}
\newcommand{\hx}{{\mathbf{x}}^{*}(\la,\z)}
\newcommand{\fx}{\hat{\mathbf{x}}(\la,\z)}
\newcommand{\hw}{{\mathbf{w}}^{*}(\la,\z)}
\newcommand{\wh}{\hat{\mathbf{w}}}
\newcommand{\fw}{\hat{\mathbf{w}}(\la,\z)}
\newcommand{\ffw}{\hat{\mathbf{w}_r}(\la,\z)}
\newcommand{\xn}{{\mathbf{x}}_0}

%hat, tilde, etc
\newcommand{\tih}{\tilde{\mathbf{h}}}
\newcommand{\tz}{\tilde{\z}}
\newcommand{\tw}{\tilde{\w}}
\newcommand{\bT}{{\bar{T}}}
\newcommand{\bP}{{\bar{P}_\w}}
\newcommand{\wP}{{{P}_\w}}
\newcommand{\bS}{{\bar{S}}}

%Calligrophic
\newcommand{\Sc}{{\mathcal{S}}}
\newcommand{\map}{{\text{map}}}
\newcommand{\call}{{\text{calib}}}
\newcommand{\Rca}{{\mathcal{R}}}
\newcommand{\Rco}{{\mathcal{R}}_{\text{ON}}}
\newcommand{\Rcf}{{\mathcal{R}}_{\text{OFF}}}
\newcommand{\Rci}{{\mathcal{R}}_{\infty}}
\newcommand{\Rr}{\mathbf{\rho}}
\newcommand{\Bc}{{\mathcal{B}}}
\newcommand{\Yc}{\mathcal{Y}}
\newcommand{\Dc}{\mathcal{D}}
\newcommand{\Lc}{\hat{\mathcal{L}}}
\newcommand{\Lco}{\mathcal{{L}}}

\newcommand{\wo}{\mathbf{\w}^*}

\newcommand{\Fco}{\mathcal{F}}
\newcommand{\Uco}{\mathcal{{U}}}
\newcommand{\fone}{L}
\newcommand{\Ldev}{\mathcal{L}_{dev}}
\newcommand{\Xc}{\mathcal{X}}
\newcommand{\Zc}{\mathcal{Z}}
\newcommand{\Uc}{\hat{\mathcal{U}}}
\newcommand{\pol}{^\circ}
\newcommand{\Kc}{\mathcal{K}}
\newcommand{\Nc}{\text{Null}}
\newcommand{\Rc}{\text{Range}}
\newcommand{\Nn}{\mathcal{N}}
\newcommand{\Cc}{\mathcal{C}}
\newcommand{\Gc}{\mathcal{G}}
\newcommand{\Ac}{\mathcal{A}}
\newcommand{\Pc}{\mathcal{P}}

%Math BB
\newcommand{\R}{\mathbb{R}}
\newcommand{\Pro}{{\mathbb{P}}}
\newcommand{\E}{{\mathbb{E}}}

%Greek, etc.
\newcommand{\xu}{\hat{x}_i}
\newcommand{\xz}{x_{0,i}}
\newcommand{\wi}{\hat{w}_i}

\newcommand{\ti}{\tilde}
\newcommand{\el}{\eta_{DN}(\la,\z)}
\newcommand{\paf}{\pa f(\x_0)}
\newcommand{\pef}{\pa f(\x)}
\newcommand{\dir}{D_{f,\x}(\w)}
\newcommand{\din}{D_{f,\x_0}(\w)}
\newcommand{\paw}{\pa f'_{\x_0}(\w)}
\newcommand{\fp}{f'_{\x_0}}
\newcommand{\ff}{\hat{f}_{\x_0}}
\newcommand{\ffx}{\hat{f}_{\x}}
\newcommand{\ffr}{f_{r,\x_0}}
\newcommand{\ffrx}{f_{r,\x}}
\newcommand{\la}{{\lambda}}
\newcommand{\lab}{{\bm{\lambda}}}
\newcommand{\eps}{\epsilon}
\newcommand{\om}{\omega}
\newcommand{\si}{\sigma}
\newcommand{\st}{\star}
\newcommand{\Fc}{\hat{{\mathcal{F}}}}
\newcommand{\Fch}{\hat{{\mathcal{F}}}}
\newcommand{\Tc}{\mathcal{T}}
\newcommand{\pa}{\partial}

\newcommand{\mub}{\boldsymbol{\mu}}
%Rest
\newcommand{\sg}{\text{sgn}}
\newcommand{\hg}{\hat{g}}
\newcommand{\Du}{\text{dual}}
\newcommand{\cl}{\text{Cl}}
\newcommand{\tr}{\text{trace}}
\newcommand{\rk}{\text{rank}}
\newcommand{\cn}{\text{cone}}
\newcommand{\dbt}{\text{\bf{d}}}
\newcommand{\dt}{\text{{dist}}}
\newcommand{\dtR}{\dt_{\R^+}(\h)}
\newcommand{\vs}{\vspace}
\newcommand{\hs}{\hspace}
\newcommand{\has}{\hat{\s}}
\newcommand{\nn}{\nonumber}
\newcommand{\li}{\left<}
\newcommand{\ri}{\right>}
\newcommand{\nor}{\|\cdot\|}

%----------------------------------------------------------------------------------------
%	ABSTRACT
%----------------------------------------------------------------------------------------
\begin{abstract} 

We consider the problem of estimating an unknown signal $\x_0$ from noisy linear observations $\y = \A\x_0 + \z\in \mathbb{R}^m$. In many practical instances of this problem,  $\x_0$ has a certain structure that can be captured by a structure inducing function $f(\cdot)$. For example, $\ell_1$ norm can be used to encourage a sparse solution. To estimate $\x_0$ with the aid of a convex $f(\cdot)$, we consider three variations of the widely used $LASSO$ estimator and provide sharp characterizations of their performances.
 Our study falls under a generic framework, where the entries of the measurement matrix $\A$ and the noise vector $\z$ have zero-mean normal distributions with variances $1$ and $\sigma^2$, respectively.
For the LASSO estimator $\x^*$, we ask: ``What is the precise estimation error as a function of the noise level $\sigma$, the number of observations $m$ and the structure of the signal?". In particular, we attempt to calculate the Normalized Square Error (NSE) defined as $\frac{\|\x^*-\x_0\|_2^2}{\sigma^2}$. We show that, the structure of the signal $\x_0$ and choice of the function $f(\cdot)$ enter the error formulae through the summary parameters $\Delxf$ and $\Dlf$, which are defined as the ``Gaussian squared-distances'' to the subdifferential cone and to the $\lambda$-scaled subdifferential of $f$ at $\x_0$, respectively. The first estimator assumes a-priori knowledge of $f(\x_0)$ and is given by $\arg\min_{\x}\left\{{\|\y-\A\x\|_2}~\text{subject to}~f(\x)\leq f(\x_0)\right\}$. We prove that its worst case NSE is achieved when $\sigma\rightarrow 0$ and concentrates around $\frac{\Delxf}{m-\Delxf}$. Secondly, we consider $\arg\min_{\x}\left\{\|\y-\A\x\|_2+\la f(\x)\right\}$, for some penalty parameter $\la\geq 0$. This time, the NSE formula depends on the choice of $\lambda$ and is given by $\frac{\Dlf}{m-\Dlf}$ over a range of $\la$. The last estimator is $\arg\min_{\x}\left\{\frac{1}{2}\|\y-\A\x\|_2^2+ \sigma\tau f(\x)\right\}$.
 We establish a mapping between this and the second estimator and propose a formula for its NSE. As useful side results, we find
      explicit formulae for the optimal estimation performance
      and the optimal penalty parameters $\labb$ and $\tau_{best}$. Finally, for a number of important structured signal classes, we translate our abstract formulae to closed-form upper bounds on the NSE.

\vspace{5pt}

\end{abstract}

\newpage

%
%
%
%
%\begin{document}
%\title{\Large{\bf{The Squared-Error of Generalized LASSO:\\A Precise Analysis}}\vs{15pt}}
%%Other title possibilities
%%The risk of noisy linear inverse algorithms
%
%%\author{Samet Oymak\hs{120pt}Babak Hassibi\vs{10pt}\\Department of Electrical Engineering \\ Caltech, Pasadena - 91125}
%%\address{}
%
%\author{Samet Oymak\hs{60pt}Christos Thrampoulidis\hs{60pt}Babak Hassibi\vs{10pt}\\Department of Electrical Engineering \\ Caltech, Pasadena -- 91125
%\thanks{
%Emails: {\tt(soymak, cthrampo, hassibi)@caltech.edu}. 
%This work was supported in part by the National Science Foundation under grants CCF-0729203, CNS-0932428 and CIF-1018927, by the Office of Naval Research under the MURI grant N00014-08-1-0747, and by a grant from Qualcomm Inc.}
%}
%
%\date{}
%\maketitle

%\center
%\today
%\center

%{\bf{Problems/corrections/fixes:}}
%\begin{itemize}
%%\item Split section 4, introduction to tech vs real tech -chris
%\item Fixing main result, approximation and rest - samet
%%\item Balance the main result verbally move to contrib if necessary - samet 
%\item Rita Foygel Lester Mackey - chris
%\item Definitions - samet
%\item LASSO related things, matrix LASSO, noisy matrix completion, Montanari!! - chris, samet
%\item ell22-LASSO - samet
%\item Figure is important subdif - samet chris
%\item Simulations - samet
%\item first 2-3 section, make a pass
%%\item justification - chris
%%\item Mention somewhere that recovering $\x_0$ robustly means error proportional to the noise level.
%%\item Fix Abstract.
%%\item Define $\Tc_f(\x_0)$, $\mathcal{N}(0,\mathbb{I})$. -chris
%\end{itemize}

\setcounter{tocdepth}{2}
\tableofcontents

%\section{Things to do (in correct order)}
%
%\begin{itemize}
%\item Doing the proofs for the main case under first order approximation!!!
%\begin{itemize}
%\item Lower bound
%\item Upper bound
%\item Restricted lower bounds
%\item Proof for the constrained problem
%\end{itemize}
%\item Showing asymptotic equality of first order and regular.
%\begin{itemize}
%\item Also for the constrained problem
%\item Can we show ``first order approximation'' gives more error than the regular problem?
%\end{itemize}
%\item Relation between the performance of LASSO and the simple denoising problem.
%\begin{itemize}
%\item How to map denoising regularizer to LASSO regularizer?
%%\item Formula for optimal $\la^*$ that does as good as constrained and its relation to minimax denoising.
%\end{itemize}
%
%
%
%\end{itemize}
\newpage

\section{Introduction}

\subsection{The Generalized LASSO Problem}
Recovering a structured signal $\x_0\in\mathbb{R}^n$ from a vector of limited and noisy linear observations $\y=\A\x_0+\z\in\mathbb{R}^m$, is a problem of fundamental importance
%central problem of interest in numerous applications ranging from 
encountered in several disciplines including machine learning, signal processing, network inference and many more \cite{CandesMain,candes-tao,Comp,Don1}.
 A typical approach for estimating the structured signal $\x_0$ from the measurement vector $\y$, is picking some proper structure inducing function $f(\cdot)$ and solving the following problem
\beq
\x^*_{LASSO} = \operatorname*{argmin}_{\x}\left\{~\frac{1}{2}\|\y-\A\x\|_2^2+\la f(\x)~\right\}\label{general model},
\eeq
for some nonnegative penalty parameter $\lambda$.

In case $\x_0$ is a sparse vector, the associated structure inducing function is the $\ell_1$ norm, i.e. $f(\x)=\|\x\|_1$. The resulting $\ell_1$-penalized quadratic program in \eqref{general model} is known as the  LASSO in the statistics literature. LASSO was originally introduced in \cite{Tikh} and has since then been subject of great interest as a natural and powerful approach to do noise robust compressed sensing (CS), \cite{Tikh, BicRit,BunTsy,BayMon,Mon,ExampleLAS,HighLAS,OracleLAS,ZhaoLasso,WainLasso,DonLasso}. There are also closely related algorithms such as SOCP variations and the Dantzig selector \cite{CanTao,BoydSOCP}. 
Of course, applications of \eqref{general model} are not limited to sparse recovery; they extend to various problems including the recovery of block sparse signals \cite{generalLAS,blocklasso}, the matrix completion problem \cite{CandesMatrixComp1,Koltc} and the total variation minimization \cite{TotalCompress1}. In each application, $f(\cdot)$ is chosen in accordance to the structure of $\x_0$. See \cite{Cha} for additional examples and a principled approach to
constructing such penalty functions. 
%, and are by no means limited to
In this work, we consider arbitrary convex penalty functions $f(\cdot)$ and we commonly refer to this generic formulation
in \eqref{general model} 
% of the noisy compressed sensing problem
  as the ``Generalized LASSO" or simply ``LASSO" problem.

\subsection{Motivation}
The LASSO problem can be viewed  as a ``merger" of two closely related problems, which have both recently attracted a lot of attention by the research community;
the problems of noiseless CS and that of proximal denoising.

%Hilighting those recent developments and establishing their relevance to our work, is meant to not only provide the reader with the necessary background, but to also motivate and emphasize the nature of  our results. 

\subsubsection{Noiseless compressed sensing}
In the noiseless CS problem one wishes to recover $\x_0$ from the random linear measurements $\y=\A\x_0$. A common approach is solving the following convex optimization problem
\beq
\min_{\x} f(\x)~~~\text{subject to}~~~\y=\A\x\label{LinInv}.
\eeq
A critical performance criteria for the problem \eqref{LinInv} concerns the minimum number of measurements needed to guarantee successful recovery of $\x_0$ \cite{Sto1,StojSharp,DonPhaseTrans,DonNonNegative,DonCentSym, Cha,McCoy}. Here, success means that $\x_0$  is the unique minimizer of \eqref{LinInv}, with high probability, over the realizations of 
the random matrix 
$\A$. 

%The minimum such number of measurements (often termed the ``phase transition" point of problem \eqref{LinInv})  has recently been show \cite{McCoy} to coincide with $\Delxf$, i.e. the statistical dimension of the tangent cone of $f(\cdot)$ at $\x_0$.

\vp

%\noindent$\bullet$ {\bf{Estimation from noisy observations:}} 
%\begin{comment}
\subsubsection{Proximal denoising}
The proximal denoising problem tries to estimate $\x_0$ from noisy but uncompressed observations $\y=\x_0+\z$, $\z\sim\Nn(0,\sigma^2 \mathbf{I}_{n})$, where we write $\mathbf{I}_{k}$ for  the identity matrix of size $k\times k$, $k\in\mathbb{Z}^+$. In particular, it solves,
\beq
\min_{\x}\left\{ \frac{1}{2}\|\y-\x\|_2^2\label{DirDen} + \la\sigma f(\x)\right\}.
\eeq
A closely related approach to estimate $\x_0$, which requires prior knowledge $f(\x_0)$ about the signal of interest $\x_0$, is solving the constrained denoising problem:
\beq
\min_{\x}\|\y-\x\|_2^2~~~\text{subject to}~~~ f(\x)\leq f(\x_0)\label{DirCon}.
\eeq
The natural question to be posed 
%here
in both cases 
is how well can one estimate $\x_0$ via \eqref{DirDen}
 (or \eqref{DirCon})
 \cite{soft-thresh,DonAcu,ChaJor,Oymak,OymAller}? The minimizer $\x^*$ of \eqref{DirDen} 
 (or \eqref{DirCon})
  is a function of the noise vector $\z$ and the common measure of performance, is the normalized mean-squared-error which is defined as %to be equal to
$
\frac{\E\|\x^*-\x_0\|_2^2}{\sigma^2}.
$
%\end{comment}
%In \cite{Oymak}, it has been shown that, the denoiser risk is equal to $\Dlf$. Hence, the ``Gaussian squared distance to the scaled subdifferential'' $\Dlf$, is a measure of the estimation performance as a function of the regularizer $\la$. 
%%$\Dlf$ and $\Delxf$ is known to be closely related. While the exact details are deferred to the upcoming sections, we have,
%%\beq
%%\Dlf\geq \Delxf~~~\text{and}~~~\min_{\la\geq 0}\Dlf\approx \Delxf
%%\eeq
%%In fact,
%Moreover, it has been shown in \cite{Oymak} that the risk of the constrained denoiser,
%\beq
%\min_{\x}\|\y-\x\|_2^2~~~\text{subject to}~~~ f(\x)\leq f(\x_0)\label{DirCon},
%\eeq
%is equal to $\Delxf$.
%%\noindent$\bullet$ {\bf{The merger, LASSO:}} 
\subsubsection{The ``merger" LASSO}
The Generalized LASSO problem is naturally merging the problems of noiseless CS and proximal denoising.  The compressed nature of measurements, poses the question of finding the minimum number of measurements required to recover $\x_0$ \emph{robustly}, that is with error proportional to the noise level. When recovery is robust, it is of importance to be able to explicitly characterize how good the estimate is. In this direction, when $\z\sim\Nn(0,\sigma^2 \mathbf{I}_{m})$, a common measure of performance for the LASSO estimate $\x^*_{LASSO}$ is defined to be the \textbf{normalized squared error} (NSE) :
$$NSE = \frac{1}{\sigma^2}\|\x^*_{LASSO} - \x_0\|^2_2.$$ This is exactly the main topic of this work: \emph{proving precise bounds for the NSE of the Generalized LASSO problem.}

In the specific case of $\ell_1$-penalization in \eqref{general model}, researchers have considered other performance criteria additional to the NSE \cite{WainLasso,DonLasso,ZhaoLasso}. As an example, we mention the support recovery  criteria \cite{WainLasso}, which measures how well \eqref{general model} recovers the subset of nonzero indices of $\x_0$. However, under our general setup, where we allow arbitrary structure to the signal $\x_0$, the NSE serves as the most natural measure of performance and is, thus, the sole focus in this work. In the relevant literature,  researchers have dealt with the analysis of the NSE of \eqref{general model} under several settings (see Section \ref{sec:relLit}). Yet, \emph{we still lack a general theory that would yield precise bounds for the squared-error of \eqref{general model} for arbitrary convex regularizer $f(\cdot)$. This paper aims to close this gap.} Our answer involves inherent quantities regarding the geometry of the problem which, in fact, have recently appeared in the related literature, \cite{Cha,McCoy,Foygel,Mon,BayMon,Oymak}.

\subsection{Three Versions of the LASSO Problem}
Throughout the analysis, we assume $\A\in\R^{m\times n}$ has independent standard normal entries and $\z\sim\Nn(0,\sigma^2\mathbf{I}_{m})$. 
Our approach 
%We also 
tackles various forms of the LASSO all at once, and relates them to each other. In particular, we consider the following three versions:
\vspace{-2pt}
\begin{itemize}
\item[$\star$] \textbf{C-LASSO}: Assumes a-priori knowledge of $f(\x_0)$ and solves,
\beq
 \x_c^*(\A,\z) = \text{arg}\min_{\x}\|\y-\A\x\|_2~~~\text{subject to}~~~f(\x)\leq f(\x_0).\label{cmodel}\vspace{-1pt}
 \eeq
\item[$\star$] \textbf{$\ell_2$-LASSO}: Uses $\ell_2$-penalization rather than $\ell_2^2$ and solves,
\beq
 \x_{\ell_2}^*(\la,\A,\z) = \text{arg}\min_{\x}\left\{ ~\|\y-\A\x\|_2 + \la f(\x) ~\right\}\label{ell2model}.\vspace{-8pt}
 \eeq
\item[$\star$] \textbf{$\ell_2^2$-LASSO}: the original form given in \eqref{general model}
%\footnote{For the rest of the paper we use $\la$ to denote the penalty parameter of the $\ell_2$-LASSO and $\tau$ for the corresponding parameter of the $\ell_2^2$-LASSO.}
:
\begin{align}
\x_{\ell_2^2}^*(\tau,\A,\z) = \text{arg}\min_{\x}\left\{~\frac{1}{2}\|\y-\A\x\|_2^2+\sigma\tau f(\x)~\right\}.\label{ell22model}
\end{align}
\end{itemize}
C-LASSO in \eqref{cmodel} stands for ``Constrained LASSO". This version of the LASSO problem assumes some a-priori knowledge about $\x_0$, which makes the analysis of the problem arguably simpler than that of the other two versions, in which the role of the penalty parameter (which is meant to  compensate for the lack of a-priori knowledge) has to be taken into consideration. 
%Our discussion will mostly focus on providing precise characterizations for the NSE of the $\ell_2$-LASSO and the $\ell_2^2$-LASSO. 
To distinguish between the $\ell_2$-LASSO and the  $\ell_2^2$-LASSO, we use $\la$ to denote the penalty parameter of the former and $\tau$ for the penalty parameter of the latter.
Part of our contribution is establishing useful connections between these three versions of the LASSO problem. We will often drop the arguments $\la,\tau,\A,\z$ from the LASSO estimates defined in \eqref{cmodel}--\eqref{ell22model}, when clear from context.

\subsection{Relevant Literature}\label{sec:relLit}
Precise characterization of the NSE of the LASSO is closely related to the precise performance analysis of noiseless CS and proximal denoising. To keep the discussion short, we defer most of the comments on the connections of our results to these problems to the main body of the paper. Table \ref{table:lit} provides a summary of the relevant literature and highlights the area of our contribution.

%\vspace{-15pt}
\begin{center}
\begin{table}[h]
\begin{center}
  \begin{tabular}{ | c |  c  |  c |}
     \cline{2-3}
     \multicolumn{1}{c|}{}&\pbox{20cm}{\vp{\bf{{{Convex functions}}}}\vp}& \pbox{20cm}{\bf{{{$\ell_1$-minimization}}}}\\ \hline
    {\color{darkred}{{Noiseless CS}}} & \pbox{20cm}{ \vp Chandrasekaran et al. \cite{Cha}\\ Amelunxen et al. \cite{McCoy}\vp}  &\pbox{20cm}{Donoho and Tanner, \cite{DonPhaseTrans}\\ Stojnic, \cite{Sto1}} \\    \hline
     {\color{darkred}{{\begin{tabular}{l} Proximal \\ denoising\end{tabular}}}}  & \pbox{20cm} {\vp Donoho et al. \cite{DonAcu}\\Oymak and Hassibi\cite{Oymak}} & \pbox{20cm}{\vp Donoho \cite{soft-thresh}\vp} \\ \hline    
     {\color{darkred}{{LASSO}}}  & \color{red}{Present paper} & \pbox{20cm}{\vp Bayati and Montanari, \\ \cite{BayMon}, \cite{Mon}\\ Stojnic, \cite{StojLAS}\vp} \\
    \hline
  \end{tabular}
  \end{center}
    \caption{Relevant Literature.}
  \label{table:lit}

\end{table}
  \end{center}

%*************************************
%************* EXTRA ***************
%*************************************
%%\chris{Maybe we want to add a few things about literature in CS and Denoising here.}
%%\begin{itemize}
%%\item Donoho and Tanner, '06: phase transitions for compressed signal recovery (weak, strong and sectional thresholds); neighborly polytopes and Grassman angles (very complicated to compute---see Xu and Hassibi, '08)
%%\item Stojnic, '09: new framework based on "escape through mesh" theorem and Gaussian widths; much simpler; can be generalized to other norms--see Oymak and Hassibi, '11
%%\item Chandrasekaran, Recht, Parrilo, Willsky, '10: provided a general framework for structured signal recovery using atomic norms and a generalization of Stojnic's approach\item Amelunxen, Lotz, McCoy, Tropp, '13: proved the tightness of Stojnic's Gaussian width bounds; reinterpreted it as {\em statistical dimension}
%%\item Donoho, Johnstone, Montanari, '12: for $\ell_1$ optimization showed that the normalized denoising error is equivalent to the compressed sensing threshold; provided empirical evidence for other norms
%%\item Oymak, Hassibi, '13: proved the equivalence for arbitrary convex norms
%%\end{itemize}

%\vspace{-9pt}
The works closest in spirit to our results include \cite{BayMon,Mon,StojLAS,Bar},  which focus on the exact analysis of the LASSO problem, while restricting the attention on  sparse recovery where
 $f(\x)=\|\x\|_1$ .
% In all these works, the researchers limit their attention on sparse recovery and the analysis of the $\ell_1$-regularized minimization. 
 In \cite{BayMon,Mon}, 
 %\samet{We should refer correctly} 
 Bayati and Montanari are able to show that the mean-squared-error of the LASSO problem is equivalent to the one achieved by a properly defined ``Approximate Message Passing'' (AMP) algorithm. Following this connection and after evaluating the error of the AMP algorithm, they obtain an explicit expression for the mean squared error of the LASSO algorithm in an asymptotic setting. In \cite{Bar}, Maleki et al. proposes Complex AMP, and characterizes the performance of LASSO for sparse signals with complex entries.
% uses ``Approximate Message Passing'' (AMP) algorithms and gives formulas for \eqref{general model}.
In \cite{StojLAS}, Stojnic's approach relies on results on Gaussian processes \cite{Gor,Gor2} to derive sharp bounds for the \emph{worst case NSE} of the $\ell_1$-constrained LASSO problem in \eqref{cmodel}.
% makes use of results on Gaussian processes, \cite{Gordon}, and analyzes the Constrained-LASSO problem \eqref{cmodel} with $\ell_1$ minimization. 
Our approach in this work builds on the framework proposed by Stojnic, but extends the results in multiple directions as noted in the next section.
%and we make use of the properties of Gaussian processes. 

\subsection{Contributions}
This section summarizes our main contributions.
% An extended and detailed discussion is provided in the following sections.
 In short, this work:
\begin{itemize}
%\item \emph{simplifies} the framework proposed by Stojnic in \cite{StojLAS}.
\item \emph{generalizes} the results of \cite{StojLAS} on the constrained LASSO for arbitrary convex functions; proves that the worst case NSE is achieved when the noise level $\sigma\rightarrow 0$, and derives sharp bounds for it.
\item \emph{extends} the analysis to the NSE of the more challenging $\ell_2$-LASSO; provides 
bounds as a function of the penalty parameter $\la$, which are sharp when $\sigma\rightarrow 0$.
%   bounds as functions of the penalty parameter $\la$.
\item identifies a \emph{connection} between the $\ell_2$-LASSO to the $\ell_2^2$-LASSO; proposes a formula for precisely calculating the NSE of the latter when $\sigma\rightarrow 0$.
%the squared estimation error of the latter.
\item provides simple \emph{recipes} for the optimal tuning of the penalty parameters $\la$ and $\tau$ in the $\ell_2$ and $\ell_2^2$-LASSO problems.
\item analyzes the regime in which \emph{stable} estimation of $\x_0$ fails. 
\end{itemize}

\subsection{Motivating Examples}
Before going into specific examples, it is instructive to consider the scenario where $f(\cdot) = 0$. This reduces the problem to a regular least-squares estimation problem, the analysis of which is easy to perform. When $m<n$, the system is underdetermined, and one cannot expect $\x^*$ to be a good estimate. When $m\geq n$, the estimate can be given by $\x^*=(\A^T\A)^{-1}\A^T\y$. In this case, the normalized mean-squared-error takes the form,
			\beq
			\frac{\E\|\x^*-\x_0\|^2}{\sigma^2}=\frac{\E[\z^T\A(\A^T\A)^{-2}\A^T\z]}{\sigma^2}=\E[{\tr(\A(\A^T\A)^{-2}\A^T)}]=\E[{\tr((\A^T\A)^{-1})}]. \nn
			\eeq
			$\A^T\A$ is a Wishart matrix and its inverse is well studied. In particular, when $m\geq n+2$, we have $\E[(\A^T\A)^{-1}]=\frac{\Iden_n}{m-n-1}$ (see \cite{multiv}). Hence,
%			\beq
%			\E[(\A^T\A)^{-1}]=\frac{\Iden_n}{m-n-1}~~~\text{and}~~~\E[{\tr((\A^T\A)^{-1})}]=\frac{n}{m-n-1}.\nn
%			\eeq
\beq
\frac{\E\|\x^*-\x_0\|^2}{\sigma^2}=\frac{n}{m-n-1}. \label{eq:LS}
\eeq

\begin{center}
\emph{ How does this result change when a nontrivial convex function $f(\cdot)$ is introduced?}
\end{center}

\noindent Our message is simple: 		
%			
%			This is highly consistent with results of this work. For $f(\cdot)=0$, we have, $\paf=0$, which yields, $\Delxf=\Dlf=n$ for all $\la\geq 0$. Hence, based on Theorem \ref{thm:CLASSO}, we expect the NSE to be around $\frac{n}{m-n}$ when $m>n$. On the other hand, 
%			
\emph{when $f(\cdot)$ is an arbitrary convex function, the LASSO error formula is obtained by simply replacing the ambient dimension $n$ in \eqref{eq:LS} with a summary parameter $\Delxf$ or $\Dlf$ }. These parameters are defined as the expected squared-distance of a standard normal vector in $\mathbb{R}^n$ to the conic hull of the subdifferential  $\text{cone}(\paf)$ and to the scaled subdifferential $\la\paf$, respectively. They summarize the effect of the structure of the signal $\x_0$ and choice of the function $f(\cdot)$ on the estimation error.

%Let $\A\in\Nn(0,\Iden_{m\times n})$ and $\z\sim\Nn(0,\sigma^2\Iden_m)$. We will consider two examples to illustrate our results.

To get a flavor of the (simple) nature of our results, we briefly describe how they apply in three commonly encountered settings, namely the ``sparse signal",  ``low-rank matrix" and ``block-sparse signal" estimation problems. For simplicity of exposition, let us focus on the C-LASSO estimator in \eqref{cmodel}. A more elaborate discussion, including estimation via $\ell_2$-LASSO and $\ell_2^2$-LASSO, can be found in Section \ref{sec:closedForm}. The following statements are true with high probability in $\A,\vb$ and hold under mild assumptions.

%\noindent$\bullet $
 \noindent{$~~\emph1$}.  {\bf{Sparse signal estimation:}} Assume $\x_0\in\R^n$ has $k$ nonzero entries. In order to estimate $\x_0$, use the Constrained-LASSO and pick $\ell_1$-norm for $f(\cdot)$. Let $m>2k(\log\frac{n}{k}+1)$. 
% Denote the LASSO estimate by $\x^*$. 
 Then, 
 %under mild assumptions, 
\beq\frac{\|\x^*_c-\x_0\|_2^2}{\sigma^2}\lesssim \frac{2k(\log{\frac{n}{k}+1)}}{m-2k(\log\frac{n}{k}+1)}.\label{eq:clo1}\eeq

%\noindent$\bullet$
 \noindent{$~~\emph2$}. {\bf{Low-rank matrix estimation:}} Assume $\X_0\in\R^{d\times d}$ is a rank $r$ matrix, $n=d\times d$. This time, $\x_0\in\R^n$ corresponds to vectorization of $\X_0$ and $f(\cdot)$ is chosen as the nuclear norm $\|\cdot\|_\st$ (sum of the singular values of a matrix) \cite{Fazel,RechtFazel}. Hence, we observe $\y=\A\cdot\text{vec}(\X_0)+\z$ and solve,% i.e.%, as $f(\cdot)$ and solve,
\beq
\min_{\X\in\R^{d\times d}} \|\y-\A\cdot\text{vec}(\X)\|_2~~~\text{subject to}~~~\|\X\|_\st\leq \|\X_0\|_\st\nn
\eeq
 Let $m>6dr$. Denote the LASSO estimate by $\X_c^*$ and use $\|\cdot\|_F$ for the Frobenius norm of a matrix. Then, 
 %under mild assumptions,
\beq\frac{\|\X_c^*-\X_0\|^2_F}{\sigma^2}\lesssim \frac{6dr}{m-6dr}.\label{eq:clo2}\eeq

 \noindent{$~~\emph3$}. {\bf{Block sparse estimation:}} Let $n=t\times b$ and assume the entries of $\x_0\in\R^n$ can be grouped into $t$ known blocks of size $b$ so that only $k$ of these $t$ blocks are nonzero. To induce the structure, the standard approach is to use the $\ell_{1,2}$ norm which sums up the $\ell_2$ norms of the blocks, \cite{Yonina,Block,StojBlock,RechtBlock}. In particular, denoting the subvector corresponding to $i$'th block of a vector $\x$ by $\x_i$, the $\ell_{1,2}$ norm is equal to $\|\x\|_{1,2}=\sum_{i=1}^t\|\x_i\|_2$.
Assume $m>4k(\log\frac{t}{k}+b)$ .
%and denote the estimate of Constrained-LASSO with $\ell_{1,2}$ minimization by $\x^*_c$. 
Then,
% under mild assumptions,
\beq\frac{\|\x_c^*-\x_0\|^2_2}{\sigma^2}\lesssim \frac{4k(\log\frac{t}{k}+b)}{m-4k(\log\frac{t}{k}+b)}.\label{eq:clo3}\eeq

Note how \eqref{eq:clo1}-\eqref{eq:clo3} are similar in nature to \eqref{eq:LS}.

\section{Our Approach}
\label{sec:app}
%serves to motivate the nature of our results and
In this section we introduce the main ideas that underlie our approach. This will also allow us to introduce important concepts from convex geometry required for the statements of our main results in Section \ref{sec:mainResults}. 
%
%The purpose of this section is twofold. First, it gives a brief sketch of our proof strategy. Due to space limitations, most of the technical details are omitted and only the backbone ideas are presented. This serves the second purpose, to motivate the nature of our results and also introduce important concepts from convex geometry that underlie our analysis. 
The details of most of the technical discussion in this introductory section are deferred to later sections. To keep the discussion concise, we focus our attention on the $\ell_2$-LASSO. 
Throughout, we use boldface lowercase letters to denote vectors and boldface capital letters to denote matrices. 
%We write $\mathbf{I}_{m\times n}$ and $\mathbf{I}_{m}$ for  the identity matrices of sizes $m\times n$ and $m\times m$, respectively. 
Also, to simplify the notation the $\ell_2$-norm will be denoted as $\| \cdot \|$ from now on. %instead of $\| \cdot \|_2$ for the $\ell_2$ norm of a vector.
 %, in order to assure that concepts are introduced as clear as possible
 
%As already mentioned the results of this work span three different (but closely related) versions of the LASSO problem. In this section, we focus on the $\ell_2$-LASSO

\subsection{First-Order Approximation}
Recall the $\ell_2$-LASSO problem introduced in \eqref{ell2model}:
\beq\label{eq:ell2}
\x^*_{\ell_2} = \text{arg}\min_{\x}\left\{ ~\|\y-\A\x\| + \la f(\x) ~\right\}.\vspace{-8pt}
\eeq
A key idea behind our approach is using the linearization of the convex structure inducing function $f(\cdot)$ around the vector of interest $\x_0$ \cite{Roc70,Lewis}:
\beq\label{eq:fodef}
\fh(\x)=f(\x_0)+\sup_{\s\in\la\paf}\s^T(\x-\x_0).
\eeq
$\partial f(\x_0)$ denotes the subdifferential of  $f(\cdot)$ at $\x_0$ and is always a compact and convex set \cite{Roc70}. Throughout, we assume that $\x_0$ is not a minimizer of $f(\cdot)$, hence, $\paf$ does not contain the origin. From convexity of $f(\cdot)$, $f(\x)\geq \hat f(\x)$, for all $\x$. What is more, when $\|\x-\x_0\|$ is sufficiently small, then $\hat f(\x)\approx f(\x)$.
We substitute $f(\cdot)$ in \eqref{eq:ell2} by its first-order approximation $\fh(\cdot)$, to get a corresponding ``\emph{Approximated LASSO}" problem.  To write the approximated problem in an easy-to-work-with format, recall that $\y=\A\x_0 + \z = \A\x_0 + \sigma\vb$, for $\vb\sim\Nn(0,\mathbf{I}_{m})$ and change the optimization variable from $\x$ to $\w=\x-\x_0$:
\begin{align}\label{eq:ell2Approx}
\hat\w_{\ell_2}(\la,\sigma,\A,\vb) = \text{arg}\min_{\w}\left\{ ~\|\A\w-\sigma\vb\| + \sup_{\s\in\la\paf}\s^T\w ~\right\}.\vspace{-8pt}
\end{align}
We will often drop all or part of the arguments $\la,\sigma,\A,\vb$ above, when it is clear from the context. We denote $\hat\w_{\ell_2}$ for the optimal solution of the approximated problem in \eqref{eq:ell2Approx} and $\w^*_{\ell_2}=\x^*_{\ell_2} - \x_0$ for the optimal solution of the original problem in \eqref{eq:ell2}\footnote{We follow this conventions throughout the paper: use the symbol ``$~\hat{}~$'' over variables that are associated with the approximated problems. To distinguish, use the symbol ``$~^*~$''  for the variables associated with the original problem
.}. Also, denote the optimal cost achieved in \eqref{eq:fodef} by $\hat\w_{\ell_2}$, as $\Fc_{\ell_2}(\A,\vb)$.

Taking advantage of the simple characterization of $\fh(\cdot)$ via the subdifferential $\paf$, we are able to \emph{precisely} analyze the optimal cost and the normalized squared error of the resulting approximated problem. The approximation is tight when $\|\x_{\ell_2}^*-\x_0\|\rightarrow 0$ and we later show that this is the case when the noise level $\sigma\rightarrow 0$. This fact allows us to translate the results obtained for the Approximated LASSO problem to corresponding \emph{precise} results for the Original LASSO problem, in the small noise variance regime.

\subsection{Importance of $\sigma\rightarrow 0$}
In this work, we focus on the precise characterization of the NSE. While we show that the first order characteristics of the function, i.e. $\paf$, suffice to provide sharp and closed-form bounds for small noise level $\sigma$, we believe that higher order terms are required for such precise results when $\sigma$ is arbitrary. On the other hand, we empirically observe that the worst case NSE for the LASSO problem is achieved when $\sigma\rightarrow 0$. While we do not have a proof for the validity of this statement for the $\ell_2$- and $\ell_2^2$-LASSO, we \emph{do prove} that this is indeed the case for the C-LASSO problem. Interestingly, the same phenomena has been observed and proved to be true for related estimation problems, for example for the proximal denoising problem \eqref{DirDen} in \cite{Oymak,DonAcu,Matan} and, closer to the present paper, for the LASSO problem with $\ell_1$ penalization (see Donoho et al. \cite{NoiseSense}).

Summarizing, for the C-LASSO problem, we derive a formula that sharply characterizes its NSE for the small $\sigma$ regime and we show that the same formula upper bounds the NSE when $\sigma$ is arbitrary. Proving the validity of this last statement for the $\ell_2$- and $\ell_2^2$-LASSO would ensure that our corresponding NSE formulae for small $\sigma$ provide upper bounds to the NSE for arbitrary $\sigma$. 

\subsection{Gordon's Lemma}

Perhaps the most important technical ingredient of the analysis presented in this work is a lemma proved by Gordon in \cite{Gor}. Gordon's Lemma establishes a very useful (probabilistic) inequality for Gaussian processes.
\begin{lem}[Gordon \cite{Gor}]\label{lemma:Gordon}
Let $\Gb\in\R^{m\times n}, g\in\R, \g\in\R^m, \h\in\R^n$ be independent of each other and have independent standard normal entries. Also, let $\mathcal{S}\subset\mathbb{R}^n$ be an arbitrary set and $\psi:\mathcal{S}\rightarrow\mathbb{R}$ be an arbitrary function. Then, for any $c\in \mathbb{R}$,
\begin{align}\label{eq:gorap}
&\Pro\left( \min_{\x\in\mathcal{S}}\left\{ \| \Gb\x \| + \|\x\| g - \psi(\x) \right\} \geq c \right) 
\geq
\Pro\left( \min_{\x\in\mathcal{S}}\left\{ \|\x\|\|\g\| - \h^T\x - \psi(\x) \right\} \geq c \right) .
\end{align}
\end{lem}
It is worth mentioning that the ``escape through a mesh" lemma, which has been the backbone of the approach introduced by Stojnic \cite{Sto1} (and subsequently refined in \cite{Cha}) for computing an asymptotic upper bound to the minimum number of measurements required in the Noiseless CS problem, is a corollary of Lemma \ref{lemma:Gordon}
% and further requires only a simple refinement to make it fit our setup

For the purposes of our analysis, we require a slight modification of this lemma. To avoid technicalities at this stage, we defer its precise statement  to Section \ref{sec:gorMain}. Here, it suffices to observe that the original Gordon's Lemma \ref{lemma:Gordon} is (almost) directly applicable to the LASSO problem in \eqref{eq:ell2Approx}. First, write $\|\A\w-\sigma\vb\|=\max_{\|\ab\|=1}{\ab^T[\A, -\vb] \begin{bmatrix} \w \\ \sigma \end{bmatrix}}$ and take function $\psi(\cdot)$ in the lemma to be $\sup_{\s\in\la\paf} \s^T\w$. Then, the optimization problem in the left hand side of \eqref{eq:gorap} takes the format of the LASSO problem in \eqref{eq:ell2Approx}, except for the ``distracting" factor $\|\x\| g$. A simple argument shows that this term can be discarded without affecting the essence of the probabilistic statement of Lemma \ref{lemma:Gordon}. Details being postponed to the later sections (cf. Section \ref{sec:intro2tech}), Corollary \ref{cor:keylow} below summarizes the result of applying Gordon's Lemma to the LASSO problem.
\begin{cor}[Lower Key Optimization]\label{cor:keylow}
Let $\g\sim\Nn(0,\mathbf{I}_{m})$, $\h\sim\Nn(0,\mathbf{I}_{n})$ and $h\sim\Nn(0,1)$ be independent of each other. Define the following optimization problem:
\begin{align}\label{eq:keylow}
\hspace*{-3pt}\Lc(\g,\h) = \min_{\w}\left\{ \sqrt{\|\w\|^2_2 + \sigma^2}\|\g\| - \h^T\w + \max_{\s\in\la\paf} \s^T\w \right\}. 
\end{align}
Then, for any $c\in\mathbb{R}$:
\begin{align}\notag
\Pro\left( ~ \Fc_{\ell_2}(\A,\vb) \geq c~\right) \geq 2\cdot\Pro\left( ~ \Lc(\g,\h) - h\sigma \geq c ~ \right) - 1.
\end{align}
\end{cor}

Corollary \ref{cor:keylow} establishes a probabilistic connection between the LASSO problem and the minimization \eqref{eq:keylow}. In the next section, we argue that the latter is much easier to analyze than the former. Intuitively, the main reason is that instead of an $m\times n$ matrix, \eqref{eq:keylow} only involves two vectors of sizes $m\times 1$ and $n\times 1$. Even more, those vectors have independent standard normal entries and are independent of each other, which greatly facilitates probabilistic statements about the value of $\Lc(\g,\h)$. Due to its central role in our analysis, we often refer to problem \eqref{eq:keylow} as ``\emph{key optimization}" or ``\emph{lower key optimization}". The term ``lower" is attributed to the fact that analysis of  \eqref{eq:keylow} results in a probabilistic lower bound for the optimal cost of the LASSO problem.
%
%Typically, the optimization problem in the left hand side of \eqref{eq:inGor} will be the problem of immediate interest to us, but also the one which is hard to directly analyze. Instead, Gordon's Lemma suggests to carry the analysis over the, typically amenable to direct analysis, problem on the right hand side of \eqref{eq:inGor}. Observe, that the optimization in the left hand side of \eqref{eq:inGor} involves an $m\times n$ gaussian matrix $\Gb$. On the other hand, the optimization in the right hand side is a function of only two independently generated gaussian vectors $\g$ and $\h$ of sizes $m\times 1$ and $n\times 1$, correspondingly. %This observation is meant to help build up some intuition to why the latter problem is easier to analyze and to help appreciate the power of the Gordon's Lemma.

\subsection{Analyzing the Key Optimization}

\subsubsection{Deterministic Analysis}
First, we perform the deterministic analysis of $\Lc(\g,\h)$ for fixed $\g\in \mathbb{R}^m$ and $\h\in\mathbb{R}^n$.
In particular, we reduce the optimization in \eqref{eq:keylow} to a \emph{scalar} optimization. To see this, perform the optimization over a fixed $\ell_2$-norm of $\w$ to equivalently write
\begin{align*}
\hspace{-3pt}\Lc(\g,\h) = \min_{\alpha\geq 0}\left\{
\sqrt{\alpha^2+\sigma^2}\|\textcolor{black}{\g}\| - 
\max_{\|\w\|=\alpha}~\min_{\s\in\la\paf}(\h-\s)^T\w
%\underbrace{\max_{\|\w\|=\alpha}~\min_{\s\in\la\paf}(\h-\s)^T\w}_{ = \alpha\cdot \dt( \h , \la\paf )}
 \right\}.
\end{align*} 
The maximin problem that appears in the objective function of the optimization above has a simple solution. It can be shown that 
\begin{align*}
\hspace*{-5pt}\max_{\|\w\|=\alpha}~\min_{\s\in\la\paf}(\h-\s)^T\w &= 
\min_{\s\in\la\paf}~\max_{\|\w\|=\alpha}(\h-\s)^T\w \\
&= \alpha \min_{\s\in\la\paf} \|\h-\s\|.
\end{align*}
%and it is not hard to see that  
%% the minimax problem above achieves an optimal value of
%  $\min_{\s\in\la\paf}~\max_{\|\w\|=\alpha}(\h-\s)^T\w = \alpha \min_{\s\in\la\paf} \|\h-\s\|.$ 
  % := \alpha \cdot \dt( \h , \la\paf )$.
This reduces \eqref{eq:keylow} to a scalar optimization problem over $\alpha$, for which one can compute the optimal value $\hat{\alpha}$ and the corresponding optimal cost. The result is summarized in Lemma \ref{lemma:lowdet} below. For the statement of the lemma, for any vector $\vb\in\mathbb{R}^n$ define its projection and its distance to a convex and closed set $\Cc\in\R^n$ as
\vspace{-6pt}
\begin{align*}
&\hspace{28pt}\bu(\vb,\Cc) := \operatorname{argmin}_{\s\in\Cc}{\| \vb-\s \| }\quad
%\| \x-\bu(\x,\Cc) \| \leq \| \x-\s \|, ~~\text{for all } \s\in\Cc.
\text{ and }\quad\dt(\vb,\Cc) := \| \vb - \bu(\vb,\Cc) \|.
\end{align*}
  
\begin{lem}[Deterministic Result] \label{lemma:lowdet}
Let $\wh(\g,\h)$ be a minimizer of the problem in \eqref{eq:keylow}.
%Then,
%\begin{align*}
%\mathcal{L}^*({\g},{\h}) = \min_\alpha\left\{
%\sqrt{\alpha^2+\sigma^2}\|{\g}\| - 
%\alpha\cdot \dt(\h,\la\paf)
% \right\}.
%\end{align*} 
%Furthermore, 
If $\|\g\| > \dt(\h,\la\paf)$, then,
\begin{align*}
%\hspace{-20pt}\|\w^*\|^2 =
\text{a)}~~~ &\wh(\g,\h) = \sigma\frac{\h-\bu( \h , \la\paf )}{\sqrt{\|{\g}\|^2-\dt^2( \h , \la\paf )}},\\
\text{b)}~~~ &\|\wh(\g,\h) \|^2 = \sigma^2\frac{\dt^2( \h , \la\paf )}{\|{\g}\|^2-\dt^2( \h , \la\paf )},\\
 \\ \text{c)}~~~ &
\Lc(\textcolor{black}{\g},\textcolor{black}{\h}) = \sigma\sqrt{{\|\g}\|^2-\dt^2( \h , \la\paf )}.
\end{align*}
\end{lem}

\hspace{0pt}
\subsubsection{Probabilistic Analysis}\label{sec:prob}
Of interest is making probabilistic statements about $\Lc(\g,\h)$ and the norm of its minimizer $\|\wh(\g,\h)\|$. Lemma \ref{lemma:lowdet} provided closed form deterministic solutions for both of them, which only involve the quantities $\|\g\|^2$ and $\dt^2(\h,\la\paf)$. For $\g\sim\Nn(0,\mathbf{I}_{m})$ and $\h\sim\Nn(0,\mathbf{I}_{n})$, standard results on Gaussian concentration show that, these quantities concentrate nicely around their means $\E\left[ \|\g\|^2 \right] = m$ and $\E\left[ \dt^2(\h,\la\paf) \right] =: \Dlf$, respectively. Combining these arguments with Lemma \ref{lemma:lowdet}, we conclude with Lemma \ref{lemma:lowprob} below. 
%
%To facilitate the statement of the lemma, define\footnote{Observe that the dependence of $\eta$ and $\gamma$ on $\la$, $m$, $f(\cdot)$ and $\x_0$, is implicit in this definition. }
%\begin{align*}
%\eta = \sqrt{ m - \Dlf } ~~,~~ \gamma =  \frac{\Dlf}{{ m - \Dlf }}.
%\end{align*}
\begin{lem}[Probabilistic Result] \label{lemma:lowprob}
Assume that $(1-\eps_L)m\geq\Dlf  \geq \eps_Lm$ for some constant $\eps_L> 0$.  Define\footnote{Observe that the dependence of $\eta$ and $\gamma$ on $\la$, $m$ and $\paf$, is implicit in this definition. },
\begin{align*}
\eta = \sqrt{ m - \Dlf } \quad \text{ and } \quad \gamma =  \frac{\Dlf}{{ m - \Dlf }}.
\end{align*}
Then, for any $\eps>0$, there exists a constant $c>0$ such that, for sufficiently large $m$, with probability  $1 - \exp(-c m)$,
\begin{align*}
\left| \Lc(\g,\h) - \sigma\eta \right| \leq \eps \sigma\eta,\quad\text{ and }\quad
\left| \frac{\|\wh(\g,\h)\|^2}{\sigma^2} - \gamma \right| \leq \eps\gamma.
\end{align*}
\end{lem}
\noindent{\emph{Remark}:} In Lemma \ref{lemma:lowprob}, the condition ``$(1-\eps_L)m\geq\Dlf$''  ensures that $\|\g\|>\dt(\h,\la\paf)$ (cf. Lemma \ref{lemma:lowdet}) with high probability over the realizations of $\g$ and $\h$.

\subsection{Connecting back to the LASSO: The ``Predictive Power of Gordon's Lemma''}\label{sec:predictivePower}
Let us recap the last few steps of our approach. Application of Gordon's Lemma to the approximated LASSO problem in \eqref{eq:ell2Approx} introduced the simpler lower key optimization \eqref{eq:keylow}. Without much effort, we found in Lemma \ref{lemma:lowprob} that its cost $\Lc(\g,\h)$ and the normalized squared norm of its minimizer $\frac{\|\wh(\g,\h)\|^2}{\sigma^2}$ concentrate around $\sigma\eta$ and $\gamma$, respectively. This brings the following question:
\begin{itemize}
\item[-] \emph{To what extent do such results on $\Lc(\g,\h)$ and $\wh(\g,\h)$ translate to useful conclusions about $\Fc_{\ell_2}(\A,\vb)$ and $\hat\w_{\ell2}(\A,\vb)$? }
\end{itemize}
Application of Gordon's Lemma as performed in Corollary \ref{cor:keylow} when combined with Lemma \ref{lemma:lowprob}, provide a preliminary answer to this question: 
%under mild assumptions, 
%for specific conditions on $\la$ and m, 
$\Fc_{\ell_2}(\A,\vb)$ is lower bounded  by $\sigma\eta$ with overwhelming probability. Formally,
\begin{lem}[Lower Bound] \label{lemma:cost_lowbound}
Assume $(1-\eps_L)m\geq \Dlf \geq \eps_L m$ for some constant $\eps_L>0$ and $m$ is sufficiently large. Then,
for any $\eps>0$, there exists a constant $c>0$ such that, with probability $1-\exp(-cm)$,
\begin{align*}%\label{eq:conobjVal}
\Fc_{\ell_2}(\A , \vb ) \geq (1-\eps)\sigma \eta.
\end{align*}
\end{lem}
But is that all? A major part of our technical analysis in the remainder of this work involves showing that the connection between  the LASSO problem and the simple optimization
\eqref{eq:keylow} is much \emph{deeper} than Lemma \ref{lemma:cost_lowbound} predicts. In short, under certain conditions on $\la$ and $m$ (similar in nature to those involved in the assumption of Lemma \ref{lemma:cost_lowbound}),
%under the assumptions of Lemma \ref{lemma:cost_lowbound}, 
we prove that the followings are true:
\begin{itemize}
\item Similar to $\Lc(\g,\h)$, the optimal cost $\Fc_{\ell_2}(\A,\vb)$ of the approximated $\ell_2$-LASSO concentrates around $\sigma\eta$.
\item Similar to $\frac{\|\wh(\g,\h)\|^2}{\sigma^2}$, the NSE of the approximated $\ell_2$-LASSO $\frac{\|\hat\w_{\ell_2}(\A,\vb)\|^2}{\sigma^2}$ concentrates around $\gamma$.
\end{itemize}
In some sense, $\Lc(\g,\h)$ ``predicts" $\Fc_{\ell_2}(\A,\vb)$ and  $\|\wh(\g,\h)\|$ ``predicts" $\|\hat\w_{\ell_2}(\A,\vb)\|$, which attributes Gordon's Lemma (or more precisely to the lower key optimization) a ``predictive power". This power is not necessarily restricted to the two examples above. In Section \ref{any sigma}, we extend the applicability of this idea to prove that worst case NSE of the C-LASSO is achieved when $\sigma\rightarrow 0$ . Finally, in Section \ref{justifyl22} we rely on this predictive power of Gordon's Lemma to motivate our claims regarding the $\ell_2^2$-LASSO.
%
%This power is not necessarily restricted to the examples discussing above of predicting the optimal cost and the norm of the optimal value. 
%
%\chris{Discuss that this is more generally true. eg. worst case error, ell Squared}

The main idea behind the framework that underlies the proof of the above claims was originally introduced by Stojnic in his recent work \cite{StojLAS} in the context of the analysis of the $\ell_1$-constrained LASSO. While the fundamentals of the approach remain similar, we significantly extend the existing results in multiple directions by analyzing the more involved $\ell_2$-LASSO and $\ell_2^2$-LASSO problems and by generalizing the analysis to arbitrary convex functions. A synopsis of the framework is provided in the next section, while the details are deferred to later sections.
%Due to space limitations, we only provide a synopsis of the framework and its main steps in the next section. The detailed analysis can be found in \cite{OTH}.

\subsection{Synopsis of the Technical Framework}\label{synops}
We highlight the main steps of the technical framework.% that underlies our proofs.
\begin{enumerate}
\item Apply Gordon's Lemma to $\Fc_{\ell_2}(\A,\vb)$ to find a \emph{high-probability} \emph{lower bound} for it. (cf. Lemma \ref{lemma:cost_lowbound})
\item Apply Gordon's Lemma to the \emph{dual} of $\Fc_{\ell_2}(\A,\vb)$ to find a \emph{high-probability} \emph{upper bound} for it.
\item Both lower and upper bounds can be made arbitrarily close to $\sigma\eta$. Hence, $\Fc_{\ell_2}(\A,\vb)$ concentrates with high probability around $\sigma\eta$ as well.
\item Assume $\frac{\|\hat\w_{\ell_2}\|^2}{\sigma^2}$ deviates from $\gamma$. 
%Then, a
A third application of Gordon's Lemma shows that such a deviation would result in a \emph{significant increase} in the optimal cost, namely $\Fc_{\ell_2}(\A,\vb)$ would be significantly larger than $\sigma\eta$.
\item From the previous step, conclude that $\frac{\|\hat\w_{\ell_2}\|^2}{\sigma^2}$ concentrates with high probability around $\gamma$.
\end{enumerate}

\subsection{Gaussian Squared Distance and Related Quantities}
The Gaussian squared distance to the $\la$-scaled set of subdifferential of $f(\cdot)$ at $\x_0$, 
\beq
\Dlf := \E\left[ \dt^2(\h,\la\paf) \right], \label{eq:dlf_def}
\eeq 
has been key to our discussion above. Here, we explore some of its  useful properties and introduce some other relevant quantities that altogether capture the (convex) geometry of the problem. Given a set $\Cc\in\R^n$, denote its conic hull by $\text{cone}(\Cc)$. Also, denote its polar cone by $\Cc^\circ$, which is the closed and convex set $\{\ub\in\R^n\big|\ub^T\vb\leq 0~\text{for all}~\vb\in\Cc\}$.

%For any $\h\in\mathbb{R}^n$ define its 
%projection to a set $\Cc$ as
%$$
%\bu(\h,\Cc) := \operatorname{argmin}_{\s\in\Cc}{\| \h-\s \| }.
%%\| \x-\bu(\x,\Cc) \| \leq \| \x-\s \|, ~~\text{for all } \s\in\Cc.
%$$ 
% Then,
%distance to a set $\Cc$ as
%$\dt(\h,\Cc) = \min_{\s\in\Cc}\| \h - \s \|$. Furthermore, 
Let $\h\sim\mathcal{N}(0,\mathbf{I}_n)$. 
%Besides $\Dlf$, the following quantities are of central interest throughout the paper:
Then, define,
%\begin{subequations}\label{eq:defn01}
\begin{align}
&\Clf := \E\left[\left( \h - \bu(\h,\la\paf) \right)^T \bu(\h,\la\paf)\right] \label{eq:clf_def},
\\%\label{gcorr}\\
&\Delxf := \E\left[~ \dt^2(\h,\text{cone}(\paf))~ \right]. \label{eq:delxf_def}
% \\\PC &:= \E\left[~\|\bi(\x,\Cc) \|^2~\right], 
\end{align} 
%\end{subequations}

From the previous discussion, it has become clear how $\Dlf$ appears in the analysis of the NSE of the $\ell_2$-LASSO. $\Delxf$ replaces $\Dlf$ in the case of C-LASSO. This correspondence is actually not surprising as the approximated C-LASSO problem can be written in the format of the problem in \eqref{eq:ell2Approx} by replacing $\la\paf$ with $\text{cone}(\paf)$. 
While $\Delxf$ is the only quantity that appears in the analysis of the C-LASSO, the analysis of the $\ell_2$-LASSO requires considering not only $\Dlf$ but also $\Clf$. $\Clf$ appears in the analysis during the second step  of the framework described in Section \ref{synops}. In fact, $\Clf$ is closely related to $\Dlf$ as the following lemma shows.

\begin{lem}[\cite{McCoy}]\label{lem:Dlf}
 Suppose $\paf$ is nonempty and does not contain the origin.
% \footnote{Equivalently, $\x_0$ is not a minimizer of $f(\cdot)$. We will use this assumption throughout the paper.}.
 Then, 
\begin{enumerate}
\item $\Dlf$ is a strictly convex function of $\la\geq 0$, and is differentiable for $\la>0$.
\item $\frac{\partial\Dlf}{\partial\la} = -\frac{2}{\la}\Clf$.
\end{enumerate}
\end{lem}
\vspace{2pt}

%We make use of Lemma \ref{lem:Dlf} and various other properties of $\Dlf$, the details on which can be found in \cite{OTH}.

As a last remark, the quantities $\Delxf$ and $\Dlf$ also play a crucial role in the analysis of the Noiseless CS and the Proximal Denoising problems. Without going into details, we mention that it has been recently proved in \cite{McCoy}\footnote{The authors in \cite{McCoy} coined the term ``statistical dimension'' of a cone $\mathcal{K}$ to denote the expected squared distance of a gaussian vector to its polar cone $\mathcal{K}\pol$. In that terminology, $\Delxf$ is the statistical dimension of the $(\text{cone}(\paf))\pol$, or equivalently (see Lemma \ref{tsg}) of the descent cone of $f(\cdot)$ at $\x_0$.
} that the noiseless compressed sensing problem \eqref{LinInv} exhibits a transition from ``failure'' to ``success'' around $m\approx\Delxf$. Also, \cite{Oymak,DonAcu,ChaJor} shows that $\Dlf$ and $\Delxf$ are equal to the worst case normalized mean-squared-error of the proximal denoisers \eqref{DirDen} and \eqref{DirCon} respectively. It is known that under mild assumptions, $\Delxf$ relates to $\Dlf$ as follows \cite{McCoy,Foygel,Oymak},
\begin{align}\label{eq:Delxf2Dlf}
\min_{\la\geq 0}\Dlf \approx \Delxf.
\end{align}
%. Similarly, the  MSE risk of the constrained simple denoiser \eqref{DirCon} is characterized by $\Delxf$. 
%\chris{add comment that you can actually calculate those quantities. Add reference to the section that talks about it}
% is exactly the number of measurements $m$ at which a phase transition in the success of \eqref{LinInv} occurs .

%\input{Contributions.tex}

%\input{Background}

%\section{Overview of Main results}
\section{Main Results}\label{sec:mainResults}

\begin{figure}
  \begin{center}
{\includegraphics[scale=0.25]{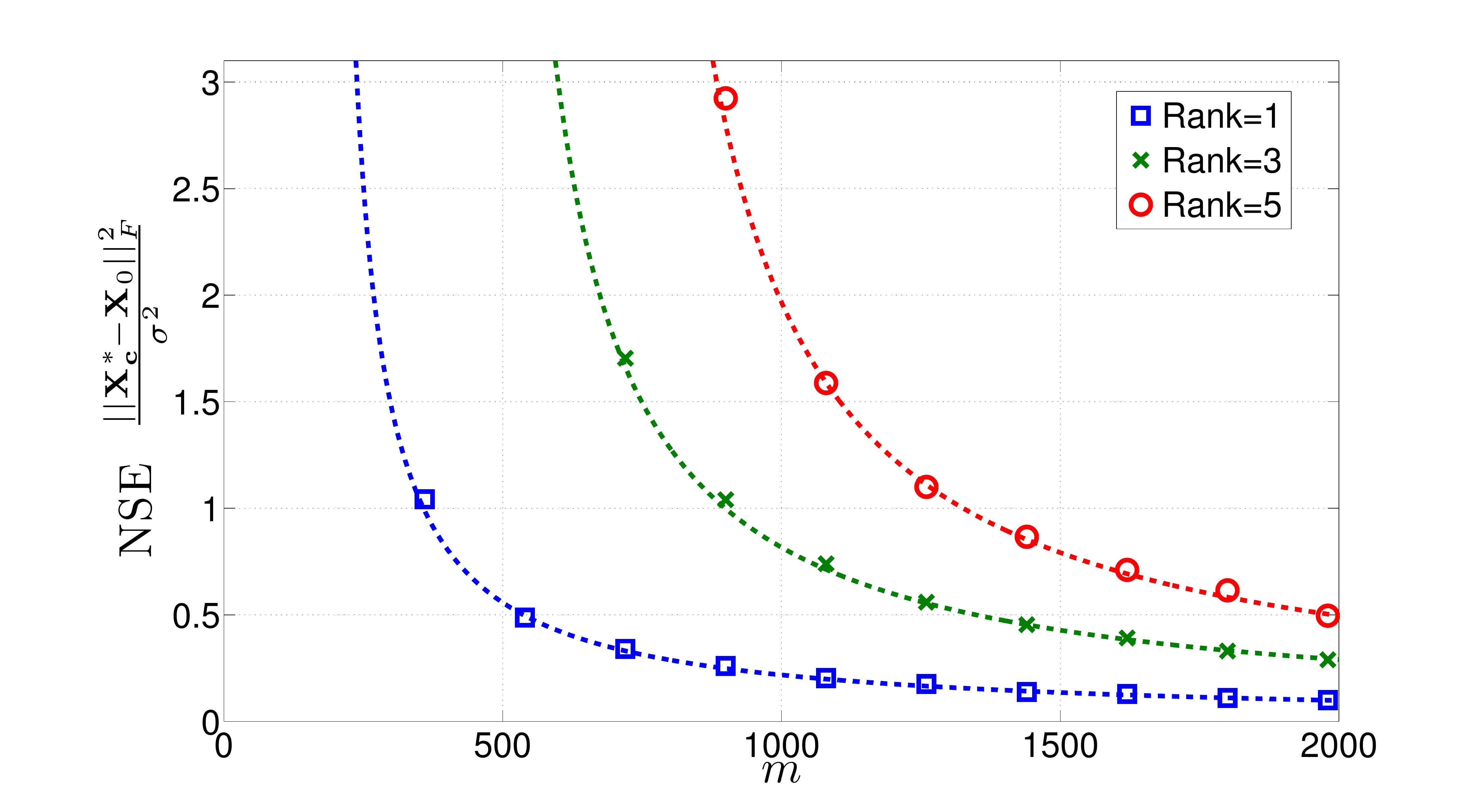}}
  \end{center}
  \caption{\footnotesize{We have considered the Constrained-LASSO with nuclear norm minimization and fixed the signal to noise ratio $\frac{\|\X_0\|_F^2}{\sigma^2}$ to $10^5$. Size of the underlying matrices are $40\times 40$ and their ranks are $1,3$ and $5$. Based on \cite{Oym,Matan}, we estimate ${\bf{D}}_f(\X_0,\R^+)\approx 179,450$ and $663$ respectively. As the rank increases, the corresponding ${\bf{D}}_f(\X_0,\R^+)$ increases and the normalized squared error increases.}}
%  \vspace*{-10pt}
\label{figCont1}
\end{figure}

This section provides the formal statements of our main results. 
A more elaborate discussion follows in Section \ref{sec:discussion}. 
%A detailed discussion follows in Section \ref{sec:discussion}. 

\subsection{Setup}
Before stating our results, we repeat our basic assumptions on the model of the LASSO problem. Recall the definitions of the three versions of the LASSO problem as given in \eqref{cmodel}, \eqref{ell2model} and \eqref{ell22model}. Therein, assume:
\begin{itemize}
\item $\A\in\R^{m\times n}$ has independent standard normal entries,
\item $\z\sim\Nn(0,\sigma^2\mathbf{I}_{m})$,
%$\z:=\sigma\vb\sim\Nn(0,\sigma^2\mathbf{I}_{m})$,
\item $f:\mathbb{R}^n\rightarrow\mathbb{R}$ is convex and continuous,
%\item $\x_0\in\mathbf{R}^n$ is \emph{not} a minimizer of $f(\cdot)$.
\item $\paf$ does \emph{not} contain the origin.%is nonempty, compact and 
\end{itemize}
The results to be presented hold with high probability over the realizations of the measurement matrix $\A$ and the noise vector $\vb$. Finally, recall the definitions of the quantities $\Dlf, \Clf$ and $\Delxf$ in \eqref{eq:dlf_def}, \eqref{eq:clf_def} and \eqref{eq:delxf_def}, respectively.

% It is known that success of the noiseless program exhibits a transition around $m\approx \Delxf$.  and $\Dlf$ can be considered as measures of structure of the signalhis is the regime where number of nonzero Before the formal statement of our results, we define the setting in which they hold. We  assume that $m,n$ and $\Delxf$ are large enough and grow proportionally and that $m>\Delxf$.

%\begin{defn} [Linear regime]\label{def:linear}\eqname{def:linear}
% There exists a constant $\eps_L$ satisfying,
%\beq
%\frac{\Delxf}{\eps_L}>m>(1+\eps_L)\Delxf
%\eeq
%\end{defn}

%
\subsection{C-LASSO}
%
%We will state one result for each of the three LASSO types we have introduced.
%\samet{we should merge the multiple probabilistic statements in these theorems into one}
\begin{restatable}[NSE of C-LASSO]{thm}{CLASSO} \label{thm:CLASSO}Assume there exists a constant $\eps_L>0$ such that, 
$(1-\eps_L)m\geq \Delxf\geq \eps_L m$
%$\frac{\Delxf}{\eps_L}>m>(1+\eps_L)\Delxf$. 
and $m$ is sufficiently large. For any $\eps>0$, there exists a constant $C=C(\eps,\eps_L)>0$ such that, with probability $1-\exp(-Cm)$,
\begin{align}\label{eq:cp2}
\frac{\|\x^*_c - \x_0 \|^2}{\sigma^2} \leq (1+\eps)\frac{\Delxf}{m-\Delxf},
\end{align}
Furthermore, there exists a deterministic number $\sigma_0>0$ (i.e. independent of $\A,\vb$) such that, if $\sigma\leq \sigma_0$, with the same probability,
\begin{align}\label{eq:cp1}
\left| \frac{\|\x^*_c - \x_0 \|^2}{\sigma^2} \times \frac{m-\Delxf}{\Delxf}-1\right|<\eps.
\end{align}
%and for all $\sigma$,
\end{restatable}

\subsection{$\ell_2$-LASSO}
\begin{defn}[$\Rco$]\label{defn:Ron}
Suppose $m> \min_{\la\geq 0}\Dlf$. Define $\Rco$ as follows,
\begin{align}\nn%\label{eq:Rco}
\Rco = \left\{ \la> 0 ~|~ m-\Dlf> \max\{0,  \Clf\}  \right\}.
\end{align}
%We say that,
%% $\la\in\Rco$, iff 
%\begin{align}\label{eq:Rco}
%\la\in\Rco \quad \text{iff} \quad m\geq \max\{ \Dlf , \Dlf + \Clf \}.
%\end{align}
%there exists constant $\eps_L>0$ such that,
%\begin{align}\label{eq:Rco}
%m\geq \max\{ \Dlf , \Dlf + \Clf \} + \eps_L m.
%\end{align}
\emph{Remark:} Section \ref{sec:ell2LASSO} fully characterizes $\Rco$ and shows that it is an open interval.
\end{defn}

\begin{restatable}[NSE of $\ell_2$-LASSO in $\Rco$]{thm}{RONLASSO}\label{thm:ell2LASSO}
%Let $m\geq \min_{\la\geq 0}\Dlf$ and 
% and $\la\in\Rco$. 
%Further assume there exists constant $\eps_L$ such that, $\Dlf>\eps_L m$ 
Assume there exists a constant $\eps_L>0$ such that $(1-\eps_L)m\geq \max\{\Dlf,$ $\Dlf+\Clf\}$ and $\Dlf\geq \eps_Lm$. Further, assume that $m$ is sufficiently large. Then, for 
any $\eps>0$, there exists a constant $C=C(\eps,\eps_L)>0$ and a deterministic number $\sigma_0>0$ (i.e.~independent of $\A,\vb$) such that, whenever $\sigma\leq \sigma_0$, with probability $1-\exp(-C\min\{m,\frac{m^2}{n}\})$,
\beq
\left|\frac{\|\x^*_{\ell_2}-\x_0\|^2}{\sigma^2} \times \frac{m-\Dlf}{\Dlf}-1\right|<\eps.
\eeq
\end{restatable}

%\begin{itemize}
%\item Consider the $\ell_2$--LASSO problem and assume $\la\in\Rco$. For any $\eps>0$, there exists a constant $C=C(\eps,\eps_L)$ and $\sigma_0=\sigma_0(\eps,\eps_L,n)$ such that, whenever $\sigma\leq \sigma_0$, with probability $1-\exp(-C^2\min\{m,\frac{m^2}{n}\})$, we have,
%\beq
%\left|\zeta_{\ell_2}-\frac{\Dlf}{m-\Dlf}\right|<\eps
%\eeq
%\item Consider the approximated $\ell_2$--LASSO problem and assume $\la\in\Rco$. For any $\eps>0$, there exists a constant $C=C(\eps,\eps_L)$ such that, with probability $1-\exp(-C^2\min\{m,\frac{m^2}{n}\})$, we have,
%\beq
%\left|\geta_{\ell_2}-\frac{\Dlf}{m-\Dlf}\right|<\eps
%\eeq
%
%\end{itemize}
%\end{restatable}

% and \ref{thm:ell2LASSO}, observe that, the exponent in the probability of success grows as $\min\{m,\frac{m^2}{n}\}$. This would imply that, for success, we require $m$ to grow at least linearly in $\sqrt{n}$.
\begin{figure}
  \begin{center}
{\includegraphics[scale=0.25]{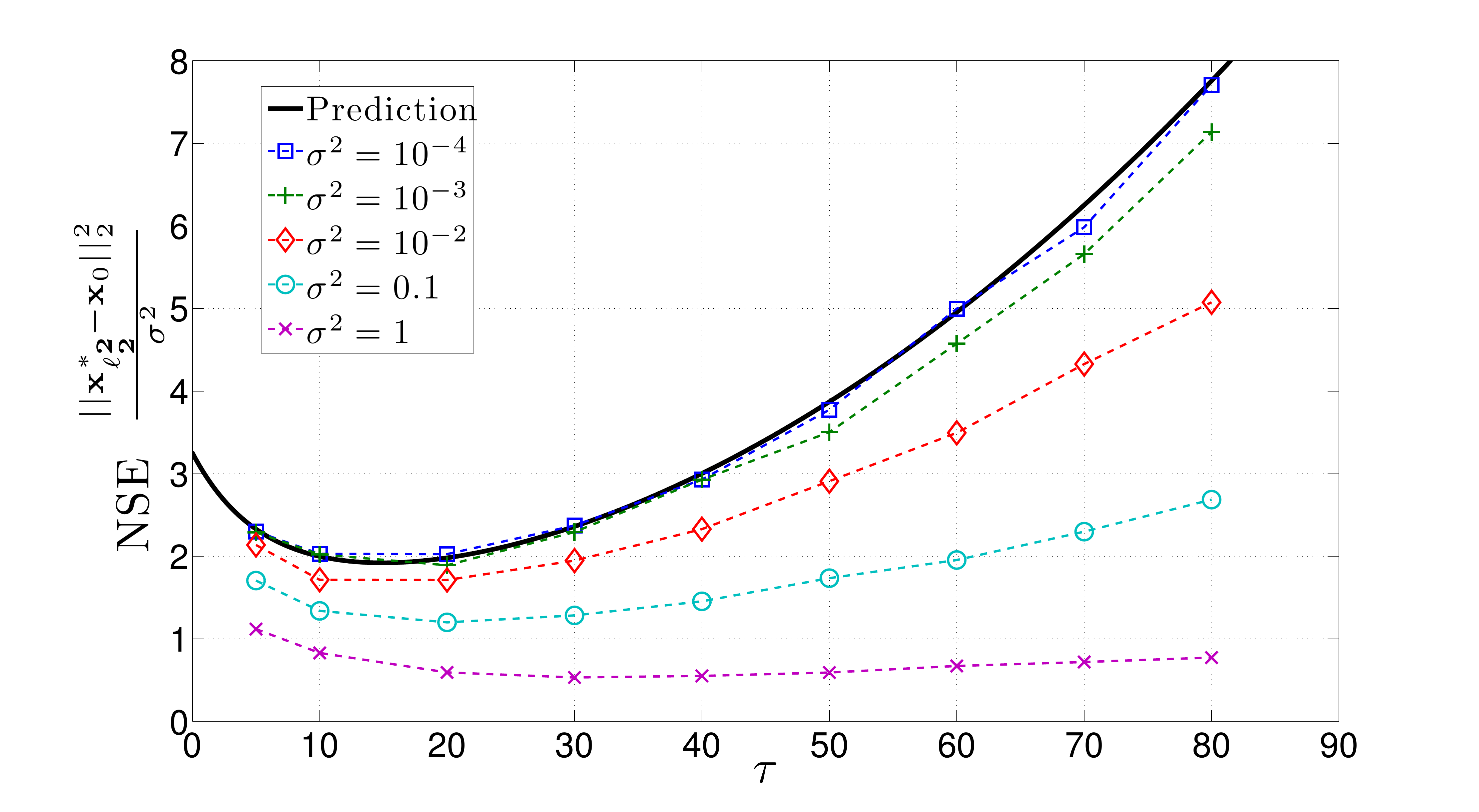}}
  \end{center}
  \caption{\footnotesize{We considered $\ell_2^2$-LASSO problem, for a $k$ sparse signal of size $n=1000$. We let $\frac{k}{n}=0.1$ and $\frac{m}{n}=0.5$ and normalize the signal power by setting $\|\x_0\|=1$. $\tau$ is varied from $0$ to $80$ and the signal-to-noise ratio (SNR) $\frac{\|\x_0\|^2}{\sigma^2}$ is varied from $1$ to $10^4$. We observe that, for high SNR $(\sigma^2\leq 10^{-3})$, the analytical prediction matches with simulation. Furthermore, the lower SNR curves are upper bounded by the high SNR curves. This behavior is fully consistent with what one would expect from Theorem \ref{thm:CLASSO} and Formula \ref{form:ell22LASSO}.}}
%  \vspace*{-10pt}
\label{L22VarVar}
\end{figure}
\subsection{$\ell_2^2$-LASSO}
\begin{defn}[Mapping Function]\label{defn:map}
For any $\la\in\Rco$, define
\beq\label{eq:map}
\map(\la)=\la~ \frac{m-\Df{\la})-\Cf{\la})}{\sqrt{m-\Df{\la})}}.
\eeq
\end{defn}

\begin{thm}[Properties of $\map(\cdot)$] \label{thm:map}
Assume $m>\min_{\la\geq 0}\Dlf$. The function $\map(\cdot):\Rco\rightarrow \R^+$ is strictly increasing, continuous and bijective. Thus, its inverse function $\map^{-1}(\cdot):\R^+\rightarrow\Rco$ is well defined.
\end{thm}
%
%For the $\ell_2^2$-LASSO problem, we propose the following formula to characterize its error performance.

\begin{form} [Conjecture on the NSE of $\ell_2^2$-LASSO]\label{form:ell22LASSO}
%Consider the $\ell_2^2$-LASSO problem. 
Assume $(1-\eps_L)m\geq  \min_{\la\geq 0}\Dlf\geq \eps_Lm$ for a constant $\eps_L>0$ and $m$ is sufficiently large. For any value of the penalty parameter $\tau>0$, we claim that, the expression,
\beq\nn
 \frac{\Df\map^{-1}(\tau))}{m-\Df\map^{-1}(\tau))},
\eeq
provides a good prediction of the NSE $\frac{\|\x_{\ell_2^2}^*-\x_0\|^2}{\sigma^2}$ for sufficiently small $\sigma$. Furthermore, we believe that the same expression upper bounds the NSE for arbitrary values of $\sigma$.
%\begin{itemize}
%\item $\zeta_{\ell_2^2}$ for sufficiently small $\sigma$,
%%\item $\geta_{\ell_2^2}$ for any value of $\sigma$,
%%\end{itemize}
%where $\map^{-1}(\tau)$ is the unique point in $\Rco$ obtained by taking the inverse of the strictly increasing, bijective function $\map(\cdot):\Rco\rightarrow \R^+$,
\end{form}

%\chris{We haven't formally defined $R_{ON}$ yet.}
%\begin{form} [Result on $\ell_2^2$-LASSO]\label{form:ell22LASSO}\eqname{form:ell22LASSO}
%Consider the $\ell_2^2$-LASSO problem. We claim that,
%\beq\nn
% \frac{\Df\map^{-1}(\tau))}{m-\Df\map^{-1}(\tau))},
%\eeq
%provides a good prediction of,
%\begin{itemize}
%\item $\zeta_{\ell_2^2}$ for sufficiently small $\sigma$,
%\item $\geta_{\ell_2^2}$ for any value of $\sigma$,
%\end{itemize}
%where $\map^{-1}(\tau)$ is the unique point in $\Rco$ obtained by taking the inverse of the strictly increasing, bijective function $\map(\cdot):\Rco\rightarrow \R^+$,
%\beq\label{eq:map}
%\map(\la)=\la\cdot\frac{m-\Df{\la})-\Cf{\la})}{\sqrt{m-\Df{\la})}}
%\eeq
%\end{form}

\subsection{Converse Results}
\begin{defn} A function $f(\cdot):\R^n\rightarrow \R$ is called Lipschitz continuous if there exists a constant $L>0$ such that, for all $\x,\y\in\R^n$, we have $|f(\x)-f(\y)|\leq L\|\x-\y\|$.
\end{defn}
\noindent {\emph{Remark:}} 
%$\ell_1$ norm, $\ell_{1,2}$ norm and the nuclear norm are Lipschitz continuous functions. In general,
Any norm in $\R^n$ is Lipschitz continuous \cite{normlip}.
%Until this point, we have considered the scenario where $m$ is sufficiently large. When $m$ is small, one naturally expects $\zeta(\A,\vb)$ to be large. One question is how small should it be so that LASSO will not be able to estimate $\x_0$ robustly. To answer this, consider the noiseless linear inverse problem \eqref{}. Intuitively, if \eqref{} fails to recover $\x_0$ with no noise, it is not reasonable to expect robustness in the noisy setup. From \cite{McCoy}, it is known that, the noiseless algorithm \eqref{} fails when $m<\Delxf$ as $\Delxf$ corresponds to a threshold for noiseless compressed sensing. To complement our results in Theorems \ref{thm:CLASSO}, \ref{thm:ell2LASSO}, we show that, LASSO algorithms are not robust in this regime either. The next theorem formalizes our results.
\begin{thm}[Failure of Robust Recovery] \label{not robust}
Let $f(\cdot)$ be a Lipschitz continuous convex function
% and let $\A\in\R^{m\times n}$ have independent standard normal entries and $\z\sim\Nn(0,\sigma^2\Iden_m)$.
  Assume $m<\Delxf$. Then, for any $C_{max}>0$, there exists a positive number $\sigma_0:=\sigma_0(m,n,f,\x_0,C_{max})$ such that, if $\sigma\leq \sigma_0$, with probability $1-8\exp(-\frac{(\Delxf-m)^2}{4n})$ 
%  over the generation of $\A,\z$
  , we have,
\beq
 \frac{\|\x_{\ell_2}^*(\A,\z)-\x_0\|^2}{\sigma^2}\geq C_{max},
\quad\text{ and }\quad
 \frac{\|\x_{\ell_2^2}^*(\A,\z)-\x_0\|^2}{\sigma^2}\geq C_{max}.
\eeq
\end{thm}

\subsection{Remarks}
A detailed discussion of the results follows in Section \ref{sec:discussion}. Before this, the following remarks are in place.

\noindent{$\bullet~~$} Known results in the noiseless CS problem \eqref{LinInv} quantify the minimum number of measurements required for successful recovery of the signal of interest. Our Theorems \ref{thm:CLASSO} and \ref{thm:ell2LASSO} hold in the regime where this minimum number of measurements required grows proportional to the actual number of measurements $m$. As Theorem \ref{not robust} shows, when $m$ is less than the minimum number of measurements required, then the LASSO programs fails to stably estimate $\x_0$.
%
%Our results hold in the regime where the ``structured sparsity'' of the problem is proportional to the number of measurements. While typically sparsity corresponds to the number of nonzero entries, we need a more general notion to measure sparsity. In particular, this will be $\Delxf$ for C-LASSO and $\Dlf$ for $\ell_2$-LASSO programs. LASSO programs will transition from ``not noise robust'' to ``noise robust'' around $m\approx \Delxf\text{ or }\Dlf$. \samet{structured sparsity or minimum number of measurements?}

%\noindent{\bf{Remark 1:}} 
\noindent{$\bullet~~$} In Theorem \ref{thm:ell2LASSO}, the exponent in the probability expression grows as $\min\{m,\frac{m^2}{n}\}$. This implies that, we require $m$ to grow at least linearly in $\sqrt{n}$.

\noindent{$\bullet~~$} Theorem \ref{thm:CLASSO} suggests that  the NSE of the Constrained-LASSO is maximized as $\sigma\rightarrow 0$. While we believe, the same statement is also valid for the $\ell_2$- and $\ell_2^2$-LASSO, we do not have a proof yet. Thus, Theorem \ref{thm:ell2LASSO} and Formula \ref{form:ell22LASSO} lack this guarantee.

\noindent{$\bullet~~$} As expected the NSE of the $\ell_2$-LASSO depends on the particular choice of the penalty parameter $\la$. Theorem \ref{thm:ell2LASSO} sharply characterizes the NSE (in the small $\sigma$ regime) for all values of the penalty parameter $\la\in\Rco$. In Section \ref{sec:discussion} we elaborate on the behavior of the NSE for other values of the penalty parameter. Yet, the set of values $\Rco$ is the most interesting one for several reasons, including but not limited to the following:
\begin{enumerate} [(a)]
\item The optimal penalty parameter $\labb$ that minimizes the NSE is in $\Rco$.
\item The function $\map(\cdot)$ defined in Definition \ref{defn:map} proposes a bijective mapping from $\Rco$ to $\R^+$. The inverse of this function effectively maps any value of the penalty parameter $\tau$ of the $\ell_2^2$-LASSO to a particular value in $\Rco$. Following this mapping, the exact characterization of the NSE of the $\ell_2$-LASSO for $\la\in\Rco$, translates (see Formula \ref{form:ell22LASSO}) to a prediction of the NSE of the $\ell_2^2$-LASSO for any $\tau\in\R^+$.
\end{enumerate}

\noindent{$\bullet~~$} We don't have a rigorous proof of Formula \ref{form:ell22LASSO}. Yet, we provide partial justification and explain the intuition behind it in Section \ref{justifyl22}. Section \ref{justifyl22} also shows that, when $m>\min_{\la\geq 0}\Dlf$, $\ell_2^2$-LASSO will stably recover $\x_0$ for any value of $\tau>0$, which is consistent with Formula \ref{form:ell22LASSO}. See also the discussion in Section \ref{sec:discussion}. We, also, present numerical simulations that support the validity of the claim.

\noindent{$\bullet~~$}
Theorem \ref{not robust} proves that  both in the $\ell_2$- and $\ell_2^2$-LASSO problems,  the estimation error does not grow proportionally to the noise level $\sigma$, when the number of measurements is not large enough. This result can be seen as a corollary of Theorem 1 of \cite{McCoy}. A result of similar nature holds for the C-LASSO, as well. For the exact statement of this result and the proofs see Section \ref{notenoughm}. 

%  We have a slightly different result for C-LASSO which is given in Section \ref{notenoughm}. We should emphasize that, in order to end up with infinity, one in general has to let $\sigma\rightarrow 0$. For example when $f(\cdot)$ is a norm, it is not difficult to show that, as long as $\sigma$ is finite, $\|\x^*-\x_0\|$ is finite as well. The reason is $f(\x^*)$ will grow as a function of $\|\x^*\|$ and optimality of $\x^*$ ensures that $f(\x^*)$ can't be large.

\subsection{Paper Organization}
Section \ref{sec:discussion} contains a detailed discussion on our results and on their interpretation. %The remaining sections contain the technical details of this work. 
Sections \ref{sec:intro2tech} and \ref{sec:KeyIdeas} contain the technical details of the framework as it was summarized in Section \ref{synops}.
% main ideas of our technical approach. 
 In Sections \ref{sec:ConstrainedLASSO} and \ref{any sigma}, we prove the two parts of Theorem \ref{thm:CLASSO} on the NSE of the C-LASSO. 
 Section \ref{sec:ell2LASSO} analyzes the $\ell_2$-LASSO and Section \ref{sec:ell2LASSO NSE} proves Theorem \ref{thm:ell2LASSO} regarding the NSE over $\Rco$. Section \ref{justifyl22} discusses the mapping between $\ell_2$ and $\ell_2^2$-LASSO, proves Theorem \ref{thm:map} and motivates Formula \ref{form:ell22LASSO}. In Section \ref{notenoughm} we focus on the regime where robust estimation fails and prove Theorem \ref{not robust}. Simulation results presented in
Section \ref{sec:num} support our analytical predictions. Finally, directions for
future work are discussed in Section \ref{sec:conc}. Some of the technical details are deferred to the Appendix.

\section{Discussion of the Results}\label{sec:discussion}

This section contains an extended discussion on the results of this work. We elaborate on their interpretation and implications. 
%All proofs of the statements claimed in this section are given in later sections. 
% Focus is mainly on the $\ell_2$-LASSO. To supplement Definition \ref{} and Theorem \ref{} we discuss

\subsection{C-LASSO}
%The first convex programming approach to estimating the signal of interest $\x_0$ from noisy compressed measurements  makes use of available a-priori knowledge of the value $f(\x_0)$. It explicitly constraints the resulting estimate $\x^*$ to satisfy $f(\x^*)\leq f(\x_0)$. This is the rationale behind the term ``Constrained-LASSO" or ``C-LASSO" assigned to this first algorithm.
%\begin{defn} [C-LASSO] \label{def:c-lasso}The original signal $\x_0$ is estimated by 
%\begin{align}\label{eq:con_defn}\eqname{eq:con_defn}
%\x_c^*(\A , \vb ) &= \operatorname*{argmin}_{\x}\|\y-\A\x\|_2,\notag \\
%&\ \ \ \ \ \text{s.t.}~~~~ f(\x) \leq f(\x_0).
%\end{align}
%%The optimal cost of the above optimization is denoted by $G^*_c(\A,\vb)$. 
%The normalized squared error is defined to be
%$$
%\zeta_c(\A, \vb) = \frac{\|x_c^*(\A , \vb) - \x_0\|_2^2}{\sigma^2}.
%$$
%%The risk of the Constrained LASSO problem is defined  to be:
%%\beq\label{eq:con_risk}
%%\eta_{C}( \sigma )=\E\left[  \zeta_{C}( \A, \vb, \sigma ) \right].
%%\eeq
%\end{defn}
%%Notice that solving the C-LASSO problem requires extra a-priori knowledge of the value $f(\x_0)$. 
%In the rest of the paper we shall often omit the arguments $\A,\vb$ when clear from context.
%Notice that C-LASSO is in similar nature to the constrained denoising problem \eqref{DirCon}. It will arguably be the simplest problem of our analysis. For it, 

We are able to characterize the estimation performance of the Constrained-LASSO in \eqref{cmodel} solely based on $\Delxf$. Whenever $m>\Delxf$, for sufficiently small $\sigma$, %in high dimensions where $\Delxf$ and $m$ are sufficiently large and proportional, 
we prove that,
\beq
\frac{\| \x_c^* - \x_0 \|^2}{\sigma^2} \approx\frac{\Delxf}{m-\Delxf}\label{maineq1}.
\eeq
Furthermore, \eqref{maineq1} holds for arbitrary values of $\sigma$ when $\approx$ is replaced with $\lesssim$. 
Observe in \eqref{maineq1} that as $m$ approaches $\Delxf$, the NSE increases and when $m=\Delxf$, $\text{NSE}=\infty$. This behavior is not surprising as when $m<\Delxf$, one cannot even recover $\x_0$ from noiseless observations via \eqref{LinInv} hence it is futile to expect noise robustness. 
%\chris{@Samet:do we keep this last part?} 
For purposes of illustration, notice that \eqref{maineq1} can be further simplified for certain regimes as follows:
\begin{align*}
\frac{\| \x_c^* - \x_0 \|^2}{\sigma^2}\approx \begin{cases}~~~1&\text{when}~~~m=2\Delxf,\\
\frac{\Delxf}{m} &\text{when}~~~m\gg\Delxf.\end{cases}
\end{align*}
%\samet{The asymptotic NSE is around $\frac{\Delxf}{m-\Delxf}$ which is only $1$ when we set $m=2\Delxf$. talk about literature here}

\subsubsection{Relation to Proximal Denoising}\label{sec:relationc}
We want to compare the NSE of the C-LASSO in \eqref{cmodel} to the MSE risk of the constrained proximal denoiser in \eqref{DirCon}. For a fair comparison, the average signal power $\E[\|\A\x_0\|^2]$ in \eqref{cmodel} should be equal to $\|\x_0\|^2$. This is the case for example when $\A$ has independent $\Nn(0,\frac{1}{m})$ entries. This is equivalent to amplifying the noise variance to $m\sigma^2$ while still normalizing the error term $\|\x_c^*-\x_0\|^2$ by $\sigma^2$. Thus, in this case, the formula \eqref{maineq1} for the NSE is multiplied by $m$ to result in $\Delxf\cdot \frac{m}{m-\Delxf}$ (see Section \ref{sec:trans} for further explanation).
%
%The reader will observe that, since we are using an i.i.d. $\Nn(0,1)$ measurement matrix, we are actually amplifying the signal power. This naturally results in a small amount of NSE. For example, consider the Constrained-LASSO. The asymptotic NSE is around $\frac{\Delxf}{m-\Delxf}$ which is only $1$ when we set $m=2\Delxf$. \samet{talk about literature here}. To keep the signal power fixed, we can use a matrix $\A$ with i.i.d. $\Nn(0,\frac{1}{m})$ entries. This is equivalent to amplifying noise variance to $m\sigma^2$ but still normalizing the error term $\|\x^*-\x_0\|^2$ by $\sigma^2$. Hence, it would amplify the NSE by $m$ and we would end up with $\frac{m\Delxf}{m-\Delxf}$. Equivalently let us write this as $\Delxf\cdot \frac{m}{m-\Delxf}$.
Now, let us compare this with the results known for proximal denoising. There \cite{Oymak,ChaJor}, it is known that the normalized MSE is maximized when $\sigma\rightarrow0$ and is equal to $\Delxf$. Hence, we can conclude that the NSE of the LASSO problem is amplified compared to the corresponding quantity of proximal denoising by a factor of $\frac{m}{m-\Delxf}>1$. This factor can be interpreted as the penalty paid in the estimation error for using linear measurements.
%
%Hence, we have two problems with the same signal power one of which has linear measurements. While both normalized error terms have $\Delxf$, LASSO error has an additional factor of $\frac{m}{m-\Delxf}$ which is strictly greater than $1$. This intuitively corresponds to a penalty for using linear measurements.

%Gaussian squared distances to be 

\subsection{$\ell_2$-LASSO}
%We use \eqref{ell2model} for estimation. Denote the estimate by $\x_{\ell_2}$. 
%The performance is characterized in terms of $\la$ and the region $\la$ lies.
 Characterization of the NSE of the $\ell_2$-LASSO is more involved than that of the NSE of the C-LASSO. For this problem, choice of $\la$ naturally plays a critical role. We characterize three distinct ``regions of operation'' of the $\ell_2$-LASSO, depending on the particular value of $\la$.
% , in the sense that the NSE behaves differently in each of the three regions. 
\subsubsection{Regions Of Operation}
 First, we identify the regime in which the $\ell_2$-LASSO can robustly recover $\x_0$. In this direction, the number of measurements should be large enough to guarantee at least noiseless recovery in \eqref{LinInv}, which is the case when $m>\Delxf$ \cite{Cha,McCoy}. To translate this requirement in terms of $\Dlf$, recall \eqref{eq:Delxf2Dlf} and Lemma \ref{lem:Dlf}, and define
$\labb$ to be the \emph{unique} minimizer of $\Dlf$ over $\la\in\R^+$. We, then, write the regime of interest as $m>\Df\labb)\approx\Delxf.$
%We will consider the scenario in which $\Df\labb)<\min\{m,n\}$. This assumption is quite reasonable when we consider the fact that $\Delxf\approx\Df\labb)$. $\Delxf<m$ ensures that, $m$ is sufficiently large to guarantee noiseless recovery, and $\Delxf<n$ ensures that, the underlying signal has some sort of structure that we can exploit. 

Next, we identify three important values of the penalty parameter $\la$, needed to describe the distinct regions of operation of the estimator. 
%There will three important regularizers we will use to explain our results.
\begin{enumerate}[a)]
\item \emph{$\labb$~}: We show that $\labb$ is optimal in the sense that the NSE is minimized for this particular choice of the penalty parameter. This also explains the term ``best" we associate with it.
% basically be the best regularizer that minimizes the error.% We assume $\labb>0$ since $\labb=0$ is not an interesting scenario as it will be clear soon.
\item \emph{$\lam$~}:  Over $\la\geq \labb$, the equation $m=\Dlf$ has a unique solution. We denote this solution by $\lam$. For values of $\la$ larger than $\lam$, we have $m\leq \Dlf$.% and we have no robust recovery.
\item \emph{$\lac$~}: Over $0\leq\la\leq\labb$, if $m\leq n$, the equation $m-\Dlf=\Clf$ has a unique solution which we denote $\lac$. Otherwise, it has no solution and $\lac:=0$.
% \cite{OTH}
% to be the smallest nonnegative value of $\la$ for which $m-\Dlf\geq\Clf$. 
% We show in \cite{OTH} that, when $m\leq n$ this value is strictly positive and $m-\Df\lac)=\Cf\lac)$. Otherwise, $\lac=0$.
%
%Over $0\leq \la\leq \labb$, the equation $m-\Dlf= \Clf$ has a unique solution when $m\leq n$ and has no solution when $m>n$. If the solution exists, we denote it by $\lac$. Otherwise, we set $\lac=0$. 
%
%The value is critical in that the NSE behaves differently for  $\la\leq\lac$ compared to its behavior for $\la\geq\lac$.% may or may not exist. While we define $\lac$ in this way, we observe that there is at most one $\la\leq \labb$ satisfying $m-\Dlf=\Clf$. 
\end{enumerate}
Based on the above definitions, we recognize the three distinct regions of operation of the $\ell_2$-LASSO, as follows,
\begin{enumerate}[a)]
\item $\Rco=\{\la\in\R^+\big|\lac<\la<\lam\}$.%=\{\la\geq 0\big|m-\Dlf> \max\{\Clf,0\}\}$.
\item $\Rcf=\{\la\in\R^+\big|\la\leq\lac\}$.%=\{\la\leq \labb\big|m-\Dlf\leq \Clf\}$.
\item $\Rci=\{\la\in\R^+\big|\la\geq \lam\}$.%=\{\la\geq \labb\big| m\leq \Dlf\}$.
\end{enumerate}
See Figure \ref{LASSOintro2} for an illustration of the definitions above and Section \ref{sec:ell2LASSO} for the detailed proofs of the statements.

\begin{figure}
  \begin{center}
{\includegraphics[scale=0.45]{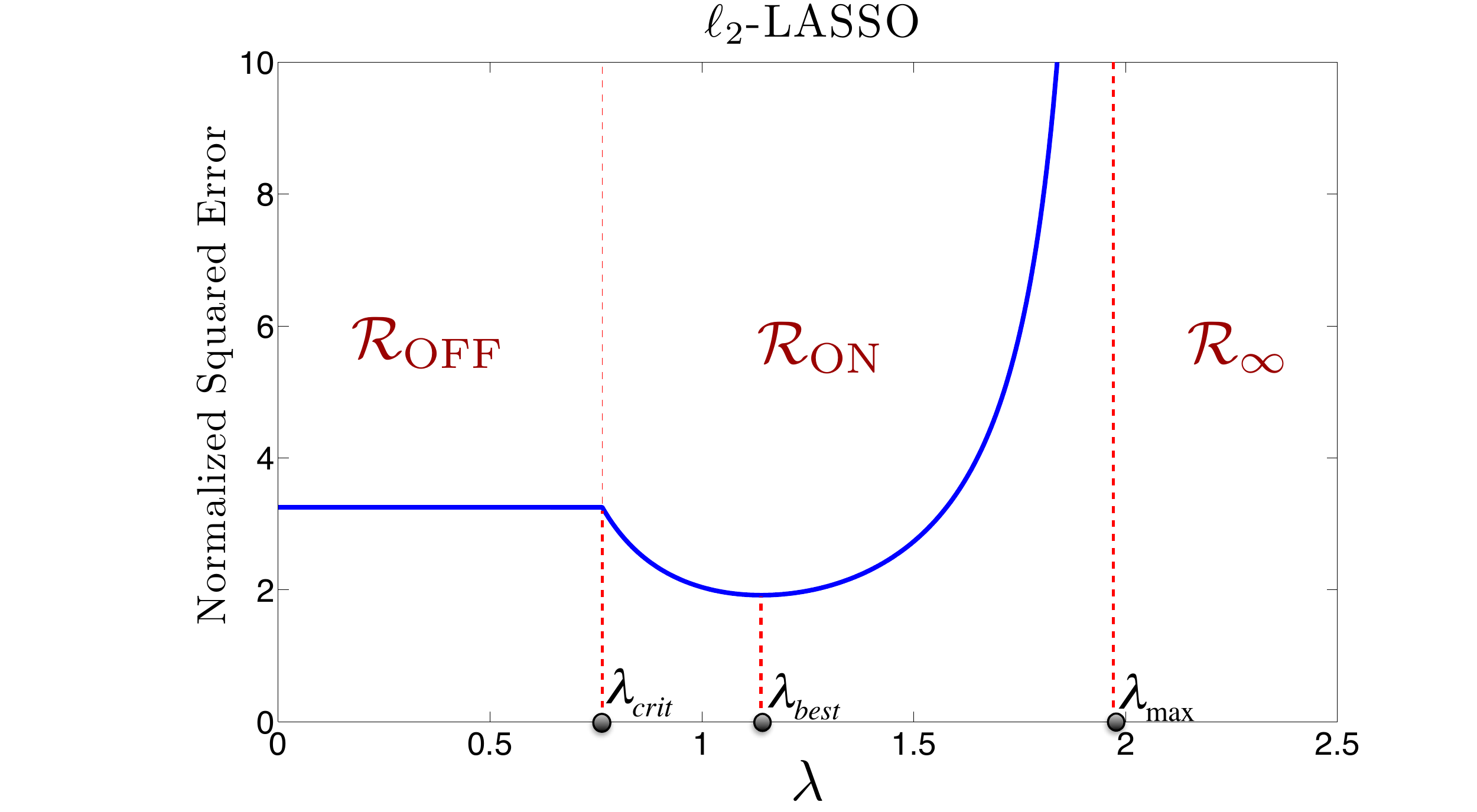}}  
  \end{center}
  \caption{\footnotesize{We consider the $\ell_1$-penalized $\ell_2$-LASSO problem for a $k$ sparse signal in $\R^n$. $x$-axis is the penalty parameter $\la$. For $\frac{k}{n}=0.1$ and $\frac{m}{n}=0.5$, we have $\lac\approx0.76$, $\labb\approx1.14$, $\lam\approx1.97$.}}
%  \vspace*{-10pt}
\label{figCont2}
\end{figure}

\subsubsection{Characterizing the NSE in each Region}

Our main result on the $\ell_2$-LASSO is for the region $\mathcal{R}_{ON}$ as stated in Theorem \ref{thm:ell2LASSO}. We also briefly discuss on our observations regarding $\Rcf$ and $\Rci$:
\begin{itemize}
\item $\Rcf$: 
%The region $0<\la<\lac$. 
For $\la\in\Rcf$, we empirically observe that the LASSO estimate $\x^*_{\ell_2}$ satisfies $\y=\A\x^*_{\ell_2}$ and the optimization \eqref{ell2model} reduces to:
\beq\nn
\min_\x f(\x)~~~\text{subject to}~~~\y=\A\x\label{unreg},
\eeq
which is the standard approach to solving the noiseless linear inverse problems (recall \eqref{LinInv}).
We prove that this reduction is indeed true for values of $\la$ sufficiently small (see Lemma \ref{lemma:Roff}), while our empirical observations suggest that the claim is valid for all $\la \in \Rcf$.
Proving the validity of the claim would show that when $\sigma\rightarrow 0$, the NSE is $\frac{\Df\lac)}{m-\Df\lac)}$, for all $\la\in\Rcf$. Interestingly, this would also give the NSE formula for the particularly interesting problem \eqref{unreg}. Simulation results in Section \ref{sec:num} validate the claim.
%%Observe that \eqref{unreg} is the standard approach to solving the noiseless linear inverse problems.
%
%
%the solution of the $\ell_2$-LASSO is same to the solution of the linear inverse problem \eqref{LinInv} which forces $\y=\A\x$. In particular, the NSE is $\frac{\Df\lac)}{m-\Df\lac)}$, for all $\la\in\Rcf$.
%For this regime, we have the following observations for which we provide partial theoretical justification. When $\la\in\Rcf$, the estimate $\x^*_{\ell_2}(\la)$ is such that $\y=\A\x^*_{\ell_2}(\la)$, thus the optimization \eqref{eq:ell_defn} reduces to:
%\beq
%\min_\x f(\x)~~~\text{subject to}~~~\y=\A\x\label{unreg}.
%\eeq
%The NSE is then calculated for all $\la\leq\lac$ as,
%\beq\nn
%\zeta_{\ell_2}(\la)\approx\frac{\Df\lac)}{m-\Df\lac)}.
%\eeq
%Observe that \eqref{unreg} is the standard approach to solving the noiseless linear inverse problems. Interestingly, we are able to characterize the estimation performance of this program via an apparently unrelated quantity $\lac$.

%\item $\Rco$: $\lac<\la<\lam$. Contains $\labb$. We show, $\lim_{\sigma\rightarrow 0} \frac{\|\x_{\ell_2}-\x_0\|^2}{\sigma^2}=\frac{\Dlf}{m-\Dlf}$.

\item $\Rco$: Begin with observing that $\Rco$ is a nonempty and open interval. In particular, $\labb\in\Rco$ since $m>\Df\labb)$. 
% This paper will focus on the rigorous characterization of the behavior of $\ell_2$-LASSO in $\Rco$. 
% 
 We prove that for all $\la\in \Rco$ and $\sigma$ is sufficiently small, 
 \beq\label{eq:rote}
\frac{\|\x_{\ell_2}^*-\x_0\|}{\sigma^2}\approx\frac{\Dlf}{m-\Dlf}.
\eeq
%We observe that this regime corresponds to the region $\lac\leq \la\leq \lam$.
%Let $\lam=\$ b
Also, empirical observations suggest that \ref{eq:rote} holds for arbitrary $\sigma$ when $\approx$  replaced with $\lesssim$. Finally, we should note that the NSE formula $\frac{\Dlf}{m-\Dlf}$ is a convex function of $\la$ over $\Rco$. 

\item $\Rci$: Empirically, we observe that the stable recovery of $\x_0$ is not possible for $\la\in\Rci$.
\end{itemize}

\subsubsection{Optimal Tuning of the Penalty Parameter}
It is not hard to see that the formula in \eqref{eq:rote} is strictly increasing in $\Dlf$. Thus, when $\sigma\rightarrow0$,
 the NSE achieves its minimum value when the penalty parameter is set to
%The choice of $\la$ minimizing the NSE as $\sigma\rightarrow0$ is,
$
\labb.
$ Now, recall that $\Df\labb)\approx \Delxf$ and compare the formulae in \eqref{maineq1} and \eqref{eq:rote}, to conclude that the C-LASSO and $\ell_2$-LASSO can be related by choosing $\la=\labb$. In particular, we have,
\beq\label{eq:above}
\frac{\|\x_{\ell_2}^*(\labb)-\x_0\|^2}{\sigma^2}\approx\frac{\Df\labb)}{m-\Df\labb)}\approx \frac{\Delxf}{m-\Delxf}\approx \frac{\|\x_{c}^*-\x_0\|^2}{\sigma^2}.
\eeq
%The relation is similar in nature to the one between the denoisers \eqref{DirDen} and \eqref{DirCon}; when optimally tuned, the risk of \eqref{DirDen} is $\Df\labb)$ which is close to the risk $\Delxf$ of \eqref{DirCon}.

%\samet{For fixed regularizer the risk is independent of the noise variance. Or rather, the regularizer choice doesn't depend on $\sigma$.}

%
%\subsubsection{} %   Recall from \eqref{eq:Delxf2Dlf} that $\Df\labb)\approx\Delxf$.
%Following \eqref{eq:above}, this minimum value of the NSE of the $\ell_2$-LASSO is asymptotically equal to the NSE of the Constrained LASSO.

%Since $\labb\in\Rco$ and  $\Df\labb)\approx\Delxf$, it follows from Theorems \ref{thm:C} and \ref{thm:ell2} that this minimum value of the NSE of the $\ell_2$-LASSO is asymptotically equal to the NSE of the Constrained LASSO.
%

%\subsubsection{Relation to the unpenalized problem}
%We observe that, for the regime $\Rcf$, \eqref{eq:ell_defn} behaves same as the unpenalized problem \eqref{unreg}. : The noise performance at $\la=\lac$, namely, $\frac{\Df\lac)}{m-\Df\lac)}$.

\subsection{$\ell_2^2$-LASSO}

\subsubsection{Connection to $\ell_2$-LASSO}
We propose a mapping between the penalty parameters $\la$ of the $\ell_2$-LASSO program \eqref{ell2model} and $\tau$ of the $\ell_2^2$-LASSO program \eqref{ell22model}, for which the NSE of the two problems behaves the same. The mapping function was defined in Definition \ref{defn:map}. Observe that $\map(\la)$ is well-defined over the region $\Rco$, since $m>\Dlf$ and $m-\Dlf>\Clf$ for all $\la\in\Rco$. 
Theorem \ref{thm:map} proves that $\map(\cdot)$ defines a bijective mapping from $\Rco$ to $\R^+$. Other useful properties of the mapping function include the following:
\begin{itemize}
\item $\map(\lac)=0$,%On the left border of $\Rco$, we have $\map(\lac)=0$.
\item $\lim_{\la\rightarrow\lam}\map(\la)=\infty$,%On the right border, $\lim_{\la\rightarrow\lam}\call(\la)=\infty$.
%\item $\map(\la)$ is an increasing function of $\la$ in $\Rco$.
\end{itemize}
Section \ref{justifyl22} proves these properties and more, and contains a short technical discussion that motivates the proposed mapping function.

\subsubsection{Proposed Formula}
%\begin{form} [Result on $\ell_2^2$-LASSO]\label{form:ell22LASSO}\eqname{form:ell22LASSO}
%Consider the $\ell_2^2$-LASSO problem and denote the estimate returned by \eqref{general model} by $\x^*_{\ell_2^2}$. 
We use the mapping function in \eqref{eq:map} to translate our results on the NSE of the $\ell_2$-LASSO over $\Rco$ (see formula \eqref{eq:rote}) to corresponding results on the $\ell_2^2$-LASSO for $\tau\in\R^+$.
Assume $m>\Df\labb)$.  We suspect that for any $\tau>0$,
\beq\nn
 \frac{\Df\map^{-1}(\tau))}{m-\Df\map^{-1}(\tau))}\label{ell22min},
\eeq
%provides a good prediction for,
%\begin{itemize}
%\item
accurately characterizes $\frac{\|\x^*_{\ell_2^2}-\x_0\|^2}{\sigma^2}$ for sufficiently small $\sigma$,
%\item 
and upper bounds $\frac{\|\x^*_{\ell_2^2}-\x_0\|^2}{\sigma^2}$ for arbitrary $\sigma$.

\begin{figure}
\centering
\mbox{
\subfigure[]{
\includegraphics[width=3.1in]{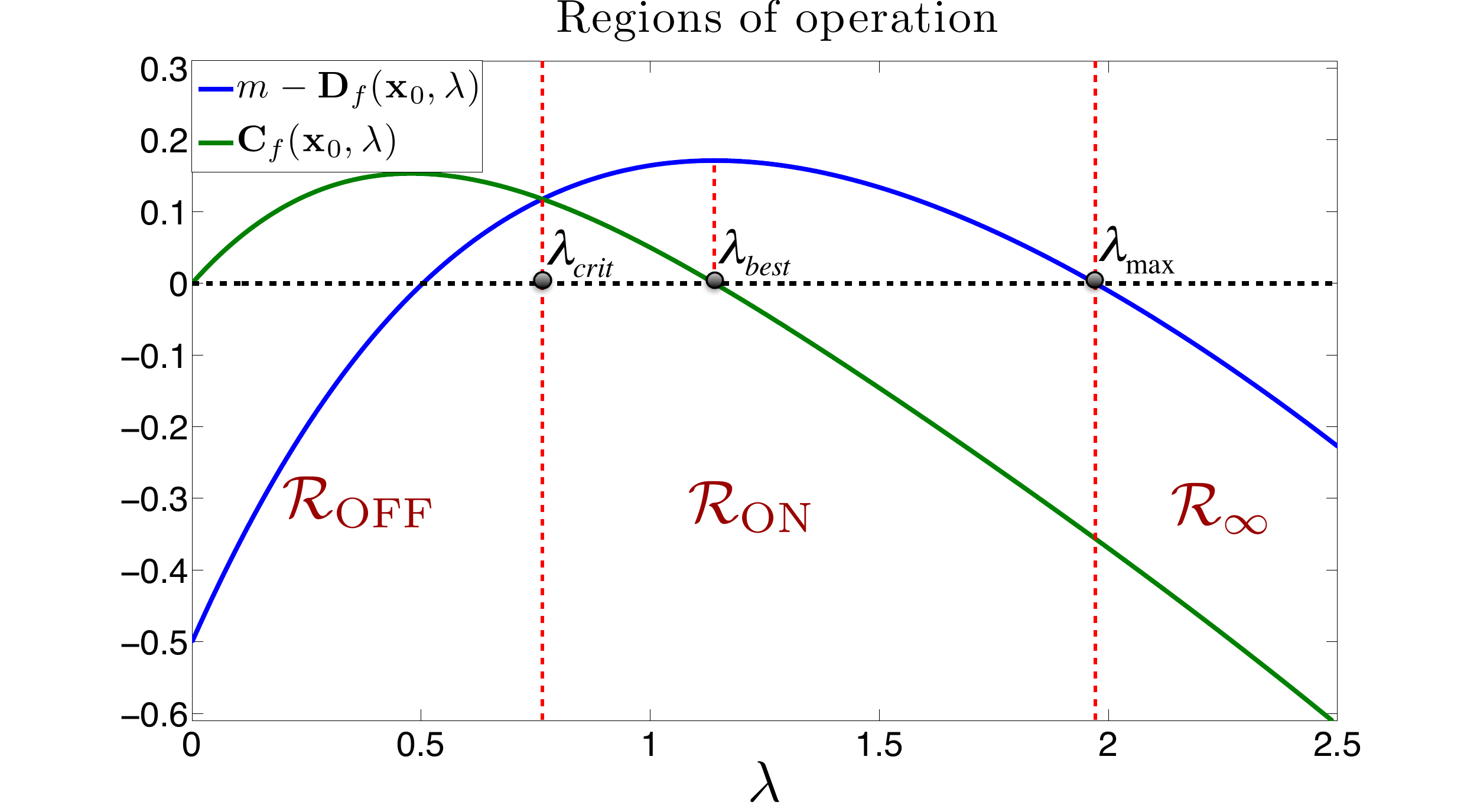}
\label{figCont31}
}\quad
\subfigure[]{
\includegraphics[width=3.1in]{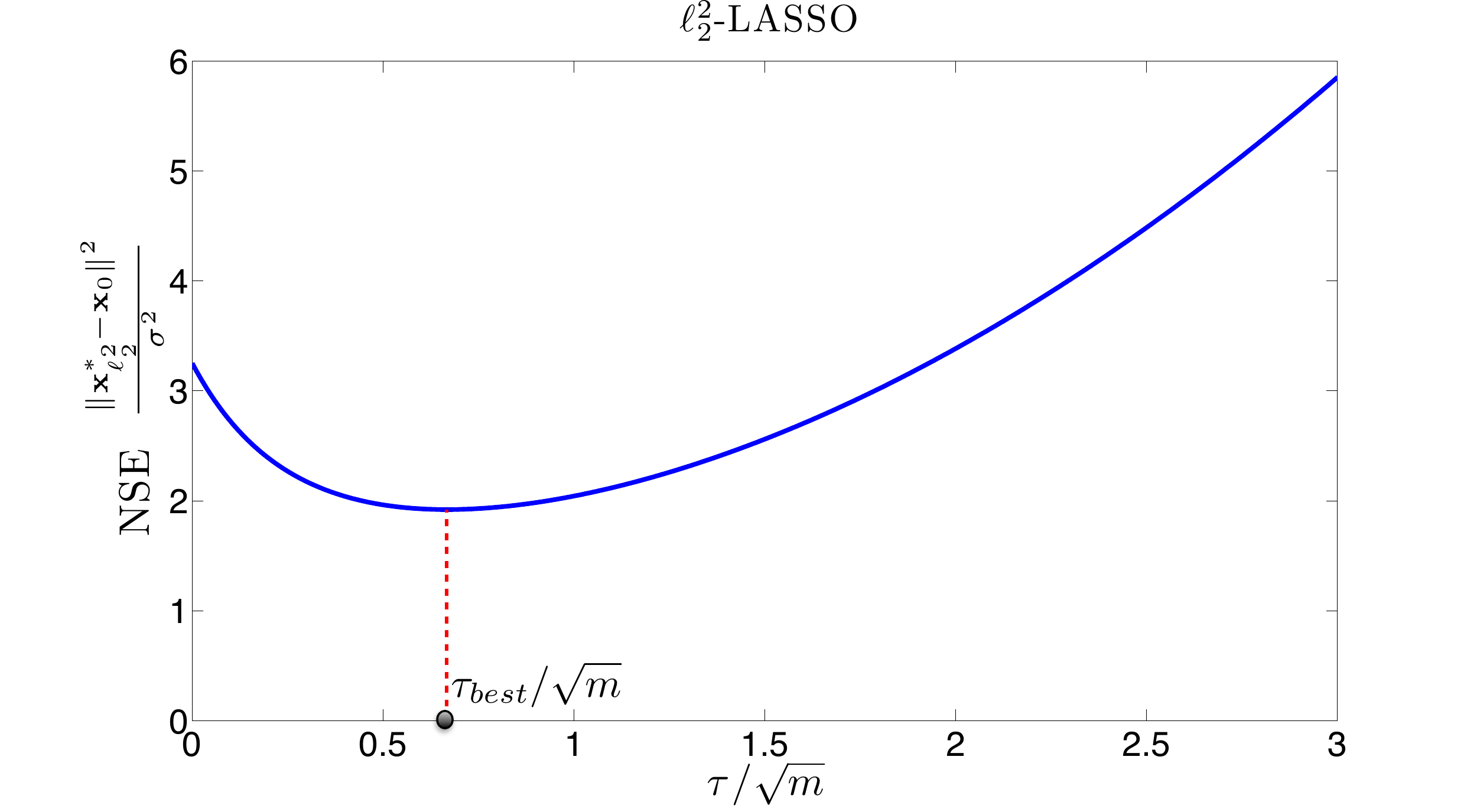}
\label{figCont32}
}
}
\centering
\caption{\small{We consider the exact same setup of Figure \ref{figCont2}. a) We plot $m-\Dlf$ and $\Clf$ as a function of $\la$ to illustrate the important penalty parameters $\lac,\labb,\lam$ and the regions of operation $\Rcf,\Rco,\Rci$. b) We plot the $\ell_2^2$-LASSO error as a function of $\frac{\tau}{\sqrt{m}}$ by using the $\map(\cdot)$ function. The normalization is due to the fact that $\tau$ grows linearly in $\sqrt{m}$.}}
\label{LASSOintro2}
\end{figure}

%
%\subsubsection{Properties of the mapping}
%
%%Hence, the calibration is active over $\la\in\Rco$. 
%Moreover, we show that $\map(\la):\Rco\rightarrow\R^+$ is bijective and strictly increasing. Some other properties of the mapping function include the following:
%\begin{itemize}
%\item $\map(\lac)=0$.%On the left border of $\Rco$, we have $\map(\lac)=0$.
%\item $\lim_{\la\rightarrow\lam}\call(\la)=\infty$.%On the right border, $\lim_{\la\rightarrow\lam}\call(\la)=\infty$.
%\item $\call(\la)$ is an increasing function of $\la$ in $\Rco$.
%\end{itemize}
%For the purpose of illustration, we discuss a quick proof of the fact that $\map(\lac)=0$. Recall that $\lac=0$ or it is the solution of $m=\Dlf+\Clf$.
%%As an example, to see $\map(\lac)=0$, recall that, $\lac=0$ or it is the solution of $m=\Dlf+\Clf$. 
%Hence, either $\lac$ or $\call(\lac)$ is equal to $0$, completing the proof.

\subsubsection{A rule of thumb for the optimal penalty parameter}
%The optimal value of the penalty parameter $\tau$ is the one for which the NSE is minimum.

% Apart from the theoretical interest, being able to optimally tune $\tau$ is of key importance in practical applications of the $\ell_2^2$-LASSO for obvious reasons. 
Formula \ref{form:ell22LASSO} provides a simple recipe for computing the optimal value of the penalty parameter, which we call $\tau_{{best}}$.
Recall that $\labb$ minimizes the error in the $\ell_2$-LASSO. Then, the proposed mapping between the two problems, suggests that $\tau_{best}=\map(\labb)$. To evaluate $\map(\labb)$ we make use of Lemma \ref{lem:Dlf} and the fact that $\frac{d\Dlf}{d\la} = -\frac{2}{\la}\Clf$ for all $\la\geq 0$. Combine this with the fact that $\labb$ is the unique minimizer of $\Dlf$, to show that $\Cf\labb)=0$, and to conclude with,
%
%Lemma \ref{lemma:keyProp} in Section \ref{sec:ell2LASSO} lists a collection of useful properties and relations between $\Dlf, \Plf$ and $\Clf$. To our interest here is its 7th statement according to which $\frac{d\Dlf}{d\la} = -\frac{2}{\la}\Clf$ for all $\la\geq 0$. Now, $\labb$ was defined as the unique minimizer of $\Dlf$, hence $\frac{d\Df\labb)}{d\la}=0$. Combining, $\Cf\labb)=0$, and thus,
%% In Section \ref{}, and in particular in Lemma \ref{} we prove that $\Cf\labb)=0$. 
%
%A critical question to ask is ``what is the optimal regularizer?''. Given that LASSO type approaches are being used in countless papers and applications, it is of prime importance to know the answer. We are able to give a simple answer to this question based on the properties of the mapping function. Clearly, $\labb$ minimizes the error in \eqref{eq:ell_defn} hence, we need to find the point $\tau=\taub$ in \eqref{eq:ell2_defn} where $\labb$ is mapped to. $\Dlf$ is minimized at $\labb$ and a trivial relation between $\Dlf$ and $\Clf$ reveals that $\Cf\labb)=0$. Consequently,
%$
%\call(\labb)=\sqrt{m-\Df\labb)}.
%$
%This gives, 
\begin{align}\label{eq:optRegForm}
\taub=\labb \sqrt{m-\Df\labb)}.
\end{align}
% as the optimal regularizer for the $\ell_2^2$ problem. Hence, to find the optimal regularizer, one needs to do the following.
%%Our recipe then for computing the optimal value of the penalty parameter for the $\ell_2^2$-LASSO can be summarized with the following two step procedure:
%%\begin{enumerate}
%%\item Minimize $\Dlf$ subject to $\la$ to compute $\labb$.
%%\item Return $\labb \sqrt{m-\Df\labb)}$.
%%\end{enumerate}
As a last comment, \eqref{eq:optRegForm} simplifies even further if one uses the fact $\Df\labb)\approx \Delxf$, which is valid under reasonable assumptions, \cite{Foygel,Oymak,McCoy}. In this case, $\tau_{best}\approx\labb\sqrt{m-\Delxf}$.
%, which arguably looks more accessible.% and simpler.% and relates the optimal penalty parameter to the compressed sensing phase transition of the noiseless linear inverse problem \eqref{LinInv}.

%\subsection{Normalization issues and the cost of linear measurements}
%The reader will observe that, since we are using an i.i.d. $\Nn(0,1)$ measurement matrix, we are actually amplifying the signal power. This naturally results in a small amount of NSE. For example, consider the Constrained-LASSO. The asymptotic NSE is around $\frac{\Delxf}{m-\Delxf}$ which is only $1$ when we set $m=2\Delxf$. \samet{talk about literature here}. To keep the signal power fixed, we can use a matrix $\A$ with i.i.d. $\Nn(0,\frac{1}{m})$ entries. This is equivalent to amplifying noise variance to $m\sigma^2$ but still normalizing the error term $\|\x^*-\x_0\|^2$ by $\sigma^2$. Hence, it would amplify the NSE by $m$ and we would end up with $\frac{m\Delxf}{m-\Delxf}$. Equivalently let us write this as $\Delxf\cdot \frac{m}{m-\Delxf}$.
%
%Now, let us compare this with the direct denoising. In the simple denoising problem, the asymptotic MSE is known to be $\Delxf$. Hence, we have two problems with the same signal power one of which has linear measurements. While both normalized error terms have $\Delxf$, LASSO error has an additional factor of $\frac{m}{m-\Delxf}$ which is strictly greater than $1$. This intuitively corresponds to a penalty for using linear measurements.

%********************************************
%%

%\subsection{Closed form bounds}
\subsection{Closed Form Calculations of the Formulae}\label{sec:closedForm}

\begin{center}
\begin{table}[h]
\begin{center}
  \begin{tabular}{ | c |  c  |  }
     \cline{2-2}
     \multicolumn{1}{c|}{}&\pbox{20cm}{\vp{\bf{{Normalized Squared Error}}}\vp}\\ \hline
    {\color{darkred}{C-LASSO}} & \pbox{20cm}{ \vp{\Large{$\frac{\Delxf}{m-\Delxf}$}}\vp}   \\    \hline
     {\color{darkred}{$\ell_2$-LASSO}}  &  \pbox{20cm}{ \vp{\Large{$\frac{\Dlf}{m-\Dlf}$}}~{\text{for}~$\la\in\Rco$}\vp}  \\\hline
          {\color{darkred}{$\ell_2^2$-LASSO}}  & \pbox{20cm}{ \vp{\Large{$\frac{\Df\map^{-1}(\tau))}{m-\Df\map^{-1}(\tau))}$}}~{\text{for}~${\tau}\in\R^+$}\vp}   \\\hline
  \end{tabular}
\end{center}
  \caption{Summary of formulae for the NSE.}
  \label{table:formulae}
\end{table}
\end{center}

Table \ref{table:formulae} summarizes the formulae for the NSE of the three versions of the LASSO problem. While simple and concise, it may appear to the reader that the formulae are rather abstract, because of the presence of  $\Delxf$ and $\Dlf$ ($\Clf$ is also implicitly involved in the calculation of  $\map^{-1}(\cdot)$) which were introduced to capture the convex geometry of the problem. However, as discussed here, for certain critical regularizers $f(\cdot)$, one can calculate (tight) upper bounds or even explicit formulas for these quantities. For example, for the estimation of a $k$-sparse signal $\x_0$ with $f(\cdot)=\|\cdot\|_1$, it has been shown that $\Delxf\lesssim 2k (\log\frac{n}{k}+1)$. Substituting this into the formula for the NSE of the C-LASSO results in the ``closed-form" upper bound given in \eqref{eq:clo1}, i.e. one expressed only in terms of $m$,$n$ and $k$. Analogous results have been derived  \cite{Cha, Oym,StojBlock,Foygel} for other well-known signal models as well, including low rankness (see \eqref{eq:clo2}) and block-sparsity (see \eqref{eq:clo3}). The first row of Table \ref{table:closed} summarizes some of the results for $\Delxf$ found in the literature (see \cite{Cha,Foygel}). The second row provides our closed form results on $\Dlf$ when $\la$ is sufficiently large. The reader will observe that, by setting $\la$ to its lower bound in the second row, one approximately obtains the corresponding result in the first row. For a related discussion on $\Dlf$ and closed form bounds, the reader is referred to \cite{Foygel}. The derivation of these results can be found in Section \ref{explicit_form} of the Appendix. In the same section, we also provide exact formulas for $\Dlf$ and $\Clf$ for the same signal models.  Based on those formulas and Table \ref{table:closed}, one simply needs to substitute $\Delxf$ or $\Dlf$ with their corresponding value to reach the error bounds. 
%For example, for an $n\times n$, rank $r$ matrix, the NSE is upper bounded by $\frac{6nr}{m-6nr}$. 
We should emphasize that, examples are not limited to the ones discussed here (see for instance \cite{Cha}). 
%For instance, \cite{Cha} considers a good amount of examples regarding structured signals. 

%\chris{@Samet: Could you fill in the table?}
%$((\sqrt{d}+\sqrt{2\log\frac{t}{k}})^2+d+2)k$
\begin{table}[h]
\begin{center}
\hspace*{-5pt}  \begin{tabular}{ | c |  c  |  c | c |}
     \cline{2-4}
     \multicolumn{1}{c|}{}&\pbox{20cm}{\vp{\bf{{{$k$-sparse, $\x_0\in\R^n$}}}}\vp}& \pbox{20cm}{\bf{{{Rank $r$, $\X_0\in\R^{d\times d}$}}}}   & $k$-{\bf{block sparse}}, $\x_0\in\R^{tb}$ \\ \hline
    {\color{darkred}{{$\Delxf$}}} & \pbox{20cm}{ \vp $2k(\log \frac{n}{k}+1)$\vp}  &\pbox{20cm}{$6dr$} &  $4k(\log\frac{t}{k}+b)$\\    \hline
     {\color{darkred}{{$\Dlf$}}}  & $ (\la^2+3)k$ ~for~ $\la\geq \sqrt{2\log \frac{n}{k}}$ & \pbox{20cm}{\vp \small{$\la^2r+2d(r+1)$ ~for~ $\la\geq 2\sqrt{d}$}\vp} & \small{$(\la^2+b+2)k$~ for~ $\la \geq \sqrt{b}+\sqrt{2\log\frac{t}{k}}$} \\
    \hline
  \end{tabular}
  \caption{Closed form upper bounds for $\Delxf$ (\cite{Cha,Foygel}) and $\Dlf$ corresponding to \eqref{eq:clo1}, \eqref{eq:clo2} and \eqref{eq:clo3}.}
  \label{table:closed}
\end{center}
\end{table}

It follows from this discussion, that establishing new and tighter analytic bounds for $\Dlf$ and $\Delxf$ for more regularizers $f$ is certainly an interesting direction for future research. In the case where such analytic bounds do not already exist in literature or are hard to derive, one can numerically estimate $\Dlf$ and $\Delxf$ once there is an available characterization of the set of subdifferentials $\paf$. More in detail, it is not hard to show that, when $\h\sim\Nn(0,\Iden_n)$, $\dt^2(\h,\la\paf)$ concentrates nicely around $\Dlf$ (see Lemma \ref{lemma:concentrationALL}) . Hence to compute $\Dlf$:
\begin{enumerate}[(a)]
\item draw a vector $\h\sim\Nn(0,\mathbf{I}_n)$,
\item return the solution of the convex program $\min_{\s\in\paf}{\|\h-\la\s\|^2}$.
\end{enumerate}
Computing $\Delxf$ can be built on the same recipe by writing $\dt^2(\h,\text{cone}(\paf))$ as $\min_{\la\geq 0,\s\in\paf}{\|\h-\la\s\|^2}$.

Summing up, our proposed formulae for the NSE of the LASSO problems can be effectively calculated, either analytically or numerically.

\subsection{Translating the Results}\label{sec:trans}
Until this point, we have considered the scenario, in which the measurement matrix $\A$ has independent standard normal entries, and the noise vector $\z$ is equal to $\sigma\vb$ with $\vb\sim\Nn(0,\Iden_{m})$. In related literature, the entries of $\A$ are often assumed to have variance $\frac{1}{m}$ or $\frac{1}{n}$, \cite{Mon,BayMon,CanHow}.
% do not have unit variance \cite{Mon,BayMon,CanHow}. Instead, its entries have variance $\frac{1}{m}$ or $\frac{1}{n}$.
  For example, a variance of $\frac{1}{m}$ 
  %helps to ensure
  ensures that in expectation $\|\A\x\|^2$ is same as $\|\x\|^2$. Hence, it is important to understand, how our setting can be translated to those.
% Similarly, variance of $\frac{1}{n}$ ensures that, rows of $\A$ have approximately unit variance which is consistent with
To distinguish our setup from the ``non-unit variance'' setup, we introduce the ``non-unit variance'' variables $\A',\sigma',\la'$ and $\tau'$. Let entries of $\A'$ have variance $\frac{1}{m}$ and consider the $\ell_2$-LASSO problem with these new variables, which can be equivalently written as,
\beq
\min_{\x}\|\A'\x_0+\sigma'\vb-\A'\x\|+\la'f(\x).\nn
\eeq
Multiplying the objective with $\sqrt{m}$, we obtain,
\beq
\min_{\x}\|\sqrt{m}\A'\x_0+\sqrt{m}\sigma'\vb-\sqrt{m}\A'\x\|+\sqrt{m}\la'f(\x).\nn
\eeq
Observe that, $\sqrt{m}\A'$ is now statistically identical to $\A$. Hence, Theorem \ref{thm:ell2LASSO} is applicable under the mapping $\sigma\leftarrow\sqrt{m}\sigma'$ and $\la\leftarrow\sqrt{m}\la'$. Consequently, the NSE formula for the new setting for $\sqrt{m}\la'\in \Rco$ can be given as,
\beq
\frac{\|\x^*_{\ell_2}-\x_0\|^2}{({\sqrt{m}\sigma'})^2}=\frac{\|\x^*_{\ell_2}-\x_0\|^2}{{\sigma}^2}\lesssim \frac{\Dlf}{m-\Dlf}=\frac{\Df\sqrt{m}\la')}{m-\Df\sqrt{m}\la')}.\nn
\eeq
Identical arguments for the Constrained-LASSO and $\ell_2^2$-LASSO results in the following NSE formulas,
\beq
\frac{\|\x^*_c-\x_0\|^2}{m{\sigma'}^2}\lesssim \frac{\Delxf}{m-\Delxf} \quad \text{ and } \quad\frac{\|\x^*_{\ell_2^2}-\x_0\|^2}{m{\sigma'}^2}\lesssim \frac{\Df\map^{-1}(m\tau'))}{m-\Df\map^{-1}(m\tau'))}.\nn
\eeq
In general, reducing the signal power $\|\A\x_0\|^2$ by a factor of $m$, amplifies the proposed NSE upper bound by $m$ times and the penalty parameters should be mapped as $\tau\longleftrightarrow m\tau'$ and $\la\longleftrightarrow\sqrt{m}\la'$.

\section{Applying Gordon's Lemma}
\label{sec:intro2tech}
%\eqname{sec:intro2tech}

%
%\subsection{Notation}
%
First, we introduce the basic notation that is used throughout the technical analysis of our results. Some additional notation, specific to the subject of each particular section is introduced later therein. 
%We use boldface lowercase letters to denote vectors and boldface capital letters to denote matrices. We write $\mathbf{I}_{m\times n}$ and $\mathbf{I}_{m}$ for  the identity matrices of sizes $m\times n$ and $m\times m$, respectively.
To make explicit the variance of the noise vector $\z$, we denote $\z=\sigma\vb$, where $\vb\sim\mathcal{N}(0,\mathbf{I}_m)$. Also, we reserve the variables $\h$ and $\g$ to denote i.i.d. Gaussian vectors in $\mathbf{R}^n$ and $\mathbf{R}^m$, respectively. In similar flavor, reserve the variable $\s$ to describe the subgradients of $f$ at $\x_0$. Finally, the Euclidean unit ball and unit sphere are respectively denoted as
\begin{align}
\Bc^{n-1} := \left\{ \x\in\mathbb{R}^n~ | ~\| \x\|\leq 1\right\} \quad \text{ and } \quad \Sc^{n-1} := \left\{ \x\in\mathbb{R}^n ~|~ \| \x\|= 1\right\} .\nn
\end{align}

\subsection{Introducing the Error Vector}

%Our analysis aims to precisely characterize the norm of the error vector $\x_{LASSO}^*-\x_0$. 
%Next, we describe a useful (notational) transformation of tÄhe LASSO algorithms which solves for the error vector $\
%x-\x_0$ instead of $\x$. In particular, 

For each candidate solution $\x$ of the LASSO algorithm, denote $\w=\x-\x_0$. Solving for $\w$ is clearly equivalent to solving for $\x$, but simplifies considerably the presentation of the analysis. Under this notation, $\|\y-\A\x\| = \|\A\w-\sigma\vb\|.$ Furthermore, it is convenient to subtract the constant factor $\la f(\x_0)$ from the objective function of the LASSO problem and their approximations. In this direction, define the following \emph{``perturbation" functions}:
\begin{align}
f_p(\w) &= f(\x_0 + \w) - f(\x_0),\label{eq:perdef1} \\
\hat{f}_p(\w) &= \hat{f}(\x_0 + \w) - f(\x_0) =  \sup_{\s\in\paf} \s^T\w. \label{eq:perdef2}
\end{align}
Then, the $\ell_2$-LASSO  will write as
\begin{align}\label{eq:well2}
\w^*_{\ell_2} = \arg\min_{\w}\left\{ \|\A\w-\sigma\vb\| + \la  f_p(\w)  \right\}.
\end{align}
and the C-LASSO  as
\begin{align*}
%\hat{w}_{c}^*(\A,\vb,\sigma) ~=~ &\arg
\w_c^* = &\arg\min_{\w}~\|\A\w-\sigma\vb\|~ \\
&~~~~~~~\text{s.t.}~~~ {f}_p(\w)\leq 0 .
\end{align*}
or, equivalently, 
\begin{align}\label{eq:wC}
%\hat{w}_{c}^*(\A,\vb,\sigma) ~=~ &\arg
\w^*_c = \arg\min_{\w}\left\{~\|\A\w-\sigma\vb\| + \max_{\la\geq 0} \la {f}_p(\w) ~\right\} .
\end{align}

\subsection{The Approximate LASSO Problem}

In Section \ref{sec:app}, and in particular in \eqref{eq:ell2Approx} we introduced the approximated $\ell_2$-LASSO problem. We repeat the definition here, and also, we define accordingly the approximate C-LASSO.
%Accordingly, recall the definition of the first-order approximation $\hat{f}(\cdot)$ of $f(\cdot)$ in \eqref{eq:fodef} to write for 
The approximated $\ell_2$-LASSO writes:
\begin{align}\label{eq:foEll2_w}
\hat{\w}_{\ell_2} = \arg\min_{\w}\left\{ \|\A\w-\sigma\vb\| + \la \hat{f}_p(\w) \right\}.
\end{align}
%and we recover $\x_{\ell_2}^* = \x_0 + \w_{\ell_2}^*$.
Similarly, the approximated C-LASSO writes
%\begin{align*}
%\hat{w}_{c}^*(\A,\vb,\sigma) ~=~ &\arg\min_{\w}~\|\A\w-\sigma\vb\|~ \\
%&~\text{s.t.}~~ \hat{f}_p(\w)\leq 0 .
%\end{align*}
%or, equivalently, 
\begin{align}\label{eq:foC_w}
\hat{\w}_{c} ~=~ &\arg\min_{\w}\left\{~\|\A\w-\sigma\vb\| + \max_{\la\geq 0} \la \hat{f}_p(\w) ~\right\} .
\end{align}
%Compare the approximated problems above to their corresponding original problems in \eqref{eq:well2} and \eqref{eq:wC}.  
 Denote $\Fc_c(\A,\vb)$ and $\Fc_{\ell2}(\A,\vb)$ the optimal costs of problems \eqref{eq:foC_w} and \eqref{eq:foEll2_w}, respectively.  
Note our convention to use the symbol ``$~\hat{}~$" over variables that are associated with the approximate problems. To distinguish, we use the symbol ``$~^*~$"  for the variables associated with the original problems. 
%Finally, when clear from context, we often drop all or some of the arguments $\A,\vb,\sigma$ and $\la$ in the definitions above.
%

%The present Section and Section \ref{sec:KeyIdeas} focus explicitly on the approximated LASSO 

%
\subsection{Technical Tool: Gordon's Lemma}\label{sec:gorMain}
 As already noted the most important technical ingredient underlying our analysis is a Lemma proved by Gordon in \cite{Gor}; recall Lemma \ref{lemma:Gordon} in  Section \ref{sec:app}. %Gordon's Lemma establishes a very useful inequality for Gaussian processes.
%%In general, Gaussian inequalities play a fundamental role in the study of high dimensional probability \cite{Cha,Sto1}. 
%Stojnic was the first to use Gordon's Lemma for the purposes of analyzing the LASSO problem \cite{StojLAS,StojSOCP}. In fact, he even proposed the application of the Lemma to the analysis of more general classes of optimization problems, \cite{StojReg}. Even though the results of this work are more general than those in \cite{StojLAS} in many directions, a significant part of this work adopts the key features of Stojnic's technique . Thus, Gordon's lemma proves to be essential for our analysis.  
In fact, Gordon's key Lemma \ref{lemma:Gordon} is a Corollary of a more general theorem which  establishes a probabilistic comparison between two centered Gaussian processes. The theorem was proved by Gordon in \cite{Gor2} and is stated below for completeness.

\begin{thm}[Gordon's Theorem, \cite{Gor}]\label{thm:GordonMain}
Let $\left\{X_{ij}\right\}$ and $\left\{Y_{ij}\right\}$, $1\leq i\leq n$, $1\leq j\leq m$, be two centered Gaussian processes which satisfy the following inequalities for all choices of indices
\begin{enumerate}
\item $\E\left[X_{ij}^2\right] = \E\left[Y_{ij}^2\right]$,
\item $\E\left[X_{ij}X_{ik}\right] \geq \E\left[Y_{ij}Y_{ik}\right]$,
\item $\E\left[X_{ij}X_{\ell k}\right] \leq \E\left[Y_{ij}Y_{\ell k}\right]$, \quad if $i\neq\ell$.
\end{enumerate}
Then,
$$
\Pro\left( \cap_i \cup_j \left[ Y_{ij} \geq \lambda_{ij} \right] \right) \geq \Pro\left( \cap_i \cup_j \left[ X_{ij} \geq \lambda_{ij} \right] \right),
$$
for all choices of $\lambda_{ij}\in\mathbf{R}$.
\end{thm}
Application of Gordon's Theorem \ref{thm:GordonMain} to specific Gaussian processes results in Gordon's Lemma \ref{lemma:Gordon} \cite{Gor}. 
%We repeat the Lemma here for completeness. 
%\begin{lem}[modified Gordon's Lemma]\label{lemma:Gor}
%Let $\mathbf{G}\in\mathbb{R}^{m\times n}$ be an i.i.d. matrix with standard normal entries. Further let $\g\sim\Nn(0,\mathbf{I}_m)$, $\h\sim\Nn(0,\mathbf{I}_n)$ and $g\sim\Nn(0,1)$ and assume all $\mathbf{G},\g,\h,g$ are independently generated. Finally, let $\Phi_1\subset\mathbf{R}^m$ be arbitrary set, $C>0$ some constant and $\psi:\Phi_1\times\mathbb{R}^m\rightarrow\mathbb{R}$ arbitrary function. Then, for any $c\in\mathbb{R}$:
%\begin{align}\label{eq:inGor}
%\Pro\left(  \min_{\x\in\Phi_1}~~\max_{\|\ab\| = C}~\left\{ \ab^T\Gb\x  - \psi({\x,\ab} ) \right\} \geq c\right) \geq 2\Pro\left(  \min_{\x\in\Phi_1}~~\max_{\|\ab\| = C}~\left\{ \|\x \| \g^T\ab  -  C \h^T\x - \psi({\x,\ab})  \right\}\geq c  \right) - 1.
%\end{align}
%\end{lem}
%\samet{too much repetition, same lemma 3 times!}
In this work, we require a slightly modified version of this lemma, namely Lemma \ref{lemma:Gor}. The key idea is of course the same as in the original lemma, but the statement is modified to fit the setup of the current paper. 

\begin{restatable}[Modified Gordon's Lemma]{lem}{Gor}\label{lemma:Gor}%\eqname{lemma:Gor}
Let $\mathbf{G}$, $\g$, $\h$ be defined as in Lemma \ref{lemma:Gordon} and let $\psi(\cdot,\cdot):\R^n\times \R^m\rightarrow \R$ . Also, let $\Phi_1\subset\mathbb{R}^n$ and $\Phi_2\subset\mathbb{R}^m$ such that either both $\Phi_1$ and $\Phi_2$ are compact or $\Phi_1$ is arbitrary and $\Phi_2$ is a scaled unit sphere.
% $C>0$ be some constant.
  Then, for any $c\in\mathbb{R}$:
\begin{align*}
&\Pro\left(  \min_{\x\in\Phi_1}~\max_{\ab\in\Phi_2}~\left\{ \ab^T\Gb\x  - \psi({\x,\ab} ) \right\} \geq c\right) \geq \notag
2\Pro\left(  \min_{\x\in\Phi_1}\max_{\ab\in\Phi_2}\left\{ \|\x \| \g^T\ab  -  \|\ab\| \h^T\x - \psi({\x,\ab})  \right\}\geq c  \right) - 1.
%\label{eq:inGor}
\end{align*}
\end{restatable}

The proof of Lemma \ref{lemma:Gor} closely parallels the proof of Lemma \ref{lemma:Gor} in \cite{Gor}. We defer  the proof to Section \ref{sec:proofGor} in  the Appendix.

\subsection{Simplifying the LASSO objective through Gordon's Lemma }	% <--------------------
%

%The first fundamental idea behind this work is the use of Gordon's Lemma in order to simplify the objective function of the LASSO problem. This idea is attributed to Stojnic. Stonjic used Gordon's Lemma in a very recent paper \cite{StojLAS} in order to analyze the Constrained LASSO problem for the case where the structure induced function $f(\cdot)$ is the $\ell_1$ norm $\|\cdot\|_1.$ In this work we borrow and extend this idea to first generalize Stojnic's analysis of the Constrained LASSO to any convex structured induced function $f(\cdot)$ and also to analyze the more popular $\ell_2$ and $\ell_2^2$ LASSO problems.

Section \ref{synops} introduced the technical framework. Key feature in this framework is the application of Gordon's Lemma. In particular, we apply Gordon's Lemma three times: once each for the purposes of the lower bound, the upper bound and the deviation analysis. Each application  results in a corresponding simplified problem, which we call \emph{``key optimization"}. The analysis is carried out for that latter one as opposed to the original and more complex LASSO problem. In this Section, we show the details of applying Gordon's Lemma and we identify the corresponding key optimizations. Later, in Section \ref{sec:KeyIdeas}, we focus on the approximate LASSO problem and we show that in that case, the key optimizations are amenable to detailed analysis.

To avoid unnecessary repetitions, we treat the original and approximate versions of both the C-LASSO and the $\ell_2$-LASSO, in a common framework, by defining the following problem:
\begin{align}\label{eq:generic}
\Fco(\A,\vb) = \min_\w{\left\{~ \|\A\w-\sigma\vb \| + p(\w) ~\right\}},
\end{align}
where $p:\mathbf{R}^n\rightarrow \mathbf{R}\cup\infty$ is a proper convex function \cite{Roc70}. Choose the penalty function $p(\cdot)$  in the generic formulation \eqref{eq:generic} accordingly to end up with \eqref{eq:well2}, \eqref{eq:wC}, \eqref{eq:foEll2_w} or \eqref{eq:foC_w}. To retrieve \eqref{eq:wC} and \eqref{eq:foC_w}, choose $p(\w)$ as the indicator function of the sets $\left\{ \w | f_p(\w)\leq 0 \right\}$ and $\left\{ \w | \hat{f}_p(\w)\leq 0 \right\}$ \cite{Boyd}.

%In what follows,  In the subsequent section, we describe in details useful properties each one of them.

%As previously mentioned the analysis is mainly carried out for the corresponding approximated problems.  We show here how to use Gordon's Lemma for the purpose of simplifying the analysis of those problems. 

%Before doing so, and in an effort to present the ideas of this section as concise and general as possible, it is useful to view both the approximated C-LASSO and the approximated $\ell_2$-LASSO under a common framework. 

\subsubsection{Lower Bound}
The following corollary is a direct application of Lemma \ref{lemma:Gor} to $\Fco(\A,\vb)$ in \eqref{eq:generic}. 
%It establishes a probabilistic connection between $\Fco$ and a simpler to analyze optimization problem. We denote the optimal cost of this latter problem as $\Lco$. Loosely speaking, a high probability lower bound for $\Lco$ translates to a high probability lower bound for $\Fco$.
\begin{cor}\label{cor:low}
Let $\g\sim\Nn(0,\mathbf{I}_m)$, $\h\sim\Nn(0,\mathbf{I}_n)$ and $h\sim\Nn(0,1)$ and assume all $\g,\h,h$ are independently generated. Let
\begin{align}\label{eq:lc}
\Lco(\g,\h) = \min_{\w}\left\{ \sqrt{\|\w\|^2 + \sigma^2}\|\g\| - \h^T\w + p(\w) \right\}. 
\end{align}
Then, for any $c\in\mathbb{R}$:
\begin{align}\notag
\Pro\left( ~ \Fco(\A,\vb) \geq c~\right) \geq 2\cdot\Pro\left( ~ \Lco(\g,\h) - h\sigma \geq c ~ \right) - 1.
\end{align}
\end{cor}
\begin{proof} 
Notice that $\| \A \w-\sigma\vb  \| =  \| \A_{\vb} \w_\sigma \|$, where $\A_\vb:=[\A \ - \vb]$ is a matrix with i.i.d. standard normal entries of size $m\times (n+1)$ and $\w_\sigma=[\w^T~\sigma]^T\in\R^{n+1}$. Apply the modified Gordon's Lemma \ref{lemma:Gor}, with $\x=\w_\sigma$, $\Phi_1=\{\w_\sigma\big|\w\in\R^n\}$, $\Phi_2=\mathcal{S}^{m-1}$, $\Gb=\A_\vb $, $\psi(\w_\sigma) =  p(\w)$.
%, $\g\sim \Nn(0,I_{m\times m})$.
% and $\h\sim \Nn(0,I_{n+1\times n+1})$. 
Further perform the trivial optimizations over $\ab$ on both sides of the inequality. Namely, 
$
\max_{\|\ab\|=1}\ab^T\A_\z[\w^T~\sigma]^T=\|\A_\z\w_\sigma\|
$
and, $\max_{\|\ab\|=1}\g^T\ab=\|\g\|$.
\end{proof}

%$\Fco(\A,\vb)$ is a function of the realization of an $m\times n$ gaussian matrix $\A$ and an $m\times 1$ gaussian vector $\vb$. On the other hand, $\Lco(\g,\h)$ is a function of only two independently generated gaussian vectors $\g$ and $\h$ of sizes $m\times 1$ and $n\times 1$, correspondingly.  This is the main reason why $\Lco$ is much more amenable to direct analysis than $\Fco$. We fully characterize $\Lco$ in the next section. This characterization is of great importance for the purposes of the analysis since we can then use it to characterize $\Fco$ according to Corollary \ref{cor:low}.

%\begin{lem}
% Further assume that $m\geq\DC + \eps_0 m$ for some $\eps_0\geq 0$. Then, for any $\eps>0$, there exist $c_1,c_2>0$ such that
%\begin{align*}
%\Pro\left( \Lco(\g,\h) \geq (1-\eps) \sigma\sqrt{ m - \DC } \right) \geq 1 - c_1\exp(-c_2 m).
%\end{align*}
%\end{lem}

%
\subsubsection{Upper Bound}\label{sec:thisUp}
Similar to the lower bound derived in the previous section, we derive an upper bound for $\Fco(\A,\vb)$. For this, we need to apply Gordon's Lemma to $-\Fco(\A,\vb)$ and use the dual formulation of it.  Lemma \ref{sec:dual} in the Appendix shows that the dual of the minimization in \eqref{eq:generic} can be written as
\begin{align}
-\Fco(\A,\vb) = \min_{\| \mub \|\leq 1}\max_{\w} \left\{ \mub^T\left( \A\w - \sigma\vb\right)- p(\w) \right\}.
\end{align}  
Lemma \ref{lemma:Gor} requires the set over which maximization is performed to be compact. We thus apply Lemma \ref{lemma:Gor} to the restricted problem,
 $$\min_{\| \mub \|\leq 1}~\max_{\|\w\|\leq C_{up}} \left\{ \mub^T\left( \A\w - \sigma\vb\right)- p(\w) \right\}.$$ 
 Notice, that this still gives a valid lower bound to $-\Fco(\A,\vb)$ since the optimal cost of this latter problem is no larger than $-\Fco(\A,\vb)$. 
 In Section \ref{sec:KeyIdeas}, we will choose $C_{up}$ so that the  resulting lower bound is as tight as possible.

\begin{cor}\label{cor:up}
Let $\g\sim\Nn(0,\mathbf{I}_m)$, $\h\sim\Nn(0,\mathbf{I}_n)$ and $h\sim\Nn(0,1)$ and assume all $\g,\h,h$ are independently generated. Let,
\begin{align}\label{eq:uc}
\Uco(\g,\h) = -\min_{\|\mub\|\leq 1}~\max_{\|\w\|\leq C_{up}}\left\{ \sqrt{\|\w\|^2 + \sigma^2}~\g^T\mub + \|\mub\|\h^T\w - p(\w) \right\}. 
\end{align}
Then, for any $c\in\mathbb{R}$:
\begin{align}\notag
\Pro\left( ~ \Fco(\A,\vb) \leq c~\right) \geq 2\cdot\Pro\left( ~ \Uco(\g,\h) - \min_{0\leq \alpha\leq1} \alpha\sigma h  \leq c ~ \right) - 1.
\end{align}
\end{cor}
\begin{proof}
Similar to the proof of Corollary \ref{cor:low} write $\| {\sigma}\vb - \A \w \| =  \| \A_{\vb} \w_\sigma \|$. Then, apply the modified Gordon's Lemma \ref{lemma:Gor}, with $\x=\mub$, $\alpha = \w_{\sigma}$, $\Phi_1=\mathcal{B}^{m-1}$, $\Phi_2= \left\{\w_\sigma~ |~ \frac{1}{C_{up}}\w\in\Bc^{n-1}\right\}$, $\Gb=\A_\vb $, $\psi(\w_\sigma) =  p(\w)$, to find that for any $c\in \mathbb{R}$:
\begin{align*}\notag
\Pro\left( ~ -\Fco(\A,\vb) \geq -c~\right) &\geq 2\cdot\Pro\left( ~ \min_{\|\mub\|\leq 1}~\max_{\|\w\|\leq C_{up}}\left\{ \sqrt{C_{up}^2 + \sigma^2}~\g^T\mub + \|\mub\|\h^T\w - p(\w) +\| \mub\|\sigma h \right\} \geq -c ~ \right) - 1\\
&\geq 2\Pro\left( ~ -\Uco(\g,\h) + \min_{\|\mu\|\leq1}\| \mub\|\sigma h \geq -c ~ \right) -1.
\end{align*}
%This concludes the proof of the Corollary.
\end{proof}
%
%Similarly to $\Lco$, we are able to fully analyze $\Uco^*$ as defined in the Corollary above, and thus deduce a high probability upper bound for $\Fco$. 

%
\subsubsection{Deviation Analysis}\label{sec:thisDev}
Of interest in the deviation analysis of the LASSO problem (cf. Step 4 in Section \ref{synops})  is the analysis of a restricted version of the $LASSO$ problem, namely
\begin{align}\label{eq:res_problem}
\min_{\|\w\|\in S_{dev}}~\left\{ \|\A\w - \sigma\vb \| + p(\w) \right\} %\eqname{eq:res_problem}
\end{align}
where
 $$S_{dev} := \left\{ \ell~ | ~ \left| \frac{ \ell }{C_{dev}} -1 \right|\geq \delta_{dev}\right\}.$$
$\delta_{dev}>0$ is any arbitrary small constant and $C_{dev}>0$ a constant that will be chosen carefully for the purpose of the deviation analysis . We establish a high probability lower bound for \eqref{eq:res_problem}.
As usual,  we apply Lemma \ref{lemma:Gor} to our setup, to conclude the following.
\begin{cor}\label{cor:dev}
Let $\g\sim\Nn(0,\mathbf{I}_m)$, $\h\sim\Nn(0,\mathbf{I}_n)$ and $h\sim\Nn(0,1)$ and assume all $\g,\h,h$ are independently generated. Let
\begin{align}
\Lco_{dev}(\g,\h) = \min_{\|\w\|\in S_{dev}}\left\{ \sqrt{\|\w\|^2 + \sigma^2}\|\g\| - \h^T\w + p(\w) \right\}. \label{eq:ldev}
\end{align}
Then, for any $c\in\mathbb{R}$:
\begin{align}\notag
\Pro\left( ~ \min_{\|\w\|\in S_{dev}}~\left\{ \|\A\w - \sigma\vb \| + p(\w) \right\}  \geq c~\right) \geq 2\cdot\Pro\left( ~ \Lco_{dev}(\g,\h) - h\sigma \geq c ~ \right) - 1.
\end{align}
\end{cor}
%
%\begin{cor}\label{cor:res_Gordon}
%Let $\g\sim \Nn(0,I_{m\times m})$, $\h\sim \Nn(0,I_{n\times n})$ and $g\sim\Nn(0,1)$, all independently generated. Then, for all $c>0$:
%\beq\label{eq:res_useful}
%\Pro\left(\min_{\w \in T_\Cc(\x_0) , \|\w\|\in S_{dev}} \| {\sigma}\vb - \A \w \| +\sqrt{\|\w\|^2+\sigma^2}g\geq c \right) \geq \Pro\left(\min_{\w \in T_\Cc(\x_0) ,  \|\w\|\in S_{dev} }   \sqrt{\|\w\|^2+\sigma^2}\|\g\|-\h^T\w-h_{n+1}\sigma\geq c\right)
%\eeq
%\end{cor}
\begin{proof} 
Follows from Lemma \ref{lemma:Gor} following exactly the same steps as in the proof of Corollary \ref{cor:low}.
\end{proof}
The reader will observe that $\Lco$ is a special case of $\Lco_{dev}$ where $S_{dev}=\R^+$.

%\subsection{Further simplification of the corollaries}
\subsubsection{Summary}

We summarize the results of Corollaries  \ref{cor:low}, \ref{cor:up} and \ref{cor:dev} in Lemma \ref{lemma:applied_Gordon}. Adding to a simple summary, we perform a further simplification of the corresponding statements. In particular, we discard the ``distracting" term $\sigma h$ in Corollaries \ref{cor:low} and \ref{cor:dev}, as well as the term $\min_{0\leq \alpha\leq1} \alpha\sigma h$ in Corollary \ref{cor:up}. Recall the definitions of the key optimizations $\Lco$, $\Uco$ and $\Lco_{dev}$ in \eqref{eq:lc}, \eqref{eq:uc} and \eqref{eq:ldev}.

\begin{lem}
\label{lemma:applied_Gordon}
%$\eqname{lemma:applied_Gordon}$
 Let  $\g\sim\mathcal{N}(0,{\bf{I}}_m)$ and $\h\sim\mathcal{N}(0,{\bf{I}}_n)$ be independently generated. Then,  for any positive constant $\eps>0$,  the following are true:
 \begin{align}\notag
&\text{1.~~}
%\text{1. Low:~~}
\Pro\left( ~ \Fco(\A,\vb) \geq c~\right) \geq 2~\Pro\left( ~ \Lco(\g,\h) - \sigma\eps\sqrt{m} \geq c ~ \right) - 4\exp\left(-\frac{\eps^2m}{2}\right)-1.\\
&\text{2.~~}
%\text{2. Up:~~~}
\Pro\left( ~ \Fco(\A,\vb) \leq c~\right) \geq 2~\Pro\left( ~ \Uco(\g,\h) +\sigma\eps\sqrt{m} \leq c ~ \right) - 4\exp\left(-\frac{\eps^2m}{2}\right)-1.\nn\\
&\text{3.~~}
%\text{3. Dev:~~}
\Pro\left( ~ \min_{\|\w\|\in S_{dev}}~\left\{ \|\A\w - \sigma\vb \| + p(\w) \right\}  \geq c~\right) \geq 2~\Pro\left( ~ \Lco_{dev}(\g,\h) - \sigma\eps\sqrt{m} \geq c ~ \right) - 4\exp\left(-\frac{\eps^2m}{2}\right)-1.\nn
\end{align}

\end{lem}
\begin{proof}
For $h\sim\mathcal{N}(0,1)$ and all $\eps>0$,  
\begin{align}\label{eq:use154}
\Pro\left( |h|\leq \eps\sqrt{m} \right) \geq 1-2\exp(-\frac{\eps^2m}{2}).
\end{align}
Thus,
\begin{align}
\Pro\left( ~ \Lco(\g,\h) - h\sigma \geq c ~ \right)&\geq \Pro\left( \Lco(\g,\h) - \eps\sigma\sqrt{m} \geq c ~,~ h\leq \eps\sqrt{m}~ \right)\nn\\
&\geq \Pro(\Lco(\g,\h) - \eps\sigma\sqrt{m} \geq c)-2\exp(-\frac{\eps^2m}{2})\nn.
\end{align}
Combine this with Corollary \ref{cor:low} to conclude with the first statement of Lemma \ref{lemma:applied_Gordon}. The proof of the third statement of the Lemma follows the exact same steps applied this time to Corollary \ref{cor:dev}.
% 
% In Corollaries \ref{cor:low}, \ref{cor:up}, \ref{cor:dev}, using Lemma \ref{lemma:conc4}, $|h|\leq \eps\sqrt{m}$ with probability $1-2\exp(-\frac{\eps^2m}{2})$. Hence,
%The bound for $\Lco_{dev}$ can be obtained in the exact same manner. 
For the second statement write,
\begin{align}
\Pro\left( ~ \Uco(\g,\h) - \min_{\|\mu\|\leq1}\| \mub\|\sigma h \leq c ~ \right)&\geq \Pro\left( ~ \Uco(\g,\h) +\sigma |h| \leq c ~ \right)\nn\\
&\geq \Pro\left(~\Uco(\g,\h) + \eps\sigma\sqrt{m} \leq c ~,~ |h|\leq \eps\sqrt{m}~\right)\nn,
\end{align}
and use \eqref{eq:use154} as above. To conclude, combine with the statement of Corollary \ref{cor:up}.
\end{proof}

%\section{Key Technical Analysis}
\section{After Gordon's Lemma: Analyzing the Key Optimizations}
\label{sec:KeyIdeas}
%\eqname{sec:KeyIdeas}

\subsection{Preliminaries}

This Section is devoted to the analysis of the three key optimizations introduced in the previous section. In particular, we focus on the approximated C-LASSO and $\ell_2$-LASSO problems, for which a detailed such analysis is tractable. Recall that the approximated C-LASSO and $\ell_2$-LASSO are obtained from the generic optimization in \eqref{eq:generic} when substituting $p(\w) = \max_{\la\geq 0} \max_{\s\in\la\paf}\s^T\w =\max_{\s\in\text{cone}(\paf)}\s^T\w$ and $p(\w) = \max_{\s\in\la\paf}\s^T\w$, respectively. Considering this and recalling the definitions in \eqref{eq:lc}, \eqref{eq:uc} and \eqref{eq:ldev}, we will be analyzing the following \emph{key optimizations},
\begin{subequations}\label{eq:keyApp}
\begin{align}
\Lc(\g,\h) &= \min_{\w}\left\{ \sqrt{\|\w\|^2 + \sigma^2}\|\g\| - \h^T\w + \max_{\s\in\Cc} \s^T\w \right\}, \\
\Uc(\g,\h) &= -\min_{\|\mub\|\leq 1}~\max_{\|\w\| = C_{up}}\left\{ \sqrt{\|\w\|^2 + \sigma^2}~\g^T\mub + \|\mub\|\h^T\w -  \max_{\s\in\Cc} \s^T\w  \right\}, \label{eq:Uc2}\\
\Lc_{dev}(\g,\h) &= \min_{\|\w\|\in S_{dev}}\left\{ \sqrt{\|\w\|^2 + \sigma^2}\|\g\| - \h^T\w + \max_{\s\in\Cc} \s^T\w \right\},
\end{align}
\end{subequations}
where $\Cc$ is taken to be either $\text{cone}(\paf)$ or $\la\paf$, corresponding to the C-LASSO and $\ell_2$-LASSO, respectively. Notice that in \eqref{eq:Uc2} we have constrained the feasible set of the inner maximization to the scaled sphere rather than ball. Following our discussion, in Section \ref{sec:thisUp} this does not affect the validity of Lemma \ref{lemma:applied_Gordon}, while it facilitates our derivations here.

To be consistent with the definitions in \eqref{eq:keyApp}, which treat the key optimizations of the C-LASSO and $\ell_2$-LASSO under a common framework with introducing a generic set $\Cc$, we also define
\begin{align}\label{eq:com}%\eqname{eq:com}
\Fc(\A,\vb) = \min_{\w} \left\{ \| \A\w - \sigma\vb \| + \max_{\s\in\Cc} \s^T\w \right\},
\end{align}
to correspond to \eqref{eq:foC_w} and \eqref{eq:foEll2_w}, when setting $\mathcal{C} = \text{cone}(\paf)$ and $\mathcal{C} = \la\paf$, respectively.

\subsection{Some Notation}\label{sec:snot}
%
%\subsubsection{Subdifferential, Tangent and Normal Cones}
%%In the following, let $f(\cdot)$ be a convex function and $\x_0\in\text{dom} f\subset\mathbb{R}^n$. 
%The subdifferential of  $f(\cdot)$ at $\x_0$ is denoted by $\paf$. $\partial f(\x_0)$ is always a closed convex set \cite{Roc70} and for all $\s\in\paf$,
%$$
%f(\x)\geq f(\x_0) + \s^T(\x-\x_0), ~~\text{for all } \x\in\text{dom}f.
%$$
% The tangent cone of $f(\cdot)$ at $\x_0$ is defined as
%\begin{align}
%T_f(\x_0) = \text{cone}\left(\left\{ \x-\x_0 ~|~ f(x)\leq f(\x_0) \right\}\right),
%\end{align}
%where $\text{cone}\left(\cdot \right)$ denotes the conic hull.
%
%The following useful Lemma relates the tangent cone to the cone of subdifferentials.
%\begin{lem} [\cite{Roc70}]\label{tsg}
%Assume $f(\cdot):\R^n\rightarrow\R$ is a convex function and assume $\x_0\in\R^n$ is not a minimizer of it. Then,
%\beq
%\Tc_f(\x)=\text{cone}(\pa f(\x_0))
%\eeq
%\end{lem}
%%
%%When $\x_0$ is not a minimizer of $f(\cdot)$, then the polar cone of the tangent cone ( often called the normal cone \cite{Roc70}) is the cone of the subdifferential \cite{Roc70}, i.e.
%%$$
%%( T_f(\x_0) ) = \text{cone}(\paf),
%%$$
%%where we denote $\mathcal{K}$ the polar cone of a cone $\mathcal{K}$.
%\subsubsection{Distance, Projection \& Correlation}
Let $\Cc\subset\mathbb{R}^n$ be a closed and nonempty convex set. For any vector $\x\in\mathbb{R}^n$, we denote its (unique) projection onto $\Cc$ as $\bu(\x,\Cc)$, i.e.
$$
\bu(\x,\Cc) := \operatorname{argmin}_{\s\in\Cc}{\| \x-\s \| }.
%\| \x-\bu(\x,\Cc) \| \leq \| \x-\s \|, ~~\text{for all } \s\in\Cc.
$$ 
It will also be convenient to denote,
 $$\bi(\x,\Cc) := \x-\bu(\x,\Cc).$$
 The distance of $\x$ to the set $\Cc$ can then be written as,
$$\dt(\x,\Cc) := \| \bi(\x,\Cc) \|.$$
Finally, we denote,
$$\corr(\x,\Cc):= \li \bu(\x,\Cc), \bi(\x,\Cc) \ri.$$
Now, let $\h\sim\mathcal{N}(0,\mathbf{I}_n)$. The following quantities are of central interest throughout the paper:
\begin{subequations}\label{eq:defn01}
\begin{align}
\DC &:= \E\left[~ \dt^2(\h,\Cc)~ \right], \\
\PC &:= \E\left[~ \| \bu(\h,\Cc) \|^2~ \right], \\
\CC &:= \E\left[~ \corr(\h,\Cc) ~\right], 
\end{align} 
\end{subequations}
where the $\E[\cdot]$ is over the distribution of the Gaussian vector $\h$. It is easy to verify that
$
n = \DC + \PC + 2\CC.
$
Under this notation,
\begin{align*}
\Dlf &= \mathbf{D}(\la\paf),\\
\Clf &= \mathbf{C}(\la\paf), \\
\Delxf &= \mathbf{D}(\text{cone}(\paf)).
\end{align*}
On the same lines, define $\Plf:= \mathbf{P}(\la\paf)$.

\subsection{Analysis}\label{sec:after}

We perform a detailed analysis of the three key optimization problems $\Lc$, $\Uc$ and $\Lc_{dev}$. For each one of them we summarize the results of the analysis in Lemmas \ref{lemma:lowKey}, \ref{lemma:upKey} and \ref{lemma:dev} below.
%
%%
%%Lemma \ref{lemma:applied_Gordon} summarizes the (probabilistic) connection between the LASSO problem $\Fc(\A,\vb)$ and the three key optimization problems, namely $\Lc$, $\Uc$ and $\Lc_{dev}$. Here, we perform a detailed analysis of the properties of each one of those problems. 
%
%For each one of them, we perform a detailed analysis of its properties and we summarize the results in corresponding Lemmas in this Section. 
%
Each Lemma includes three statements. In the first, we reduce the corresponding key optimization problem to a scalar optimization. Next, we compute the optimal value of this optimization in a deterministic setup. We convert this into a probabilistic statement in the last step, which is directly applicable in Lemma \ref{lemma:applied_Gordon}. Eventhough, we are eventually interested only in this last probabilistic statement, we have decided to include all three steps in the statement of the Lemmas in order to provide some further intuition into how they nicely build up to the desired result. All proofs of the lemmas are deferred to Section \ref{sec:proofOfKeyOpt} in the Appendix.

\subsubsection{Lower Key Optimization}
%
%In Lemma \ref{lemma:lowKey} we summarize all properties of $\Lc$. The proof of the lemma is deferred to the Appendix.
%\samet{When there is no lemma name let us write lemma small. For example, ``Lemma 5'' vs ``next lemma''.}
%\samet{I think there is redundancy as $\Lc_{dev}$ already covers $\Lc$. But let us deal with it later.}
\begin{lem}[Properties of $\Lc$]\label{lemma:lowKey}%\eqname{lemma:lowKey}
Let $\g\sim\Nn(0,\mathbf{I}_m)$ and $\h\sim\Nn(0,\mathbf{I}_n)$ and
\begin{align}\label{eq:lemFlow}
\Lc(\g,\h) = \min_{\w}\left\{ \sqrt{\|\w\|^2 + \sigma^2}\|\g\| - \h^T\w + \max_{\s\in\Cc} \s^T\w \right\},
\end{align}
Denote $\hat{\w}_{low}(\g,\h)$ its optimal value. 
The following are true:
\begin{enumerate}
\item{\emph{Scalarization:}} $\Lc(\g,\h) = \min_{\alpha\geq 0}\left\{ \sqrt{\alpha^2 + \sigma^2}\|\g\| - \alpha\cdot\dt(\h,\Cc) \right\}$
%\chris{change everywhere $\alpha\geq 0$}
\item {\emph{Deterministic result:}} If
$\|\g\|^2 > \dt(\h,\Cc)^2, $ then,
$$\Lc(\g,\h) = \sigma\sqrt{ \|\g\|^2 - \dt^2(\h,\Cc) },$$
and,
 $$\|\hat{\w}_{low}(\g,\h)\|^2 = \sigma^2\frac{\dt^2(\h,\Cc)}{\|\g\|^2 - \dt^2(\h,\Cc)}.$$
\item {\emph{Probabilistic result:}} Assume that $m\geq\DC + \eps_L m$ for some $\eps_L\geq 0$. Then, for any $\eps>0$, there exist $c_1,c_2>0$ such that, for sufficiently large $m$,
\begin{align*}
\Pro\left( \Lc(\g,\h) \geq (1-\eps) \sigma\sqrt{ m - \DC } \right) \geq 1 - c_1\exp(-c_2 m).
\end{align*}
\end{enumerate}
\end{lem}

%Corollary \ref{cor:low} and the discussion that accompanied it, made clear that we are eventually interested in the very last property of $\Lc$ stated in Lemma \ref{lemma:lowKey}, since it yields a high probability lower bound for $\Fc$. None the less, the intention  behind the current statement of the Lemma is to emphasize critical intermediate results that build up to the last statement. 

%
\subsubsection{Upper Key Optimization}
%
%The following Lemma summarizes the deterministic properties of $\Uc$. Its proof is deferred to the Appendix.
\begin{lem}[Properties of $\Uc$]\label{lemma:upKey}%\eqname{lemma:upKey}
Let $\g\sim\Nn(0,\mathbf{I}_m)$, $\h\sim\Nn(0,\mathbf{I}_n)$ and
\begin{align}
\Uc(\g,\h) = -\min_{\|\mu\| \leq 1}~\max_{\|\w\|=C_{up}}\left\{ \sqrt{C_{up}^2 + \sigma^2}~\g^T\mub + \|\mub\|\h^T\w - \max_{\s\in\Cc} \s^T\w \right\}.
\end{align}
The following hold true:
\begin{enumerate}
\item {\emph{Scalarization:}}
$\Uc(\g,\h) = -\min_{0\leq\alpha\leq 1}\left\{ -\alpha\cdot\sqrt{C_{up}^2 + \sigma^2}~\|\g\| + C_{up}\dt(\alpha\h,\Cc)\right\}.$
\item {\emph{Deterministic result:}}
If $\h \notin \Cc$ and
\begin{align}\label{eq:ass_up1}%\eqname{{eq:ass_up1}}
%C_{up}\dt(\h,\Cc) + C_{up}\frac{ \li \bi(\h,\Cc) , \bu(\h,\Cc) \ri  }{ \dt(\h,\Cc) } < \sqrt{C_{up}^2+\sigma^2}\|\g\|
C_{up}\dt(\h,\Cc) + C_{up}\frac{ \corr(\h,\Cc)  }{ \dt(\h,\Cc) } < \sqrt{C_{up}^2+\sigma^2}\|\g\|,
\end{align}
then,
\beq \Uc(\g,\h) =  \sqrt{C_{up}^2+\sigma^2}\|\g\|-C_{up}\dt(\h,\Cc) \label{alpha=1}.\eeq

\item {\emph{Probabilistic result:}}
Assume $m\geq  \max\left\{ \DC , \DC+\CC \right\} +\eps_L m$ for some $\eps_L>0$.  Set 
$$C_{up} = \sigma\sqrt{\frac{ \DC  }{ m - \DC  }}.
$$
Then, for any $\eps>0$, there exist $c_1,c_2>0$ such that for sufficiently large $\DC$,
\begin{align*}
\Pro\left( \Uc(\g,\h) \leq (1+\eps) \sigma\sqrt{ m - \DC } \right) \geq 1 - 
 c_1\exp\left(-c_2 \gamma(m,n)\right).
%c_1\exp(-c_2 \min\{\DC,n,\frac{m^2}{n}\}).
\end{align*}
where $\gamma(m,n)=m$ if $\Cc$ is a cone and $\gamma(m,n)=\min\left\{m,\frac{m^2}{n}\right\}$ otherwise.
\end{enumerate}
\end{lem}

\subsubsection{Deviation Key Optimization}
\begin{lem}[Properties of $\Lc_{dev}$]\label{lemma:dev}
Let $\g\sim\Nn(0,\mathbf{I}_m)$ and $\h\sim\Nn(0,\mathbf{I}_n)$ and
\begin{align}
\Lc_{dev}(\g,\h) = \min_{\|\w\|\in S_{dev}}\left\{ \sqrt{\|\w\|^2 + \sigma^2}\|\g\| - \h^T\w + \max_{\s\in\Cc} \s^T\w \right\},
\end{align}
where
 $$S_{dev} := \left\{ \ell~ | ~ \left| \frac{ \ell }{C_{dev}} -1 \right|\geq \delta_{dev}\right\},$$
$\delta_{dev}>0$ is any arbitrary small constant and $C_{dev}>0$.
The following are true:
\begin{enumerate}
\item {\emph{Scalarization:}} $\Lc_{dev}(\g,\h) = \min_{\alpha\in S_{dev}}\left\{ \sqrt{\alpha^2 + \sigma^2}\|\g\| - \alpha\cdot\dt(\h,\Cc) \right\}.$

\item {\emph{Deterministic result:}} If 
\begin{align}\label{eq:ass_dev}%\eqname{eq:ass_dev}
\frac{\sigma\cdot\dt(\h,\Cc)}{\sqrt{\| \g \|^2-\dt^2(\h,\Cc)}}\notin S_{dev},
\end{align}
then,
\begin{align*}
\Lc_{dev}(\g,\h) = \sqrt{ (1\pm\delta_{dev})^2C_{dev}^2 + \sigma^2 }\|\g\| - (1\pm\eps)C_{dev}\dt(\h,\Cc).
\end{align*}
\item {\emph{Probabilistic result:}} Assume $(1-\eps_L)m>\DC>\eps_L m$, for some $\epsilon_0>0$ and set
\begin{align*}
C_{dev} = \sigma\sqrt{\frac{\DC}{m-\DC}}.
\end{align*}
Then, for all $\delta_{dev}>0$ there exists $t>0$ and $c_1,c_2>0$ such that,
\begin{align}\label{eq:conc_res1}
\Pro\left( \Lc_{dev}(\g,\h) \geq (1+ t )\sigma\sqrt{m-\DC}   \right) \geq 1 - c_1\exp{(-c_2 m)}.
\end{align}

\end{enumerate}
\end{lem}

%
%\subsection{Merging upper and lower bounds}
\subsection{Going Back: From the Key Optimizations to the Squared Error of the LASSO}
Application of Gordon's Lemma to $\Fc(\A,\vb)$ introduced the three key optimizations in Lemma \ref{lemma:applied_Gordon}. Next, in Lemmas \ref{lemma:lowKey}, \ref{lemma:upKey} and \ref{lemma:dev} we carried out the analysis of those problems. Here, we combine the results of the four Lemmas mentioned above in order to evaluate $\Fc(\A,\vb)$ and to compute an exact value for the norm of its optimizer $\hat{\w}(\A,\vb)$. Lemma \ref{thm:unified} below formally states the results of the analysis and the proof of it follows.

\begin{lem} \label{thm:unified}%\eqname{thm:unified}
Assume 
%$(1-\eps_L)m\geq \DC\geq \eps_L\m$
$m\geq \max\left\{ \DC, \DC+\CC \right\} + \eps_L m$ and $\DC\geq\eps_L m $ for some $\eps_L>0$. Also, assume $m$ is sufficiently large and let $\gamma(m,n)=m$ if $\Cc$ is a cone and $\min\{m,\frac{m^2}{n}\}$ else. Then, the following statements are true.
\begin{enumerate}
\item  For any $\eps>0$, there exist constants $c_1,c_2>0$ such that
\begin{align}\label{eq:conobjVal}
\left| {{\Fc}(\A , \vb )} - \sigma\sqrt{m - \DC} \right| \leq \eps\sigma\sqrt{m - \DC}.
\end{align}
with probability $1-c_1\exp(-c_2\gamma(m,n))$.
\item For any $\delta_{dev}>0$ and all $\w\in{\Cc}$ satisfying
\begin{align}\label{eq:assSr}%\eqname{eq:assSr}
\left| {\|\w\|}-   {\sigma} \sqrt{\frac{\DC}{m-\DC}} \right|\geq \delta_{dev}{\sigma} \sqrt{\frac{\DC}{m-\DC}},
\end{align}
there exists constant $t(\delta_{dev})>0$ and $c_1,c_2>0$ such that
\begin{align}\label{eq:con_deviation}
%\| \y-\A\x \| \geq {F}(\A , \vb ) + \eps_3\sigma\sqrt{m},
\| \A\w - \sigma\vb \| + \max_{\s\in\Cc} \s^T\w  \geq {\Fc}(\A , \vb ) + t \sigma\sqrt{m},
\end{align}
with probability $1-c_1\exp(-c_2\gamma(m,n))$.
\item For any $\delta>0$, there exist constants $c_1,c_2>0$ such that
\begin{align}\label{eq:connormW}
\left| \|\hat{\w}(\A,\vb)\| - \sigma\sqrt{\frac{\DC}{m-\DC}} \right|  \leq \delta\sigma\sqrt{\frac{\DC}{m-\DC}}.
\end{align}
with probability $1-c_1\exp(-c_2\gamma(m,n))$.
\end{enumerate}

\end{lem}

% COMMENT
\begin{comment}
\begin{thm}\label{thm:unified} 
Assume the linear regime.
For any $\eps_1,\eps_2>0$, there exist $c_1,c_2>0$ such that with probability $1-c_1\exp(-c_2 n)$,
\begin{align}\label{eq:conobjVal}
\left| \frac{{\Fc}_c(\A , \vb )}{\sigma\sqrt{m - \Delxf} } - 1 \right| \leq \eps_1,
\end{align}
and
\begin{align}\label{eq:con_normW}
\left| \hat\zeta_c(\A,\vb) - \frac{\Delxf}{m-\Delxf} \right|  \leq \eps_2.
\end{align}
Furthermore, for all $\x\in\hat{\Cc}(\x_0)$ and $\delta_{dev}>0$ such that 
\begin{align}\label{eq:assSr}\eqname{eq:assSr}
\left| \frac{\|\x-\x_0\|}{\sigma} -   \sqrt{\frac{\Delxf}{m-\Delxf}} \right|\geq \delta_{dev}\sqrt{\frac{\Delxf}{m-\Delxf}},
\end{align}
there exist $\eps_3>0$ and $c_3,c_4>0$ such that
\begin{align}\label{eq:con_deviation}
\Pro\left( \| \y-\A\x \| \geq {\Fc}_c(\A , \vb ) + \eps_3\sigma\sqrt{m} \right) \geq 1 -c_3\exp(-c_4n).
\end{align}
\end{thm}
\end{comment}
% COMMENT
%\subsubsection{Proof of Theorem \ref{thm:unified}}

%
\begin{proof}

We prove  each one of the three statements of Theorem \ref{thm:unified} sequentially. Assume the regime where $m\geq \max\left\{ \DC, \DC+\CC \right\} + \eps_L m$ and $\DC\geq\eps_L m $ for some $\eps_L>0$ and also $m$ is sufficiently large.

\noindent$\text{1.~~}$ \emph{Proof of \eqref{eq:conobjVal}:}
Consider any $\eps'>0$.
First, we establish a high probability lower bound for $\Fc(\A,\vb)$. 
From  Lemma \ref{lemma:lowKey}, %there exist $c'_1,c'_2>0$ such that,
\begin{align}\notag
\Lc(\g,\h) \geq (1-\eps') \sigma \sqrt{ m - \DC },
\end{align}
with probability  $1-\exp(-\order{m}).$ %$1 - c'_1\exp(-c'_2 m)$. 
%Also, for $h\sim\mathcal{N}(0,1)$ and all $\eps>0$, $\mathbb{P}(h\leq \eps\sqrt{m})\geq 1 - 2\exp{(-\frac{\eps^2 m}{2})}$. Combine this with \ref{eq:pp11} and apply Corollary \ref{cor:low} 
Combine this with the first statement of Lemma \ref{lemma:applied_Gordon}
to conclude that
\begin{align}\label{eq:pplow}
\Fc(\A,\vb)\geq  (1-\eps')\sigma\sqrt{m-\DC} - \eps'\sigma\sqrt{m},
\end{align}
with the same probability.

Similarly, for a high probability upper bound for $\Fc(\A,\vb)$ we have from Lemma \ref{lemma:upKey}, that %there exist $c''_1,c''_2>0$ such that,
\begin{align}\notag
\Uc(\g,\h) \leq (1+\eps') \sigma \sqrt{ m - \DC },
\end{align}
with probability $1-\exp\left(-\order{\gamma(m,n)}\right)$.
  %$1 - c''_1\exp(-c''_2 m)$.
Combine this with the second statement of Lemma \ref{lemma:applied_Gordon}
to conclude that
\begin{align}\label{eq:ppup}
\Fc(\A,\vb)\leq  (1+\eps')\sigma\sqrt{m-\DC} + \eps'\sigma\sqrt{m},
\end{align}
with the same probability. %$1 - c_1''\exp(-c_2'' m)$.
To conclude the proof of \eqref{eq:conobjVal} fix any positive constant $\eps>0$, and observe that by choosing $\eps'=\eps\frac{\sqrt{\eps_L}}{1+\sqrt{\eps_L}}$ in \eqref{eq:pplow} and \eqref{eq:ppup} we ensure that
$\eps'\left( 1 + \frac{\sqrt{m}}{\sqrt{m-\DC}} \right) \leq \eps$. It then follows from  \eqref{eq:pplow} and \eqref{eq:ppup} that there exist $c_1,c_2>0$ such that
\begin{align}\label{eq:state1}
\left| \frac{{\Fc}(\A , \vb )}{\sigma\sqrt{m - \DC}} - 1  \right| \leq \eps,
\end{align}
with probability $1 - c_1\exp\left(-c_2{\gamma(m,n)}\right)$.
%
% It will follow from a combination of Corollary \ref{cor:low}, Lemma \ref{lemma:lowKey} and the fact that for $h\sim\mathcal{N}(0,1)$ and all $\eps>0$, $\mathbb{P}(h\leq \eps\sqrt{m})\leq 1 - 2\exp{(-\frac{\eps^2 m}{2})}$. In particular, for any $\eps>0$
%\begin{align}
%\mathbb{P}\left( \Fc_c(\A,\vb)\geq  (1-\eps)\sigma\sqrt{m-\DC}) - \eps\sigma\sqrt{m} \right) &\geq 2 \mathbb{P}\left( \Lc(\g,\h) - h\sigma \geq  (1-\eps)\sigma\sqrt{m-\DC}) - \eps\sigma\sqrt{m} \right) - 1 \notag\\
%&
%\end{align}

\noindent$\text{2.~~}$ \emph{Proof of \eqref{eq:con_deviation}:}
Fix any $\delta_{dev}>0$. In accordance to its definition in previous sections define the set $$
S_{dev}=\left\{ \ell~ | ~ \left| \ell - \sigma\sqrt{\frac{\DC}{m-\DC}} \right| \leq \delta_{dev}\sigma\sqrt{\frac{\DC}{m-\DC}}\right\}.
$$
Clearly, for all $\w$ such that $\|\w\|\in S_{dev}$ we have,
\begin{align*}
\| \A\w - \sigma\vb \| + \max_{\s\in\Cc} \s^T\w \geq \min_{\|\w\|\in S_{dev}}\left\{ \|\A\w-\sigma\vb\| + \max_{s\in\Cc} s^T\w \right\}.
\end{align*}
Combining this with the third statement of Lemma \ref{lemma:applied_Gordon}, it suffices for the proof of \eqref{eq:con_deviation} to show that there exists constant $t(\delta_{dev})>0$ such that
\begin{align}\label{eq:2showKey}
\Lc_{dev}(\g,\h) \geq \Fc(\A,\vb) + 2t\sigma\sqrt{m},
\end{align}
with probability $1-\exp\left(-\order{m}\right)$. 

To show \eqref{eq:2showKey}, start from Lemma \ref{lemma:dev} which gives that here exists $t'(\delta_{dev})>0$, such that
\begin{align}\label{eq:61}
\Lc_{dev}(\g,\h) \geq (1+t')\sigma\sqrt{m-\DC},
\end{align}
with probability $1-\exp(-\order{m}).$ 
Furthermore, from the first statement of Lemma \ref{thm:unified}, 
\begin{align}\label{eq:62}
\Fc(\A,\vb)\leq (1+\frac{t'}{2})\sigma\sqrt{m-\DC},
\end{align}
with probability $1-\exp\left(-\order{\gamma(m,n)}\right)$. Finally, choose $t=\frac{t'}{4}\sqrt{\eps_L}$ to ensure that 
\begin{align}\label{eq:63}
2t\sigma\sqrt{m} \leq \frac{t'}{2}\sigma\sqrt{m-\DC}.
\end{align}
Combine \eqref{eq:61}, \eqref{eq:62} and \eqref{eq:63} to conclude that \eqref{eq:2showKey} indeed holds with the desired probability.

\noindent$\text{3.~~}$ \emph{Proof of \eqref{eq:connormW}:}
The third statement of Lemma \ref{thm:unified} is a simple consequence of its second statement. Fix any $\eps>0$. The proof is by contradiction. Assume that $\hat{\w}(\A,\vb)$ does not satisfy \eqref{eq:connormW}. It then satisfies \eqref{eq:assSr} for $\delta_{dev}=\eps$. Thus, it follows from the second statement of Lemma \ref{thm:unified}, that there exists $t(\eps)>0$ such that
\begin{align}
\Fc(\A,\vb) \geq \Fc(\A,\vb) + t\sigma\sqrt{m},
\end{align}
with probability $1-\exp\left(-\order{\gamma(m,n)}\right)$. This is a contradiction and completes the proof.

\end{proof}

\section{The NSE of the C-LASSO}\label{sec:ConstrainedLASSO}
%\eqname{sec:ConstrainedLASSO}
In this section, we  prove the second statement of Theorem \ref{thm:CLASSO}, namely \eqref{eq:cp1}. We restate the theorem here for ease of reference.

\CLASSO*

First, in Section \ref{sec:approxCLASSO} we focus on the approximated C-LASSO  and prove that its NSE concentrates around $\frac{\Delxf}{m-\Delxf}$ for arbitrary values of $\sigma$. Later in Section \ref{sec:originalCLASSO}, we use that result and fundamental properties of the approximated problem to prove \eqref{eq:cp1}, i.e. that the NSE of the original problem concentrates around the same quantity for small enough $\sigma$.

%\subsection{Tangent Cone and the Cone of Subdifferentials}
%In the following, let $f(\cdot)$ be a convex function and $\x_0\in\text{dom} f\subset\mathbb{R}^n$. 
%The subdifferential of  $f(\cdot)$ at $\x_0$ is denoted by $\paf$. $\partial f(\x_0)$ is always a closed convex set \cite{Roc70} and for all $\s\in\paf$,
%$$
%f(\x)\geq f(\x_0) + \s^T(\x-\x_0), ~~\text{for all } \x\in\text{dom}f.
%$$
%

\subsection{Approximated C-LASSO Problem}\label{sec:approxCLASSO}

Recall the definition of the approximated C-LASSO problem in \eqref{eq:foC_w}.
As it has been argued previously, this is equivalent to the generic problem \eqref{eq:com} with $\Cc=\text{cone}\{\paf\}$. Hence, to calculate its NSE we will simply apply the results we obtained throughout Section \ref{sec:KeyIdeas}. We first start by mapping the generic formulation in Section \ref{sec:KeyIdeas} to the C-LASSO.

\begin{lem}\label{lemma:gen2CLASSO}
Let $\Cc = \text{cone}\{\paf\}$. Then,
\begin{itemize}
\item $\DC = \Delxf$,
\item $\corr(\h,\Cc) = 0$, for all $\h\in\mathbb{R}^n$,
\item $\CC = 0$.
\end{itemize}
\end{lem}
\begin{proof}
The first statement follows by definition of the quantities involved. The second statement is a direct consequence of Moreau's decomposition theorem (Fact \ref{more}) applied on the closed and convex cone $\text{cone}\{\paf\}$. The last statement follows easily after taking expectation in both sides of the equality in the second statement.
\end{proof}

With this mapping, we can directly apply Lemma \ref{thm:unified}, where $\Cc$ is a cone, to conclude with the desired result. The following corollary summarizes the result.
%\chris{I am changing the assumption on the regime, to be in accordance to Lemma \ref{thm:unified}.}
\begin{cor} \label{thm:fo_CLASSO}Assume 
$(1-\eps_L)m \geq \Delxf \geq \eps_L m $, 
 for some $\eps_L>0$. Also, assume $m$ is sufficiently large.
Then, for any constants $\eps_1,\eps_2>0$, there exist constants $c_1,c_2>0$ such that with probability $1-c_1\exp(-c_2m)$,
\begin{align*}%\label{eq:con_objVal}
\left| \frac{\Fc_c(\A , \vb )}{\sigma\sqrt{m - \Delxf} } - 1 \right| \leq \eps_1,
\end{align*}
and
\begin{align*}%\label{eq:con_normW}
\left| \frac{\|\hat\w_c(\A,\vb)\|^2}{\sigma^2} - \frac{\Delxf}{m-\Delxf} \right|  \leq \eps_2.
\end{align*}
\end{cor}

\subsection{Original C-LASSO  Problem}\label{sec:originalCLASSO}

In this section we prove \eqref{eq:cp1}. For the proof we rely on Corollary \ref{thm:fo_CLASSO}. First, we require the introduction of some useful concepts from convex analysis.

\subsubsection{Tangent Cone and Cone of the Subdifferential} 
Consider any \emph{convex} set  $\Cc\subset\mathbb{R}^n$ and $\x^*\in\Cc$. We define the set of feasible directions in $\Cc$ at $\x^*$ as 
$${F}_\Cc(\x^*) := \left\{ \ub~ | ~ (\x^* + \ub) \in \Cc \right\}. $$
The tangent cone of $\Cc$ at $\x^*$ is defined as %$f(\cdot)$ at $\x_0$ is defined as
\begin{align}\nn
%\Tc_\Cc(\x^*) := \text{cone}\left(\left\{ \ub ~|~ (\x^*+\ub)\in\Cc \right\}\right).
\Tc_\Cc(\x^*) := \text{Cl}\left( \text{cone}( {F}_\Cc(\x^*) ) \right),
\end{align}
where $Cl(\cdot)$ denotes the closure of a set.
By definition, tangent cone $\Tc_\Cc(\x^*)$ and feasible set $F_\Cc(\x^*)$ should be \emph{close} to each other around a small neighborhood of $0$. The following proposition is a corollary of Proposition F.1 of \cite{Oymak} and shows that the elements of tangent cone, that are close to the origin, can be \emph{uniformly} approximated by the elements of the feasible set.

\begin{propo}[Approximating the tangent cone, \cite{Oymak}]\label{prop:F2T}
 Let $\Cc$ be a closed convex set and $\x^*\in\Cc$. For any $\delta>0$, there exists $\eps>0$ such that
\beq\notag
\dt(\ub,F_\Cc(\x^*))\leq \delta\|\ub\|,
\eeq
for all $\ub\in \Tc_\Cc(\x^*)$ with $\|\ub\|\leq\eps$.
\end{propo}
Assume $\Cc$ is the descent set of $f$ at $\x_0$, namely, $\Cc=\left\{ \x ~|~ f(\x)\leq f(\x_0)  \right\}$ for some convex function $f(\cdot)$. In this case, we commonly refer to $\Tc_{\Cc}(\x_0)$ as the ``tangent cone of $f(\cdot)$ at $\x_0$" and denote it by $\Tc_{f}(\x_0)$.
Under the condition that $\x_0$ is not a minimizer of $f(\cdot)$, the following lemma relates $\Tc_{f}(\x_0)$ to the cone of the subdifferential.

\begin{lem} [\cite{Roc70}]\label{tsg}
Assume $f(\cdot):\R^n\rightarrow\R$ is convex and $\x_0\in\R^n$ is not a minimizer of it . 
%Consider the set $\Cc=\left\{ \x ~|~ f(\x)\leq f(\x_0)  \right\}$. 
Then,
\beq
(\Tc_f(\x_0))\pol=\text{cone}(\pa f(\x_0))\nn.
\eeq
\end{lem}

%
%%For a convex function $f(\cdot)$, we denote $\Tc_f(\x_0)$ the tangent cone of the set $\Cc = \left\{ \x ~| ~ f(\x) \leq f(\x_0) \right\}$ at $\x_0$. 
%The tangent cone relates to the cone of subdifferentials as follows:
%\begin{lem} [\cite{Roc70}]\label{tsg}
%Assume $f(\cdot):\R^n\rightarrow\R$ is convex and $\x_0\in\R^n$ is not a minimizer of it .Then,
%\beq
%(\Tc_f(\x))\pol=\text{cone}(\pa f(\x_0))\nn
%\eeq
%\end{lem}
%
%Lemma \ref{prop:F2T} below, fixes a set $\Cc$ and a point $\x_0\in\Cc$, and provides an approximation of the set of feasible directions ${F}_\Cc(\x_0) = \left\{ \ub~ | ~ (\x_0 + \ub) \in \Cc \right\} $ via the tangent cone $\Tc_\Cc(\x_0)=Cl(\text{cone}(F_\Cc(\x_0)))$, which is a corollary of Proposition F.1 in \cite{Oymak}.  $Cl(\cdot)$ denotes the closure of a set.
%\begin{lem}[Approximating the feasible set, \cite{Oymak}]\label{prop:F2T}
% Let $\Cc$ be a convex set and $\x_0\in\Cc$. For any $\delta>0$, there exists $\eps>0$ such that
%\beq\notag
%\dt(\w_c,F_\Cc(\x_0))\leq \delta\|\w_c\|,
%\eeq
%for all $\w_c\in \Tc_\Cc(\x_0)$ with $\|\w_c\|\leq\eps$.
%\end{lem}

\subsubsection{Proof of Theorem \ref{thm:CLASSO}: Small $\sigma$ regime}
We  prove here the second part of Theorem \ref{thm:CLASSO}, namely \eqref{eq:cp1}. For a proof of \eqref{eq:cp2} see Section \ref{any sigma}. For the purposes of the proof, we will use $\Cc=\{\x\big|f(\x)\leq f(\x_0)\}$. Recall that we denote the minimizers of the C-LASSO and approximated C-LASSO by $\w_c^*$ and $\hat{\w}_c$, respectively. Also, for convenience denote
 $$\eta_{c} = \frac{\Delxf}{m-\Delxf}.$$ 
 %, where, recall that  the noise is $\z=\sigma\vb\sim\Nn(0,\sigma^2\Iden)$. 
Recalling the definition of the approximated C-LASSO problem in \eqref{eq:foC_w}, we may write
\begin{align*}
\hat{\w}_{c} &= \arg\min_{\w}\left\{~\|\A\w-\sigma\vb\| + \max_{\la\geq 0} \la \hat{f}_p(\w) ~\right\} \\
&= \arg\min_{\w}\left\{~\|\A\w-\sigma\vb\| + \max_{\s\in\text{cone}(\paf)} \s^T\w ~\right\}\\
&= \arg\min_{\w\in\Tc_\Cc(\x_0)}~\|\A\w-\sigma\vb\|,
\end{align*}
where for the last equality we have used Lemma \ref{tsg}. Hence,
\beq\label{eq:inTan}
\hat{\w}_c\in\Tc_\Cc(\x_0).
\eeq
 At the same time, clearly,
\begin{align}\label{eq:inFeas}
\w^*_c\in F_\Cc(\x_0).
\end{align}
After Corollary \ref{thm:fo_CLASSO}, $\|\hat{\w}_c\|^2$ concentrates around $\sigma^2\eta_c$. We will argue that, in the small noise regime, we can translate our results to the original problem in a smooth way.
Assume that the statements of Corollary \ref{thm:fo_CLASSO}, hold with high probability for some arbitrary $\eps_1,\eps_2>0$.
It suffices to prove that for any $\eps_3>0$ there exists $\sigma_0>0$ such that 
\begin{align}\label{eq:2prove_43}
\left|\frac{\|\w_c^*\|^2}{\sigma^2} - \eta_c \right| \leq \eps_3,
\end{align}
for all $\sigma<\sigma_0$.
To begin with, fix a $\delta>0$, the value of which is to be determined later in the proof. As an immediate implication of Proposition \ref{prop:F2T},  there exists $\sigma_0$ such that
\begin{align}
\dt(\w , F_\Cc(\x_0)) \leq \delta\| \w \|\label{smalldist}
\end{align}
for all $\w\in \Tc_\Cc(\x_0)$ satisfying $\|\w\|\leq C=C(\sigma_0,\eps_2) := \sigma_0\sqrt{(1+\eps_2)\eta_c}$.

Now, fix any $\sigma < \sigma_0 $.
% and consider the LASSO problem with noise \samet{check $\z$}$\z=\sigma\vb\sim\Nn(0,\sigma^2\Iden)$. 
We will make use of the fact that the following three events hold with high probability.
\begin{itemize}
\item Using Corollary \ref{thm:fo_CLASSO}, with high probability $\hat{\w_c}$ satisfies,
\beq
\|\hat{\w}_c\|\leq \sigma\sqrt{(1+\eps_2)\eta_c}\leq C.\label{smallhatw}
\eeq
\item $\A$ has independent standard normal entries. Hence, its spectral norm satisfies $\|\A\|_2\leq 2(\sqrt{n}+\sqrt{m})$ with probability $1-\exp(-\order{\max\{m,n\}})$, \cite{Vers}. 
\item Using \eqref{eq:con_deviation} of Lemma \ref{thm:unified} with $\Cc=\text{cone}(\paf)$, there exists a constant $t=t(\eps_3)$ so that for all $\w$ satisfying $|\frac{\|\w\|^2}{\sigma^2}-\eta_c|\geq \eps_3$, we have,
\beq
\|\A\w-\sigma\vb\| + \max_{\s\in\text{cone}(\paf)} \s^T\w \geq \Fc_c(\A,\vb) +  t(\eps_3)\sigma\sqrt{m}\label{statement3}.
\eeq
\end{itemize}
Consider the projection of $\hat{\w}_c$ on the set of feasible directions $F_\Cc(\x_0)$,
\beq
\p(\hat{\w}_c): = \bu(\hat{\w}_c,F_\Cc(\x_0)) = \hat{\w}_c - \mathbf\Pi(\hat{\w}_c, F_\Cc(\x_0))\label{projhatw}.
\eeq

First, we show that $\| \A\p(\hat{\w}_c) - \sigma\vb \|$ is not much larger than the objective of the approximated problem, namely $\Fc_c(\A,\vb)$. Indeed,
\begin{align}\label{eq:o1}
\| \A\p(\hat{\w}_c) -\sigma\vb \| &\leq \|\A\hat{\w}_c -\sigma\vb\| +\|\A\hat{\w_c}-\A\p(\hat{\w}_c)\|  \notag \\
&\leq\Fc_c(\A,\vb) + \|\A\|_2\dt(\hat{\w}_c,F_\Cc(\x_0)) \notag \\
&\leq \Fc_c(\A,\vb)+ \|\A\|_2 \sigma\delta\sqrt{(1+\eps_2)\eta_c}\nn\\
&\leq \Fc_c(\A,\vb)+ 2(\sqrt{m}+\sqrt{n}) \sigma\delta\sqrt{(1+\eps_2)\eta_c}.
\end{align}
The first inequality is an application of the triangle inequality and the second one follows from \eqref{projhatw}. For the third inequality, we have used \eqref{eq:inTan} and combined \eqref{smalldist} with \eqref{smallhatw}.
% holds, hence, using \eqref{smalldist}
%\begin{align*}
%\dt(\hat\w_c , F_\Cc(\x_0)) &\leq \delta\| \hat\w_c  \| \notag \\
%&\leq \sigma\delta\sqrt{ (1+\eps_2)\eta_c}.
%\end{align*}
%Denote, $\p = \bu(\hat{\w_c},\Fc_\Cc(\x_0)) = \hat{\w_c} - \mathbf\Pi(\hat{\w_c}, \Fc_\Cc(\x_0))$ .Then,
%\begin{align}\label{eq:o1}
%\| \A\p -\sigma\vb \| &\leq \hat f^*_c(\sigma) + \|\A\|\dt(\hat{\w_c},\Fc_\Cc(\x_0)) \notag \\
%&\leq \hat f^*_c(\sigma) + \delta\|\A\| (1+\eps_2)\sigma\sqrt{\eta_c}.
%\end{align}
%The first inequality is an application of the triangle inequality and the Cauchy-Schwarz inequality. The second one follows from \eqref{eq:qwe3}.
 
 Next, we show that if \eqref{eq:2prove_43}  was not true then a suitable choice of $\delta$ would make $\| \A\p(\hat{\w}_c) - \sigma\vb \|$ much larger than the optimal $\Fc_c(\A,\vb)$ than \eqref{eq:o1} allows. Therefore, concluding a desired contradiction. More precisely, assuming \eqref{eq:2prove_43} does not hold, we have
\begin{align}\label{eq:o2}
\| \A\p(\hat{\w}_c) - \sigma\vb \| &\geq 
%\Fc^*_c(f,\A,\vb)=
\|\A\w_c^*-\sigma\vb\| \notag\\
&\geq \Fc_c(\A,\vb)+t(\eps_3)\sigma\sqrt{m}.
\end{align}
The first inequality above follows since $\p(\hat{\w}_c)\in F_\Cc(\x_0)$ and from the optimality of $\w_c^*\in F_\Cc(\x_0)$. To get the second inequality, recall that \eqref{eq:2prove_43} is not true. Also, from \eqref{eq:inFeas}, $\max_{\s\in\text{cone}(\paf)}\s^T\w^*_c = \max_{\s\in(\Tc(\x_0))\pol}\s^T\w^*_c = 0$. Combine these and invoke \eqref{statement3}.
% and the fact that $\w_c^*$ deviates from $\sigma\eta_c$, i.e. \eqref{eq:2prove_43}.

%Combining \eqref{eq:o1} and \eqref{eq:o2} leads to a contradiction. To see this, we write
%\begin{align}
%\| \A\p -\sigma\vb \| &\geq f^*_c(\sigma) \notag \\
%&\geq \hat{f}^*_c(\sigma)+\eps_3\sigma\sqrt{m} \notag\\
%&\geq \| \A\p -\sigma\vb \| + \underbrace{( \eps_3\sigma\sqrt{m}  - \delta\|\A\| (1+\eps_2)\sigma\sqrt{\eta_c})}_{:=\rho}
%\end{align}
% The first inequality above follows since $(\x_0+ \p)\in\Cc$. 
% 
To conclude, choose $\sigma_0$ sufficiently small to ensure $\delta<\frac{t(\eps_3)\sqrt{m}}{2(\sqrt{m}+\sqrt{n}) \sqrt{(1+\eps_2)\eta_c}}$ and combine \eqref{eq:o1} and \eqref{eq:o2} to obtain the following contradiction.
\begin{align*}
\Fc_c(\A,\vb)+ 2(\sqrt{m}+\sqrt{n}) \delta\sigma\sqrt{(1+\eps_2)\eta_c}&\geq \| \A\p(\hat{\w}_c) - \sigma\vb \|\\
&\geq  \Fc_c(\A,\vb)+t(\eps_3)\sigma\sqrt{m}.
\end{align*}
$\sigma_0$ is a deterministic number that is a function of $m,n,f,\x_0,\eps_3$.

\section{$\ell_2$-LASSO: Regions of Operation}\label{sec:ell2LASSO}

The performance of the $\ell_2$-regularized LASSO clearly depends on the particular choice of the parameter $\la$. A key contribution of this work is that we are able to fully characterize this dependence. In other words, our analysis predicts the performance of the $\ell_2$-LASSO estimator for all values $\la\geq 0$. To facilitate our analysis we divide the range $[0,\infty)$ of possible values of $\la$  into three  distinct regions. We call the regions $\mathcal{R}_{OFF}$, $\mathcal{R}_{ON}$ and $\mathcal{R}_{\infty}$.  Each region has specific performance characteristics and the analysis is the same for all $\la$ that belong to the same region. In this Section, we formally define those distinct regions of operation.The analysis of the value of the NSE for each one of them is then deferred to Section \ref{sec:ell2LASSO NSE}.

\subsection{Properties of Distance, Projection and Correlation}
For the purpose of defining the distinct regions of operation of the $\ell_2$-LASSO, it is first important to explore some useful properties of the Gaussian squared distance $\Dlf$, projection $\Plf$ and correlation $\Clf$. Those quantities are closely related to each other and are of key importance to our analysis. We choose to enlist all their important properties in a single Lemma, which serves as a reference for the rest of the Section.
%In the following Lemma we list important properties of them that will be used for defining the distinct regions of operation of the $\ell_2$-LASSO.
\begin{lem}\label{lemma:keyProp}
Consider fixed $\x_0$ and $f(\cdot)$. Let $\partial f(\x_0)$ be a nonempty, compact set of $\mathbb{R}^n$ that does not contain the origin. Then, the following properties hold
\begin{enumerate}

\item $\Dlf + 2\Clf + \Plf = n$.

\item $\Df0) = n$ ,  $\Pf 0) = 0$, and $\Cf 0) = 0$.

\item 
$\lim_{\la\rightarrow\infty} \Dlf=\infty,~\lim_{\la\rightarrow\infty} \Plf=\infty,\text { and } \lim_{\la\rightarrow\infty} \Clf=-\infty.$

\item $\Plf$, $\Clf$ and $\Dlf$ are all continuous functions of $\la\geq 0$.

\item $\Dlf$ is strictly convex and attains its minimum at a unique point. Denote $\labb$ the unique minimizer of $\Dlf$. 

\item $\Plf$ is an increasing function for $\la\geq 0$.

\item $\Dlf$ is differentiable for $\la>0$. For $\la> 0$,
\begin{align}\notag
\frac{d\Dlf}{d\la} = -\frac{2}{\la}\Clf.
\end{align}
For $\la=0$, interpret $\frac{d\Dlf}{d\la}$ as a right derivative.

\item 
\beq\notag
\Clf\begin{cases}
\geq 0 &, \la\in[0,\labb]  \\
 = 0  &,\la = \labb\\
\leq 0 &,\la\in[\labb,\infty)
\end{cases}
\eeq
%$\Clf\geq 0$, for all $\la\in[0,\labb]$, $\Cf\labb)  = 0$ and $\Clf\leq 0$ for all $\la\in[\labb,\infty)$.

\item $\Dlf + \Clf$ is strictly decreasing for $\la\in[0,\labb]$. 

\end{enumerate}
%Further assume that $\Df\labb)< \min\left\{m,n\right\}$. Then,
\end{lem}
Some of the statements in Lemma \ref{lemma:keyProp} are easy to prove, while others require more work. Statements $5$ and $7$ have been recently proved in \cite{McCoy}. We defer the proofs of all statements to Appendix \ref{appendixe}.
% Lemma \ref{lemma:keyProp} will prove to be a very powerful tool towards defining the key values of the parameter $\la$ in the $\ell_2$-LASSO in the next section.

\subsection{Key Values of the Penalty Parameter}
We define three key values of the regularizer $\la$. The main work is devoted to showing that those definitions are well established. 

%We end the section emphasizing that those key parameters can be practically calculated for a wide class of structure inducing functions $f(\cdot)$. 

\subsubsection{$\la_{\text{best}}$}
The first key parameter is $\labb$ which was defined in Lemma \ref{lemma:keyProp} to be the unique minimum of $\Dlf$ over $\la\in[0,\infty).$ The rationale behind the subscript ``best" associated with this parameter is that
% as will be proved in subsequent section, 
the estimation error is minimized for that particular choice of  $\la$. In that sense, $\labb$ is the optimal penalty parameter. We formally prove this fact in Section \ref{sec:ell2LASSO NSE}, where we explicitly calculate the NSE. 
%\chris{Check this sentence:}
 In what follows, we assume that 
$\Df\labb)< m$ to ensure that there exists $\la\geq 0 $ for which estimation of $\x_0$ is robust.
%This assumption \dots
Also, observe that, $\Df\labb)\leq \Df 0)= n$.

%\begin{defn}
%$$\la_{\text{best}} = \operatorname*{argmin}_{\la\geq 0}\Dlf.$$
%\end{defn}
%
%\begin{lem}
%$\la_{\text{best}}$ is unique.
%\end{lem}
%\begin{proof}
%From Lemma \ref{lemma:Dlf} proved in \cite{McCoy}, $\Dlf$ is a strictly convex function of $\la\geq 0$. Thus is has a unique minimum.
%\end{proof}

\subsubsection{$\la_{\text{max}}$}
The second key parameter $\lam$ is defined as the unique $\la\geq\labb$ that satisfies $\Dlf = m$. We formally repeat this definition in the following Lemma.
\begin{lem}\label{lemma:lam}
Suppose $\Df\labb)< m$ and consider the following equation over $\la\geq\labb$:
\begin{align}\label{eq:Dlf=m}
\Dlf = m, \quad\la\geq\labb.
\end{align}
Equation \eqref{eq:Dlf=m} has a unique solution, which we denote $\lam$.
\end{lem} 
\begin{proof}
We make use of Lemma \ref{lemma:keyProp}. First, we show that equation \eqref{eq:Dlf=m} has at most one solution: $\Dlf$ is a strictly convex function of $\la\geq 0$ and thus strictly increasing for $\la\geq\labb$.
Next, we show that \eqref{eq:Dlf=m} has at least one solution. From assumption, $\mathbf{D}(\x_0,\labb) < m$. Also, $\lim_{\la\rightarrow\infty}\mathbf{D}(\x_0,\labb) = \infty$. Furthermore, $\Dlf$ is continuous in $\la$. Combining those facts and using the intermediate value theorem we conclude with the desired result.
\end{proof}

\subsubsection{$\lac$}
The third key parameter $\lac$ is defined to be the unique $\la\leq\labb$ that satisfies $m-\Dlf = \Clf$ when $m\leq n$ or to be 0 when $m>n$. We formally repeat this definition in the following Lemma.
\begin{lem}\label{lemma:lac}
 Suppose $\mathbf{D}(\x_0,\labb)<m$
and consider the following equation over $0\leq\la\leq\labb$:
\begin{align}\label{eq:TheEq}
m-\Dlf=\Clf, \quad 0\leq\la\leq\labb.
\end{align}
\begin{itemize}
\item If $m\leq n$, then \eqref{eq:TheEq} has a unique solution, which we denote as $\lac$.
\item If $m>n$, then \eqref{eq:TheEq} has no solution. Then $\lac = 0$.
\end{itemize}
\end{lem}
\begin{proof}
We repeatedly make use of Lemma \ref{lemma:keyProp}.
For convenience define the function 
$$g(\la) = \Dlf + \Clf,$$
for $\la\in[0,\labb)$.
The function $g(\la)$ has the following properties over $\la\in[0,\labb]$:
\begin{itemize}
\item[-] it is strictly decreasing,
\item[-] $g(0) = n$,
\item[-] $g(\labb) = \Df\labb) <  m$.
\end{itemize}

If $m\leq n$, from the intermediate value Theorem it follows that \eqref{eq:TheEq} has at least one solution. This solution is unique since $g(\la)$ is strictly decreasing.

If $m>n$, since $g(\la)\leq n$ for all $\la\in[0,\labb]$, it is clear that \eqref{eq:TheEq} has no solution.
\end{proof}

%First, we prove that $g(\cdot)$ is a strictly decreasing function of $\la$ in $[0,\labb]$, i.e. for any $0\leq \la_1<\la_2\leq\labb$,
%\begin{align}\label{eq:o67}
%g(\la_1)> g(\la_2).
%\end{align}
%From Lemma \ref{lemma:Dlf}, $\Dlf$ is strictly decreasing for $\la\in[0,\labb]$. Thus,
%\begin{align}\label{eq:o68}
%\Df\la_1) > \Df\la_2).
%\end{align}
%Furthermore, from Lemma \ref{lemma:Plf}, $\Plf$ is an increasing function of $\la$. Thus, 
%\begin{align}\label{eq:o69}
%\Df\la_1) + 2\Cf\la_1)  \geq \Df\la_2) + 2\Cf\la_2).
%\end{align} 
%where we have used  the fact that $\Dlf+2\Clf = n - \Plf$. 
%Combining \eqref{eq:o68} and \eqref{eq:o69}, we conclude with \eqref{eq:o67}, as desired.

%Since $g(\la)$ is strictly monotone, equation \eqref{eq:TheEq} has at most one solution.To complete the proof of the Lemma we use the facts that
%$g(0) = n$ and $ g(\labb) = \Df\labb) <  m$.

%\subsubsection{Calculating $\labb$, $\lam$ and $\lac$}
%\chris{???}
\subsection{Regions of Operation: $\mathcal{R}_{OFF}$, $\mathcal{R}_{ON}$, $\mathcal{R}_{\infty}$ }
Having defined the key parameters $\labb$,$\lac$ and $\lam$, we are now ready to define the three distinct regions of operation of the $\ell_2$-LASSO problem. 

\begin{defn}\label{defn:Ron2}
Define the following regions of operation for the $\ell_2$-LASSO problem:
\begin{itemize}
\item $\mathcal{R}_{OFF} = \left\{ \la~ | ~ 0\leq\la\leq\lac \right\},$
\item $\mathcal{R}_{ON} = \left\{ \la~ | ~ \lac<\la < \lam \right\},$
\item $\mathcal{R}_{\infty} = \left\{ \la ~ | ~ \la\geq\lam \right\}.$
\end{itemize}
\end{defn}

\noindent{\textbf{Remark}}: The definition of $\Rco$ in Definition \ref{defn:Ron2} is consistent to the Definition in \ref{defn:Ron}. In other words, $\lac\leq\la\leq\lam$ if and only if $m\geq \max\{\Dlf, \Dlf+\Clf\}$. This follows after combining Lemmas \ref{lemma:lam} and \ref{lemma:lac} with the Lemma \ref{lemma:props} below.
%
%Based on this definition and combining with the properties of $\labb$,$\lac$ and $\lam$ discussed in the previous chapter, we deduce the following about the three regions of operation.
\begin{lem}
\label{lemma:props}
The following hold:
\begin{enumerate} 
\item $m-\Dlf \leq \Clf$ for all  $\la\in\mathcal{R}_{OFF}$ if $\lac\neq 0$.
\item $m-\Dlf> \max\{0 , \Clf\}$ for all $\la\in\mathcal{R}_{ON}$,
\item $m\leq\Dlf$ for all $\la\in\mathcal{R}_{\infty}$.
\end{enumerate}
\end{lem}
\begin{proof}
We prove the statements in the order they appear. We use Lemma \ref{lemma:keyProp} throughout.

%\noindent{\textbf{First Statement}:}
\noindent{{\emph1.~}}
The function $\Dlf+\Clf$ is strictly decreasing in $[0,\labb]$. Thus, assuming $\lac\neq 0$, $\Dlf+\Clf \geq \Df\lac)+\Cf\lac)=m$ for all $\la\in[0,\lac]$.

%\noindent{\textbf{Second Statement}:}
\noindent{{\emph2.~}}
Since $\Dlf$ is strictly convex, $m-\Dlf$ is strictly concave and has a unique maximum at $\labb$. Therefore, for all $\la\in[\lac,\lam]$, $$m-\Dlf \geq \max\{~ \underbrace{m-\Df\lac)}_{=\Cf\lac)\geq 0} ~,~ \underbrace{m-\Df\lam)}_{=0}~ \}\geq 0.$$ 
Furthermore, $\Dlf+\Clf$ is strictly decreasing in $[0,\labb]$. Thus, $\Dlf+\Clf < \Df\lac)+\Cf\lac)\leq m$ for all $\la\in(\lac,\labb]$. For $\la\in[\labb,\lam)$, we have $m-\Dlf> 0\geq \Clf$.

%\noindent{\textbf{Third Statement}:}
\noindent{{\emph3.~}}
$\Dlf$ is strictly convex. Hence, $m-\Dlf$ is strictly decreasing in $[\labb,\infty)$. This proves that $m-\Dlf\leq m-\Df\lam) = 0$ for all $\la\geq\lam$.

\end{proof}

%
%\begin{proof} Furthermore, $\Clf$ is nonnegative as $\Dlf'\leq 0$ over $\la\leq \labb$. We will first show,
%
%\noindent {\bf{Claim:}} $\Dlf+\Clf$ is strictly decreasing over the region $\Dlf>\Df\labb)$ and $\la\leq \labb$.
%\begin{proof} The assumptions on $\la$ implies $\Dlf'<0$ and $\Dlf$ is strictly decreasing over this region. Now, for any such $\la_1>\la_2$, let $e=\Df\la_1)-\Df\la_2)<0$. Now, from Lemma \ref{}, it is known that $\Plf$ increases as a function of $\la$. This gives,
%\beq
%\Df\la_1)+2\Cf\la_1)\leq \Df\la_2)+2\Cf\la_2)
%\eeq
%Combining both, we obtain $\Cf\la_1)-\Cf\la_2)\leq-\frac{e}{2}$ and,
%\beq
%[\Df\la_1)+\Cf\la_1)]-[\Df\la_2)+\Cf\la_2)]\leq \frac{e}{2}<0
%\eeq
%\end{proof}
%Hence, $\Dlf+\Clf$ is a strictly decreasing continuous function that starts at $\Df0)+\Cf0)=n$ and ends at $\Df\labb)+\Cf\labb)=\Df\labb)<m$. If $n<m$, there is no solution, as the function is decreasing and at $0$ it is less than $n$. Otherwise, there is at least one solution from intermediate value theorem and the solution is unique thanks to strict decrease.
%\end{proof}

\section{The NSE of the $\ell_2$-LASSO}\label{sec:ell2LASSO NSE}
%\eqname{sec:ell2LASSO NSE}
We split our analysis in three sections, one for each of the three regions $\Rcf$, $\Rco$ and $\Rci$. We start from $\Rco$, for which the analysis is similar in nature to C-LASSO.

\subsection{$\mathcal{R}_{ON}$}\label{sec:Ron}
%\eqname{sec:Ron}

%\chris{say briefly why this region is the most important and interesting:
%connection to $\ell_2^2$ regularized LASSO. full mapping....}

In this section we prove Theorem \ref{thm:ell2LASSO} which characterizes the NSE of the $\ell_2$-LASSO in the region $\Rco$. We repeat the statement of the theorem here, for ease of reference.

\RONLASSO*

As usual, we first focus on the approximated $\ell_2$-LASSO problem in Section \ref{sec:ell2Ap}. Next, in Section \ref{sec:ell2Or}, we translate this result to the original $\ell_2$-LASSO problem.

%We first start by proving the second statement which gives the result on the approximated problem. 
\subsubsection{Approximated $\ell_2$-LASSO}\label{sec:ell2Ap}
The approximated $\ell_2$-LASSO problem is equivalent to the generic problem \eqref{eq:com} after taking $\Cc=\la\paf$. Hence, we simply need to apply the result of Lemma \ref{thm:unified}. %Let us first describe the mapping between the generic problem \eqref{} and $\ell_2$-LASSO. We will let $\Cc=\la\paf$. Then, 
with $\DC$ and $\CC$ corresponding to $\Dlf$ and $\Clf$. We conclude with the following result.

\begin{cor} \label{cor:ell2}
Let $m\geq \min_{\la\geq 0}\Dlf$ and 
% and $\la\in\Rco$. 
%Further assume there exists constant $\eps_L$ such that, $\Dlf>\eps_L m$ 
assume there exists constant $\eps_L>0$ such that $(1-\eps_L)m\geq \max\{\Dlf,$ $\Dlf+\Clf\}$ and $\Dlf\geq \eps_Lm$. Further assume that $m$ is sufficiently large. Then, for any constants $\eps_1,\eps_2>0$, there exist constants $c_1,c_2>0$ such that with probability $1-c_1\exp(-c_2\min\{m,\frac{m^2}{n}\})$,
\begin{align}\label{eq:con_objVal}
\left| \frac{\Fc_{\ell_2}(\A , \vb )}{\sigma\sqrt{m - \Dlf} } - 1 \right| \leq \eps_1,
\end{align}
and
\begin{align}\label{eq:con_normW}
\left| \frac{\|\hat\w_{\ell_2}(\A,\vb)\|^2}{\sigma^2} - \frac{\Dlf}{m-\Dlf} \right|  \leq \eps_2.
\end{align}
\end{cor}
%\samet{write the remaining statements and connect to the next proof}
%\end{cor}

\subsubsection{Original $\ell_2$-LASSO: Proof of Theorem \ref{thm:ell2LASSO}}\label{sec:ell2Or}
Next, we use Corollary \ref{cor:ell2} to prove Theorem \ref{thm:ell2LASSO}. To do this, we will first relate $f(\cdot)$ and $\hat{f}(\cdot)$. The following result shows that, $f(\cdot)$ and $\hat{f}(\cdot)$ are close around a sufficiently small neighborhood of $\x_0$.

\begin{propo}[Max formula, \cite{Borwein,Lewis}] \label{max form}Let $f(\cdot):\R^n\rightarrow \R$ be a convex and continuous function on $\R^n$. Then, any point $\x$ and any direction $\vb$ satisfy,
\beq
\lim_{\eps\rightarrow 0^+}\frac{f(\x+\eps\vb)-f(\x)}{\eps}=\sup_{\s\in \pa f(\x)} \li\s,\vb\ri.\nn
\eeq 
In particular, the subdifferential $\pa f(\x)$ is nonempty.
\end{propo}

Proposition \ref{max form} considers a fixed direction $\vb$, and compares $f(\x_0+\eps\vb)$ and $\hat{f}(\x_0+\eps\vb)$. We will need a slightly stronger version which says $\hat{f}(\cdot)$ is a good approximation of $f(\cdot)$ at all directions simultaneously. The following proposition is a restatement of Lemma 2.1.1 of Chapter VI of \cite{Urru}.%While we believe this to be a standard result on convex functions, we provide a proof in  Appendix \ref{sec max form}.
\begin{propo}[Uniform max formula]\label{prop:F3T}
 Assume $f(\cdot):\R^n\rightarrow \R$ is convex and continuous on $\R^n$ and $\x_0\in\R^n$. Let $\hat{f}(\cdot)$ be the first order approximation of $f(\cdot)$ around $\x_0$ as defined in \eqref{eq:fodef}. Then, for any $\delta>0$, there exists $\eps>0$ such that,
\beq
f(\x_0+\w)-\hat{f}(\x_0+\w)\leq \delta \|\w\|,
\eeq
for all $\w\in\R^n$ with $\|\w\|\leq\eps$.
\end{propo}
%\begin{propo} Assume $g(\cdot):\R^n\rightarrow \R$ is a convex function and $\x\in\R^n$. Let $g_x(\cdot)$ be the first order approximation of $g(\cdot)$ around $x$. For any $\delta>0$, there exists $\eps>0$ such that, for all $\w\in\R^n$, $\|\w\|\leq\eps$ we have,
%\beq
%g(\x+\w)-\hat{g}_\x(\x+\w)\leq \delta \|\w\|
%\eeq
%\end{propo}
%%%%%\samet{use a probabilistic explanation}
%\subsubsubsection{Proof of Theorem \ref{thm:ell2LASSO}}
%\begin{lem} [Regularized result]Denote the minimizer of the original problem and the approximated problem by $\w^*$ and $\hat{\w}$ respectively. Assume the following holds for some $\eps,\alpha>0$,
%\begin{itemize}
%\item $\hat{\w}$ satisfies the following,
%\beq
%1-\eps\leq\frac{\|\hat{\w}\|}{\sigma\eta_{LASSO}}\leq 1+\eps
%\eeq
%\item For all $\|\w\|\geq ( 1+\eps) \sigma\eta_{LASSO}$ or $\|\w\|\leq ( 1-\eps) \sigma\eta_{LASSO}$, we have $\Fh(\w)\geq \Fh(\hat{\w})+\alpha\sigma$.\samet{correct F notation}
%\end{itemize}
%Then, there exists a $\sigma_0>0$ such that, whenever $\sigma\leq \sigma_0$, we have,
%\beq
%1-\eps\leq\frac{\|\w^*\|}{\sigma\eta_{LASSO}}\leq 1+\eps\label{we want w^*}
%\eeq
%%and for all $\w$ satisfying $|\|\w\|-\eta_{LASSO}|>\eps$, $f^*(\w)\geq f^*(\w^*)+\delta\eta_{LASSO}$. Finally, let $\mu$ be such that, for all $\|\w\|\leq \eta_{LASSO}$, we have:
%%\beq
%%f_{orig}(\w)-f^*(\w)\leq \mu\|\w\|
%%\eeq
%%Then, for the original problem, when $\w^*$ is sufficiently small, we have:
%%\beq
%%|\|\w_{orig}\|-\eta_{LASSO}|<\eps
%%\eeq
%\end{lem}
%\begin{proof}
 
Recall that we denote the minimizers of the $\ell_2$-LASSO and approximated $\ell_2$-LASSO by $\w_{\ell_2}^*$ and $\hat{\w}_{\ell_2}$, respectively. Also, for convenience denote,
 $$\eta_{\ell_2} = \frac{\Dlf}{m-\Dlf}.$$
% Choose $\sigma_0$ so that, for all $\|\w\|\leq C:=(1+\eps)\sigma_0\eta_{LASSO}$, $f(\w)-\fh(\w)\leq \delta\|\w\|$ where $\delta$ is to be determined. 
% 
After Corollary \ref{cor:ell2}, $\|\hat{\w}_{\ell_2}\|^2$ concentrates around $\sigma^2\eta_{\ell_2}$. We will argue that, in the small noise regime, we can translate our results to the original problem in a smooth way.
Assume that the statements of Corollary \ref{cor:ell2} hold with high probability for some arbitrary $\eps_1,\eps_2>0$.
It suffices to prove that for any $\eps_3>0$ there exists $\sigma_0>0$ such that 
\begin{align}\label{we want w^*}
\left|\frac{\|\w_{\ell_2}^*\|^2}{\sigma^2} - \eta_{\ell_2} \right| \leq \eps_3,
\end{align}
for all $\sigma<\sigma_0$.
 To begin with, fix a $\delta>0$, the value of which is to be determined later in the proof. As an immediate implication of Proposition \ref{prop:F3T},  there exists $\sigma_0$ such that
\begin{align}
f(\x_0+\w)-\hat{f}(\x_0+\w) \leq \delta\| \w \|\label{smalldist2}
\end{align}
for all $\w$ satisfying  $\|\w\|\leq C = C(\sigma_0 , \eps_2) :=\sigma_0\sqrt{(1+\eps_2)\eta_{\ell_2}}$.
Now, fix any $\sigma<\sigma_0$.
We will make use of the fact that the following three events hold with high probability.
\begin{itemize}
\item Using Corollary \ref{cor:ell2}, with high probability $\hat{\w}_{\ell_2}$ satisfies,
\beq
\|\hat{\w}_{\ell_2}\|\leq \sigma\sqrt{(1+\eps_2)\eta_{\ell_2}}\leq C.\label{smallhatw2}
\eeq
\item Using \eqref{eq:con_deviation} of Lemma \ref{thm:unified} with $\Cc=\la\paf$, there exists a constant $t=t(\eps_3)$ so that for any $\w$ satisfying $|\frac{\|\w\|^2}{\sigma^2}-\eta_{\ell_2}|\geq \eps_3$, we have,
\beq
\|\A\w-\sigma\vb\| + \max_{\s\in\la\paf} \s^T\w \geq \Fc_{\ell_2}(\A,\vb) +  t(\eps_3)\sigma\sqrt{m}\label{statement32}.
\eeq
\end{itemize}
Combine \eqref{smallhatw2} with \eqref{smalldist2} to find  that
\begin{align}
\|\A\hat{\w}_{\ell_2}-\sigma\vb\| + \la ( f(\x_0+\hat{\w}_{\ell_2}) - f(\x_0) ) &\leq 
\underbrace{\|\A\hat{\w}_{\ell_2}-\sigma\vb\| + \la ( \hat{f}(\x_0+\hat{\w}_{\ell_2}) - f(\x_0) ) }_{= \Fc_{\ell_2}(\A,\vb)} + \delta\|\hat{\w}_{\ell_2}\|  \nn \\
&\leq {\Fc_{\ell_2}}(\A,\vb) + \delta\sigma\sqrt{(1+\eps_2)\eta_{\ell_2}}. \label{eq:triton}
%  \F(\hat{\w})\leq \Fh(\hat{\w})+\delta'\la\|\hat{\w}\|
\end{align}
%\beq
% \F(\hat{\w})\leq \Fh(\hat{\w})+\delta'\la\|\hat{\w}\|
%\eeq
Now, assume that $\|\w^*_{\ell_2}\|$ does not satisfy \eqref{we want w^*}. Then,
\begin{align}
\|\A\hat{\w}_{\ell_2}-\sigma\vb\| + \la ( f(\x_0+\hat{\w}_{\ell_2}) - f(\x_0) ) &\geq \Fco^*_{\ell_2}(\A,\vb) \label{eq:et1}\\
&\geq \|\A{\w}^*_{\ell_2}-\sigma\vb\| + \la\max_{\s\in\la\paf}\s^T\w^*_{\ell_2}\label{eq:et2}\\
 &\geq \Fc_{\ell_2}(\A,\vb) +  t(\eps_3)\sigma\sqrt{m}.\label{eq:et3}
\end{align}
\eqref{eq:et1} follows from optimality of $\w^*_{\ell_2}$. For \eqref{eq:et2} we used convexity of $f(\cdot)$ and the basic property of the subdifferential that $f(\x_0+\w)\geq f(\x_0) + \s^T\w$, for all $\w$ and $\s\in\paf$. Finally, \eqref{eq:et3} follows from \eqref{statement32}.

To complete the proof, choose $\delta< \frac{t\sqrt{m}}{\sqrt{(1+\eps_2)\eta_{\ell_2}}}$. This will result in contradiction between \eqref{eq:triton} and \eqref{eq:et3}. Observe that, our choice of $\delta$ and $\sigma_0$ is deterministic and depends on $m,\x_0,f(\cdot),\eps_3$.

\subsubsection{A Property of the NSE Formula}

Theorem \ref{thm:ell2LASSO} shows that the asymptotic NSE formula in $\Rco$ is $\frac{\Dlf}{m-\Dlf}$. The next lemma provides a useful property of this formula as a function of $\la$ on $\Rco$.
\begin{lem} $\frac{\Dlf}{m-\Dlf}$ is a convex function of $\la$ over $\Rco$.
\end{lem}
\begin{proof} From \ref{lemma:keyProp}, $\Dlf$ is a strictly convex function of $\la$. Also, $\frac{x}{m-x}$ is an increasing function of $x$ over $0\leq x<m$ and its second derivative is $\frac{m}{(m-x)^3}$ which is strictly positive over $\Rco$. Consequently, the asymptotic NSE formula is a composition of an increasing convex function with a convex function, and is thus itself convex\cite{Boyd}.
\end{proof}

\subsection{$\mathcal{R}_{OFF}$}\label{sec:thisROFF}
Our analysis, unfortunately, does not extend to $\mathcal{R}_{OFF}$, and we have no proof that characterizes the NSE in this regime.
% which is the regime $\Dlf+\Clf>m>\Dlf$. The reason is, in this regime, we are not able to conclude that $\alpha=1$ in the upper bound optimization \eqref{}. 
% 
 On the other hand, our extensive numerical experiments (see Section \ref{sec:num}) show that, in this regime, the optimal estimate $\x^*_{\ell_2}$ of \eqref{ell2model} satisfies $\y=\A\x^*_{\ell_2}$. Observe that, in this case, the $\ell_2$-LASSO reduces to the standard approach taken for the noiseless compressed sensing problem,
 \beq\min f(\x)~~~\text{subject to}~~~\y=\A\x\label{NoiselessCS2}.\eeq
Here, we provide some intuition to why it is reasonable to expect this to be the case. Recall that $\la\in\Rcf$ iff $0\leq\la\leq\lac$, and so the ``small" values of the penalty parameter $\la$ are in $\Rcf$. As $\la$ gets smaller, $\|\y-\A\x\|$ becomes the dominant term, and $\ell_2$-LASSO penalizes this term more. So, at least for sufficiently small $\la$, the reduction to problem \eqref{NoiselessCS2} would not be surprising. Lemma \ref{lemma:Roff} formalizes this idea for the small $\la$ regime.
\begin{lem}\label{lemma:Roff}%\samet{I will fix this lemma}
 Assume $m\leq \alpha n$ for some constant $\alpha<1$ and $f(\cdot)$ is a Lipschitz continuous function with Lipschitz constant $L>0$. Then, for $\la <\frac{\sqrt{n}-\sqrt{m}}{L}(1-o(1))$, the solution $\x^*_{\ell_2}$ of $\ell_2$-LASSO satisfies $\y=\A\x^*_{\ell_2}$,  with probability $1-\exp(-\order{n})$. Here, $o(1)$ term is arbitrarily small positive constant.
\end{lem}
\begin{proof} When $m\leq \alpha n$ for some constant $0<\alpha<1$, $\sqrt{n}-\sqrt{m}=\order{\sqrt{n}}$. Then, from standard concentration results (see \cite{Vers}), with probability $1-\exp(-\order{n})$, minimum singular value $\sigma_{min}(\A)$ of $\A$ satisfies
\beq
\frac{ \sigma_{min}(\A^T)}{\sqrt{n}-\sqrt{m}}\nn\geq 1-o(1). \label{eq:froom}\eeq
Take any $\la<\frac{\sqrt{n}-\sqrt{m}}{L}(1-o(1))$ and let $\p:=\y-\A\x^*_{\ell_2}$. We will prove that $\|\p\|=0$. Denote $\w_2:=\A^T(\A\A^T)^{-1}\p$. 
Using \eqref{eq:froom},  with the same probability,
\beq
\|\w_2\|^2=\p^T(\A\A^T)^{-1}\p\leq \frac{\|\p\|^2}{(\sigma_{min}(\A^T) )^2} \leq \frac{\|\p\|^2}{((\sqrt{n}-\sqrt{m})(1-o(1)))^2},\label{eq:dr1}
\eeq
Define $\x_2=\x^*_{\ell_2}+\w_2$, for which $\y-\A\x_2 = 0$ and consider the difference between the $\ell_2$-LASSO costs achieved by the minimizer $\x^*_{\ell_2}$ and $\x_2$. From optimality of $\x^*_{\ell_2}$, we have,
\begin{align}
0 &\geq  \|\p\| + \la f(\x^*_{\ell_2})-\la f(\x_2) \nn\\
&\geq \|\p\|-\la L\|\x^*_{\ell_2}-\x^*_2\| \label{eq:laos1} = \|\p\|-\la L\|\w_2\| \\
&\geq  \|\p\|(1-\la\frac{ L}{(\sqrt{n}-\sqrt{m})(1-o(1))}) \label{eq:laotz2}.
\end{align}
The inequality in \eqref{eq:laos1} follows from Lipschitzness of $f(\cdot)$, while we use \eqref{eq:dr1} to find \eqref{eq:laotz2}. For the sake of contradiction, assume that $\|\p\|\neq 0$, then \eqref{eq:laotz2} reduces to $0>0$, clearly, a contradiction.
%
%Hence, if $\|\y-\A\x^*_{\ell_2}\|=\|\p\|\neq 0$, then, we need $\la \geq \frac{\sqrt{n}-\sqrt{m}}{L}$.
\end{proof}
For an illustration of Lemma \ref{lemma:Roff}, consider the case where $f(\cdot)=\|\cdot\|_1$. $\ell_1$-norm is Lipschitz with $L=\sqrt{n}$ (see \cite{simultaneous} for related discussion). Lemma \ref{lemma:Roff} would, then, require $\la< 1-\sqrt{\frac{m}{n}}$ to be applicable. As an example, considering the setup in Figure \ref{figCont2}, Lemma \ref{lemma:Roff} would yield $\la<1-\sqrt{\frac{1}{2}}\approx 0.292$ whereas $\lac\approx 0.76$. While Lemma \ref{lemma:Roff} supports our claims on $\Rcf$, it does not say much about the exact location of the transition point, at which the $\ell_2$-LASSO reduces to \eqref{NoiselessCS2}. We claim this point is $\la=\lac$.

% % %
\subsection{$\Rci$}
In this region $m\leq \Dlf$. In this region, we expect \emph{no} noise robustness, namely, $\frac{\|\x^*_{\ell_2}-\x_0\|^2}{\sigma^2}\rightarrow \infty$ as $\sigma\rightarrow 0$. In this work, we show this under a stricter assumption, namely, $m< \Delxf$. See Theorem \ref{not robust} and Section \ref{notenoughm} for more details.
%
%The discussion can be found in Section \ref{notenoughm} along with the related converse results for C-LASSO and $\ell_2^2$-LASSO.
Our proof method relies on results of \cite{McCoy} rather than Gordon's Lemma. On the other hand, we believe, application of Gordon's Lemma can give the desired result for the wider regime $m<\Dlf$. We leave this as a future work.
%Now, we will consider the scenario where the approximated problems objective is $-\infty$ with high probability. Following lemmas shows that, the original problem is not noise robust and the normalized error will be unbounded in this case.

\section{Constrained-LASSO Analysis for Arbitrary $\sigma$}\label{any sigma}

%$S_{up}=\{\w\big|\|\w\|^2\leq 2\bseta(m,\DC)$.
%$S_{dev}=\{\w\big||\|\w\|- \bseta(m,\DC)|>\eps \bseta(m,\DC)\}$.
%$\Lco^*_{dev}=\min_{\w\in S_{dev} }\dots$.
%$\Uc^*=\max_{\mu\leq 1}\min_{\| \w \|\leq C_{up} }\dots$.

%In Sections \ref{sec:ConstrainedLASSO} and \ref{sec:ell2LASSO NSE} we proved that we can precisely predict the NSE of the C-LASSO and NSE of the $\ell_2$-LASSO, when $\sigma\rightarrow 0$. 

In Section \ref{sec:ConstrainedLASSO} we proved the first part of Theorem \ref{thm:CLASSO}, which refers to the case where $\sigma\rightarrow 0$. Here, we complete the proof of the Theorem by showing \eqref{eq:cp2}, which is to say that the worst case NSE of the C-LASSO problem is achieved as $\sigma\rightarrow 0$. In other words,  we prove that our exact bounds for the small $\sigma$ regime upper bound the squared error, for arbitrary values of the noise variance. The analysis relies, again, on the proper application of Gordon's Lemma.

%Until this point, we have rigorously shown that, whether it is the $\ell_2$-LASSO or C-LASSO, as $\sigma\rightarrow 0$, we can precisely predict the normalized-squared-error. In this section, we will show that, for C-LASSO, $\sigma\rightarrow 0$ is the worst case scenario for the NSE, i.e. it maximizes the NSE over all $\sigma>0$. This will prove the third statement of Theorem \ref{thm:CLASSO}. The analysis will again be based on the application of Gordon's Lemma.

%\samet{fill here}This belief is in fact supported by other results in similar nature. For example, in \cite{NoiseSense, Donoho, Oymak}, it has been argued that, the least favorable noise distribution that maximizes the normalized error is achieved as $\sigma\rightarrow 0$ for a simpler problem \eqref{}. Hereby, we will make use of a slightly different analysis to handle C-LASSO and prove Theorem \ref{}. %With a similar intuition, one might expect a similar behavior from the LASSO problem. While, 

\subsection{Notation}

We begin with describing some notation used throughout this section. 
First, we denote $$\dtR:=\dt(\h,\text{cone}(\paf)).$$ 
 Also, recall the definitions of the ``perturbation" functions $f_p(\cdot)$ and $\hat{f}_p(\cdot)$ in \eqref{eq:perdef1} and \eqref{eq:perdef2}.
%Next, let us define the perturbation function around $f$
%%which is
%as
% $$f_p(\w):=f(\x_0+\w)-f(\x_0),$$ 
%% Recall that 
%and, also,
% $$\hat{f}_p(\w):=\hat{f}(\x_0+\w)-f(\x_0)=\max_{\s\in\paf}\s^T\w.$$ 
% Throughout this section, 
Finally, we will be making use of the following functions:
% we will make use of the following functions.
\begin{align}
&\Fco(\w;\A,\vb):=\|\A\w-\sigma\vb\|,\nn\\
&\Lco(\w;\g,\h):=\sqrt{\|\w\|^2+\sigma^2}\|\g\|-\h^T\w\label{OriginalLow},\\
&\fone(\alpha;a,b):=\sqrt{\alpha^2+\sigma^2}a-\alpha b.\label{scalarfunc}
\end{align}
%\samet{comment on $\Lc$ with $f$ vs without}
Using this notation, and denoting the optimal cost of the (original) C-LASSO (see \eqref{cmodel}) as $\Fco_c^*(\A,\vb)$, we write
% the C-LASSO objective is given as,
\beq
\Fco_c^*(\A,\vb)=\min_{f_p(\w)\leq 0}\Fco(\w;\A,\vb) = \Fco(\w^*_c;\A,\vb).\label{eq:erE}
\eeq
%Consider the original C-LASSO problem \eqref{eq:con_defn}. 

%\subsection{Application of Gordon's Lemma to the original problem}
\subsection{Lower Key Optimization}
As a first step in our proof, we apply Gordon's Lemma to the original C-LASSO problem in \eqref{eq:erE}. Recall, that  application of Corollary \ref{cor:low} to the approximated problem resulted in the following key optimization:
\beq\label{eq:appas}
\Lc(\g,\h)=\min_{\hat{f}_p(\w)\leq 0}\left\{\sqrt{\|\w\|^2+\sigma^2}\|\g\|-\h^T\w\right\}=\min_{\hat{f}_p(\w)\leq 0} ~ \Lco(\w;\g,\h).
\eeq
Denote the minimizer of \eqref{eq:appas}, as $\wh_{low}$. Using Corollary \ref{cor:low}, the lower key optimization corresponding to the original C-LASSO has the following form:
\beq
\Lco^*(\g,\h)=\min_{f_p(\w)\leq 0}\left\{\sqrt{\|\w\|^2+\sigma^2}\|\g\|-\h^T\w\right\} = \min_{f_p(\w)\leq 0} ~ \Lco(\w;\g,\h)\label{OriginalLow}.
\eeq
% The details are of course no different to Lemma \ref{lem:
%When Gordon's Lemma is applied to the C-LASSO, we will obtain the following optimization which basically replaces the approximation $\hat{f}$ with $f$ in our previous argument.
%\beq
%\Lco^*(\g,\h)=\min_{f_p(\w)\leq 0}\sqrt{\|\w\|^2+\sigma^2}\|\g\|-\h^T\w\label{OriginalLow}
%\eeq
%On the other hand, we will still make use of the approximation $\hat{f}$ and,
%\beq
%\Lc(\g,\h)=\min_{\hat{f}_p(\w)\leq 0}\sqrt{\|\w\|^2+\sigma^2}\|\g\|-\h^T\w
%\eeq
%which is same as \eqref{eq:com} when we set $\Cc=\text{cone}(\paf)$.
Recall that in both \eqref{eq:appas} and \eqref{OriginalLow}, $\g\in\R^m$ and $\h\in\R^n$. In Lemma \ref{lemma:lowKey} in Section \ref{sec:KeyIdeas} we solved explicitly for the optimizer $\wh_{low}$ of problem \eqref{eq:appas}. In a similar nature, Lemma \ref{classolow} below identifies a critical property of the optimizer $\w^*_{low}$ of the key optimization \eqref{OriginalLow}: $\|\w^*\|$ is no larger than $\|\wh_{low}\|$.
%In a similar nature to the Lemma \ref{lemma:lowKey}, regarding $\Lco^*(\g,\h)$, we have the following results.
%{\bf{Intuition:}} While applying Gordon's Lemma for the lower bound part, we have basically 
%
\begin{lem} \label{classolow}Let $\g\in\R^m,\h\in\R^n$ be given and $\|\g\|>\dtR$. Denote the minimizer of the problem \eqref{OriginalLow} as $\wo_{low}=\wo_{low}(\g,\h)$. Then,
\beq
\frac{\|\wo_{low}\|^2}{\sigma^2}\leq \frac{\dtR^2}{\|\g\|^2-\dtR^2} = \frac{\|\wh_{low}\|^2}{\sigma^2}.
\label{largesiglow}
\eeq
\end{lem}

For the proof of Lemma \ref{classolow}, we require the following result on the tangent cone of the feasible set of \eqref{OriginalLow}.
\begin{lem}\label{tan cone} Let $f(\cdot):\R^n\rightarrow \R$ be a convex function and $\x_0\in\R^n$ that is not a minimizer of $f(\cdot)$. Consider the set $\Cc=\{\w\big|f(\x_0+\w)\leq f(\x_0)\}$. Then, for all $\w^*\in\Cc$,
%the polar cone of the tangent cone of $\Cc$ at $\x$ is given as,
\beq
\Tc_\Cc(\w^*)\pol=\begin{cases}
\text{cone}(\pa f(\x_0 + \w^*)) &\text{if}~f(\x_0 + \w^*)=f(\x_0), \\
\{0\} &\text{if}~f(\x_0 + \w^*)<f(\x_0).\end{cases}
\eeq
\end{lem}
\begin{proof} We need to characterize the feasible set $F_\Cc(\w^*)$. 

Suppose $f(\x_0+\w^*)<f(\x_0)$. Since $f(\cdot)$ is continuous, for all directions $\ub\in\mathbb{R}^n$, there exists sufficiently small $\eps>0$ such that $f(\x_0+\w^*+\eps\ub)\in\Cc$. Hence, $\Tc_\Cc(\w^*)=\text{cone}(\text{Cl}(F_\Cc(\w^*)))=\R^n \implies (\Tc_\Cc(\w^*))\pol = \{0\}$ in this case.

Now, assume $f(\x_0+\w^*)=f(\x_0)$. Then, $F_\Cc(\w^*)=\{\ub\big|f(\x_0+\w^*+\ub)\leq f(\x_0)=f(\x_0+\w^*)\} = F_{\Cc'}(\x_0+\w^*)$, where $F_{\Cc'}(\x_0+\w^*)$ denotes the set of feasible directions in  $\Cc' := \{ \x | f(\x)\leq f(\x_0+\w^*) \}$ at $\x_0+\w^*$. Thus, $\Tc_\Cc(\w^*)=\Tc_{\Cc'}(\x_0+\w^*)= \text{cone}(\pa f(\x_0+\w^*))\pol\nn
$, where the last equality follows from Lemma \ref{tsg}, and the fact that $\x_0+\w^*$ is not a minimizer of $f(\cdot)$ as $f(\x_0)=f(\x_0+\w^*)$.
%which means, it is the descent set of $f(\cdot)$ at $\x$, namely, $F_f(\x)$. Consequently,
%\beq
%\Tc_\Cc(\x)=\text{cone}(\text{Cl}(F_f(\x)))=\Tc_f(\x)=\text{cone}(\pa f(\x))^*\nn
%\eeq
\end{proof}

\begin{proof} [Proof of Lemma \ref{classolow}]

We first show that, $\wo_{low}$ exists and is finite.
From the convexity of $f(\cdot)$, $\hat{f}_p(\w)\leq f_p(\w)$, thus, every feasible solution of \eqref{OriginalLow} is also feasible for \eqref{eq:appas}. This implies that $\Lco^*(\g,\h)\geq \Lc(\g,\h)$. Also, from Lemma \ref{lemma:lowKey}, $\Lc(\g,\h)=\sigma\sqrt{\|\g\|^2-\dtR^2}$. Combining,
\beq
\Lco^*(\g,\h) \geq \sigma\sqrt{\|\g\|^2-\dtR^2}>0.\label{eq:ara}
\eeq
Using the scalarization result of Lemma \ref{lemma:lowKey} with $\Cc=\text{cone}(\paf)$, for any $\alpha\geq 0$,
\beq
\min_{\substack{\hat{f}_p(\w)\leq 0 \\ \|\w\|=\alpha}}\Lco(\w;\g,\h)=L(\alpha,\|\g\|,\dtR).\nn
\eeq
Hence, using Lemma \ref{lemma:pert1} in the appendix shows that, when $\|\g\|>\dtR$,
\beq
\lim_{C\rightarrow\infty}\min_{\substack{
   \|\w\|\geq C \\
   f_p(\w)\leq 0
  }}
  \Lco(\w;\g,\h) =\lim_{C\rightarrow\infty}\min_{\alpha\geq C} L(\alpha,\|\g\|,\dtR)=\infty.\nn
%
%\sqrt{\|\w\|^2+\sigma^2}\|\g\|-\h^T\w+\max_{\la\geq 0}\la \hat{f}_p(\w)\rightarrow\infty
\eeq
Combining this with \eqref{eq:ara} shows that $\Lco^*(\g,\h)$ is strictly positive, and that $\|\wo_{low}\|$ and $\wo_{low}$ is finite.

The minimizer $\wo_{low}$ satisfies the KKT optimality conditions of \eqref{OriginalLow}\cite{Bertse}:
\beq\nn
\frac{\wo_{low}}{\sqrt{\|\wo_{low}\|^2+\sigma^2}}\|\g\|=\h-\s^*,
\eeq
or, equivalently,
\beq
\wo_{low}=\sigma\frac{\h-\s^*}{\sqrt{\|\g\|^2-\|\h-\s^*\|^2}},\label{eq:from}
\eeq
where, from Lemma \ref{tan cone},
\begin{align}
\s^*\in\begin{cases}
 \text{cone}\left(\partial f(\x_0 + \wo_{low})\right) & \text{ if } f_p(\w_{low})=0, \\
\{0\} &\text{ if } f_p(\w_{low})<0.
\end{cases}\label{eq:cases}
\end{align}
%
%Also, $\w_{low}\in \Tc_{f}(\x_0)$ since $f_p(\w_{low})\leq 0$. Based on these, we argue that $\|\w_{low}\|$ is small.
%In particular, solving \eqref{eq:2s} for $\w_{low}$, we obtain,

First, consider the scenario in \eqref{eq:cases} where $f_p(\wo_{low})<0$ and $\s^*=0$.  Then, from \eqref{eq:from}  $\h=c_h\wo_{low}$ for some constant $c_h>0$. But, from feasibility constraints, $\wo_{low}\in \Tc_{f}(\x_0)$, hence, $\h\in \Tc_f(\x_0)\implies \h=\dtR$ which implies equality in \eqref{largesiglow}.

Otherwise, $f(\x_0+\wo_{low})=f(\x_0)$ and $\s^*\in\text{cone}\left(\partial f(\x_0 + \wo_{low})\right)$. For this case, we argue that $\|\h-\s^*\|\leq \|\dtR\|$. 
To begin with, there exists scalar $\theta>0$ such that $\theta\s^*\in\partial f(\x_0 + \wo_{low})$. Convexity of $f(\cdot)$, then, implies that,
\begin{align}\label{eq:11}
f(\x_0+\wo_{low}) = f(\x_0) \geq f(\x_0+\wo_{low}) - \li\theta\s^*,\wo_{low}\ri \implies \li\s^*,\wo_{low}\ri \geq 0.
\end{align}
%%by making use of the fact that $f_p(\w_{low})=f_p(0)=0$. 
%Let $\x_{low}=\x_0+\w_{low}$. Since $\s^*\in \Tc_f(\x_{low})^*$, we may write,
%\beq
%f(\x_0)=f(\x_{low}-\w_{low})=f(\x_{low})\implies -\w_{low}\in \Tc_f(\x_{low})\implies  \li-\w_{low},\s^*\ri\leq 0
%\eeq
%This gives $\li\w_{low},\s^*\ri\geq 0$. 
Furthermore, $\wo_{low}\in \Tc_f(\x_0)$ and $\s_0 := \bu(\h,\text{cone}(\paf))$, thus
\beq\label{eq:12}
\li\wo_{low},\s_0 \ri\leq 0.
\eeq
Combine \eqref{eq:11} and \eqref{eq:12}, and further use \eqref{eq:from} to conclude that
$$
\li\wo_{low},\s^*-\s_0\ri\geq 0\implies \li\h-\s^*,\s^*-\s_0\ri\geq 0.
$$
We may then write,
\beq
(\dtR)^2=\| (\h-\s^*)+(\s^*-\s_0) \|^2 \geq \|\h-\s^*\|^2,
\eeq
%
%Conversely, pick $\s_0=\bu_{\R^+}(\h)$. Since $\w_{low}\in \Tc_f(\x_0)$ and $\s_0\in \Tc^*_f(\x_0)$, this time, we have,
%\beq
%\li\w_{low},\s_0\ri\leq 0\implies \li\w_{low},\s^*-\s_0\ri\geq 0\implies \li\h-\s^*,\s^*-\s_0\ri\geq 0
%\eeq
%This gives,
%\beq
%\dtR^2=\li\h-\s^*+(\s^*-\s_0),\h-\s^*+(\s^*-\s_0)\ri\geq \|\h-\s^*\|^2
%\eeq
and combine with the fact that the function $f(x,y) = \frac{x}{\sqrt{y^2 - x^2}}, x\geq 0,y> 0$ is nondecreasing in the regime $x<y$, to complete the proof.
%
%To conclude, observe that, $\|\h-\la\s^*\|\leq \dtR$ and $\frac{1}{\sqrt{\|\g\|^2-\|\h-\la\s^*\|^2}}\leq \frac{1}{\sqrt{\|\g\|^2-\dtR^2}}$ together implies \eqref{largesiglow}.
\end{proof}

%
%Hence, there exists $\s^*\in\pa f(\x_0+\w_{low})$ and $\la\geq 0$ such that,
%\beq
%\frac{\w_{low}}{\sqrt{\|\w_{low}\|^2+\sigma^2}}\|\g\|=\h-\la\s^*
%\eeq
%\beq
%\frac{\w_{low}}{\sqrt{\|\w_{low}\|^2+\sigma^2}}\|\g\|-\h+\s^*
%\eeq
%where $\s^*\in\pa (\x_0+\w_{low})$. $\s_0=\bu_{\R^+}(\h)$. $\x_{low}=\x_0+\w_{low}$.
%\chris{this seems incorrect :/, $f(\x_0)\geq f(\x_0 + \w_{low})$}
%\beq
%0\geq f(\x_0)-f(\x_{low})=f(\x_{low}-\w_{low})-f(\x_{low})\geq \li-\w_{low},\s^*\ri
%\eeq
% This yields,
%\beq
%\w_{low}=\sigma\frac{\h-\la\s^*}{\sqrt{\|\g\|^2-\|\h-\la\s^*\|^2}}=\Lco(\hat{f},\g,\h)
%\eeq
%We will now argue that, this constraints the length of $\w_{low}$ in terms of $\dtR$. To show this, we first claim that $\dtR\geq \|\h-\la\s^*\|$.
%
%\begin{proof} Let $\bu_\la(\h)=\la\s_0$ for some $\s_0\in\paf$. From fundamental subgradient property, we have,
%\beq
%\li\w_{low},\s^*\ri\geq f(\x_0+\w_{low})-f(\x_0)\geq \li\w_{low},\s_0\ri\implies \li\w_{low},\s^*-\s_0\ri\geq0\implies \li\h-\la\s^*,\s^*-\s_0\ri
%\eeq
%\end{proof}

%The reader will observe that, the right hand side of \eqref{largesiglow} is the NSE we found for the approximated problem. Hence, Lemma \ref{classolow} indicates that, the NSE of the original problem is upper bounded by the NSE of the approximated problem. 

%., takes this observation one step further and gives a deviation guarantee for \eqref{OriginalLow}

\subsection{Upper Key Optimization}
In this section we 
%apply Gordon's Lemma to the dual of \eqref{eq:erE} to 
find a high probability upper bound for $\Fco_c^*(\A,\vb)$. Using Corollary \ref{cor:up} of Section \ref{sec:thisUp}, application of Gordon's Lemma to the dual of the C-LASSO results in the following key optimization:
\beq
\Uco^*(\g,\h)=\max_{\|\mu\|\leq 1}\left\{\min_{\substack{ f_p(\w)\leq 0 \\ \| \w \|\leq C_{up} }} \sqrt{\|\w\|^2+\sigma^2}\mu^T\g-\|\mu\|\h^T\w\right\}, \label{eq:upp1}
\eeq
where $$C_{up}=2\sqrt{\frac{\Delxf}{m-\Delxf}}.$$
%
%Th section deals with finding a tight upper bound on the LASSO objective in a similar nature to Lemma \ref{lemma:upKey}.
%
%
%Let $S_{up}=\{\w\big|\|\w\|\leq \sigma C_{up}\}$ where $C_{up}=2\sqrt{\frac{\Delxf}{m-\Delxf}}$. Application of Gordon's Lemma to the dual of the C-LASSO gives the following.
%\beq
%\Uco^*(\g,\h)=\max_{\mu\leq 1}\min_{\| \w \|\leq C_{up},f_p(\w)\leq 0} \sqrt{\|\w\|^2+\sigma^2}\mu^T\g-\|\mu\|\h^T\w
%\eeq
Normalizing the inner terms in \eqref{eq:upp1} by $\|\mu\|$ for $\mu\neq 0$, this can be equivalently be written as,
\begin{align}
\Uco^*(\g,\h)&=\max_{ \|\mu\|\leq 1}\left\{ \|\mu\|\min_{\substack{ f_p(\w)\leq 0 \\ \| \w \|\leq C_{up} }}\left\{ \sqrt{\|\w\|^2+\sigma^2}\|\g\|-\h^T\w\right\}\right\}\nn\\
&=\max\left\{0,~\min_{\substack{ f_p(\w)\leq 0 \\ \| \w \|\leq C_{up} }} \Lco(\w;\g,\h)\right\}\nn\\
&=\max\left\{0,~\Lco^*_{up}(\g,\h)\right\}\label{lcup},
\end{align}
where we additionally defined 
\beq\label{eq:addD}
\Lco^*_{up}(\g,\h):=\min_{\substack{f_p(\w)\leq 0 \\ \| \w \|\leq C_{up}}} \Lco(\w;\g,\h).
\eeq
 Observe the similarity of the upper key optimization \eqref{lcup} to the lower key optimization \eqref{OriginalLow}.
The next lemma proves that $\Lco^*_{up}(\g,\h)$ and $\Uco^*(\g,\h)$ are Lipschitz functions. 
\begin{lem} [Lipschitzness of $\Uco^*(\g,\h)$] \label{lipschitzlem}$\Lco^*_{up}(\g,\h)$ and,  consequently, $\Uco^*(\g,\h)$ are Lipschitz with Lipschitz constants at most $2\sigma\sqrt{C_{up}^2+1}$.
\end{lem}
\begin{proof}
First, we prove that $\Lco^*_{up}(\g,\h)$ is Lipschitz. Given pairs $(\g_1,\h_1),(\g_2,\h_2)$, denote $\w_1$ and $\w_2$ the corresponding optimizers in problem \eqref{eq:addD}. W.l.o.g., assume that $\Lco^*_{up}(\g_1,\h_1)\geq\Lco^*_{up}(\g_2,\h_2)$. Then, 
\begin{align}
\Lco^*_{up}(\g_1,\h_1)-\Lco^*_{up}(\g_2,\h_2)&=\Lco(\w_1;\g_1,\h_2)-\Lco(\w_2;\g_2,\h_2)\nn\\&\leq \Lco(\w_2;\g_1,\h_1)-\Lco(\w_2;\g_2,\h_2)\nn\\
&=\sqrt{\|\w_2\|^2+\sigma^2}(\|\g_1\|-\|\g_2\|)-(\h_1-\h_2)^T\w_2\nn\\
&\leq \sqrt{\sigma^2C_{up}^2+\sigma^2}\|\g_1-\g_2\|+\|\h_1-\h_2\|\sigma C_{up}\label{use cup},
\end{align}
where, we have used the fact that $\|\w_2\|\leq \sigma C_{up}$. From \eqref{use cup}, it follows that $\Lco^*_{up}(\g,\h)$ is indeed Lipschitz and
\begin{align*}%\label{eq:lip143}
| \Lco^*_{up}(\g_1,\h_1)-\Lco^*_{up}(\g_2,\h_2) | \leq 2\sigma\sqrt{C_{up}^2+1} \sqrt{\|\g_1-\g_2\|^2+\|\h_1-\h_2\|^2}.
\end{align*}
%
% with Lipschitz constant at most $\sigma C_{up}+\sqrt{\sigma^2 C_{up}^2+\sigma^2}\leq 2\sigma\sqrt{C_{up}^2+1}$.
% $\w$'s minimizing $\Lco(\w,\g_1,\h_1)$, $\Lco(\w,\g_2,\h_2)$ over $S_{up}$ by $\w_1,\w_2$ respectively. First, 
% 
To prove that $\Uco^*(\g,\h)$ is Lipschitz with the same constant, assume w.l.o.g that $\Uco^*(\g_1,\h_1)\geq \Uco^*(\g_2,\h_2)$. Then, from \eqref{lcup}, %\chris{check}
%It takes no more than considering all possible sign combinations of $\Lco^*_{up}(\g_1,\h_1)$, $\Lco^*_{up}(\g_2,\h_2)$ in \eqref{lcup}, to write,
%observe that,
\beq\nn
|\Uco^*(\g_1,\h_1)-\Uco^*(\g_2,\h_2)|\leq |\Lco^*_{up}(\g_1,\h_1)-\Lco^*_{up}(\g_2,\h_2)|.
%=|\Lco(\w_1;\g_1,\h_1)-\Lco(\w_2;\g_2,\h_2)|.
\eeq
\end{proof}

\subsection{Matching Lower and Upper key Optimizations}
Comparing \eqref{OriginalLow} to \eqref{lcup}, we have already noted that the lower and upper key optimizations have similar forms. The next lemma proves that their optimal costs match, in the sense that they concentrate with high probability over the same quantity, namely $\E[\Lco^*_{up}(\g,\h)]$.

\begin{lem}\label{relateuplow} Let $\g\sim\Nn(0,\Iden_m),\h\sim\Nn(0,\Iden_n)$ and independently generated.  Assume $(1-\eps_0)m\geq \Delxf\geq \eps_0m$ for some constant $\eps_0>0$ and $m$ sufficiently large. For any $\eps>0$, there exists $c>0$ such that, with probability $1-\exp(-cm)$, we have,
\begin{enumerate}
\item $|\Uco^*(\g,\h)-\E[\Lco^*_{up}(\g,\h)]|\leq\eps\sigma\sqrt{m}$.
\item $|\Lco^*(\g,\h)-\E[\Lco^*_{up}(\g,\h)]|\leq\eps\sigma\sqrt{m}$.
\end{enumerate}
\end{lem}
In Lemma \ref{lipschitzlem} we proved that $\Lco^*_{up}(\g,\h)$ is Lipschitz. Gaussian concentration of Lipschitz functions (see Lemma \ref{fact:lipIneq}) implies, then, that $\Lco^*_{up}(\g,\h)$ concentrates with high probability around its mean $\E[\Lco^*_{up}(\g,\h)]$. According to Lemma \ref{relateuplow}, under certain conditions implied by its assumptions, $\Uco^*(\g,\h)$ and $\Lco^*(\g,\h)$ also concentrate around the same quantity $\E[\Lco^*_{up}(\g,\h)]$. The way to prove this fact is by showing that when these conditions hold, $\Uco^*(\g,\h)$ and $\Lco^*(\g,\h)$ are equal to $\Lco^*_{up}(\g,\h)$  with high probability. Once we have shown that, we require the following result to complete the proof.
%
%Lemma \ref{relateuplow} is in fact a corollary of the following result, which says in short that if two functions are equal with high probability and one of them concentrates, then the second one also concentrates around the same value.

\begin{lem}\label{usefullem} Let $f_1,f_2:\R^n\rightarrow \R$ and $\h\sim \Nn(0,\Iden_n)$. Assume $f_1$ is $L$-Lipschitz and,
$
\Pro(f_1(\g)=f_2(\g))>1-\eps$.
Then, for all $t>0$,
$$\Pro\left(|f_2(\g)-\E[f_1(\g)]|\leq t\right)>1-\eps-2\exp\left(-\frac{t^2}{2L^2}\right).$$
\end{lem}
\begin{proof} From standard concentration result on Lipschitz functions (see Lemma \ref{fact:lipIneq}), for all $t>0$, $|f_1(\g)-\E[f_1(\g)]|<t$ with probability $1-2\exp(-\frac{t^2}{2L^2})$. Also, by assumption $f_2(\g)=f_1(\g)$ with probability $1-\eps$.
% One can then union bound the event $|f_2(\g)-\E[f_1(\g)]|\leq t$ to conclude that it
% If the following events hold, we can conclude,
%\begin{itemize}
%\item $f_2(\g)=f_1(\g)$.
%\item $|f_1(\g)-\E[f_1(\g)]|<t$.
%\end{itemize}
%Each event holds with probabilities $1-\eps,1-2\exp(-\frac{t^2}{2L^2})$. 
%Union bounding, we find the desired result.
Combine those facts to complete the proof as follows,
\begin{align*}
\Pro\left(|f_2(\g)-\E[f_1(\g)]|\leq t\right) &\geq \Pro\left(|f_2(\g)-\E[f_1(\g)]|\leq t ~| ~ f_1(\g) = f_2(\g) \right)\Pro\left(f_1(\g) = f_2(\g) \right)\\
&=\Pro\left(|f_1(\g)-\E[f_1(\g)]|\leq t \right)\Pro\left(f_1(\g) = f_2(\g) \right)\\
&\geq \left(1-2\exp(-\frac{t^2}{2L^2})\right)(1-\eps).
\end{align*}
\end{proof}
Now, we complete the proof of Lemma \ref{relateuplow} using the result of Lemma \ref{usefullem}.
\begin{proof}[Proof of Lemma \ref{relateuplow}]
We prove the two statements of the lemma in the order they appear.

\noindent\emph{1.} First, we prove that under the assumptions of the lemma, $\Uco^*=\Lco^*_{up}$ w.h.p.. By \eqref{lcup}, it suffices to show that $\Lco^*_{up}\geq 0$ w.h.p..  Constraining the feasible set of a minimization problem cannot result in a decrease in its optimal cost, hence,
\beq\label{eq:ine}
\Lco^*_{up}(\g,\h)\geq \Lco^*(\g,\h)\geq \Lc(\g,\h).
\eeq
where recall $\Lc(\g,\h)$ is the lower key optimization of the approximated C-LASSO (see \eqref{eq:appas}). From Lemma \ref{lemma:lowKey}, since $m\geq\Delxf+\eps_0 m$, we have that 
$$
\Lc(\g,\h) \geq (1-\eps)\sigma\sqrt{m-\Delxf} \geq 0,
$$
with $1-\exp(-\order{m})$. Combine this with \eqref{eq:ine} to find that
$\Lco^*_{up}(\g,\h)\geq 0$ or $\Uco^*=\Lco^*_{up}$ with probability $1-\exp(-\order{m})$.
%
%
%assumpe $\|\g\|\geq \dtR$, $\Lc(\g,\h)=\sigma \sqrt{\|\g\|^2-\dtR^2}$ which holds with probability $1-\exp(-\order{m})$, as $m\geq \Delxf+\eps_0 m$. 
Furthermore, from Lemma \ref{lipschitzlem}, $\Lco^*_{up}(\g,\h)$ is  Lipschitz  with constant $L=2\sigma \sqrt{C_{up}^2+1}$.
We now apply Lemma \ref{usefullem} setting $f_1=\Lco^*_{up}(\g,\h)$, $f_2 = \Uco^*$ and $t=\eps\sigma\sqrt{m}$, to find that
%
% Consequently, choosing $t=\eps\sigma\sqrt{m}$ in Lemma \ref{usefullem}, with probability $1-\exp(-\order{m})$, we have,
\beq
|\Uco^*(\g,\h)-\E[\Lco_{up}^*(\g,\h)]|\leq \eps\sqrt{m},\nn
\eeq
with probability $1-\exp(-\order{m})$.
In writing the exponent in the probability as $\order{m}$, we made use of the fact that $C_{up}=2\sqrt{\frac{\Delxf}{m-\Delxf}}$ is bounded below by a constant, since $(1-\eps_0)m\geq \Delxf\geq \eps_0m$.
%
%Let us start from the first statement. $f_1,f_2$ in Lemma \ref{usefullem} will correspond to $\Uc^*,\Lco^*_{up}$ respectively and we will make use of the fact that $\Lco^*_{up}$ is Lipschitz.

\noindent\emph{2.}
As in the first statement, we apply Lemma \ref{usefullem}, this time setting $f_1 = \Lco^*_{up}$, $f_2 = \Lco^*$ and $t=\eps\sigma\sqrt{m}$. The result is immediate after application of the lemma, but first we need to show that $\Lco^*(\g,\h)=\Lco^*_{up}(\g,\h)$ w.h.p.. We will show equivalently that the minimizer $\wo_{low}$  of \eqref{OriginalLow} satisfies $\wo_{low}\in S_{up}$.  From Lemma \ref{classolow}, $\|\wo_{low}\|\leq \frac{\dtR}{\|\g\|-\dtR}$. On the other hand, using standard concentration arguments (Lemma \ref{lemma:conc4}), with probability $1-\exp(-\order{m})$, $\frac{\dtR}{\|\g\|-\dtR}\leq \frac{2\Delxf}{m-\Delxf}=C_{up}$. Combining these completes the proof.
\end{proof}

\subsection{Deviation Bound}
%We now will argue that NSE cannot be more than $\sigma\sqrt{\frac{\Delxf}{m-\Delxf}}$. 
Resembling the approach developed in Section \ref{sec:KeyIdeas}, we show that if we restrict the norm of the error vector $\|\w\|$ in \eqref{eq:erE} as follows
\beq\label{eq:Sdev}
\|\w\|\in S_{dev}:=\left\{\ell \big| \ell \geq(1+\eps_{dev})\sigma\sqrt{\frac{\Delxf}{m-\Delxf}}\right\},
\eeq
then, this results in a significant increase in the cost of C-LASSO. To lower bound the deviated cost, we apply Corollary \ref{cor:dev} of Section \ref{sec:thisDev} to the restricted original C-LASSO, which yields the following key optimization 
\beq
\Lco_{dev}^*(\g,\h)=\min_{\substack{ f_p(\w)\leq 0 \\ \|\w\|\in S_{dev}}}\Lco(\w;\g,\h)\label{optimmr1}.%\sqrt{\|\w\|^2+\sigma^2}\|\g\|-\h^T\w+\la f(\w)
\eeq

%resulting from application of Gordon's Lemma to the restricted original C-LASSO problem deviates from the optimal cost of \eqref{eq:erE}. Formally,
\begin{lem} \label{cordev}Let $\g\sim\Nn(0,\Iden_m),\h\sim\Nn(0,\Iden_n)$. Assume $(1-\eps_L)m>\Delxf>\eps_Lm$ and $m$ is sufficiently large. Then, there exists a constant $\delta_{dev}=\delta_{dev}(\eps_{dev})>0$ such that, with probability $1-\exp(-\order{m })$, we have,
\beq
\Lco^*_{dev}(\g,\h)-\E[\Lco^*_{up}(\g,\h)]\geq \sigma\delta_{dev}\sqrt{m}\label{probdev}.
\eeq
\end{lem}
%\chris{Do we need prob $\order{m }$ or  $\order{\min\{m\}}$? I remember we needed the former in because of upper bound, not deviation. }
As common, our analysis begins with a deterministic result, which builds towards the proof of the probabilistic statement in Lemma \ref{cordev}.   

\subsubsection{Deterministic Result}
For the statement of the deterministic result, we introduce first some notation. In particular, denote
$$
\eta_d :=\sigma\sqrt{\frac{\Delxf}{m-\Delxf}},
$$
and, for fixed $\g\in\R^m,\h\in\R^n$,
$$
\eta_s=\eta_s(\g,\h) := \sigma\frac{\dtR}{\sqrt{\|\g\|^2-\dtR^2}}.
$$
Also, recall the definition of the scalar function $L(\alpha;a,b)$ in \eqref{scalarfunc}.
\begin{lem} \label{classodev}
Let $\g\in\R^m$ and $\h\in\R^n$ be such that $\|\g\|>\dtR$ and  $\eta_s(\g,\h)\leq (1+\eps_{dev}) \eta_d$. Then,
\beq
\Lco_{dev}^*(\g,\h)-\Lco^*(\g,\h)\geq L\left((1+\eps_{dev})\eta_d;\|\g\|,\dtR\right)-L\left(\eta_s(\g,\h);\|\g\|,\dtR\right)\label{simpleeqwithL}
\eeq
\end{lem}
\begin{proof} 
First assume that $\Lco^*_{dev}(\g,\h)=\infty$. Since $\Lco^*(\g,\h)\leq \Lc(\mathbf{0};\g,\h)=\sigma\|\g\|$ and the right hand side of \eqref{simpleeqwithL} is finite, we can easily conclude with the desired result. 

Hence, in the following assume that $\Lco^*_{dev}(\g,\h)< \infty$ and denote $\wo_{dev}$ the minimizer of the restricted problem \eqref{optimmr1}. From feasibility constraints,  we have $f_p(\w_{dev})\leq 0$ and $\|\wo_{dev}\|\in S_{dev}$. Define $\bar{\w}_{dev}=c\wo_{dev}$ where $c:=\frac{\eta_s}{\|\wo_{dev}\|}$. Notice, $\|\wo_{dev}\|\geq(1+\eps_{dev})\eta_d\geq\eta_s(\g,\h)$, thus, $c\leq 1$. Then, from convexity of $f(\cdot)$,
$$ f_p(\bar{\w}_{dev} ) = f_p( c\wo_{dev} ) \leq c f_p(\wo_{dev} ) + (1-c) \underbrace{f_p(\mathbf{0})}_{=0} \leq 0. $$
This shows that $\bar{\w}_{dev}$ is feasible for the minimization \eqref{OriginalLow}. Hence, %Recalling that we use $\wo_{low}$ to denote the optimizer of \eqref{OriginalLow}, we can conclude that
%
%Now, from Lemma \ref{classolow}, $\|\w_{low}\|\leq \eta_s(\g,\h)$ and so 
%\begin{align}
%\|\wh_{dev}\|\geq \|\w_{low}\|.
%\end{align}
\beq
\Lco(\bar{\w}_{dev},\g,\h)\geq 
%\Lco(\w_{low},\g,\h)=
\Lco^*(\g,\h).\nn
\eeq 
%From convexity, this implies, $f_p(\wh_{dev})\leq 0$ as well. 
Starting with this, we write,
\begin{align}
\Lco_{dev}^*(\g,\h)-\Lco^*(\g,\h)&\geq \Lco(\wo_{dev};\g,\h)-\Lco(\bar{\w}_{dev};\g,\h)\nn\\
%&\geq (\sqrt{\|\wo_{dev}\|^2+\sigma^2}-\sqrt{\|\bar{\w}_{dev}\|^2+\sigma^2})-\h^T(\wo_{dev}-\bar{\w}_{dev})\nn\\
&= (\sqrt{\|\wo_{dev}\|^2+\sigma^2}-\sqrt{\|\bar{\w}_{dev}\|^2+\sigma^2})\|\g\|-\h^T(\wo_{dev}-\bar{\w}_{dev})\nn\\
&= (\sqrt{\|\wo_{dev}\|^2+\sigma^2}-\sqrt{\|\bar{\w}_{dev}\|^2+\sigma^2})\|\g\|-(1-c)\h^T\wo_{dev}\label{eq:ena}.
\end{align}
Since, $f_p(\wo_{dev})\leq 0$, $\wo_{dev}\in\Tc_f(\x_0)$. Hence, and using Moreau's decomposition Theorem (see Fact \ref{more}), we have
\begin{align}
\h^T\wo_{dev} &= \li \bu(\h,\Tc_f(\x_0)) , \wo_{dev} \ri + \underbrace{\li \bu(\h,\left(\Tc_f(\x_0)\right)\pol) , \wo_{dev}\ri}_{\leq 0}\nn \\&\leq\dtR\|\wo_{dev}\|.\label{eq:dio}
\end{align}
Use \eqref{eq:dio} in \eqref{eq:ena}, to write
\begin{align*}
\Lco_{dev}^*(\g,\h)-\Lco^*(\g,\h)
&\geq (\sqrt{\|\wo_{dev}\|^2+\sigma^2}-\sqrt{\|\bar{\w}_{dev}\|^2+\sigma^2})-\frac{\|\wo_{dev}\|-\eta_s}{\|\wo_{dev}\|}\dtR\|\wo_{dev}\|\nn\\
&=(\sqrt{\|\wo_{dev}\|^2+\sigma^2}-\sqrt{\eta_s^2+\sigma^2})\|\g\|-(\|\wo_{dev}\|-\eta_s)\dtR\nn\\
&=L(\|\wo_{dev}\|,\|\g\|,\dtR)-L(\eta_s,\|\g\|,\dtR)\nn\\
&\geq L((1+\eps)\eta_d,\|\g\|,\dtR)-L(\eta_s,\|\g\|,\dtR).
%(\sqrt{((1+\eps)\eta_d)^2+\sigma^2}-\sqrt{\eta_s^2+\sigma^2})\|\g\|-((1+\eps)\eta_d-\eta_s)\dtR\\
\end{align*}
The last inequality above follows from the that $L(\alpha;\|\g\|,\dtR)$ is convex in $\alpha$ and minimized at $\eta_s$ (see Lemma \ref{lemma:pert1}) and, also, $\|\w^*_{dev}\|\geq (1+\eps_{dev})\eta_d\geq \eta_s$.\end{proof}

\subsubsection{Probabilistic result}
We now prove the main result of the section, Lemma \ref{cordev}. 

\begin{proof}[Proof of Lemma \ref{cordev}] The proof is based on the results of Lemma \ref{classodev}. First, we show that under the assumptions of Lemma \ref{cordev}, the assumptions of Lemma \ref{classodev} hold w.h.p.. 
In this direction, using standard concentration arguments provided in Lemmas \ref{repeatlem} and \ref{lemma:concentrationALL}, we find that,
\begin{enumerate}
\item $\|\g\|\geq \dtR$,
\item $\frac{\dtR}{\sqrt{\|\g\|^2-\dtR^2}}\leq (1+\eps_{dev})\frac{m}{m-\Delxf}$.
\item For any constant $\eps>0$, 
\beq
|\|\g\|^2-m|\leq \eps m~~\text{ and }~~|\left(\dtR\right)^2-\Delxf|<\eps m\label{anyepsresult},
\eeq
\end{enumerate}
all with probability $1-\exp\left( -\order m \right)$. It follows from the first two statements that Lemma \ref{classodev} is applicable and we can use \eqref{simpleeqwithL}. Thus, it suffinces to find a lower bound for the right hand side of \eqref{simpleeqwithL}.

Lemma \ref{lemma:pert1} in the Appendix analyzes in detail many properties of the scalar function $L(\alpha;a,b)$, which appears in \eqref{simpleeqwithL}. Here, we use the sixth statement of that Lemma (in a similar manner  to the proof of Lemma \ref{lemma:dev}). In particular, apply Lemma \ref{lemma:pert1} with the following mapping:
\beq
\sqrt{m}\iff a,~\sqrt{\Delxf}\iff b,~\|\g\|\iff a',~\dtR\iff b'\nn
\eeq
Application of the lemma is valid since \eqref{anyepsresult} is true, and gives that with probability $1-\exp\left(-\order{m }\right)$,
\beq\nn
L((1+\eps)\eta_d,\|\g\|,\dtR)-L(\eta_s,\|\g\|,\dtR)\geq2\sigma \delta_{dev}\sqrt{m}
\eeq
for some constant $\delta_{dev}$. Combining this with Lemma \ref{classodev}, we may conclude 
\begin{align}
\Lco^*_{dev}(\g,\h)-\Lco^*(\g,\h)\geq2\sigma \delta_{dev}\sqrt{m}.\label{eq:2un1}
\end{align}
 On the other hand, from Lemma \ref{relateuplow},
\beq
|\Lco^*(\g,\h)-\E[\Lco^*_{up}(\g,\h)]|\leq \sigma\delta_{dev}\sqrt{m}\label{eq:2un2}
\eeq
with the desired probability. Union bounding over \eqref{eq:2un1} and \eqref{eq:2un2}, we conclude with the desired result.
\end{proof}

\subsection{Merging Upper Bound and Deviation Results}%\samet{I am still editing below}
%\subsection{Completing the Proof}
This section combines the previous sections and finalizes the proof of Theorem \ref{thm:CLASSO} by showing the second statement.
Recall the definition \eqref{cmodel} of the original C-LASSO problem and also the definition of the set $S_{dev}$ in \eqref{eq:Sdev}.
\begin{lem}
Assume there exists a constant $\eps_L$ such that, $(1-\eps_L)m\geq \Delxf\geq \eps_Lm$. Further assume, $m$ is sufficiently large. 
%For any $\eps>0$, there exist  constants $C=C(\eps,\eps_L)$ and $\sigma_0=\sigma_0(\eps,\eps_L,n)$ such that, with probability $1-\exp(-C^2\min\{m,\frac{m^2}{n}\})$, the following statements hold.
%Under the assumptions of Theorem \ref{thm:CLASSO}, the followings hold.
The following hold:
\begin{enumerate}
\item For any $\eps_{up}>0$, there exists $c_{up}>0$ such that, with probability $1-\exp(-c_{up} m)$, we have,
\beq
\Fco^*_c(\A,\vb)\leq \E[\Lco_{up}^*(f,\g,\h)]+\eps_{up}\sigma\sqrt{m}
\eeq
\item There exists constants $\delta_{dev}>0, c_{dev}>0$, such that, for sufficiently large $m$, with probability $1-\exp(-c_{dev}m)$, we have,
\beq
\min_{\|\w\|\in S_{dev},~f_p(\w)\leq 0} \Fco(\w;\A,\vb)\geq \E[\Lco_{up}^*(f,\g,\h)]+\delta_{dev}\sigma\sqrt{m}
\eeq
\item For any $\eps_{dev}>0$, there exists $c>0$ such that, with probability $1-\exp(-cm)$,
$$\|\x^*_c-\x_0\|^2\leq \sigma^2(1+\eps_{dev})\frac{\Delxf}{m-\Delxf}.$$
\end{enumerate}
\end{lem}
\begin{proof} 
%To prove this, we will make use of Lemma \ref{lemma:applied_Gordon} which relates LASSO objective function $\Fc(\w;\A,\vb)$ to the simplified objective function $\Lco(\w;\g,\h)$.
We prove the statements of the lemma in the order that they appear. 

\noindent{\emph{1.}}
For notational simplicity denote ${\xi}=\E[\Lco_{up}^*(\g,\h)]$. We combine second statement of Lemma \ref{lemma:applied_Gordon} with Lemma \ref{relateuplow}. For any constant $\eps_{up}$, we have,
\begin{align*}
\Pro(\Fco_c^*(\A,\vb) \leq {\xi}+2\sigma\eps_{up}\sqrt{m})&\geq 2\Pro(\Uco^*(\g,\h)+\sigma\eps\sqrt{m}\leq {\xi}+2\sigma\eps_{up}\sqrt{m})-1-\exp(-\order{m})\label{uppart}\\
&= 2\Pro(\Uco^*(\g,\h)\leq {\xi}+\sigma\eps_{up}\sqrt{m})-1-\exp(-\order{m})\\
&\geq 1-\exp(-\order{m}) ,
\end{align*} 
where we used the first statement of Lemma \eqref{relateuplow} to lower bound the $\Pro(\Uco^*(\g,\h)\leq  {\xi}+\sigma\eps_{up}\sqrt{m})$.
 
\noindent{\emph{2.}}
Pick a small constant $\eps>0$ satisfying $\eps<\frac{\delta_{dev}}{2}$ in the third statement of Lemma \ref{lemma:applied_Gordon}. Now, using Lemma \ref{cordev} and this choice of $\eps$, with probability $1-\exp(-\order{m})$, we have,
\begin{align*}
\Pro(\min_{\w\in S_{dev},~f_p(\w)\leq 0}\Fco(\w;\A,\vb) \geq {\xi}+\sigma\frac{\delta_{dev}}{2}\sqrt{m})&\geq 2\Pro(\Lco^*_{dev}(\g,\h)\geq {\xi}+\sigma\delta_{dev}\sqrt{m}-\eps\sigma\sqrt{m})-1-\exp(-\order{m})\\
&\geq 1-\exp(-\order{m}),
\end{align*}
where we used \eqref{probdev} of Lemma \ref{cordev}.

\noindent{\emph{3.}}
Apply Statements 1. and 2. of the lemma, choosing $\eps_{up}=\frac{\delta_{dev}}{8}$. Union bounding we find that
\beq
\Pro(\min_{\w\in S_{dev},~f_p(\w)\leq 0}\Fco(\w;\A,\vb)\geq \Fco_c^*(\A,\vb)+\sigma\frac{\delta_{dev}}{4})\geq 1-\exp(-\order{m}),\nn
\eeq
which implies with the same probability $\|\w^*_c\|\not\in S_{dev}$, i.e., $\|\w^*_c\|\leq (1+\eps_{dev})\sigma\sqrt{\frac{\Delxf}{m-\Delxf}}$.

\end{proof}

\section{$\ell_2^2$-LASSO}\label{justifyl22}
As we have discussed throughout our main results, one of the critical contributions of this paper is that, we are able to obtain a formula that predicts the performance of $\ell_2^2$-penalized LASSO. We do this by relating $\ell_2$-LASSO and $\ell_2^2$-LASSO problems. This relation is established by creating a mapping between the penalty parameters $\la$ and $\tau$. While we don't give a theoretical guarantee on $\ell_2^2$-LASSO, we give justification based on the predictive power of Gordon's Lemma.

\subsection{Mapping the $\ell_2$-penalized to the $\ell_2^2$-penalized LASSO problem }
 Our aim in this section is to provide justification for the mapping function given in \eqref{eq:map}.
The following lemma gives a simple condition for $\ell_2$-LASSO and $\ell_2^2$-LASSO to have the same solution.

\begin{lem}\label{ell22lem}
 Let $\x^*_{\ell_2}$ be a minimizer of $\ell_2$-LASSO program with the penalty parameter $\la$ and assume $\y-\A\x^*_{\ell_2}\neq 0$. Then, $\x^*_{\ell_2}$ is a minimizer of $\ell_2^2$-LASSO with penalty parameter $\tau=\la\cdot\frac{\|\A\x^*_{\ell_2}-\y\|}{\sigma}$.
\end{lem}
\begin{proof} %To establish a mapping between $\ell_2$ and $\ell_2^2$ parameters, consider the optimality conditions of each problem for same $\A$ and $\vb$. Let $\x_{\ell_2}^*=\w_{\ell_2}+\x_0$ denote the minimizer of $\ell_2$-LASSO and assume $\y-\A\x_{\ell_2}^*\neq 0$. 
%For the purpose of the proof drop the argument $\la$ from $\x^*_{\ell_2}(\la)$. 
The optimality condition for the $\ell_2^2$-LASSO problem \eqref{ell2model}, implies the existence of $\s_{\ell_2}\in\pa f(\x_{\ell_2}^*)$ such that,
\beq
\la\s_{\ell_2}+\frac{\A^T(\A\x_{\ell_2}^*-\y)}{\|\A\x_{\ell_2}^*-\y\|}=0\label{l2optimality}
\eeq
On the other hand, from the optimality conditions of \eqref{ell22model}, $\x$ is a minimizer of the $\ell_2^2$-LASSO if there exists $\s\in\pa f(\x)$ \text{such that},
%We describe a mapping between the the 
\beq
\sigma\tau\s+\A^T(\A\x-\y)=0\label{l22optimality}.
\eeq
Observe that, for $\tau=\la\cdot\frac{\|\A\x_{\ell_2}^*-\y\|}{\sigma}$, using \eqref{l2optimality}, $\x_{\ell_2}^*$ satisfies \eqref{l22optimality} and is thus a minimizer of the $\ell_2^2$-LASSO.
\end{proof}
%
%Hence, the mapping is closely related to the value of the $\|\y-\A\x_{\ell_2}\|=\|\z-\A\w_{\ell_2}\|$. We will estimate this quantity in two steps.
%\begin{itemize}
%\item From Lemma \ref{}, it is known that, $\frac{\Fc^*(\A,\vb)}{\sigma}\approx \sqrt{m-\Dlf}$ with high probability.
%\item Assume $\la \hat{f}_p(\w_{\ell_2})\approx \frac{\sigma\Clf}{\sqrt{m-\Dlf}}$. Combined with the first statement, this would yield,
%\begin{align}
%\frac{\|\z-\A\w_{\ell_2}\|}{\sigma}&=\frac{\Fc^*(\A,\vb)-\la \hat{f}_p(\w_{\ell_2})}{\sigma}\approx \sqrt{m-\Dlf}-\frac{\Clf}{\sqrt{m-\Dlf}}\nn\\
%&=\frac{m-\Dlf-\Clf}{\sqrt{m-\Dlf}}=\text{calib}(\la)
% \end{align}
%\end{itemize} 
 
 In order to evaluate the mapping function as proposed in Lemma \ref{ell22lem}, we need to estimate $\|\y-\A\x_{\ell_2}^*\|$. We do this relying again on the approximated $\ell_2$-LASSO problem in \eqref{eq:foEll2_w}. Under the first-order approximation, $\x_{\ell_2}^*\approx \x_0+\wh_{\ell_2}^*:=\hat{\x}_{\ell_2^*}$ and also define, $\hat{f}_p(\w):=\sup_{\s\in\paf}\s^T\w$. 
Then, from \eqref{eq:foEll2_w} and Lemma \ref{thm:unified},
\begin{align}
\|\y-\A\hat{\x}_{\ell_2}^*\|&=\Fc_{\ell_2}^*(\A,\vb)-\la \hat{f}_p(\wh_{\ell_2}^*)\notag\\
&\approx \sigma\sqrt{m-\Dlf}-\la \hat{f}_p(\wh_{\ell_2}^*)\label{eq:eq2}.
\end{align}
%But, Lemma \ref{lemma:lowprob} shows $\Fc^*(\A,\vb)\approx \sigma\sqrt{m-\Dlf}$.
Arguing that, 
\begin{align}\label{eq:belief}
\la \hat{f}_p(\w_{\ell_2}^*)\approx \sigma\frac{\Clf}{\sqrt{m-\Dlf}},
\end{align}
 and substituting this in \eqref{eq:eq2} will result in the desired mapping formula given in \eqref{eq:map}.

In the remaining lines we provide justification supporting our belief that \eqref{eq:belief} is true. Not surprisingly at this point, the core of our argument relies on application of Gordon's Lemma. Following the lines of our discussion in Section \ref{sec:KeyIdeas}, we use the minimizer $\w_{low}^*(\g,\h)$ of the simple optimization \eqref{eq:keylow}
%$\Lc^*(\g,\h)$
 as a proxy for $\w^*_{\ell_2}$ and expect $\hat{f}_p(\w_{\ell_2}^*)$ to concentrate around the same quantity as $\hat{f}_p(\w^*_{low}(\g,\h))$ does. Lemma \ref{lemma:justmap} below shows that
 \begin{align}
\la \hat{f}_p(\w_{low}^*(\g,\h))&=\sigma\frac{\li\Pi(\h,\la\paf),\bu(\h,\la\paf)\ri}{\sqrt{\|\g\|^2-\dt(\h,\la\paf)^2}}\nn\\
& \approx \sigma\frac{\Clf}{\sqrt{m-\Dlf}}\label{finconc}\nn,
 \end{align}
where the second (approximate) equality follows via standard concentration inequalities. 
 
% 
% Hence, we indeed reach the $\text{calib}(\la)$ function we first mention in \eqref{}. So the natural question is, how did we come up with the assumption on $\hat{f}_p(\w_{\ell_2})$? The answer is again Gordon's Lemma. We used the key optimization \eqref{} to do an educated goes on the properties of $\w_{\ell_2}$, namely $\hat{f}(\w_{\ell_2})$. Basically, we will use $\Lc^*(\g,\h)$ (recall \eqref{}) as a surrogate for the LASSO optimization, in a similar manner to what we did to estimate $\|\w\|$.
% 
% Let us again denote the minimizer of the key optimization $\Lc^*(\g,\h)$ by $\w_{low}$. The following lemma shows that $\hat{f}_p(\w_{low})$ concentrates around $\frac{\sigma\Clf}{\sqrt{m-\Dlf}}$.
 \begin{lem}\label{lemma:justmap}
Assume $(1-\eps_L)m\geq\Dlf$ and $m$ is sufficiently large. Then, for any constant $\eps>0$, with probability $1-\exp(-\order{\min\{m,\frac{m^2}{n}\}})$, 
\beq
\big|\la\hat{f}_p(\w^*_{low})-\sigma\frac{\Clf}{\sqrt{m-\Dlf}}\big|<\eps\sqrt{m}\label{yetanothereq1}.
\eeq
\end{lem}
\begin{proof}%\chris{needs cleaning}\samet{I am cleaning}
Recall that $\w^*_{low}(\g,\h)=\sigma\frac{\Pi(\h,\Cc)}{\sqrt{\|\g\|^2-\dt^2(\h,\la\paf)}}$ for $\Cc=\la\paf$. Combining this with Fact \ref{prom}, we obtain,
\beq
\hat{f}_p(\w_{low})=\max_{\s\in\Cc}\li\w_{low},\s\ri=\frac{\li\Pi(\h,\Cc),\bu(\h,\Cc)\ri}{\sqrt{\|\g\|^2-\dt(\h,\Cc)^2}}\nn.
\eeq
What remains is to show the right hand side concentrates around $\frac{\Clf}{\sqrt{m-\Dlf}}$ with the desired probability. Fix a constant $\eps>0$. Consider the denominator. Using Lemma \ref{repeatlem}, with probability $1-\exp(-\order{m})$,
\beq
|\frac{\sqrt{\|\g\|^2-\dt(\h,\Cc)^2}}{\sqrt{m-\Dlf}}-1|<\eps\label{simplemdlf}.
\eeq
We now apply Lemma \ref{lemma:concentrationALL} for $\CC$ where we choose $t=\frac{m}{\sqrt{\max\{m,n\}}}$ and use the fact that  $m>\DC$. Then, with probability $1-\exp(-\order{\min\{m,\frac{m^2}{n}\}})$, we have,
\beq
|\text{corr}(\h,\Cc)-\CC|\leq \eps m\nn.
\eeq
Combining this with \eqref{simplemdlf} choosing $\eps>0$, sufficiently small (according to $\eps_L$), we find \eqref{yetanothereq1} with the desired probability.
\end{proof}

The lemma above shows that, $\la\hat{f}_p(\w^*_{low})$ is around $\frac{\Clf}{\sqrt{m-\Dlf}}$ with high probability and we obtain the $\ell_2^2$ formula by using $\hat{f}_p(\w^*_{low})$ as a proxy for $\la\hat{f}_p(\w_{\ell_2}^*)$. Can we do further? Possibly yes. To show $\hat{f}_p(\w_{\ell_2}^*)$ is indeed around $\hat{f}_p(\w^*_{low})$, we can consider the modified deviation problem $\Lc^*_{dev}(\g,\h)=\min_{\w\in S_{dev}}\Lc(\w;\g,\h)$ where we modify the set $S_{dev}$ to,
\beq
S_{dev}=\{\w\big||\frac{\la\hat{f}_p(\w)}{\sigma}-\frac{\Clf}{\sqrt{m-\Dlf}}|>\eps_{dev}\sqrt{m}\}.\nn
\eeq
We may then repeat the same arguments, i.e., try to argue that the objective restricted to $S_{dev}$ is strictly greater than what we get from the upper bound optimization $\Uc(\g,\h)$. While this approach may be promising, we believe it is more challenging than our $\ell_2$ norm analysis of $\|\w_{\ell_2}^*\|$ and it will not be topic of this paper.

%
%\begin{figure}
%\centering
%\mbox{
%
%\includegraphics[scale=0.25]{Figures/Constant_lambda/n=1000_la=20.eps}
%\label{fixedlambdaell22}
%}
%\centering
%\caption{\small{a)  b) $\beta=\frac{1}{8}$.}}
%\label{fig-Gaussian}
%\end{figure}

The next section shows that there exists a one-to-one (monotone) mapping of the region $\mathcal{R}_{ON}$ to the entire possible regime of penalty parameters of the $\ell_2^2$-LASSO.
%\subsection{The optimal penalty parameter: Practical Implications}

%
%
%\begin{figure}
%  \begin{center}
%{\includegraphics[scale=0.25]{Figures/Lowrank/LPS_with_const_and_labest.eps}}
%  \end{center}
%  \caption{Size is $40$ $r=2$. $sp=40$, $Reg=7.89$, $\Delta=547$.}
%%  \vspace*{-10pt}
%\label{LPS_simul}
%\end{figure}
%

%\subsection{More properties of the Mapping}
\subsection{Properties of $\map(\la)$}\label{map prop}

The following result shows that ${\bf{P}}(\la\Cc),{\bf{D}}(\la\Cc),{\bf{C}}(\la\Cc)$ (see \eqref{eq:defn01}) are Lipschitz continuous and will be useful for the consequent discussion. The proof can be found in Appendix \ref{sec:concentration}. 
\begin{lem} \label{useful lemon}Let $\Cc$ be a compact and convex set. Given scalar function $g(x)$, define the local Lipschitz constant to be $L_g(x)=\lim\sup_{x'\rightarrow x}\left|\frac{g(x')-g(x)}{x'-x}\right|$. Let $\max_{\s\in\Cc}\|\s\|=R$. Then, viewing ${\bf{P}}(\la\Cc),{\bf{D}}(\la\Cc),{\bf{C}}(\la\Cc)$ as functions of $\la$, for $\la\geq 0$, we have,
\beq
\max\{L_{\bf{P}}(\la),L_{\bf{D}}(\la),L_{\bf{C}}(\la)\}\leq 2R(\sqrt{n}+\la R)\nn.
\eeq
\end{lem}
The following proposition is restatement of Theorem \ref{thm:map}. Recall the definition of $\Rco$ from Definition \ref{defn:Ron2}.
\begin{propo} Assume $m>\Df\labb)$. Recall that $\Rco=(\lac,\lam)$. $\call(\la)=\frac{m-\Dlf-\Clf}{\sqrt{m-\Dlf}}$ and $\map(\la)=\la\cdot\call(\la)$ have the following properties over $\{\lac\}\cup\Rco\rightarrow \{0\}\cup\R^+$.
\begin{itemize}
\item $\call(\la)$ is a nonnegative, increasing and continuous function over $\{\lac\}\cup\Rco$.
%\item $\call(\la)$ strictly increases at all $\la\in\Rco$, $\la\neq \labb$.
\item $\map(\la)$ is nonnegative, strictly increasing and continuous at all $\la \in\{\lac\}\cup\Rco$.
\item $\map(\lac)=0$. $\lim_{\la\rightarrow \lam}\map(\la)=\infty$. Hence, $\map(\la):\{\lac\}\cup\Rco\rightarrow \{0\}\cup\R^+$ is bijective.
\end{itemize}
\end{propo}
\begin{proof} {\emph{Proof of the first statement:}} Assume $\la\in\Rco$, from Lemma \ref{lemma:props}, $m>\max\{\Dlf,\Dlf+\Clf\}$ and $\la>0$. Hence, $\call(\la)$ is strictly positive over $\la\in\Rco$. Recall that,
\beq
\call(\la)=\frac{m-\Dlf-\Clf}{\sqrt{m-\Dlf}}=\sqrt{m-\Dlf}-\frac{\Clf}{\sqrt{m-\Dlf}}\nn.
\eeq
Let $h>0$. We will investigate the change in $\call(\la)$ by considering $\call(\la+h)-\call(\la)$ as $h\rightarrow 0^+$. Since $\Dlf$ is differentiable, $\sqrt{m-\Dlf}$ is differentiable as well and gives,
\beq
\frac{\pa\sqrt{m-\Dlf}}{\pa\la}=\frac{-\Dlf'}{2\sqrt{m-\Dlf}}\label{e0eq}.
\eeq
For the second term, consider the following,
\begin{align}
\frac{\Cf\la+h)}{\sqrt{m-\Df\la+h)}}-\frac{\Clf}{\sqrt{m-\Dlf}}=h[E_1(\la,h)+E_2(\la,h)]\nn,
\end{align}
where,
\begin{align}
&E_1(\la,h)=\frac{1}{h}[\frac{\Cf\la+h)}{\sqrt{m-\Df\la+h)}}-\frac{\Clf}{\sqrt{m-\Df\la+h)}}]\nn,\\
&E_2(\la,h)=\frac{1}{h}[\frac{\Clf}{\sqrt{m-\Df\la+h)}}-\frac{\Clf}{\sqrt{m-\Dlf}}]\nn.
\end{align}
As $h\rightarrow0^+$, we have,
\beq
\lim_{h\rightarrow 0^+}E_2(\la,h)=\Clf\frac{\pa\frac{1}{\sqrt{m-\Dlf}}}{\pa \la}=\frac{\Clf\Dlf'}{2(m-\Dlf)^{3/2}}\leq 0\label{e2eq},
\eeq
since $\sg(\Clf)=-\sg(\Dlf')$.

Fix arbitrary $\eps_D>0$ and let $R=\sup_{\s\in\paf}\|\s\|$. Using continuity of $\Dlf$ and Lemma \ref{useful lemon}, choose $h$ sufficiently small to ensure,
\beq
|\frac{1}{\sqrt{m-\Dlf}}-\frac{1}{\sqrt{m-\Df\la+h)}}|<\eps_D, ~~~|\Cf\la+h)-\Clf|<3R(\sqrt{n}+\la R)h\nn.
\eeq
We then have,
\beq
E_1(\la,h)\leq \frac{\Cf\la+h)-\Clf}{h}\frac{1}{\sqrt{m-\Dlf}}+3\eps_DR(\sqrt{n}+\la R)\label{e1eq}.
\eeq
Denote $\frac{\Cf\la+h)-\Clf}{h},\frac{\Df\la+h)-\Dlf}{h}$ by $\tC$ and $\tD$. Combining \eqref{e2eq}, \eqref{e1eq} and \eqref{e0eq}, for sufficiently small $h$, we find,
\beq
\lim\sup_{h\rightarrow0^+}\frac{\call(\la+h)-\call(\la)}{h}=\lim\sup_{h\rightarrow 0}[\frac{-\tD}{2\sqrt{m-\Dlf}}-\frac{\tC}{\sqrt{m-\Dlf}}-\frac{\Clf\Dlf'}{2(m-\Dlf)^{3/2}}+3\eps_DR(\sqrt{n}+\la R)]\nn.
\eeq
We can let $\eps_D$ go to $0$ as $h\rightarrow 0^+$ and $-\tD-2\tC$ is always nonnegative as $\Plf$ is nondecreasing due to Lemma \ref{lemma:keyProp}. Hence, the right hand side is nonnegative. Observe that the increase is strict for $\la\neq\labb$, as we have $\Clf\Dlf'>0$ whenever $\la\neq \labb$ due to the fact that $\Dlf'$ (and $\Clf$) is not $0$. Since increase is strict around any neighborhood of $\labb$, this also implies strict increase at $\la=\labb$.

Consider the scenario $\la=\lac$. Since $\call(\la)$ is continuous for all $\la\in\{\lac\}\cup\Rco$ (see next statement) and is strictly increasing at all $\la>\lac$, it is strictly increasing at $\la=\lac$ as well.% If $m\leq n$, $\call(\lac)=0$ and $\call(\la)>0$ for all $\la>\lac$ (see Lemma \ref{lemma:lac}). If $m> n$, $\lac=0$ and since $m-\Df\lac)> 0$ the identical analysis for $\la> \lac$ will apply where at $\la=0$, we consider the right derivative of $\Dlf$ (see Lemma \ref{lemma:keyProp}). Hence $\call(\la)$ strictly increases at $\la=\lac$.

To see continuity of $\call(\la)$, observe that, for any $\la \in \Rco\cup\{\lac\}$, $m-\Dlf>0$ and from Lemma \ref{useful lemon}, $\Dlf,\Clf$ are continuous functions which ensures continuity of $m-\Dlf-\Clf$ and $m-\Dlf$. Hence, $\call(\la)$ is continuous as well.

{\emph{Proof of the second statement:}} Since $\call(\la)$ is strictly increasing on $\Rco$, $\la\cdot\call(\la)$ is strictly increasing over $\Rco$ as well. Increase at $\la=\lac$ follows from the fact that $\map(\lac)=0$ (see next statement). Since $\call(\la)$ is continuous, $\la\cdot\call(\la)$ is continuous as well.

{\emph{Proof of the third statement:}} From Lemma \ref{lemma:lac}, if $\call(\lac)>0$, $\lac=0$ hence $\map(\lac)=0$. If $\call(\lac)=0$, then $\map(\lac)=\lac\cdot\call(\lac)=0$. In any case, $\map(\lac)=0$. Similarly, since $\lam>\labb$, $\Cf\lam)<0$ and as $\la\rightarrow\lam$ from left side, $\call(\la)\rightarrow\infty$. This ensures $\map(\la)\rightarrow\infty$ as well. Since $\map(\la)$ is continuous and strictly increasing and achieves the values $0$ and $\infty$, it maps $\{\lac\}\cup\Rco$ to $\{0\}\cup\R^+$ bijectively.
%, Since $\call(\la)$ is strictly increasing on $\Rco$, $\la\cdot\call(\la)$ is strictly increasing over $\Rco$ as well.
\end{proof}

\subsection{On the stability of $\ell_2^2$-LASSO}% -- \samet{I may remove this, needs work and not critical}}
As it has been discussed in Section \ref{map prop} in detail, $\map(\cdot)$ takes the interval $[\lac,\lam)$ to $[0,\infty)$ and Theorem \ref{thm:ell2LASSO} gives tight stability guarantees for $\la\in\Rco$. Consequently, one would expect $\ell_2^2$-LASSO to be stable everywhere as long as the $[\lac,\lam)$ interval exists. $\lac$ and $\lam$ is well defined for the regime $m>\Df\labb)$. Hence, we now expect $\ell_2^2$-LASSO to be stable everywhere for $\tau>0$. The next lemma shows that this is indeed the case under Lipschitzness assumption.

%\begin{defn} [Lipschitz function] \label{lip func}A function $f(\cdot):\R^n\rightarrow \R$ is called Lipschitz if there exists a number $L_f$ such that, for all $\x,\y\in\R^n$, $|f(\x)-f(\y)|\leq L_f\|\x-\y\|$.
%\end{defn}

\begin{lem} Consider the $\ell_2^2$-LASSO problem \eqref{ell22model}. Assume $f(\cdot)$ is a convex and Lipschitz continuous function and $\x_0$ is not a minimizer of $f(\cdot)$. Let $\A$ have independent standard normal entries and $\sigma\vb\sim\Nn(0,\sigma^2\Iden_m)$. Assume $(1-\eps_L)m\geq\Delxf$ for a constant $\eps_L>0$ and $m$ is sufficiently large. Then, there exists a number $C>0$ independent of $\sigma$, such that, with probability $1-\exp(-\order{m})$, 
\beq
\frac{\|\x_{\ell_2^2}^*-\x_0\|^2}{\sigma^2}\leq C\label{Csatisfies}.
\eeq
\end{lem}
\noindent{\emph{Remark:}} We are not claiming anything about $C$ except the fact that it is independent of $\sigma$. Better results can be given, however, our intention is solely showing that the estimation error is proportional to the noise variance.
\begin{proof} Consider the widening of the tangent cone defined as,
\beq
\Tc_f(\x_0,\eps_0)=\text{Cl}(\{\alpha\cdot \w\big|f(\x_0+\w)\leq f(\x_0)+\eps_0\|\w\|,~\alpha\geq 0\})\nn.
\eeq
Appendix \ref{widewidth} investigates basic properties of this set. In particular, we will make use of Lemma \ref{wide lem}. We can choose sufficiently small numbers $\eps_0,\eps_1>0$ (independent of $\sigma$) such that,
\beq
\min_{\w\in \Tc_f(\x_0,\eps_0),\|\w\|=1}\|\A\w\|\geq \eps_1\label{tangent widen},
\eeq
with probability $1-\exp(-\order{m})$ as $\sqrt{m-1}-\sqrt{\Delxf}\gtrsim (1-\sqrt{1-\eps_L})\sqrt{m}$. Furthermore, we will make use of the following fact that $\|\z\|\leq 2\sigma\sqrt{m}$ with probability $1-\exp(-\order{m})$, where we let $\z=\sigma\vb$ (see Lemma \ref{lemma:conc4}).

Assuming these hold, we will show the existence of $C>0$ satisfying \eqref{Csatisfies}. Define the perturbation function $f_p(\w)=f(\x_0+\w)-f(\x_0)$. Denote the error vector by $\w_{\ell_2^2}^*=\x_{\ell_2^2}^*-\x_0$. Then, using the optimality of $\x_{\ell_2^2}^*$ we have,
\beq
\frac{1}{2}\|\y-\A\x^*_{\ell_2^2}\|^2+\sigma\tau f(\x^*_{\ell_2^2})=\frac{1}{2}\|\z-\A\w_{\ell_2^2}^*\|^2+\sigma\tau f_p(\w_{\ell_2^2}^*)\leq \frac{1}{2}\|\z\|^2\nn.
\eeq
On the other hand, expanding the terms,
\beq
\frac{1}{2}\|\z\|^2\geq \frac{1}{2}\|\z-\A\w_{\ell_2^2}^*\|^2+\sigma\tau f_p(\w_{\ell_2^2}^*)\geq \frac{1}{2}\|\z\|^2-\|\z\|\|\A\w_{\ell_2^2}^*\|+\frac{1}{2}\|\A\w_{\ell_2^2}^*\|^2+\sigma\tau f_p(\w_{\ell_2^2}^*)\nn.
\eeq
Using $\|\z\|\leq2\sigma\sqrt{m}$, this implies,
\beq
2\sigma\sqrt{m}\|\A\w_{\ell_2^2}^*\|\geq \|\z\|\|\A\w_{\ell_2^2}^*\|\geq\frac{1}{2}\|\A\w_{\ell_2^2}^*\|^2+ \sigma\tau f_p(\w_{\ell_2^2}^*)\label{abcdef}.
\eeq
Normalizing by $\sigma$,
\beq
2\sqrt{m}\|\A\w_{\ell_2^2}^*\|\geq \frac{1}{2\sigma}\|\A\w_{\ell_2^2}^*\|^2+\tau f_p(\w_{\ell_2^2}^*).\nn
\eeq
The rest of the proof will be split into two cases.

{\bf{Case 1:}} Let $L$ be the Lipschitz constant of $f(\cdot)$. If $\w_{\ell_2^2}^*\in \Tc_f(\x_0,\eps_0)$, using \eqref{tangent widen},% and the bound on the spectral norm $\|\A\|$,
\begin{align}
 2\sqrt{m}\|\A\w_{\ell_2^2}^*\|\geq \frac{1}{2\sigma}\|\A\w_{\ell_2^2}^*\|^2-\tau L\|\w_{\ell_2^2}^*\|\geq  \frac{1}{2\sigma}\|\A\w_{\ell_2^2}^*\|^2-\frac{\tau L}{\eps_1}\|\A\w_{\ell_2^2}^*\|\nn.
\end{align}
Further simplifying, we find, $2\sigma(2\sqrt{m}+\frac{\tau L}{\eps_1})\geq \|\A\w_{\ell_2^2}^*\|\geq \eps_1\|\w_{\ell_2^2}^*\|$. Hence, indeed, $\frac{\|\w_{\ell_2^2}^*\|}{\sigma}$ is upper bound by $\frac{4\sqrt{m}}{\eps_1}+\frac{2\tau L}{\eps_1^2}$.

{\bf{Case 2:}} Assume $\w_{\ell_2^2}^*\not\in \Tc_f(\x_0,\eps_0)$. Then $f_p(\w_{\ell_2^2}^*)\geq \eps_0\|\w_{\ell_2^2}^*\|$. Using this and letting $\hat{\w}=\frac{\w_{\ell_2^2}^*}{\sigma}$, we can rewrite \eqref{abcdef} without $\sigma$ as,
 \beq
\frac{1}{2}\|\A\hat{\w}\|^2- 2\sqrt{m}\|\A\hat{\w}\|+2m+(\tau\eps_0\|\hat{\w}\|-2m)\leq 0.\nn
 \eeq
Finally, observing $\frac{1}{2}\|\A\hat{\w}\|^2- 2\sqrt{m}\|\A\hat{\w}\|+2m=\frac{1}{2} (\|\A\wh\|-2\sqrt{m})^2$, we find, 
 \beq
\tau\eps_0\|\hat{\w}\|-2m\leq 0\implies\frac{ \|\w_{\ell_2^2}^*\|}{\sigma}\leq \frac{2m}{\tau\eps_0}\nn.
 \eeq
% 
%Rewriting this as,
%\beq
%\frac{1}{2}\|\A\hat{\w}\|^2- 2\sqrt{m}\|\A\hat{\w}\|+\tau\eps_0\|\hat{\w}\|\leq 0
%\eeq
%Assuming $\w_{\ell_2^2}^*\neq 0$ and using the bound on the spectral norm $\|\A\|$, this implies,
%\beq
%\tau\eps_0\|\w\|\leq 2\sqrt{m}\|\A\hat{\w}\|\leq 2\sqrt{m}\|\A\|\|\hat{\w}\|\leq 4\sqrt{m}(\sqrt{m}+\sqrt{n})\|\hat{\w}\|
%\eeq

\end{proof}

%$\|\A\hat{\w}\|^2<2\sigma\sqrt{m}\|\A\hat{\w}\|\implies \|\A\|\|\hat{\w}\|<2\sigma\sqrt{n}$. Combining with above,
%\beq
%\eps_d\|\hat{\w}\|\leq2\sqrt{n}\|\A\hat{\w}\|- \|\A\hat{\w}\|^2\leq 2\sqrt{n}\|\A\hat{\w}\|\leq 2n\sigma
%\eeq

%\input{Leap_of_Faith.tex}
\section{Converse Results}\label{notenoughm}

Until now, we have stated the results assuming $m$ is sufficiently large. In particular, we have assumed that $m\geq \Delxf$ or $m\geq \Dlf$. 
It is important to understand the behavior of the problem when $m$ is small.
%It is also of importance to understand what happens when $m$ is small. 
Showing a converse result for $m<\Delxf$ or $m<\Dlf$ will illustrate the tightness of our analysis. In this section, we focus our attention on the case where $m<\Delxf$ and show that the NSE approaches infinity as $\sigma\rightarrow 0$. As it has been discussed previously, $\Delxf$ is the compressed sensing threshold which is 
%(up to $\order{\sqrt{n}} $ deviations) the necessary and sufficient 
the number of measurements required for the success of the noiseless problem \eqref{LinInv}:
%. In this sense, we are showing that $\Delxf$ corresponds to a stability transition as well. For the sake of this section, let us restate the noiseless linear inverse problem.
\beq
\min_\x f(\x)~~~\text{subject to}~~~\A\x=\A\x_0\label{usualCS}.
\eeq
For our analysis, we use Proposition \ref{usualCSprop} below, which is a slight modification of Theorem 1 in \cite{McCoy}.%, formalizes the necessity part of this fact.
%For the purposes of this section, we require the following slight modification of Theorem 1 in \cite{McCoy}.
\begin{propo}\label{usualCSprop}[\cite{McCoy}] Let $\A\in\R^{m\times n}$ have independent standard normal entries. Let $\y=\A\x_0$ and assume $\x_0$ is not a minimizer of $f(\cdot)$. Further, for some $t>0$, assume $m\leq \Delxf-t\sqrt{n}$. Then, $\x_0$ is not a minimizer of \eqref{usualCS} with probability at least $1-4\exp(-\frac{t^2}{4})$.
\end{propo}
%The following can be immediately deduced from \eqref{usualCSprop}.
%\begin{cor} Assume $m\leq \Delxf-t\sqrt{n}$. Then, $\x_0$ is not a minimizer of \eqref{usualCS} with probability $1-4\exp(-\frac{t^2}{4})$.
%\end{cor}
Proposition \ref{usualCSprop} leads to the following useful Corollary.%, we will make use of a corollary of this result.
\begin{cor}\label{usualcor} Consider the same setting as in Proposition \ref{usualCSprop} and denote $\x^*$ the minimizer of \eqref{usualCS}. For a given $t>0$, there exists an $\eps>0$ such that, with probability $1-8\exp(-\frac{t^2}{4})$, we have,
\beq
f(\x^*)\leq f(\x_0)-\eps\nn
\eeq
\end{cor}
\begin{proof}
Define the random variable $\chi=f(\x^*)-f(\x_0)$. $\chi$ is random since $\A$ is random. Define the events $E=\{\chi<0\}$ and $E_n=\{\chi\leq -\frac{1}{n}\}$ for positive integers $n$. From Proposition \ref{usualCSprop},  $\Pro(E)\geq 1-4\exp(-\frac{t^2}{4})$. Also, observe that,
\beq
E=\bigcup_{i=1}^{\infty} E_i~~~\text{and}~~~E_n=\bigcup_{i=1}^{n} E_i\nn,
\eeq
Since $E_n$ is an increasing sequence of events, by continuity property of probability, we have $\Pro(E)=\lim_{n\rightarrow\infty}\Pro(E_n)$. Thus, we can pick $n_0$ such that, $\Pro(E_{n_0})>1-8\exp(-\frac{t^2}{4})$. Let $\eps=n_0^{-1}$, to conclude the proof.
\end{proof}
The results discussed in this section, hold under the following assumption.
\begin{assume}\label{ass1} Assume $m_{lack}:=\Delxf-m>0$. $\x_0$ is not a minimizer of the convex function $f(\cdot)$. $f(\cdot):\R^n\rightarrow \R$ is a Lipschitz function, i.e., there exists constant $L>0$ such that, for all $\x,\y\in\R^n$, $|f(\x)-f(\y)|\leq L\|\x-\y\|$.
\end{assume}

\subsection{Converse Result for C-LASSO}
Recall the C-LASSO problem \eqref{cmodel}:
\beq
\min_\x \|\A\x_0+\sigma\vb-\A\x\|~~~\text{subject to}~~~f(\x)\leq f(\x_0)\label{C-LASSOnewform}.
\eeq
\eqref{C-LASSOnewform} has multiple minimizers, in particular, if $\x^*$ is a minimizer, so is $\x^*+\vb$ for any $\vb\in\Nn(\A)$. We will argue that when $m$ is small, there exists a feasible minimizer which is far away from $\x_0$. The following theorem is a rigorous statement of this idea.

\begin{thm}\label{convclasso} Suppose Assumption \ref{ass1} holds and let $\A,\vb$ have independent standard normal entries. For any given constant $C_{max}>0$, there exists $\sigma_0>0$ such that, whenever $\sigma\leq\sigma_0$, with probability $1-8\exp(-\frac{m_{lack}^2}{4n})$, over the generation of $\A,\vb$, there exists a minimizer of \eqref{C-LASSOnewform}, $\x^*_c$, such that,
\beq
\frac{\|\x^*_c-\x_0\|^2}{\sigma^2}\geq C_{max}\label{cwewant}
\eeq
%\emph{Remark:} $\A,\vb$ are randomly generated and fixed. Then, $\sigma$ is varied for the same $\A,\vb$.
\end{thm}
\begin{proof} From Corollary \ref{usualcor}, with probability $1-8\exp(-\frac{m_{lack}^2}{4n})$, there exists $\eps>0$ and $\x'$ satisfying $f(\x')\leq f(\x_0)-\eps$ and $\A\x'=\A\x_0$.% We additionally use that $\|\vb\|\leq 2\sqrt{m}$ with probability $1-\exp(-\order{m})$.
%\item 
%\end{itemize}
Denote $\w'=\x'-\x_0$ and pick a minimizer of \eqref{C-LASSOnewform} namely, $\x_0+\w^*$. Now, let $\w^*_2=\w^*+\w'$. Observe that $ \|\sigma\vb-\A\w^*\|=\|\sigma\vb-\A\w^*_2\|$. Hence, $\w^*_2+\x_0$ is a minimizer for C-LASSO if $f(\x_0+\w^*_2)\leq f(\x_0)$. But,
\beq
f(\x_0+\w^*_2)=f(\x'+\w^*)\leq f(\x')+L\|\w^*\|\nn,
\eeq
Hence, if $\|\w^*\|\leq \frac{f(\x_0)-f(\x')}{L}$, $\w^*_2+\x_0$ is a minimizer. Let $C_w=\min\{\frac{f(\x_0)-f(\x')}{L},\frac{1}{2}\|\w'\|\}$ and consider,
\beq
\w^*_3=\begin{cases}\w^*~~~\text{if}~~~\|\w^*\|\geq C_w,\\\w^*_2~~~\text{otherwise}. \end{cases}\nn
\eeq
From the discussion above, $\x_0+\w^*_3$ is guaranteed to be feasible and minimizer. Now, since $f(\x')\leq f(\x_0)-\eps$ and $f(\cdot)$ is Lipschitz, we have that $\|\w'\|\geq \frac{\eps}{L}$. Consequently, if $\|\w^*\|\geq C_w$, then, we have, $\frac{\|\w^*_3\|}{\sigma}\geq \frac{\eps}{2L\sigma}$. Otherwise, $\|\w^*\|\leq \frac{\|\w'\|}{2}$, and so,\beq
\frac{\|\w^*_3\|}{\sigma}=\frac{\|\w^*_2\|}{\sigma}\geq \frac{|\|\w'\|-\|\w^*\||}{\sigma}\geq \frac{\|\w'\|}{2\sigma}\geq \frac{\eps}{2L\sigma}\nn.
\eeq
In any case, we find that, $\frac{\|\w^*_3\|}{\sigma}$ is lower bounded by $\frac{\eps}{2L\sigma}$ with the desired probability. To conclude with \eqref{cwewant}, we can choose $\sigma_0$ sufficiently small to ensure $\frac{\eps^2}{4L^2\sigma_0^2}\geq C_{max}$.

%$\w^*(\A,\sigma\vb)$ is composed of two subsequences which are induced by $\w^*(\sigma)$ and $\w^*_2(\sigma)$ respectively. Let us call these $\w^*_{S_1}$ and $\w^*_{S_2}$. Clearly,
%\beq
%\lim_{\sigma\rightarrow 0}\frac{\|\w^*_{S_1}(\sigma)\|^2}{\sigma^2}\geq\lim_{\sigma\rightarrow 0} \frac{C_w^2}{\sigma^2}=\infty
%\eeq
%Similarly, 
%\beq
%\lim_{\sigma\rightarrow 0}\frac{\|\w^*_{S_2}(\sigma)\|^2}{\sigma^2}\geq\lim_{\sigma\rightarrow 0} \frac{\|\w^*_2(\sigma)\|^2}{\sigma^2}\geq \lim_{\sigma\rightarrow 0} \frac{(\|\w'\|-\|\w^*(\sigma)\|)^2}{\sigma^2}\geq \lim_{\sigma\rightarrow 0} \frac{\|\w'\|^2}{4\sigma^2}=\infty
%\eeq
%Hence, the overall sequence $\w^*(\A,\sigma\vb)$ diverges to infinity as well.
\end{proof}
%$\Delxf$ analysis corresponds to the analysis of $\Delxf$ which is a quantity smaller than $\DC$. While we believe, this is possible to do through Gordon's Lemma, we leave this for a future work and instead show a weaker result. 
%This section is dedicated to analysis of $m<\DC$ case where we show that, the optimization program will not be noise robust. This is not very surprising as $m<\DC$ is the regime, the optimization problem fails to identify $\x_0$ even the noise is $0$. The following theorem, formalizes this idea. 

\subsection{Converse Results for $\ell_2$-LASSO and $\ell_2^2$-LASSO}
This section follows an argument of similar flavor. We should emphasize that the estimation guarantee provided in Theorem \ref{thm:ell2LASSO} was for $m\geq \Dlf$. However, hereby, the converse guarantee we give is slightly looser, namely, $m\leq \Delxf$ where $\Delxf\leq \Dlf$ by definition. This is mostly because of the nature of our proof which uses Proposition \ref{usualCSprop} and we believe it is possible to get a converse result for $m\leq \Dlf$ via Gordon's Lemma. We leave this to future work. Recall $\ell_2$-LASSO in \eqref{ell2model}:
\beq
\min_\x \|\A\x_0+\sigma\vb-\A\x\|+\la f(\x)\label{L2-LASSOnewform}
\eeq
The following theorem is a restatement of Theorem \ref{not robust} and summarizes our result on the $\ell_2$-LASSO when $m$ is small.

\begin{thm} \label{ell2conv}Suppose Assumption \ref{ass1} holds and let $\A,\vb$ have independent standard normal entries. For any given constant $C_{max}>0$, there exists $\sigma_0>0$ such that, whenever $\sigma\leq\sigma_0$, with probability $1-8\exp(-\frac{m_{lack}^2}{4n})$, over the generation of $\A,\vb$, the minimizer of \eqref{L2-LASSOnewform}, $\x^*_{\ell_2}$, satisfies,
\beq
\frac{\|\x^*_{\ell_2}-\x_0\|^2}{\sigma^2}\geq C_{max}\label{goesinfty}.
\eeq

\end{thm}

\begin{proof} From Corollary \ref{usualcor}, with probability $1-8\exp(-\frac{m_{lack}^2}{4n})$, there exists $\eps>0$ and $\x'$ satisfying $f(\x')\leq f(\x_0)-\eps$ and $\A\x'=\A\x_0$. Denote $\w'=\x'-\x_0$. Let $\w^*+\x_0$ be a minimizer of \eqref{L2-LASSOnewform} and let $\w^*_2=\w^*+\w'$. Clearly,
$
\|\A\w^*_2-\sigma\vb\|=\|\A\w^*-\sigma\vb\|\nn.
$
Hence, optimality of $\w^*$ implies $f(\x_0+\w^*_2)\geq f(\x_0+\w^*)$.
Also, using the Lipschitzness of $f(\cdot)$,  
\begin{align}
&f(\x_0+\w^*_2)=f(\x'+\w^*)\leq f(\x')+L\|\w^*\|\nn,
\end{align}
and
\begin{align}
&f(\x_0+\w^*)\geq f(\x_0)-L\|\w^*\|\nn.
\end{align}
Combining those, we find,
\beq
f(\x')+L\|\w^*\|\geq f(\x_0+\w^*_2)\geq f(\x_0+\w^*)\geq f(\x_0)-L\|\w^*\|\nn,
\eeq
which implies, $\|\w^*\|\geq \frac{f(\x_0)-f(\x')}{2L}\geq \frac{\eps}{2L}$, and gives the desired result \eqref{goesinfty} when $\sigma_0\leq \frac{\eps}{4L\sqrt{C_{max}}}$.
\end{proof}

For the $\ell_2^2$-LASSO result, let us rewrite \eqref{ell22model} as,
\beq
\min_\x \frac{1}{2}\|\A\x_0+\sigma\vb-\A\x\|^2+\sigma\tau f(\x)\label{L22-LASSOnewform}
\eeq
The next theorem shows that $\ell_2^2$-LASSO does not recover $\x_0$ stably when $m<\Delxf$. Its proof is identical to the proof of Theorem \ref{ell2conv}.

\begin{thm} \label{ell22conv} Suppose Assumption \ref{ass1} holds and let $\A,\vb$ have independent standard normal entries. For any given constant $C_{max}>0$, there exists $\sigma_0>0$ such that, whenever $\sigma\leq\sigma_0$, with probability $1-8\exp(-\frac{m_{lack}^2}{4n})$, over the generation of $\A,\vb$, the minimizer of \eqref{L22-LASSOnewform}, $\x^*_{\ell_2^2}$, satisfies,
\beq
\frac{\|\x^*_{\ell_2^2}-\x_0\|^2}{\sigma^2}\geq C_{max}\label{goesinfty2}\nn.
\eeq

\end{thm}

\section{Numerical Results}\label{sec:num}
Simulation results presented in this section support our analytical predictions. We consider two standard estimation problems, namely sparse signal estimation and low rank matrix recovery from linear observations.
%
%This section is dedicated to verify our findings in standard problems such as sparse signal estimation and low rank matrix recovery from linear observations.

\begin{figure}
\centering
\mbox{
\subfigure[]{
\includegraphics[width=3.1in]{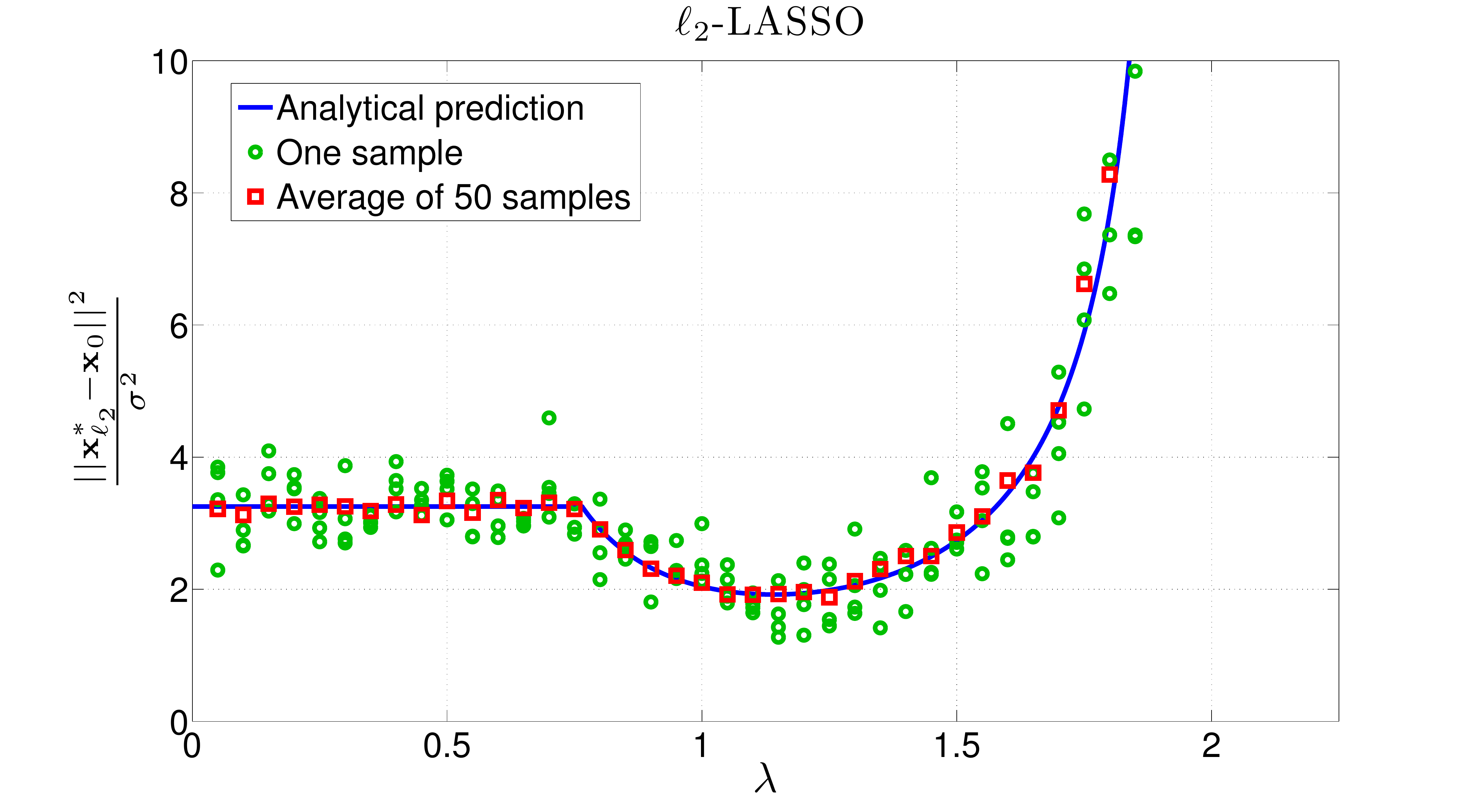}
\label{numer11}
}\quad
\subfigure[]{
\includegraphics[width=3.1in]{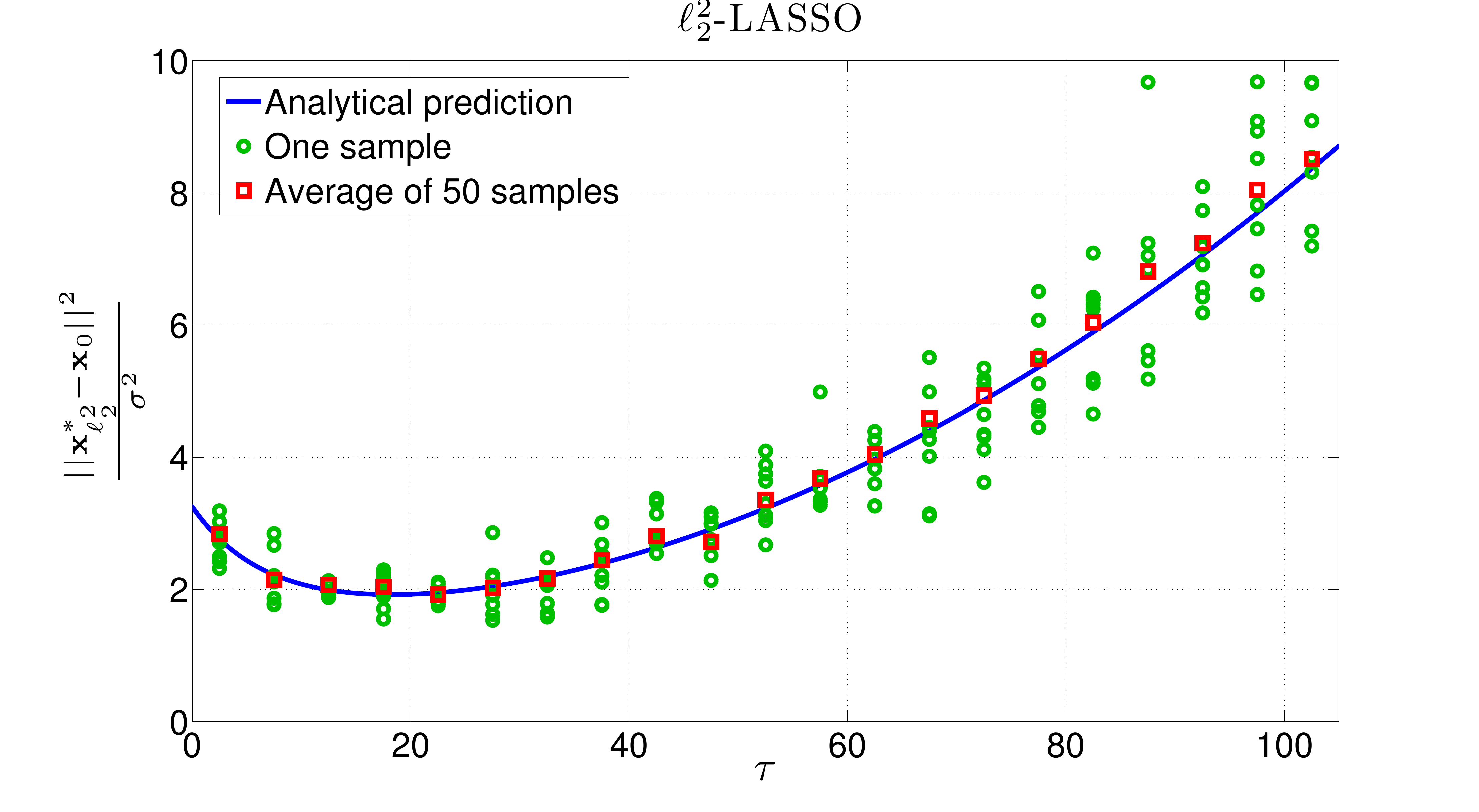}
\label{numer12}
}
}
\centering
\caption{\small{Sparse signal estimation with $n=1500,m=750,k=150$. a) $\ell_1$-penalized $\ell_2$-LASSO NSE. b) $\ell_1$-penalized $\ell_2^2$-LASSO NSE. Observe that the minimum achievable NSE is same for both (around $1.92$).}}
\label{numer1}
\end{figure}
\subsection{Sparse Signal Estimation}

First, consider the sparse signal recovery problem, where $\x_0$ is a $k$ sparse vector in $\R^n$ and $f(\cdot)$ is the $\ell_1$ norm. We wish to verify our predictions in the small noise regime.

We fix $n=1500$, $\frac{k}{n}=0.1$ and $\frac{m}{n}=0.5$. Observe that, these particular choice of ratios has also been used in the Figures \ref{figCont2} and \ref{LASSOintro2}. $\x_0\in\R^n$ is generated to be $k$ sparse with standard normal nonzero entries and then normalized to satisfy $\|\x_0\|=1$. To investigate the small $\sigma$ regime, the noise variance is set to be $\sigma^2=10^{-5}$. We observe $\y=\A\x_0+\z$ where $\z\sim\Nn(0,\sigma\Iden_m)$ and solve the $\ell_2$-LASSO and the $\ell_2^2$-LASSO problems with $\ell_1$ penalization. To obtain clearer results, each data point (red square markers) is obtained by averaging over $50$ iterations of independently generated $\A,\z,\x_0$. The effect of averaging on the NSE is illustrated in Figure \ref{numer1}. 

\noindent{\bf{$\ell_2$-LASSO:}} $\lambda$ is varied from $0$ to $2$. The analytical predictions are calculated via the formulas given in Appendix \ref{explicit_form} for the regime $\frac{k}{n}=0.1$ and $\frac{m}{n}=0.5$. We have investigated three properties.
\begin{itemize}
\item {\bf{NSE:}} In Figure \ref{numer11}, we plot the simulation results with the small $\sigma$ NSE formulas. Based on Theorem \ref{thm:ell2LASSO} and Section \ref{sec:ell2LASSO NSE}, over $\Rco$, we plotted $\frac{\Dlf}{m-\Dlf}$ and over $\Rcf$, we used $\frac{\Df\lac)}{m-\Df\lac)}$ for analytical prediction. We observe that NSE formula indeed matches with simulations. On the left hand side, observe that NSE is flat and on the right hand side, it starts increasing as $\la$ gets closer to $\lam$.%This is the regime we called $\Rco$ in which the problem is reduced to the solution of standard CS \eqref{} and we predict the NSE to be $\frac{\Df\lac)}{m-\Df\lac)}$
\item {\bf{Normalized cost:}} We plotted the cost of $\ell_2$-LASSO normalized by $\sigma$ in Figure \ref{numer22}. The exact function is $\frac{1}{\sigma}(\|\y-\A\x^*_{\ell_2}\|+\la (f(\x^*_{\ell_2})-f(\x_0)))$. In $\Rco$, this should be around $\sqrt{m-\Dlf}$ due to Theorem \ref{thm:unified}. In $\Rcf$, we expect cost to be linear in $\la$, in particular $\frac{\la}{\lac}\sqrt{m-\Df\lac)}$.

\item {\bf{Normalized fit:}} In Figure \ref{numer21}, we plotted $\frac{\|\y-\A\x_{\ell_2}^*\|}{\sigma}$, which is significant as it corresponds to the calibration function $\text{calib}(\la)$ as described in Section \ref{justifyl22}. In $\Rco$, we analytically expect this to be $\frac{m-\Dlf-\Clf}{\sqrt{m-\Dlf}}$. In $\Rcf$, as discussed in Section \ref{sec:thisROFF}, the problem behaves as \eqref{LinInv} and we have $\y=\A\x_{\ell_2}$. Numerical results for small variance verify our expectations.

% the cost of $\ell_2$-LASSO normalized by $\sigma$. The exact function is $\frac{1}{\sigma}(\|\y-\A\x^*_{\ell_2}\|+\la (f(\x^*_{\ell_2})-f(\x_0)))$. In $\Rco$, this should be around $\sqrt{m-\Dlf}$ due to Lemma \ref{}. In $\Rcf$, we expect cost to be linear in $\la$, in particular $\frac{\la}{\lac}\sqrt{m-\Df\lac)}$.

\end{itemize}

\noindent{\bf{$\ell_2^2$-LASSO:}} We consider the exact same setup and solve $\ell_2^2$-LASSO. We vary $\tau$ from $0$ to $100$ and test the accuracy of Formula \ref{form:ell22LASSO} in Figure \ref{numer12}. We find that, $\ell_2^2$-LASSO is robust everywhere as expected and the minimum achievable NSE is same as $\ell_2$-LASSO and around $1.92$ as we estimate $\Df\labb)$ to be around $330$.

%The accuracy of our formulas become apparent when we average.
%The results are plotted as a function of penalization parameter $\tau$ which is varied from $0$ to $80$. The average NSE, over 50 experiments, is plotted in Figure \ref{L22VarVar}. We used \eqref{ell22min} for the analytical prediction. We observe that, for high SNR $(\sigma^2=10^{-4})$, the analytical prediction matches with simulation. Furthermore, the lower SNR curves are upper bounded by the high SNR curve. %This illustrates Conjecture \ref{upconj} and will be the topic of Section \ref{secsigmanotsmall}.

\begin{figure}
\centering
\mbox{
\subfigure[]{
\includegraphics[width=3.1in]{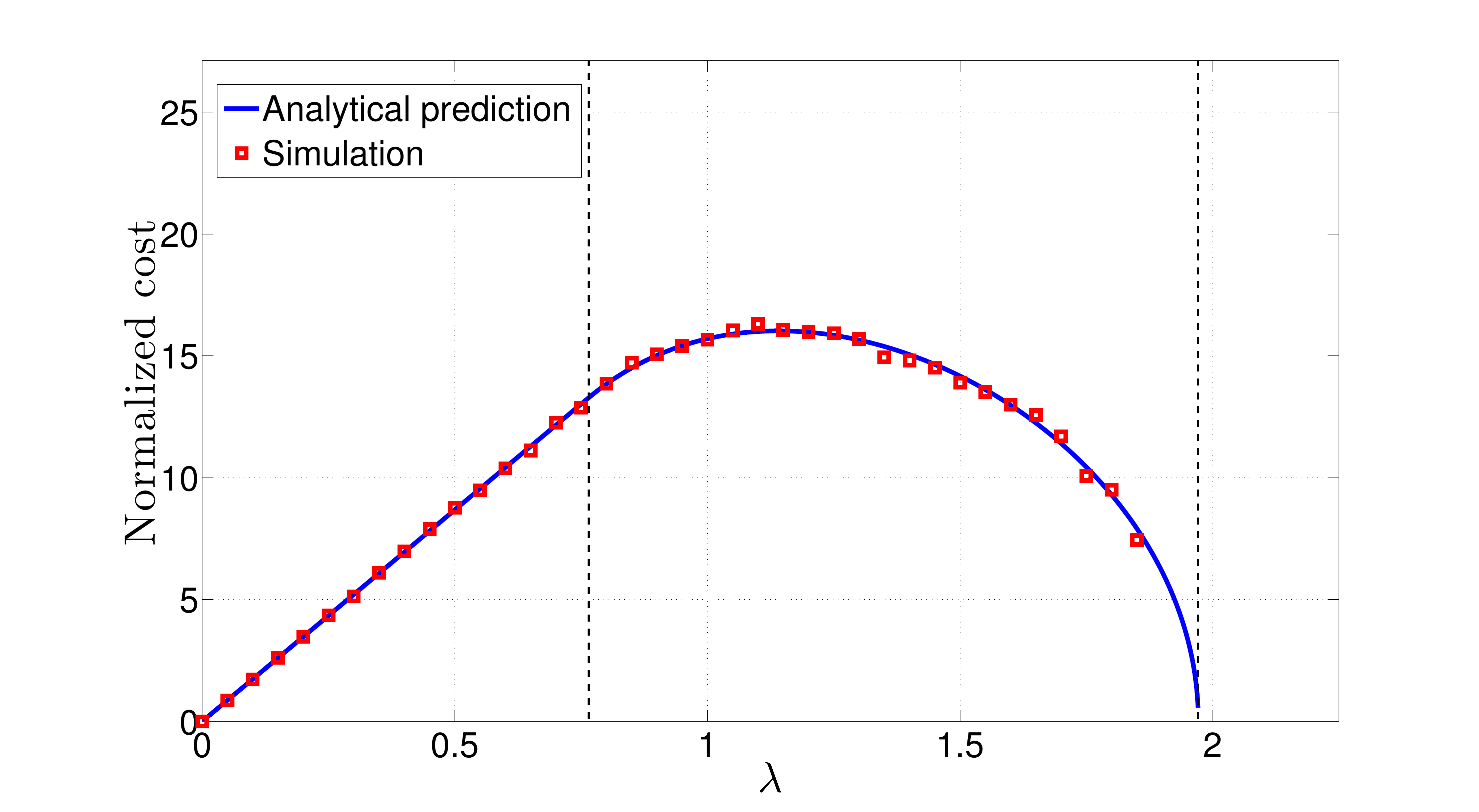}
\label{numer22}
}\quad
\subfigure[]{
\includegraphics[width=3.1in]{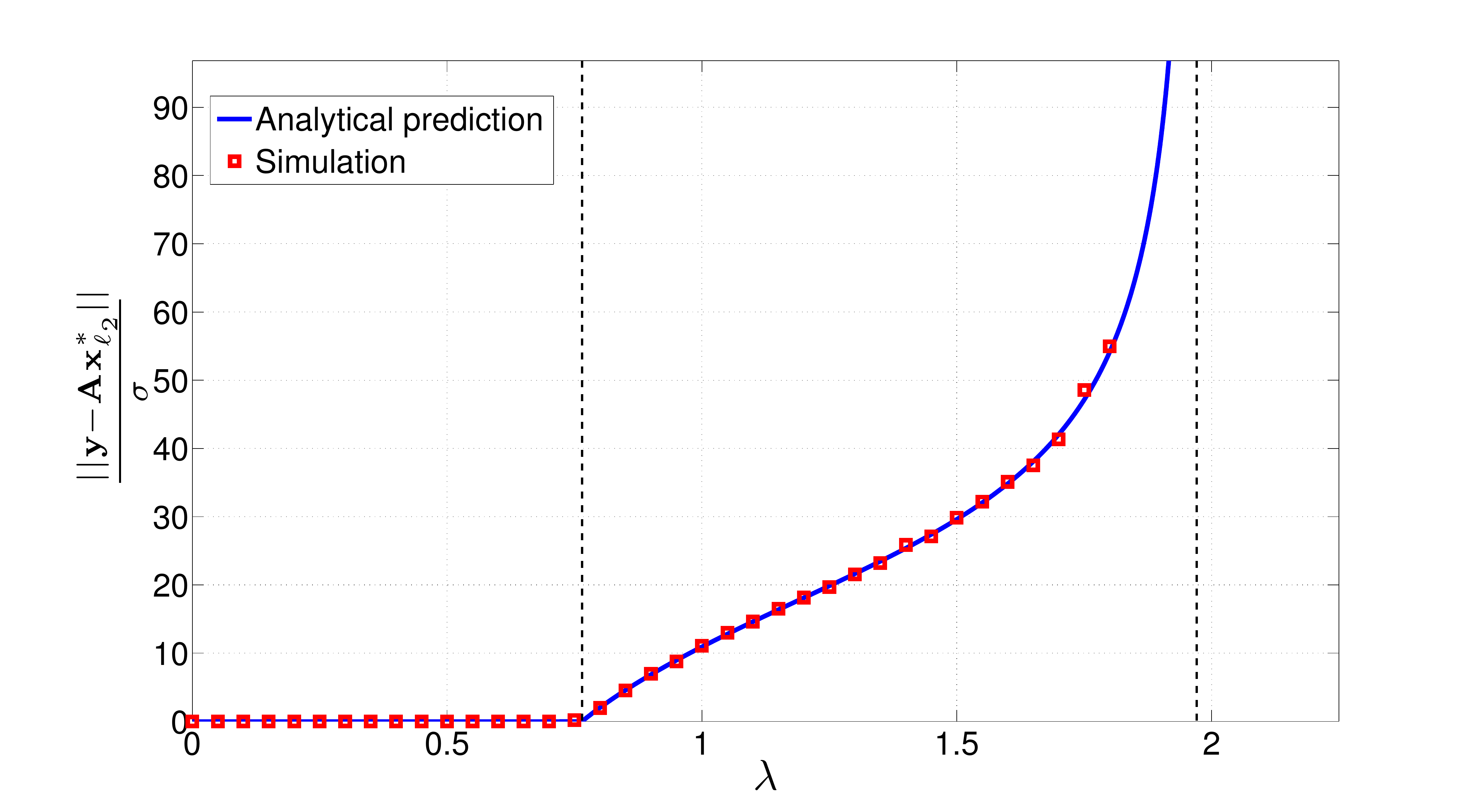}
\label{numer21}
}
}
\centering
\caption{\small{$\ell_2$-LASSO with $n=1500$, $m=750$, $k=150$. a) Normalized cost of the optimization. b) How well the LASSO estimate fits the observations $\y$. This also corresponds to the $\text{calib}(\la)$ function on $\Rco$. In $\Rcf$, ($\la\leq \lac \approx 0.76$) observe that $\y=\A\x^*_{\ell_2}$ indeed holds.}}
\label{numer2}
\end{figure}

\subsection{Low-Rank Matrix Estimation}\label{sec:sfasc}
For low rank estimation, we choose the nuclear norm $\|\cdot\|_\star$ as a surrogate for rank \cite{RechtFazel}. Nuclear norm is the sum of singular values of a matrix and basically takes the role of $\ell_1$ minimization.

Since we will deal with matrices, we will use a slightly different notation and consider a low rank matrix $\X_0\in\R^{d\times d}$. Then, $\x_0=\text{vec}(\X_0)$ will be the vector representation of $\X_0$, $n=d\times d$ and $\A$ will effectively be a Gaussian linear map $\R^{d\times d}\rightarrow \R^m$. Hence, for $\ell_2$-LASSO, we solve,
\beq
\min_{\X\in\R^{d\times d}}\|\y-\A\cdot \text{vec}(\X)\|+\la \|\X\|_\star\nn.
\eeq
where $\y=\A\cdot\text{vec}(\X_0)+\z$.

\noindent{\emph{Setup:}} We fixed $d=45$, $\text{rank}(\X_0)=6$ and $m=0.6d^2=1215$. To generate $\X_0$, we picked i.i.d. standard normal matrices ${\bf{U}},{\bf{V}}\in\R^{d\times r}$ and set $\X_0=\frac{\U\V^T}{\|\U\V^T\|_F}$ which ensures $\X_0$ is unit norm and rank $r$. We kept $\sigma^2=10^{-5}$. The results for $\ell_2$ and $\ell_2^2$-LASSO are provided in Figures \ref{lownumer1} and \ref{lownumer2} respectively. Each simulation point is obtained by averaging NSE's of $50$ simulations over $\A,\z,\X_0$.

To find the analytical predictions, based on Appendix \ref{explicit_form}, we estimated $\Dlf,\Clf$ in the asymptotic regime: $n\rightarrow\infty$, $\frac{r}{d}=0.133$ and $\frac{m}{n}=0.6$. In particular, we estimate $\Df\labb)\approx 880$ and best case NSE $\frac{\Df\labb)}{m-\Df\labb)}\approx 2.63$. Even for such arguably small values of $d$ and $r$, the simulation results are quite consistent with our analytical predictions.
% is quite consistent with the simulation results.

%
%Let $n=d^2=1600$Each simulation
%for the regime $n\rightarrow\infty$, $\frac{k}{n}=0.1$ and $\frac{m}{n}=0.5$. 
% Fix $d=40$ and choose $r$ to be $1,3,5$. . Generate noise $\z\in\R^m$ with i.i.d. $\Nn(0,\sigma^2=10^{-5})$ entries %with $\sigma^2$
% and linear Gaussian measurements $\Ac(\cdot):\R^{d\times d}\rightarrow \mathbb{R}^m$.
%
%\begin{enumerate}[(a)]
%\item Observe $\y=\Ac(\X_0)+\z$, where $\X_0,\z,\Ac(\cdot)$ are generated as above.
%\item Estimate $\X_0$ via,
%\beq
%  \vspace*{-5pt}
%\X^*_c=\arg\min_{\X}\|\y-\Ac(\X)\|~\text{subject to}~\|\X\|_\star\leq \|\X_0\|_\star\nn
%\eeq
%\item Repeat $50$ times and find average $\frac{\|\X^*_c-\X_0\|}{\sigma^2}$.
%\end{enumerate}
%We additionally estimate $\Dlf$ with the methods provided in \cite{OTH,Oym} and obtain the analytical curve for the NSE using Theorem \ref{thm:C}. The results are given in Figure \ref{figCont2} as a function of $m$ where observe that numerical experiments exhibit good match with theory and the NSE increases as a function of the rank.

\begin{figure}
\centering
\mbox{
\subfigure[]{
\includegraphics[width=3.1in]{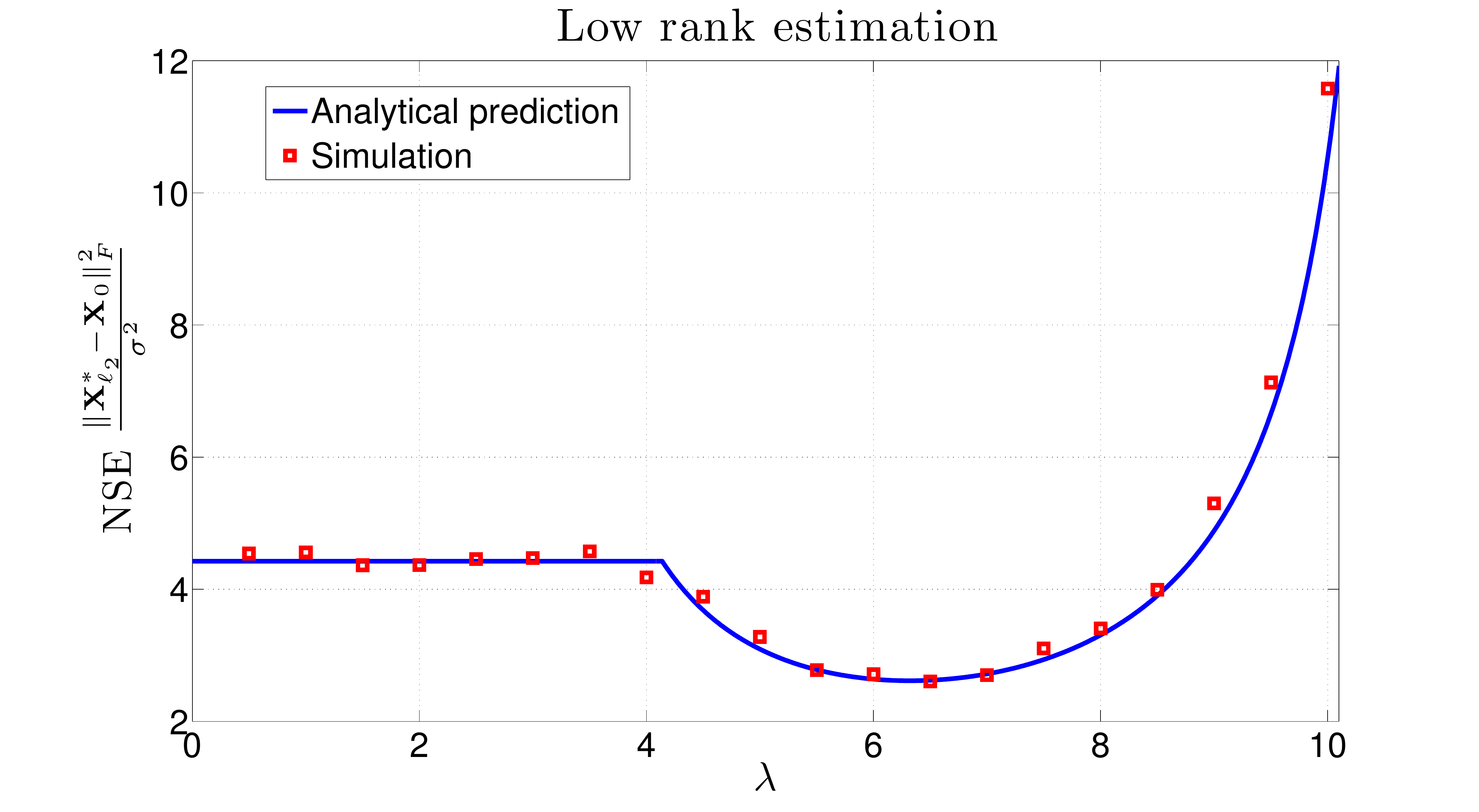}
\label{lownumer2}
}\quad
\subfigure[]{
\includegraphics[width=3.1in]{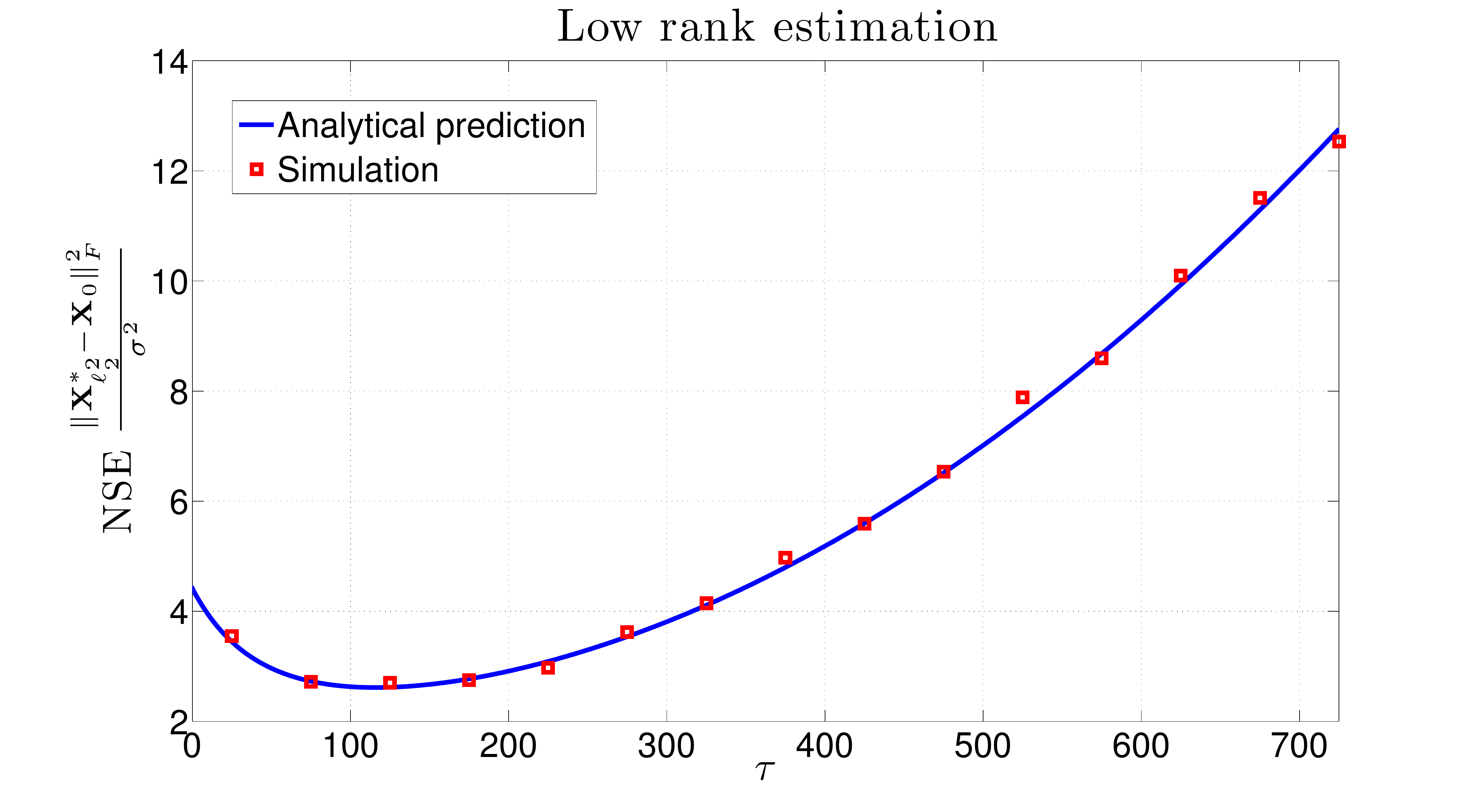}
\label{lownumer1}
}
}
\centering
\caption{\small{$d=45$, $m=0.6d^2$, $r=6$. We estimate $\Df \labb)\approx 880$. a) $\ell_2$-LASSO NSE as a function of the penalization parameter. b) $\ell_2^2$-LASSO NSE as a function of the penalization parameter.}}
\label{lownumer}
\end{figure}

\subsection{C-LASSO with varying $\sigma$}

Consider the low rank estimation problem  as in Section \ref{sec:sfasc}, but use the C-LASSO as an estimator:
\beq
\min_{\X\in\R^{d\times d}}\|\y-\A\cdot \text{vec}(\X)\|~~~\text{subject to}~~~\|\X\|_\star\leq \|\X_0\|_\star\nn.
\eeq
This time, we generate $\A$ with i.i.d. Bernoulli entries where each entry is either $1$ or $-1$, with equal probability. The noise vector$\z$, the signal of interest $\X_0$ and the simulation points are generated in the same way as in Section \ref{sec:sfasc}. Here, we used $d=40,r=4$ and varied $m$ from $0$ to $2000$ and $\sigma^2$ from $1$ to $10^{-4}$. The resulting curve is given in Figure \ref{CONMAT}. We observe that as the noise variance increases, the NSE decreases. The worst case NSE is achieved as $\sigma\rightarrow 0$, as Theorem \ref{thm:CLASSO} predicts. Our formula for the small $\sigma$ regime $\frac{\Delxf}{m-\Delxf}$ indeed provides a good estimate of NSE for $\sigma^2=10^{-4}$ and upper bounds the remaining ones. In particular, we estimate $\Delxf$ to be around $560$. Based on Theorems \ref{not robust} and \ref{thm:CLASSO}, as $m$ moves from $m<\Delxf$ to $m>\Delxf$, we expect a change in robustness. Observe that, for larger noise variances (such as $\sigma^2=1$) this change is not that apparent and the NSE is still relatively small. For $\sigma^2\leq 10^{-2}$, the NSE becomes noticeably high for the regime $m<\Delxf$.

\begin{figure}
\centering
\includegraphics[width=3.1in]{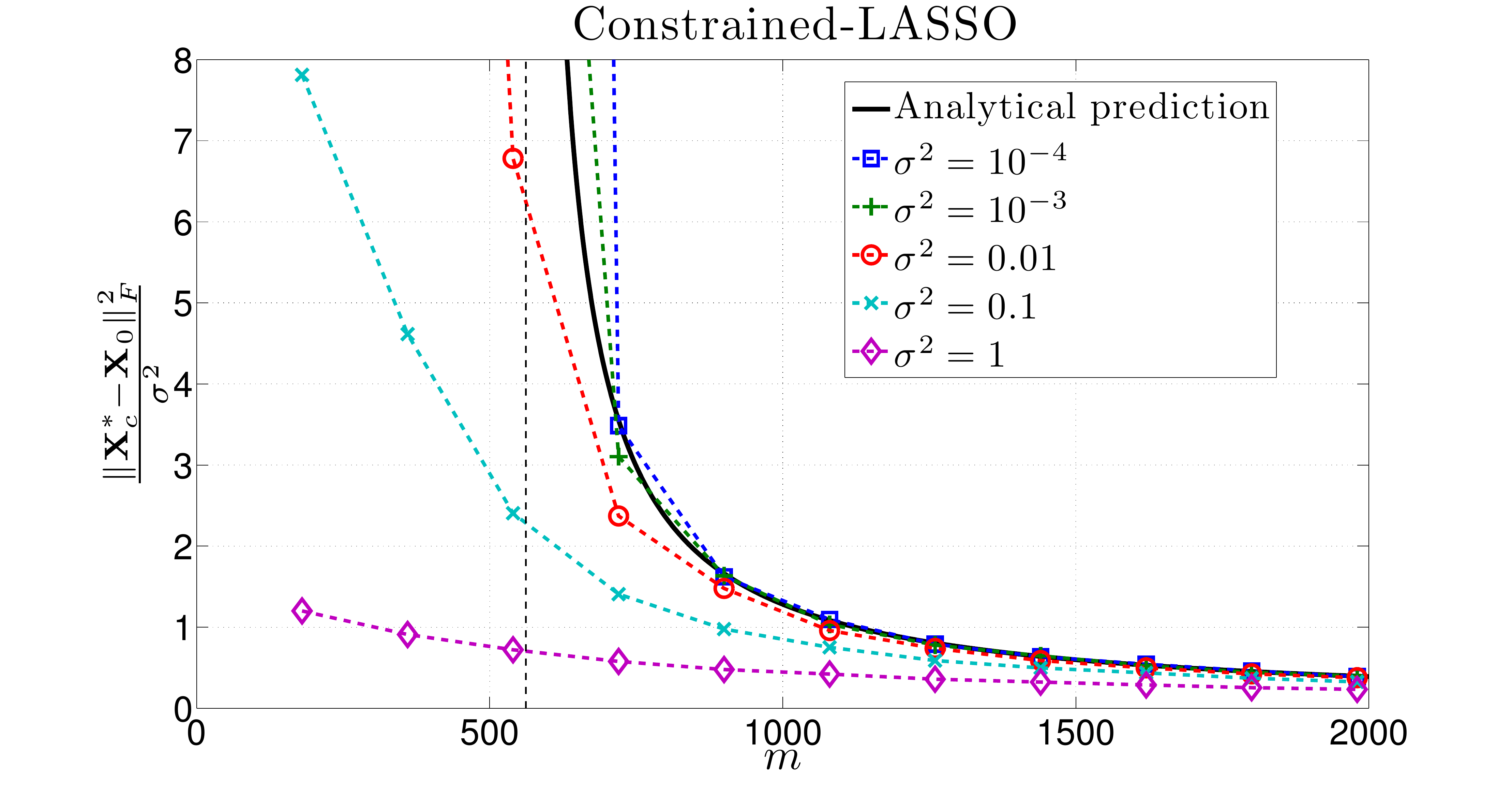}

\centering
\caption{\small{ $\X_0$ is a $40\times 40$ matrix with rank $4$. As $\sigma$ decreases, NSE increases. The vertical dashed lines marks the estimated $\Delxf$ where we expect a transition in stability.}}
\label{CONMAT}
\end{figure}

\begin{comment}
\subsubsection{singular value soft thresholding}
\subsection{Constrained minimization}

\subsection{$\ell_1$ and nuclear analytic curves}
\subsubsection{low rank plus sparse matrices}
\subsubsection{explicit formula from the other paper?}
%\subsubsection{other stupid examples -- cut matrix?}

\subsection{optimal tuning}
\subsubsection{Algorithm for optimal tuning}
\subsubsection{compare with constrained performance}

\subsection{Constrained plus regularized}
\subsubsection{Sparse \& Nonnegative signals}

\subsection{Maybe Block sparse???}
\end{comment}
%\subsection{Different measurement ensembles}
%While Gaussian measurements are significantly easier to analyze due to various nice properties, it is of prime importance to have an idea on what to expect from different measurement ensembles. We will give two examples to show how 
%\subsubsection{Binary}
%
%\subsubsection{Unitary}
%While Gaussian measurements are arguably one of the most One may considerWe have the following conjecture when the measurement matrix is a partial unitary rather than a random
\section{Future Directions}\label{sec:conc}
%While this paper presented some lengthy analysis, we believe there are many important directions to be explored. Here, we mention some of these.

%While 
We believe that our work sets up the fundamentals for a number of possible extensions. We enlist here some of those promising directions to be explored in future work. 

\begin{itemize}
\item {\bf{$\ell_2^2$-LASSO formula:}} While Section \ref{justifyl22} provides justification behind Formula \ref{form:ell22LASSO}, a rigorous proof is arguably the most important point missing in this paper. Such a proof
%The proof of $\ell_2^2$-LASSO formula 
would close the gap in this paper and will extend results of \cite{BayMon,Mon} to arbitrary convex functions.

%While we provide the justification behind Formula \ref{form:ell22LASSO}, the proof is arguably the most important point we are missing in this paper. The proof of $\ell_2^2$-LASSO formula will close the gap in this paper and will extend results of \cite{BayMon} to arbitrary convex functions.
\item {\bf{Error formulas for arbitrary $\sigma$:}} Another issue that  hasn't been fully explored in this paper is the regime where $\sigma$ is not small. For C-LASSO, we have shown that the NSE for arbitrary values of $\sigma$ is upper bounded by the NSE at $\sigma\rightarrow 0$. Empirical observations suggest that the same is true for the $\ell_2$ and $\ell_2^2$-LASSO. Proving that this is the case is one open issue. What might be even more interesting, is computing  exact error formulae for the arbitrary $\sigma$ regime.
% On the other hand, it might be interesting to give the exact error formula for arbitrary $\sigma$. 
 As we have discussed previously, we expect such formulae to not only depend on the subdifferential of the function. 

%In case of C-LASSO, a good guess would be 

\item {\bf{Extension to multiple structures:}} Throughout this work, we have focused on the recovery of a single signal $\x_0$. In general, one may consider a scenario, where we observe mixtures of multiple structures. A classic example used to motivate such problems includes estimation of matrices that can be represented as sum of a low rank and a sparse component \cite{LPS,PCA,VenPCA,McCoy2}. Another example, which is closer to our framework, is when the measurements $\A\x_0$ experience not only additive i.i.d. noise $\z$, but also sparse corruptions $\s_0$ \cite{XLi,Foygel}. In this setup, we observe $\y=\A\x_0+\s_0+\z$ and we wish to estimate $\x_0$ from $\y$. The authors in \cite{Foygel,McCoy3} provide sharp recovery guarantees for the noiseless problem, but do not address the precise noise analysis. We believe, our framework can be extended to the exact noise analysis of the following constrained problem:
\beq
\min_{\x,\s}\|\y-\A\x-\s\|~~~\text{subject to}~~~g(\s)\leq g(\s_0)~\text{ and }~f(\x)\leq f(\x_0).\nn
\eeq
where $g(\cdot)$ is typically the $\ell_1$ norm.
\item {\bf{Application specific results:}} In this paper, we focused on a generic signal-function pair $\x_0,f$ and stated our results in terms of the convex geometry of the problem. We also provided numerical experiments on NSE of sparse and low rank recovery and showed that, theory and simulations are consistent. On the other hand, it would be useful to derive case-specific guarantees other than NSE. For example, for sparse signals, we might be interested in the sparsity of the LASSO estimate, which has been considered by Bayati and Montanari \cite{Mon,BayMon}. Similarly, in low rank matrix estimation, we might care about the rank and nuclear norm of the LASSO estimate. On the other hand, our generic results may be useful to obtain NSE results for a growing set of specific problems with little effort, \cite{Yonina, McCoy2, simultaneous,French,LPS,TotalCompress1}. In particular, one can find an NSE upper bound to a LASSO problem as long as he has an upper bound to $\Dlf$ or $\Delxf$.% has already been the topic of various papers and our formulae are fully based on these quantities.

\item {\bf{Different $\A,\vb$:}} Throughout the paper, $\A$ and $\vb$ were assumed to be independent with i.i.d. standard normal entries. It might be interesting to consider different measurement ensembles such as matrices with subgaussian entries or even a different noise setup such as ``adversarial noise", in which case the error vector $\vb$ is generated to maximize the NSE. For example, in the literature of compressed sensing phase transitions, it is widely observed that measurement matrices with subgaussian entries behave same as gaussian ones, \cite{Universal,BLM12}. %Hence it wouldn't be surprising if our rest see the  interesting to give guarantees for $\A$ with subgaussian entries.

\item {\bf{Mean-Squared-Error (MSE) Analysis:}} In this paper, we focused on the $\ell_2$-norm square of the LASSO error and provided high probability guarantees. It is of interest to give guarantees in terms of mean-squared-error where we consider the expected NSE. Naturally, we expect our formulae to still hold true for the MSE, possibly requiring some more assumptions.

\end{itemize}

\section*{Acknowledgments}
Authors would like to thank Joel Tropp, Arian Maleki and Kishore Jaganathan for stimulating discussions and helpful comments. S.O. would also like to thank Adrian Lewis for pointing out Proposition \ref{prop:F3T}.

\newpage

    \addcontentsline{toc}{section}{Appendix}
    \addtocontents{toc}{\protect\setcounter{tocdepth}{-1}}

\newpage
\appendix
\vspace{7pt}
\begin{center}\Large{{APPENDIX}}\end{center}

\section{Useful Facts}

%\begin{fact}\label{fact:lipVar}
%Assume $\g\sim\Nn(0,\mathbf{I}_n)$ and let $h(\cdot):\mathbb{R}^n\rightarrow\mathbb{R}$ be an $L$-Lipschitz function. Then,
%$$
%Var( h(g) ) \leq L^2.
%$$
%\end{fact}
%
%\begin{fact}[Gaussian concentration Inequality for Lipschitz functions] \label{fact:lipIneq}
%Let $f(\cdot):\R^n\rightarrow \R$ be an $L$-Lipschitz function and $\g\sim \Nn(0,I_{n\times n})$. Then:
%\beq
%\Pro(|f(\g)-\E[f(\g)]|\geq t)\leq 2\exp(-\frac{t^2}{2L^2})
%\eeq
%\end{fact}

%
\begin{fact}[Moreau's decomposition theorem] \label{more}
%\eqname{more}
Let $\Cc$ be a closed and convex cone in $\R^n$. For any $\vb\in\R^n$, the following two are equivalent:
\begin{enumerate}
\item $\vb=\ab+\bb$, $\ab\in\Cc,\bb\in \Cc\pol$ and $\ab^T\bb=0$.
\item $\ab=\bu(\vb,\Cc)$ and $\bb=\bu(\vb,\Cc\pol)$.
\end{enumerate}
\end{fact}
\begin{fact}[Properties of the projection, \cite{Bertse,Boyd}] \label{prom}
%\eqname{prom}
Assume $\Cc\subseteq\R^n$ is a nonempty, closed and convex set and $\ab,\bb\in\R^n$ are arbitrary points. Then,
\begin{itemize}

\item The projection $\bu(\ab,\Cc)$ is the unique vector satisfying,
$
\bu(\ab,\Cc)=\arg\min_{\vb\in\Cc}\|\ab-\vb\|.\nn%\label{lem1}
$
\item
%The projection $\bu(\ab,\Cc)$ is also the unique vector $\s_0$ that satisfies,
$
\li\bu(\ab,\Cc),\ab-\bu(\ab,\Cc)\ri=\sup_{\s\in\Cc} \li\s,\ab-\bu(\ab,\Cc)\ri.\nn%\label{desiredlem2}
$
\item 
$
\|\bu(\ab)-\bu(\bb)\|\leq \|\ab-\bb\|.\nn
$
%In other words, $\ab$ and $\Cc$ lies on different half planes induced by the hyperplane that goes through $\bu(\ab,\Cc)$ and that is orthogonal to $\ab-\bu(\ab,\Cc)$.
\end{itemize}
\end{fact}

\begin{fact}[Variance of Lipschitz functions]
%\eqname{fact:lipVar}]
\label{fact:lipVar}
Assume $\g\sim\Nn(0,\mathbf{I}_p)$ and let $f(\cdot):\mathbb{R}^p\rightarrow\mathbb{R}$ be an $L$-Lipschitz function. Then,
$$
Var( f(\g) ) \leq L^2.
$$
\end{fact}

\begin{fact}[Gaussian concentration Inequality for Lipschitz functions] \label{fact:lipIneq}
Let $f(\cdot):\R^p\rightarrow \R$ be an $L$-Lipschitz function and $\g\sim \Nn(0,\mathbf{I}_{p})$. Then,
$$
\Pro\left(|f(\g)-\E[f(\g)]|\geq t\right)\leq 2\exp(-\frac{t^2}{2L^2}).
$$
\end{fact}

\section{Auxiliary Results}\label{sec:concentration}%\eqname{sec:concentration}

\begin{lem}
%\eqname{lemma:conc3}
\label{lemma:conc3}
Let $f(\cdot):\R^p\rightarrow \R$ be  an $L$-Lipschitz function and $\g\sim \Nn(0,I_{p})$. Then,
\begin{align*}
 \sqrt{ \E\left[ ( f(\g) )^2 \right] - L^2} - t \leq f(\g) \leq \sqrt{ \E\left[( f(\g) )^2 \right]} + t,
\end{align*}
with probability $1-2\exp(-\frac{t^2}{2L^2})$.
\end{lem}
\begin{proof}
From Fact \ref{fact:lipIneq}, 
\beq\label{eq:kj1}
|f(\g)-\E[f(\g)]|\leq t,
\eeq
holds with probability $1-2\exp(-\frac{t^2}{2L^2})$.
Furthermore,  
\begin{align}\label{eq:fhg1}
 \E[ ( f(\g) ) ^2] - L^2  \leq ( \E[f(\g)] )^2 \leq \E[ ( f(\g) ) ^2] .
  \end{align}
 The left hand side inequality in \ref{eq:fhg1} follows from an application of Fact \ref{fact:lipVar} and the right hand side follows from Jensen's Inequality.

Combining \eqref{eq:kj1} and \eqref{eq:fhg1} completes the proof.
\end{proof}
For the statements of the lemmas below, recall the definitions of $\DC$,$\PC$ and $\CC$ in Section \ref{sec:snot}.
\begin{lem}
%\eqname{lemma:conc4}
\label{lemma:conc4}
Let $\g\sim\Nn(0,\Iden_m),\h\sim\Nn(0,\Iden_n)$ and let $\Cc\in\R^n$ be a closed and convex set. Given $t>0$, each of the followings hold with probability $1-2\exp\left( \frac{-t^2}{2} \right)$.
\begin{itemize}
\item $\sqrt{m-1} - t \leq \| \g \|_2 \leq \sqrt{m} + t$
\item $\sqrt{\DC-1} - t \leq \dt(\h,\Cc )\leq \sqrt{\DC} + t$
\item $\sqrt{\PC-1} - t \leq \| \bu(\h,\Cc) \|_2 \leq \sqrt{\PC} + t$
\end{itemize}
\end{lem}
\begin{proof}The result is an immediate application of Lemma \ref{lemma:conc3}. The functions $\|\cdot\|$, $ \|\bu(\cdot,\Cc )\|$ and $\dt(\cdot,\Cc)$  are all $1$-Lipschitz. Furthermore, $\E[\|\g \|_2^2]=m$ and $\E[\bu(\h,\Cc)^2] = \PC$,  $\E[\dt(\h,\Cc )^2] = \DC$ by definition. %
%Using Fact \ref{}, and $1$-Lipschitzness of the projection and distance functions we have, $\text{Var}(\dt_\la(\h))=1$ and $\text{Var}(\bu_\la(\h))=1$. This gives,
%\begin{align}
%&\sqrt{\DC-1}\leq\E[\dt_\la(\h)]\leq\sqrt{\DC}\\
%&\sqrt{\PC-1}\leq\E[\|\bu_\la(\h)\|]\leq\sqrt{\PC}
%\end{align}
%Secondly using Fact \ref{}, each of the following holds with probability $1-2\exp(-\frac{t^2}{2})$
%\begin{align}
%&|\dt_\la(\h)-\E[\dt_\la(\h)]|\leq t\\
%&|\|\bu_\la(\h)\|-\|\E[\|\bu_\la(\h)\|]\||\leq t
%\end{align}
\end{proof}

\begin{lem}
%\eqname{lemma:concentrationALL}
\label{lemma:concentrationALL}
Let $\h\sim\Nn(0,\Iden_n)$ and let $\Cc\in\R^n$ be a convex and closed set. Then, given $t>0$,
\begin{itemize}
\item $|\dt(\h,\Cc)^2-\DC|\leq 2t\sqrt{\DC}+t^2+1$.
\item $|\|\bu(\h,\Cc)\|^2-\PC|\leq 3t\sqrt{n+\DC}+t^2+1$.
\item $|\corr(\h,\Cc)-\CC|\leq 3t\sqrt{n+\DC}+t^2+1$.
\end{itemize}
with probability $1-4\exp(-\frac{t^2}{2})$.
\end{lem}
\begin{proof} 
The first two statements follow trivially from Lemma \ref{lemma:conc4}. For the second statement, use again Lemma \ref{lemma:conc4} and also upper bound $\PC$ by $2(n+\DC)$ via Lemma \ref{DCbound}. To obtain the third statement, we write,
\beq
\corr(\h,\Cc)=\frac{n-(\|\bu(\h,\Cc)\|^2+\dt(\h,\Cc)^2)}{2}\nn
\eeq
and use the fact that first two statements hold with probability $1-4\exp(-\frac{t^2}{2})$. This will give,
\beq
|\corr(\h,\Cc)-\CC|\leq t(\sqrt{\DC}+\sqrt{\PC})+t^2+1,\nn
\eeq
which when combined  with Lemma \ref{DCbound} concludes the proof.
%%$\sqrt{\PC}\leq \sqrt{2(n+\DC)}$.
% and $1$-Lipschitzness of projection and distance, each of the following holds with probability $1-2\exp(-\frac{t^2}{2})$.
%\begin{align}
%|\dt(\h,\la\paf)^2-\DC|\leq \sqrt{\DC}
%\end{align}
\end{proof}

\begin{lem}\label{DCbound} 
%In $\R^n$, $\CC,\PC$ satisfy $\max\{\CC,\PC\}\leq 2(n+\DC)$.
Let $\Cc\in\R^n$ be a convex and closed set. Then, the following holds,
$$\max\{\CC,\PC\}\leq 2(n+\DC).$$
\end{lem}
\begin{proof} From triangle inequality, for any $\h\in\mathbb{R}^n$,
\beq
\|\bu(\h,\Cc)\|\leq \|\h\|+\dt(\h,\Cc).\nn
\eeq
We also have,
\beq
\E[\|\h\|\cdot\dt(\h,\Cc)]\leq\frac{1}{2}( \E[\|\h\|^2]+\E[\dt(\h,\Cc)^2])=\frac{n+\DC}{2}\nn.
\eeq
From these, we may write,
\begin{align}
\CC&=\E[\li\Pi(\h,\Cc),\bu(\h,\Cc)\ri]\nn\\%\leq \E[\dt(\h,\Cc)\|\bu(\h,\Cc)\|\nn]\\
&\leq \E[\dt(\h,\Cc)\|\bu(\h,\Cc)\|]\nn\\&\leq \frac{n+3\DC}{2}\nn.
\end{align}
Similarly, we have,
\beq
\PC = \E[\|\bu(\h,\Cc)\|^2]\le \E[\|\h\|+\dt(\h,\Cc)^2]\leq 2(n+\DC)\nn.
\eeq
\end{proof}
%
%
%We make repeated use of the following lemma throughout the paper.
\begin{lem} \label{repeatlem}Let $\g\sim\Nn(0,\Iden_m)$ and $\h\sim\Nn(0,\Iden_n)$. Let $\Cc$ be a closed and convex set in $\R^n$. Assume $m(1-\eps_L)>\DC>\eps_Lm$ for some constant $\eps_L>0$ and $m$ is sufficiently large. Then, for any constant $\eps>0$, each of the following holds with probability $1-\exp(-\order{m})$,
\begin{itemize}
\item $\|\g\|>\dt(\h,\Cc)$.
\item $\big|\frac{\|\g\|^2-\dt^2(\h,\Cc)}{m-\DC}-1\big|<\eps$.
\item $\big|\frac{\dt^2(\h,\Cc)}{\|\g\|^2-\dt^2(\h,\Cc)}\times \frac{m-\DC}{\DC}-1\big|<\eps$.
%\item 
\end{itemize}
\end{lem}
\begin{proof} Let $\delta$ be a constant to be determined. For sufficiently large $m$, using Lemma \ref{lemma:conc4}, with probability $1-\exp(-\order{m})$, we have,
\beq
|\|\g\|^2-m|<\delta m,~~~|\dt(\h,\Cc)^2-\DC|<\delta m\nn
\eeq
Now, choose $\delta<\frac{\eps_L}{2}$, which gives,
\beq
\|\g\|\geq\sqrt{m(1-\delta)}>\sqrt{\DC+\eps_L m-\delta m}>\sqrt{\DC+\delta m}\geq \dt(\h,\Cc)\nn
\eeq
This gives the first statement. For the second statement, observe that,
\beq
1+\frac{2\delta}{\eps_L}\geq \frac{m-\DC+2\delta}{m-\DC}\geq \frac{\|\g\|^2-\dt^2(\h,\Cc)}{m-\DC}\geq \frac{m-\DC-2\delta}{m-\DC}\geq 1-\frac{2\delta}{\eps_L}.\nn
\eeq
Choose $\frac{\delta}{\eps_L}<\frac{\eps}{2}$ to ensure the desired result. For the last statement, we similarly have,
\beq
\frac{1+\frac{\delta}{\eps_L}}{1-\frac{2\delta}{\eps_L}}\geq \frac{\dt^2(\h,\Cc)}{\|\g\|^2-\dt^2(\h,\Cc)}\times  \frac{m-\DC}{\DC}\geq \frac{1-\frac{\delta}{\eps_L}}{1+\frac{2\delta}{\eps_L}}\label{eq:what2}
\eeq
To conclude, notice that we can choose $\frac{\delta}{\eps_L}$ sufficiently small (constant) to ensure that the left and right bounds in \eqref{eq:what2} above are between $1\pm \eps$.

\end{proof}

\begin{proof}[{\bf{Proof of Lemma \ref{useful lemon}}}]
We will show the results for $L_{\bf{P}}(\la)$ and $L_{\bf{D}}(\la)$. $L_{\bf{C}}(\la)$ follows from the fact that ${\bf{P}}(\la\Cc)+{\bf{D}}(\la\Cc)+2{\bf{C}}(\la\Cc)=n$. Let $\h\in\R^n$. Then, for $\la+\eps,\la> 0 $,
\beq
\|\bu(\h,(\la+\eps)\Cc)\|=\frac{\la+\eps}{\la}\|\bu(\frac{\la\h}{\la+\eps},\la\Cc)\|=\|\bu(\frac{\la\h}{\la+\eps},\la\Cc)\|+\frac{\eps}{\la}\|\bu(\frac{\la\h}{\la+\eps},\la\Cc)\|\nn
\eeq 
This gives,
\beq
\big|\|\bu(\h,(\la+\eps)\Cc)\|- \|\bu(\frac{\la \h}{\la+\eps},\la\Cc)\| \big |\leq |\eps| R\nn
\eeq
Next, observe that,
\beq
\big|\|\bu(\frac{\la \h}{\la+\eps},\la\Cc)\|-\|\bu(\h,\la\Cc)\|\big|\leq \frac{|\eps|\|\h\|}{\la+\eps}\nn
\eeq
Combining, letting $\h\sim\Nn(0,\Iden_n)$ and using $\|\bu(\h,\la\Cc)\|\leq \la R$, we find,
\begin{align}
{\bf{P}}((\la+\eps)\Cc)&\leq \E[(\|\bu(\h,\la\Cc)\|+\frac{|\eps|\|\h\|}{\la+\eps}+|\eps| R)^2]\nn\\
&\leq {\bf{P}}(\la\Cc)+2\la R|\eps|(\frac{\E[\|\h\|]}{\la+\eps}+R)+|\eps|^2\E[(\frac{\|\h\|}{\la+\eps}+R)^2]\nn
\end{align}
Obtaining the similar lower bound on ${\bf{P}}((\la+\eps)\Cc)$ and letting $\eps\rightarrow 0$,
\beq
L_{\bf{P}}(\la)=\lim_{\eps\rightarrow 0}\sup\left|\frac{{\bf{P}}((\la+\eps)\Cc)-{\bf{P}}(\la\Cc)}{\eps}\right|\leq \lim_{\eps\rightarrow 0}2\la R(\frac{\E[\|\h\|]}{\la+\eps}+R+\order{|\eps|})\leq 2R(\sqrt{n}+\la R)\nn
\eeq
For $\la=0$, observe that for any $\eps>0,\h\in\R^n$, $\|\bu(\h,\eps\Cc)\|\leq \eps R$ which implies ${\bf{P}}(\eps\Cc)\leq \eps^2R^2$. Hence,
\beq
L_{\bf{P}}(0)=\lim_{\eps\rightarrow 0^+}\eps^{-1}({\bf{P}}(\eps\Cc)-{\bf{P}}(0))= 0
\eeq
Next, consider ${\bf{D}}(\la \Cc)$. Using differentiability of ${\bf{D}}(\la \Cc)$, for $\la>0$,
\beq
L_{\bf{D}}(\la)=|{\bf{D}}(\la \Cc)'|=\frac{2}{\la}|{\bf{C}}(\la \Cc)|\leq \frac{2\cdot\E[\|\bu(\h,\la\Cc)\|\cdot\dt(\h,\la\Cc)]}{\la}\leq 2R\cdot\E[\dt(\h,\la\Cc)]\leq 2R(\sqrt{n}+\la R)\nn
\eeq
For $\la=0$, see the ``Continuity at zero'' part of the proof of Lemma B.2 in \cite{McCoy}, which gives the upper bound $2R\sqrt{n}$ on $L_{\bf{D}}(0)$.
\end{proof}
\section{Proof of (modified) Gordon's Lemma}\label{sec:proofGor}

%\chris{change statement And proof to consider COMPACT set.}

In this section we prove the modified Gordon's Lemma \ref{lemma:Gor}. The Lemma is a consequence of Theorem \ref{thm:GordonMain}. We repeat the statement of the Lemma for ease of reference.
\Gor*
%

% % "x,z Statement"
%\begin{lem}[modified Gordon's Lemma]\label{lemma:Gor}
%Let $\mathbf{G}\in\mathbb{R}^{m\times n}$ be an i.i.d. matrix with standard normal entries. Further let $\g\sim\Nn(0,\mathbf{I}_m)$, $\h\sim\Nn(0,\mathbf{I}_n)$ and $g\sim\Nn(0,1)$ and assume all $\mathbf{G},\g,\h,g$ are independently generated. Finally, let $\Phi_1\subset\mathbf{R}^m$, $\Phi_2$ be arbitrary sets, $C>0$ some constant and $\psi:\Phi_1\times\Phi_2\times\mathbb{R}^m\rightarrow\mathbb{R}$ arbitrary function. Then, for any $c\in\mathbb{R}$:
%\begin{align}\label{eq:inGor}
%\Pro\left(  \min_{\x\in\Phi_1, \z\in\Phi_2}~~\max_{\|\ab\| = C}~\left\{ \ab^T\Gb\x  - \psi({\x,\z,\ab} ) \right\} \geq c\right) \geq 2\Pro\left(  \min_{\x\in\Phi_1, \z\in\Phi_2}~~\max_{\|\ab\| = C}~\left\{ \|\x \| \g^T\ab  -  C \h^T\x - \psi({\x,\z,\ab})  \right\}\geq c  \right) - 1.
%\end{align}
%\end{lem}

Our proof will closely parallel the proof of the original Gordon's Lemma $3.1$ in \cite{Gor}. 
%For the second part we require the following result.
%\begin{lem} \label{lemma:helper1}Let $a\in\mathbb{R}$ be an arbitrary random variable and let $g$ be a scalar Gaussian random variable with arbitrary variance and independent of $a$. If
%\beq\notag
%\Pro(a+g\geq t)\geq 1-p,
%\eeq
%then,
%\beq
%\Pro(a\geq t)\geq 1-2p.
%\eeq
%\end{lem}
%\begin{proof}
%We have 
%\begin{align}
%p\geq \Pro(a + g\leq t) \geq \Pro(a\leq t, g\leq 0) = \Pro(a\leq t)\Pro(g\leq 0) = \frac{\Pro(a\leq t)}{2}.
%\end{align}
%which completes the proof.
%%Observe that,
%%\beq
%%\Pro(a+g\leq t)\geq \Pro(a\leq t,g\leq 0)=\Pro(a\leq t)\Pro(g\leq 0\big|a\leq t)=\frac{\Pro(a\leq t)}{2}
%%\eeq
%%, $\Pro(a\geq t)\geq 1-2\Pro(a+g\leq t)\geq 1-2u$.
%\end{proof}
%
%We now proceed with the proof of Lemma \ref{lemma:Gor}.
%

\begin{proof}

For $\x\in\Phi_1$ and $\ab\in \Phi_2$ define the two processes,
\begin{align*}
Y_{\x,\ab} = \x^T \Gb \ab + \|\ab\| \|\x\| g\quad \text{ and }\quad
X_{\x,\ab} =  \|\x \| \g^T\ab  -  \|\ab\| \h^T\x
\end{align*}
where $\Gb,\g,\h$ are as defined in the statement of the lemma and $g\sim\mathcal{N}(0,1)$ and independent of the other. We show that the processes defined satisfy the conditions of Gordon's Theorem \ref{thm:GordonMain}:
\begin{align*}
\E[X_{\x,\ab}^2] =  \|\x \|^2 \|\ab\|^2 + \|\ab\|^2 \|\x\|^2 = \E[Y_{\x,\ab}^2],
\end{align*}
and
\begin{align*}
\E[ X_{\x,\ab}X_{\x',\ab'} ] - \E[ Y_{\x,\ab}Y_{\x',\ab'} ]  &= \|\x\|\|\x'\|(\ab^T\ab') + \|\ab\|^2(\x^T\x') -  (\x^T\x') (\ab^T\ab') - \|\ab\|\|\ab'\| \|\x\|\|\x'\| \notag \\
&= \left(\underbrace{ \|\x\|\|\x'\| -  (\x^T\x') }_{\geq 0} \right) \left( \underbrace{ (\ab^T\ab') - \|\ab\|\|\ab'\| }_{\leq 0} \right),
\end{align*}
which is non positive  and equal to zero when $x=x'$.
Also, on the way of applying Theorem \ref{thm:GordonMain} for the two processes defined above, let 
\begin{align*}
\la_{\x,\ab} = \psi(\x,\ab)+c.
\end{align*}

%Combining \eqref{eq:re1}, \eqref{eq:re2} and \eqref{eq:re3}, it is be clear how Theorem \ref{thm:GordonMain} is applicable. 
The only caveat in directly applying Theorem \ref{thm:GordonMain} is now that  it requires the processes to be discrete. This technicality is addressed by Gordon in \cite{Gor} (see  Lemma $3.1$ therein), for the case where $\Phi_1$ is arbitrary and $\Phi_2$ is a scaled unit sphere.
% For completeness, we also provide 
In Lemma \ref{discrete to cont}, we show that the minimax inequality can be translated from discrete to continuous processes, as well, in the case where both $\Phi_1$ and $\Phi_2$ are compact sets. % The argument is identical, thus, we omit  for brevity. 
%  However, following identically  the same arguments developed by Gordon in Lemma $3.1$ in \cite{Gor} we can circumvent this issue.
%We omit the details for brevity and state the result of 
To conclude, applying  Theorem \ref{thm:GordonMain} we have,
\begin{align}
&\Pro\left(  \min_{\x\in\Phi_1}~~\max_{\ab\in\Phi_2}~\left\{ \ab^T\Gb\x + \|\ab\|\|\x\| g  - \psi({\x,\ab} ) \right\} \geq c\right) \geq \notag \\
&\qquad\qquad\qquad\qquad\qquad\qquad\Pro\left(  \min_{\x\in\Phi_1}~~\max_{\ab\in\Phi_2}~\left\{ \|\x \| \g^T\ab  -  \|\ab\| \h^T\x - \psi({\x,\ab})  \right\}\geq c  \right) := q. \label{eq:bbone}
\end{align}
Since $g\sim\Nn(0,1)$, we can write the left hand side of \eqref{eq:bbone} as, $p=\frac{p_++p_-}{2}$ where we define $p_+,p_-,p_0$ as,
\begin{align}
&p_-=\Pro\left(  \min_{\x\in\Phi_1}~~\max_{\ab\in\Phi_2}~\left\{ \ab^T\Gb\x + \|\ab\|\|\x\| g  - \psi({\x,\ab} ) \right\} \geq c~\big|~g\leq 0\right),\nn\\
&p_+=\Pro\left(  \min_{\x\in\Phi_1}~~\max_{\ab\in\Phi_2}~\left\{ \ab^T\Gb\x + \|\ab\|\|\x\| g  - \psi({\x,\ab} ) \right\} \geq c~\big|~g>0\right),\nn\\
&p_0=\Pro\left(  \min_{\x\in\Phi_1}~~\max_{\ab\in\Phi_2}~\left\{ \ab^T\Gb\x   - \psi({\x,\ab} ) \right\} \geq c\right)\nn
\end{align}
By construction and independence of $g,\Gb$; $1\geq p_+\geq p_0\geq p_-$. On the other hand, $1-q\geq 1-p\geq \frac{1-p_-}{2}$ which implies, $p_-\geq 2q-1$. This further yields $p_0\geq 2q-1$, which is what we want.

%
%Now, suppose that $\x^*$ is optimal for the optimization problem,%\samet{careful here!}
%\begin{align*}
%\min_{\x\in\Phi_1}~~\max_{\ab \in\Phi_2}~\left\{ \ab^T\Gb\x - \psi({\x,\ab} )  \right\} .
%\end{align*}
%Clearly then,
%\begin{align}\label{eq:cbone}
%\min_{\x\in\Phi_1}~~\max_{\ab\in\Phi_2}~\left\{ \ab^T\Gb\x - \psi({\x,\ab} )  \right\}  + \|\ab\|\|\x^*\|g \geq \min_{\x\in\Phi_1}~~\max_{\ab\in\Phi_2}~\left\{ \ab^T\Gb\x + \|\ab\|\|\x\| g  - \psi({\x,\ab} )  \right\}.
%\end{align}
%Combining \eqref{eq:bbone} and \eqref{eq:cbone} we have that,
%\begin{align*}\label{eq:dbone}
%\Pro\left( \min_{\x\in\Phi_1}~~\max_{\ab\in\Phi_2}~\left\{ \ab^T\Gb\x - \psi({\x,\ab} )  \right\}  + \|\ab\|\|\x^*\|g \geq c \right) \geq q.
%\end{align*}
%for all $c\in\mathbb{R}$.
%To conclude the proof notice that $\|\ab\|\|\x^*\|g$ is a centered gaussian random variable independent of $\Gb$ and invoke Lemma \ref{lemma:helper1}.
\end{proof}

\begin{lem} \label{discrete to cont}Let $\Gb\in\R^{m\times n},\g\in\R^m,\h\in\R^n,g\in\R$ be independent with i.i.d. standard normal entries. Let $\Phi_1\subset\R^n,\Phi_2\subset\R^m$ be compact sets. Let $\psi(\cdot,\cdot):\R^{n}\times\R^m\rightarrow \R$ be a continuous function. Assume, for all finite sets $S_1\subset\Phi_1,S_2\subset \Phi_2$ and $c\in\R$, we have,
\beq
\Pro(\min_{\x\in S_1}~\max_{\ab\in S_2}~\{\ab^T\Gb\x-\psi(\x,\ab)\}\geq c)\geq \Pro(\min_{\x\in S_1}~\max_{\ab\in S_2}~\{\|\x\|\g^T\ab-\|\ab\|\h^T\x-\psi(\x,\ab)\}\geq c)\nn
\eeq
Then,
\beq
\Pro(\min_{\x\in\Phi_1}~\max_{\ab\in \Phi_2}~\{\ab^T\Gb\x-\psi(\x,\ab)\}\geq c)\geq \Pro(\min_{\x\in\Phi_1}~\max_{\ab\in \Phi_2}~\{\|\x\|\g^T\ab-\|\ab\|\h^T\x-\psi(\x,\ab)\}\geq c)\nn
\eeq
\end{lem}
\begin{proof} %To begin with, observe that,
%\beq
%\eeq
Let $R(\Phi_i)=\sup_{\vb\in \Phi_i}\|\vb\|$ for $1\leq i\leq 2$. Let $S_1\subset \Phi_1,S_2\subset\Phi_2$ be arbitrary $\eps$-coverings of the sets $\Phi_1,\Phi_2$ so that, for any $\vb\in \Phi_i$, there exists $\vb'\in S_i$ satisfying $\|\vb'-\vb\|\leq \eps$. Furthermore, using continuity of $\psi$ over the compact set $\Phi_1\times \Phi_2$, for any $\delta>0$, we can choose $\eps$ sufficiently small to guarantee that $|\psi(\x,\ab)-\psi(\x',\ab')|<\delta$. Here $\delta$ can be made arbitrarily small as a function of $\eps$. Now, for any $\x\in \Phi_1,\ab\in \Phi_2$, pick $\x',\ab'$ in the $\eps$-coverings $S_1,S_2$. This gives,
\beq
|[\ab^T\Gb\x-\psi(\x,\ab)]-[{\ab'}^T\Gb\x'-\psi(\x',\ab')]| \leq \eps( R(\Phi_1)+R(\Phi_2)+\eps)\|\Gb\|_2+\delta\label{c.7}
\eeq
\beq
|[\|\x\|\g^T\ab-\|\ab\|\h^T\x-\psi(\x,\ab)]-[\|\x'\|\g^T\ab'-\|\ab'\|\h^T\x'-\psi(\x',\ab')]| \leq \eps( R(\Phi_1)+R(\Phi_2)+\eps)(\|\g\|+\|\h\|)+\delta\label{c.8}
\eeq
Next, using Lipschitzness of $\|\g\|,\|\h\|,\|\Gb\|_2$ and Lemma \ref{lemma:conc4}, for $t>1$, we have,
\beq
\Pro(\max\{\|\g\|+\|\h\|,\|\Gb\|_2\}\leq t(\sqrt{n}+\sqrt{m}))\geq 1-4\exp(-\frac{(t-1)^2(m+n)}{2}):=p(t)\label{abcpt}
\eeq 
Let $C(t,\eps)=t\eps( R(\Phi_1)+R(\Phi_2)+\eps)(\sqrt{m}+\sqrt{n})+\delta$. Then, since \eqref{c.7} and \eqref{c.8} holds for all $\ab,\x$, using \eqref{abcpt},
\begin{align}
& \Pro(\min_{\x\in\Phi_1}\max_{\ab\in \Phi_2}\{\ab^T\Gb\x-\psi(\x,\ab)\}\geq c-C(t,\eps))\geq \Pro(\min_{\x\in S_1}\max_{\ab\in S_2}\{\ab^T\Gb\x-\psi(\x,\ab)\}\geq c)-p(t)\label{c.10}\\
& \Pro(\min_{\x\in S_1}\max_{\ab\in S_2}\{\|\x\|\g^T\ab-\|\ab\|\h^T\x-\psi(\x,\ab)\}\geq c)\geq  \Pro(\min_{\x\in\Phi_1}\max_{\ab\in \Phi_2}\{\|\x\|\g^T\ab-\|\ab\|\h^T\x-\psi(\x,\ab)\}\geq c+C(t,\eps))-p(t)\label{c.11}
\end{align}
Combining \eqref{c.10} and \eqref{c.11}, for all $\eps>0,t>1$, the following holds,
\beq
 \Pro(\min_{\x\in\Phi_1}\max_{\ab\in \Phi_2}\{\ab^T\Gb\x-\psi(\x,\ab)\}\geq c-C(t,\eps))\geq \Pro(\min_{\x\in\Phi_1}\max_{\ab\in \Phi_2}\{\|\x\|\g^T\ab-\|\ab\|\h^T\x-\psi(\x,\ab)\}\geq c+C(t,\eps))-2p(t)\nn
\eeq
Setting $t=\eps^{-1/2}$ and letting $\eps\rightarrow 0$, we obtain the desired result as $C(t,\eps),p(t),\delta\rightarrow0$.
\end{proof}

\section{The Dual of the LASSO} %Problem \eqref{eq:generic}}
\label{sec:dual}

To derive the dual  we write the problem in \eqref{eq:generic} equivalently as
\begin{align*}
\Fco(\A,\vb) = &\min_{\w,\bb} \left\{ \| \mathbf{b} \| + p(\w) \right\}\\
&~\text{s.t.}~~~~ \bb = \A\w - \sigma\vb,
\end{align*}
and then reduce it to 
\begin{align*}
&\min_{\w,\bb}\max_{\mub} \left\{ \| \mathbf{b} \| + \mub^T\left(\bb -  \A\w + \sigma\vb\right)+ p(\w) \right\}.
\end{align*}
The dual of the problem above is
\begin{align}\label{eq:comd1}
\max_{\mub}{\min_{\w,\bb}{ \{ \| \mathbf{b} \| + \mub^T(\bb -  \A\w + \sigma\vb ) + p(\w) \} }}.
\end{align}
The minimization over $\mathbf{b}$ above is easy to perform. A simple application of the Cauchy--Schwarz inequality gives
  \begin{align*}
 \| \mathbf{b} \| + \mub^T\mathbf{b} &\geq \| \mathbf{b} \| - \| \mathbf{b} \| \| \mub \| \\
 & = \left( 1 - \| \mub \| \right) \| \mathbf{b} \|.
 \end{align*}
Thus,
\begin{align*}
\min_{\mathbf{b}}{ \left\{ \| \mathbf{b} \| + {\mub}^T
 \mathbf{b} \right\} } = \begin{cases} 
   0  &, \|\mub\|\leq 1, \\
 -\infty &, o.w. .
  \end{cases}
  \end{align*}
Combining this with \eqref{eq:comd1} we conclude that the dual problem of the problem in \eqref{eq:generic} is the following:
\begin{align*}
&\max_{\| \mub \|\leq 1}\min_{\w} \left\{ \mub^T\left(- \A\w + \sigma\vb\right)+ p(\w) \right\}.
\end{align*}
We equivalently rewrite the dual problem in the format of  a minimization problem as follows:
\begin{align}\label{eq:comd}
-&\min_{\| \mub \|\leq 1}\max_{\w} \left\{ \mub^T\left( \A\w - \sigma\vb\right)- p(\w) \right\}.
\end{align}
If $p(\w)$ is a finite convex function from $\R^n\rightarrow \R$, the problem in \eqref{eq:generic} is convex and satisfies Slater's conditions. When $p(\w)$ is the indicator function of a convex set $\{\w\big|g(\w)\leq 0\}$, the problem can be viewed as $\min_{g(\w)\leq 0,\bb} \left\{\| \mathbf{b} \| + \mub^T\left(\bb -  \A\w + \sigma\vb\right)  \right\}$. For strong duality, we need strict feasibility, i.e., there must exist $\w$ satisfying $g(\w)<0$. In our setup, $g(\w)=f(\x_0+\w)-f(\x_0)$ and $\x_0$ is not a minimizer of $f(\cdot)$, hence strong duality holds and thus problems in \eqref{eq:generic} and \eqref{eq:comd} have the same optimal cost $\Fco(\A,\vb)$.

\section{Proofs for Section \ref{sec:KeyIdeas}}\label{sec:proofOfKeyOpt}
%\eqname{sec:proofOfKeyOpt}

\subsection{Proof of Lemma \ref{lemma:lowKey}}

We prove the statements of the Lemma in the order that they appear.
\subsubsection{Scalarization}\label{sec:F_low scalar}
The first statement of Lemma \ref{lemma:lowKey} claims that the optimization problem  in \eqref{eq:lemFlow} can be reduced into a one dimensional optimization problem. To see this begin by evaluating the optimization over $\w$ for fixed $\|\w\|$:
% In particular, we have,
\begin{align}
\Lc(\g,\h) &= \min_{\w}\left\{ \sqrt{\|\w\|^2 + \sigma^2}\|\g\| - \h^T\w + \max_{\s\in\Cc} \s^T\w \right\} \notag\\
&= \min_{\substack{
   \w:\|\w\|=\alpha \notag \\
   \alpha\geq 0
  }} \left\{ \sqrt{\|\w\|^2 + \sigma^2}\|\g\| - \h^T\w + \max_{\s\in\Cc} \s^T\w \right\} \\
  &= \min_{
   \alpha\geq 0
  } \left\{ \sqrt{\alpha^2 + \sigma^2}\|\g\| + \min_{\w:\|\w\|=\alpha} \left\{ - \h^T\w + \max_{\s\in\Cc} \s^T\w \right\}\right\} \notag\\
  &= \min_{
   \alpha\geq 0
  } \left\{ \sqrt{\alpha^2 + \sigma^2}\|\g\| - \max_{\w:\|\w\|=\alpha} \left\{\h^T\w - \min_{\s\in\Cc} \s^T\w\right\}\right\}\notag\\ \label{eq:r42}
    &= \min_{
   \alpha\geq 0
  } \left\{ \sqrt{\alpha^2 + \sigma^2}\|\g\| - \max_{\w:\|\w\|=\alpha}~\min_{\s\in\Cc}~\left\{(\h - \s)^T\w\right\}\right\}
\end{align}
To further simplify $\eqref{eq:r42}$, we use the following key observation as summarized in the Lemma below.
\begin{lem} \label{lemma:a property}Let $\Cc\in\R^n$ be a nonempty convex set in $\mathbb{R}^n$, $\h\in\R^n$ and $\alpha\geq0$. Then,
\begin{align}
\max_{\w:\|\w\|=\alpha}~\min_{\s\in\Cc}~\left\{(\h - \s)^T\w\right\} ~=~ \min_{\s\in\Cc}~\max_{\w:\|\w\|=\alpha}~\left\{(\h - \s)^T\w\right\}.\nn
\end{align}
Thus,
\begin{align}
\max_{\w:\|\w\|=\alpha}~\min_{\s\in\Cc}~\left\{(\h - \s)^T\w\right\} ~=~ \alpha\cdot\dt(\h,\Cc),\nn
\end{align}
and the optimum is attained at $\w^*=\alpha\cdot\frac{\Pi(\h,\Cc)}{\dt(\h,\Cc)}.$
\end{lem}
%
%\chris{I realized that this lemma is in fact same thing as your Lemma E.1 in the Appendix. Over there you prove it geometrically, we can keep which ever you like. }
\begin{proof}
First notice that 
$$\min_{\s\in\Cc}~\max_{\w:\|\w\|=\alpha}~(\h - \s)^T\w = \min_{\s\in\Cc}~\alpha\|\h - \s\| =  \alpha\cdot\dt(\h,\Cc).
$$ Furthermore, MinMax is never less than MaxMin \cite{Boyd}.
%\samet{cite here} 
Thus,
\begin{align*}
\max_{\w:\|\w\|=\alpha}~\min_{\s\in\Cc}~\left\{(\h - \s)^T\w\right\} ~\leq~ \min_{\s\in\Cc}~\max_{\w:\|\w\|=\alpha}~\left\{(\h - \s)^T\w\right\} = \alpha\cdot\dt(\h,\Cc).
\end{align*}
It suffices to prove that 
\begin{align}\nn
\max_{\w:\|\w\|=\alpha}~\min_{\s\in\Cc}~\left\{(\h - \s)^T\w\right\} ~\geq~  \alpha\cdot\dt(\h,\Cc).
\end{align}
Consider $\w^*=\alpha\cdot\frac{\bi(\h,\Cc)}{\dt(\h,\Cc)}$. Clearly,
\begin{align*}\label{eq:2prove53}
\max_{\w:\|\w\|=\alpha}~\min_{\s\in\Cc}~\left\{(\h - \s)^T\w\right\} ~\geq~  \min_{\s\in\Cc}~\left\{(\h - \s)^T\w^*\right\}. 
\end{align*}
But,
\begin{align}
\min_{\s\in\Cc}~\left\{(\h - \s)^T\w^*\right\} &= \frac{\alpha}{\dt(\h,\Cc)}\cdot\left( \h^T\bi(\h,\Cc) - \max_{\s\in\Cc}~\s^T \bi(\h,\Cc) \right)\\
 &= \frac{\alpha}{\dt(\h,\Cc)}\cdot\left( \h^T\bi(\h,\Cc) - \bu(\h,\Cc)^T \bi(\h,\Cc)\right)\label{eq:toproj}\\
 & = \alpha\cdot\dt(\h,\Cc),\notag
\end{align}
where \eqref{eq:toproj} follows from Fact \ref{prom}.
This completes the proof of the Lemma. %\samet{justify \eqref{toproj}}
\end{proof}

Applying the result of Lemma \ref{lemma:a property} to \eqref{eq:r42}, we conclude that
\begin{align}
\Lc(\g,\h) &= \min_{\w}\left\{ \sqrt{\|\w\|^2 + \sigma^2}\|\g\| - \h^T\w + \max_{\s\in\Cc} \s^T\w \right\} \notag\\
& = \min_{
   \alpha\geq 0
  } \left\{ \sqrt{\alpha^2 + \sigma^2}\|\g\| - \alpha\cdot\dt(\h,\Cc)\right\} \label{eq:oneDim}
\end{align}

\subsubsection{Deterministic Result}
The optimization problem in \eqref{eq:oneDim} is one dimensional and easy to handle. Setting the derivative of its objective function equal to zero and solving for the optimal $\alpha^*$, under the assumption that 
\begin{align}
\|\g\|^2 > \dt(\h,\Cc)^2\label{eq:ass_low1},
\end{align}
 it only takes a few simple calculations to prove the second statement of Lemma \ref{lemma:lowKey}, i.e. 
$$(\alpha^*)^2 = \|\w_{low}^*(\g,\h)\|^2 = \sigma^2\frac{\dt^2(\h,\Cc)}{\|\g\|^2 - \dt^2(\h,\Cc)}$$ and, 
\begin{align}
\Lc(\g,\h) = \sigma\sqrt{ \|\g\|^2 - \dt^2(\h,\Cc) }. \label{eq:det_low}
\end{align}

\subsubsection{Probabilistic Result}
%
%\chris{Definitely need to check all that: }

Next, we prove the high probability lower bound for $\Lc(\g,\h)$ implied by the last statement of Lemma \ref{lemma:lowKey}. To do this, we will make use of concentration results for specific functions of Gaussian vectors as they are stated in Lemma \ref{lemma:concentrationALL}. Setting $t=\delta\sqrt{m}$ in Lemma \ref{lemma:concentrationALL}, with probability $1-8\exp(-c_0\delta^2m)$,
\begin{align*}
&|\|\g\|^2-m|\leq 2\delta m+\delta^2m+1,\\
&| \dt^2(\h,\Cc)-\DC|\leq 2\delta\sqrt{\DC m} +\delta^2m+1\leq  2\delta m +\delta^2m+1.
\end{align*}
Combining these and using the assumption  that $m\geq\DC+\eps_L m$, we find that
\begin{align*}
\|\g\|^2-\dt^2(\h,\Cc)&\geq m-\DC-[(2\delta^2+4\delta)m+2]\notag \\
&\geq m-\DC-[(2\delta^2+4\delta)\frac{m-\DC}{\eps_L}+2]\notag\\
&\geq (m-\DC)[1-\frac{(2\delta^2+4\delta)}{\eps_L}]-2,
\end{align*}
with the same probability.
Choose $\eps'$ so that $\sqrt{1-\eps'}=1-\eps$. Also, choose $\delta$ such that $\frac{(2\delta^2+4\delta)}{\eps_L}<\frac{\eps'}{2}$ and $m$ sufficiently large to ensure $\eps_L\eps' m>4$. Combined,
\begin{align}
\|\g\|^2-\dt^2(\h,\Cc)\geq  (m-\DC)(1-\frac{\eps'}{2})-2\geq (m-\DC)(1-\eps'), \label{eq:last24}
\end{align}
with probability $1-8\exp(-c_0\delta^2m)$.
Since the right hand side in \eqref{eq:last24} is positive, it follows from the second statement of Lemma \ref{lemma:lowKey} that 
 $$
 \Lc(\g,\h)\geq \sigma \sqrt{(m-\DC)(1-\eps')}=\sigma(1-\eps)\sqrt{m-\DC},
$$
with the same probability. This concludes the proof.
\subsection{Proof of Lemma \ref{lemma:upKey}}
\subsubsection{Scalarization}
We have
\begin{align}
\Uc(\g,\h) &= -\min_{\|\mu\| \leq 1}~\max_{\|\w\|=C_{up}}\left\{ \sqrt{C_{up}^2 + \sigma^2}~\g^T\mub + \|\mub\|\h^T\w - \max_{\s\in\Cc} \s^T\w \right\} \notag \\
&= -\min_{\|\mu\| \leq 1}\left\{ \sqrt{C_{up}^2 + \sigma^2}~\g^T\mub + \max_{\|\w\|=C_{up}}\left\{ \|\mub\|\h^T\w - \max_{\s\in\Cc} \s^T\w \right\}\right\}. \label{eq:kk0}
\end{align}
Notice that
\begin{align}
\max_{\|\w\|=C_{up}}\left\{ \|\mub\|\h^T\w - \max_{\s\in\Cc} \s^T\w \right\} &=
\max_{\|\w\|=C_{up}}~\min_{\s\in\Cc} ~( \|\mub\|\h - \s )^T \w \notag\\
&= C_{up} \dt(\|\mub\|\h,\Cc).\label{eq:kk1}
\end{align}
where \eqref{eq:kk1} follows directly from Lemma \ref{lemma:a property}.
Combine \eqref{eq:kk0} and \eqref{eq:kk1} to conclude that
\begin{align}
\Uc(\g,\h) &= -\min_{\|\mu\| \leq 1}\left\{ \sqrt{C_{up}^2 + \sigma^2}~\g^T\mub +   C_{up} \dt(\|\mub\|\h,\Cc)\right\}\notag\\
&= -\min_{0\leq \alpha\leq 1}\left\{ -\alpha\cdot\sqrt{C_{up}^2 + \sigma^2}~\|\g\| + C_{up}\dt(\alpha\h,\Cc)\right\}
\label{eq:scalar_fup}.
\end{align}

\subsubsection{Deterministic Result}
For convenience denote the  objective function of problem \eqref{eq:scalar_fup} as
$$
\phi(\alpha) = C_{up} \dt(\alpha\h,\Cc)  - \alpha\sqrt{C_{up}^2+\sigma^2}\|\g\|.
$$
Notice that $\phi(\cdot)$ is convex. By way of justification, $\dt(\alpha\h,\Cc)$ is a convex function for $\alpha\geq 0$ \cite{Roc70}, and $\alpha\sqrt{C^2+\sigma^2}\|\g\|$ is linear in $\alpha$. 
Denote $\alpha^* = \operatorname{argmin}\phi{(\alpha)}$. Clearly, it  suffices to show that $\alpha^*=1$. First, we prove that $\phi(\alpha)$ is differentiable as a function of $\alpha$ at $\alpha=1$. For this, we make use of the following lemma.
\begin{lem}\label{lemma:dist_diff}
Let $C$ be a nonempty closed and convex set and $\h\notin C$. Then 
\begin{align*}
\lim_{\eps\rightarrow 0}\frac{ \dt(\h + \eps \h,C) - \dt(\h,C) }{\eps} = \langle \h, \frac{\Pi(\h,C)}{\|\Pi(\h,C)\|} \rangle,
\end{align*}
\end{lem}
\begin{proof}
%\samet{Can you check the modification?}
Let $H$ be a hyperplane of $\Cc$ at $\bu(\h,\Cc)$ orthogonal to $\Pi(\h,C)$. Using the second statement of Fact \ref{prom}, $H$ is a supporting hyperplane and $\h$ and $C$ lie on different half planes induced by $H$ (also see \cite{Boyd}).
%Let $H$ be the hyperplane passing through $\bu(\h,C)$ that is orthogonal to $\Pi(\h,C)$. 
%From Fact \ref{prom}. . 
Also, observe that $\Pi(\h,\Cc)=\Pi(\h,H)$ and $\bu(\h,\Cc)=\bu(\h,H)$. Choose $\eps>0$ sufficiently small such that $(1+\eps)\h$ lies on the same half-plane as $\h$. We then have,
\beq\label{eq:l11}
\|\Pi((1+\eps)\h,\Cc)\|\geq \|\Pi((1+\eps)\h,H)\|=\|\Pi(\h,\Cc)\|+\li\eps\h,\frac{\Pi(\h,\Cc)}{\|\Pi(\h,\Cc)\|}\ri.
\eeq
Denote the $n-1$ dimensional subspace that is orthogonal to $\bi(\h,H)$ and parallel to $H$ by $H_0$. Decomposing $\eps\h$ to its orthonormal components along $\Pi(\h,H)$ and $H_0$, we have
\beq\label{eq:l22}
\|\Pi((1+\eps)\h,C)\|^2\leq \| (1+\eps)\h - \bu(\h,C) \|^2 =\left(\|\Pi(\h,C)\|+\li\eps\h,\frac{\Pi(\h,C)}{\|\Pi(\h,C)\|}\ri\right)^2+\eps^2\|\bu(\h,H_0)\|^2.
% \|\Pi((1+\eps)\h,\bu(\h,C))\|^2
\eeq
Take square roots in both sides of \eqref{eq:l22} and apply on the right hand side the useful inequality $\sqrt{a^2+b^2}\leq a+\frac{b^2}{2a}$, which is true for all $a,b\in\mathbb{R}^+$. Combine the result with the lower bound in \eqref{eq:l11} and let $\eps\rightarrow0$ to conclude the proof.
%
% with some further manipulations we can write
%\begin{align}
%\|\Pi((1+\eps)\h,C)\| &\leq \sqrt{ \left(\|\Pi(\h,C)\|+\li\eps\h,\frac{\Pi(\h,C)}{\|\Pi(\h,C)\|}\ri\right)^2+\eps^2\|\bu(\h,H_0)\|^2 } \notag\\
%&\leq \|\Pi(\h,C)\|+\li\eps\h,\frac{\Pi(\h,C)}{\|\Pi(\h,C)\|}\ri + \frac{\eps\|\bu(\h,H_0)\|}{2\left(\|\Pi(\h,C)\|+\li\eps\h,\frac{\Pi(\h,C)}{\|\Pi(\h,C)\|}\ri\right)}
%\end{align}
%By combining the upper and lower bound and letting $\eps\rightarrow0$ we obtain the desired result. 
%%\chris{is it that obvious?}
\end{proof}

Since  $\h\notin\Cc$, it follows from Lemma \ref{lemma:dist_diff}, that $\dt(\alpha\h,\Cc)$ is differentiable as a function of $\alpha$ at $\alpha=1$, implying the same result for $\phi(\alpha)$. In fact, we have
$$
\phi'(1) = C_{up}\dt(\h,\Cc) + C_{up}\frac{ \li \bi(\h,\Cc) , \bu(\h,\Cc) \ri  }{ \dt(\h,\Cc) } - \sqrt{C_{up}^2+\sigma^2}\|\g\| < 0,
$$
where  the negativity follows from assumption \eqref{eq:ass_up1}. To conclude the proof, we make use of the following simple lemma.

\begin{lem}\label{lemma:simpleConvex_1}
Suppose $f:\mathbb{R}\rightarrow\mathbb{R}$ is a convex function, that is differentiable at $x_0\in\mathbb{R}$ and $f'(x_0)<0$. Then, 
$f(x) \geq f(x_0) $ for all $\ x\leq x_0.$
\end{lem}
\begin{proof}
By convexity of $f(\cdot)$, for all $x\leq x_0$:
\begin{align*}
f(x) &\geq f(x_0) + \underbrace{f'(x_0)}_{<0}\underbrace{(x-x_0)}_{\leq 0}\\
&\geq f(x_0)
\end{align*}
\end{proof}
Applying Lemma \ref{lemma:simpleConvex_1} for the convex function $\phi(\cdot)$ at $\alpha=1$, gives that $\phi(\alpha)\geq\phi(1)$ for all $\alpha\in[0,1]$. Therefore, $\alpha^*=1.$

%From Lemma \ref{lemma:distDiff}, for any $\eps>0$ we have
%\begin{align}
%\phi(1-\eps) - \phi(1) &= C_{up}\left( \dt_\la(\h - \eps\h) - \dt(\h,\Cc)  \right) - \eps\sqrt{C_{up}^2+\sigma^2}\|\g\|\\
%&\geq -\eps\left[ C_{up} \langle \h, \frac{\Pi(\h)}{\dt(\h,\Cc)} \rangle  - \sqrt{C_{up}^2+\sigma^2}\|\g\|\right]\\
%&= -\eps\left[ C_{up}\dt(\h,\Cc) + C_{up}\frac{\text{cor}_\la(\h)}{\dt_\la{(\h)}} - \sqrt{C_{up}^2+\sigma^2}\|\g\|\right]
%\end{align}
%\end{proof}

%
\subsubsection{Probabilistic Result}
We consider the setting where $m$ is sufficiently large and,% $m\geq\DC\gg \sqrt{n}$ and 
\begin{align}\label{eq:regime}
(1-\eps_L)m\geq \max\left({\DC+\CC,\DC}\right),~~~\DC\geq \eps_Lm
\end{align}
Choose $C_{up}=\sigma\sqrt{\frac{\DC}{m-\DC}}$ which would give $C_{up}^2+\sigma^2=\sigma^2\frac{m}{m-\DC}$. Hence, the assumption \eqref{eq:ass_up1} in the second statement of Lemma \ref{lemma:upKey} can be rewritten as,
\begin{align}\label{eq:ass_again}
\sqrt{m}\| \g \|\dt(\h,\Cc) > \sqrt{\DC}( \dt(\h,\Cc)^2 + \corr(\h,\Cc) ).
\end{align} 
 The proof technique is as follows. We first show that \eqref{eq:ass_again} (and thus \eqref{eq:ass_up1}) holds with high probability. Also, that $\h\notin\Cc$ with high probability. Then, as a last step we make use of the second statement of Lemma \ref{lemma:upKey} to compute the lower bound on $\Uc$.

\noindent $\bullet$ \emph{\eqref{eq:ass_up1} holds with high probability:}

%\noindent $\bullet$ \emph{Lower bounding the left hand side of \eqref{eq:ass_again}:} 
Using standard concentration arguments (see Lemma \ref{lemma:conc4}), we have
\begin{align*}
\sqrt{m}\| \g \|\dt(\h,\Cc) \geq \sqrt{m}(\sqrt{m-1}-t)(\sqrt{\DC-1}-t)
\end{align*}
with probability $1-4\exp\left( \frac{-t^2}{2} \right)$. Choose a sufficiently small constant $\delta>0$ and set $t=\delta\sqrt{\DC}$ to ensure,
\beq
\sqrt{m}\| \g \|\dt(\h,\Cc)\geq (1-\frac{\eps_L}{2})m\sqrt{\DC}\label{eq:rr2}
\eeq
with probability $1-\exp(-\order{m})$, where we used $(1-\eps_L)\geq \DC\geq \eps_Lm$. In particular, for sufficiently large $\DC$ we need $(1-\delta)^2>1-\frac{\eps_L}{2}$.

Equation \eqref{eq:rr2} establishes a high probability lower bound for the expression at the left hand side of \eqref{eq:ass_again}.  Next, we show that the expression at the right hand side of \eqref{eq:ass_again} is upper bounded with high probability by the same quantity.

{\bf{Case 1:}} If $\Cc$ is a cone, $\corr(\h,\Cc)=0$ and using Lemma \ref{lemma:concentrationALL} $\dt(\h,\Cc)^2\leq \DC+2t\sqrt{\DC}+t^2\leq (1-\eps_L)m+2t\sqrt{m}+t^2$ with probability $1-2\exp(-\frac{t^2}{2})$. Hence, we can choose $t=\eps\sqrt{m}$ for a small constant $\eps>0$ to ensure, $\dt(\h,\Cc)^2<(1-\frac{\eps_L}{2})m$ with probability $1-\exp(-\order{m})$. This gives \eqref{eq:ass_again} in combination with \eqref{eq:rr2}.
%\noindent$\bullet$\emph{Upper bounding the right hand side of \eqref{eq:ass_again}:} 
%%We consider two cases depending on the value of $\PC$:

{\bf{Case 2:}} Otherwise, from Lemma \ref{DCbound}, we have that $\PC\leq 2(n+\DC)$ and from \eqref{eq:regime}, $m\geq \DC$. Then, applying Lemma \ref{lemma:concentrationALL}, we have
\begin{align*}
\dt(\h,\Cc)^2 + \corr(\h,\Cc)  &\leq \DC + \CC + 3t\underbrace{\sqrt{\DC}}_{\leq \sqrt{m}} + t\underbrace{\sqrt{\PC}}_{\leq\sqrt{2(n+m)}}  +2(t^2+1)\notag \\
&\leq \DC + \CC + 3t\sqrt{m} + t\sqrt{2(n+m)}  +2(t^2+1) \\
&\leq (1-\eps_L)m + 3t\sqrt{m} + t\sqrt{2(n+m)}  +2(t^2+1). 
\end{align*}
with probability $1-4\exp\left( \frac{-t^2}{2} \right)$. Therefore, with the same probability,
\begin{align}\label{eq:rr1}
\sqrt{\DC}( \dt(\h,\Cc)^2 + \corr(\h,\Cc) )  \leq (1-\eps_L)m\sqrt{\DC} + 3t\sqrt{m}\sqrt{\DC} + t\sqrt{2(n+m)}\sqrt{\DC}  +2(t^2+1)\sqrt{\DC}
\end{align}
Comparing the right hand sides of inequalities \ref{eq:rr2} and \ref{eq:rr1} , we need to ensure that,
\beq
 3t\sqrt{m}\sqrt{\DC} + t\sqrt{2(n+m)}\sqrt{\DC}  +2(t^2+1)\sqrt{\DC}\leq \frac{\eps_L}{2}m\sqrt{\DC} \iff 3t\sqrt{m} + t\sqrt{2(n+m)}  +2(t^2+1)\leq \frac{\eps_L}{2}m.\label{upeq1}
\eeq
Choose $t=\eps\min\{\sqrt{m},\frac{m}{\sqrt{n}}\}$ for sufficiently small $\eps$ such that \eqref{upeq1} and \eqref{eq:ass_again} then hold with probability $1-\exp\left( -\order{\min\{\frac{m^2}{n},m\}} \right)$.

Combining Case 1 and Case 2, \eqref{eq:ass_again} holds with probability $1-\exp\left( -\order{\gamma(m,n)} \right)$ where $\gamma(m,n)=m$ when $\Cc$ is cone and $\gamma(m,n)=\min\{\frac{m^2}{n},m\}$ otherwise.

\noindent$\bullet$ \emph{$\h\not\in\Cc$ with high probability:} 

Apply Lemma \ref{lemma:conc4} on $\dt(\h,\Cc)$ with $t=\eps\sqrt{\DC}$ to show that $\dt(\h,\Cc)$ is strictly positive. This proves that $\h\notin\Cc$, with probability $1-\exp(-\order{\DC})$=$1-\exp(-\order{m})$.
%Next notice that $\h\notin\Cc$, with probability $1-\exp(-\order{\DC})$. This follows from application of Lemma \ref{lemma:conc4} on $\dt(\h,\Cc)$ with $t=\eps\sqrt{\DC}$ to show the strict positivity of $\dt(\h,\Cc)$.% In fact, $h\notin\text{cone}(\partial f(\x_0))$ with probability $1-\exp(-O(n))$.

\noindent$\bullet$ \emph{High probability lower bound for $\Uc$:} 

Thus far we have proved that assumptions $\h\not\in\Cc$ and \eqref{eq:ass_up1} of the second statement in Lemma \ref{lemma:upKey} hold with the desired probability. Therefore, \eqref{alpha=1} holds with the same high probability, namely,
\beq\label{eq:das1}
\Uc(\g,\h) =  \frac{\sigma}{\sqrt{m - \DC}}\left( \sqrt{m}\|\g\|- \sqrt{\DC}\dt(\h,\Cc)\right)
\eeq

% Onwards, we condition on this event.

We will use similar concentration arguments as above to upper bound the right hand side of \eqref{eq:das1}. For any $t>0$:
\begin{align*}
\sqrt{m}\|\g\| &\leq m + t\sqrt{m} \\
\sqrt{\DC}\dt(\h,\Cc) &\geq \sqrt{\DC}(\sqrt{\DC-1}-t)
\end{align*}
with probability $1-4\exp(-\frac{t^2}{2})$. Thus,
\beq
\sqrt{m}\|\g\|-\sqrt{\DC}\dt(\h,\Cc)\leq m-\DC+t(\sqrt{m}+\sqrt{\DC})+1.\label{combine6}
\eeq
For a given constant $\eps>0$, substitute \eqref{combine6} in \eqref{eq:das1} and choose $t=\eps'\sqrt{m}$ (for some sufficiently small constant $\eps'>0$), to ensure that,
$$
\Uc(\g,\h) \leq (1+\epsilon)\sigma\sqrt{m-\mathbf{D}(\x_0,\la)}
$$
with probability $1 - 4\exp{(\frac{-\eps'^2 m}{2})}$. Combining this with the high probability events of all previous steps, we obtain the desired result.
%\begin{align}
%&|\dt(\h,\la\paf)^2-\sqrt{\DC}|\leq t\\
%&|\|\bu(\h,\la\paf)^2-\sqrt{\DC}|\leq t\\
%&|\|\g\|-\sqrt{m}\leq t
%\end{align}
%, $\h,\g$ satisfies,Combining these, 
%\end{proof}

%
\subsection{Proof of Lemma \ref{lemma:dev}}
\subsubsection{Scalarization}
The reduction of $\Lc_{dev}(\g,\h)$ to an one-dimensional optimization problem follows identically the steps as in the proof for $\Lc(\g,\h)$ in Section \ref{sec:F_low scalar}. 

\subsubsection{Deterministic Result}\label{sec:detRes_dev}
From the first statement of Lemma \ref{lemma:dev},
\begin{align}
\Lc_{dev}(\g,\h) = \min_{ \alpha\in S_{dev} } \left\{ \underbrace{ \sqrt{ \alpha^2+\sigma^2}\|\g\| - \alpha\cdot\dt(\h,\Cc) }_{:=L(\alpha)} \right\},\label{eq:hara3}
\end{align}
where we have denoted the objective function as $L(\alpha)$ for notational convenience. It takes no much effort (see also statements $1$ and $2$ of Lemma \ref{lemma:pert1}) to prove that $L(\cdot)$: 
\begin{itemize}
\item
is a strictly convex function,
\item
attains its minimum at 
$$
\alpha^*(\g,\h) = \frac{\sigma\cdot\dt(\h,\Cc)}{\sqrt{\| \g \|^2-\dt^2(\h,\Cc)}}.
$$
\end{itemize}
The minimization of $L(\alpha)$ in \eqref{eq:hara3} is restricted to the set $S_{dev}$. Also, by assumption \eqref{eq:ass_dev},  $\alpha^*(\g,\h)\notin S_{dev}$.
%Thus, if $\alpha^*\in S_{dev}$, then $f^*_{r,low}(\g,\h) = L(\alpha^*)$. 
Strict convexity implies then that  the minimum of $L(\cdot)$ over $\alpha\in S_{dev}$ is attained at the boundary points of the set $S_{dev}$, i.e. at $(1\pm\delta_{dev}) C_{dev}$ \cite{Boyd}. Thus, $\Lc_{dev}(\g,\h)= L( (1\pm\delta_{dev}) C_{dev} )$, which completes the proof.

%We prove in the next section, that with high probability over the realizations of $\g$ and $\h$, $\alpha^*
%\notin S_{dev}$, and therefore only the latter case is of interest to us.

%
\subsubsection{Probabilistic Result}
%
%\samet{I will go over this tomorrow. This part is the most important :)\\}
%\chris{I modified this part. if you want to check. I plan to read your Lemma G.3 later tonight. I see you did a lot of work! }

Choose $C_{dev} = \sigma\sqrt{\frac{\DC}{m-\DC}}$ and consider the regime where 
%$m\geq(1+\eps_L)\DC$
$(1-\eps_L)m>\DC>\eps_L m$ for some constant $\eps_L>0$. $\delta_{dev}>0$ is also a constant.

$\bullet$ {\bf{Mapping $\Lc_{dev}$ to Lemma \ref{lemma:pert1}:}} It is helpful for the purposes of the presentation to consider the function 
\begin{align}
L(x) := L(x;a,b) = \sqrt{x^2+\sigma^2}a-xb,
\end{align}
over $x\geq 0$, and $a,b$ are positive parameters. Substituting $a,b,x$ with $\|\g\|,\dt(\h,\Cc),\alpha$, we can map $L(x;a,b)$ to our function of interest,
\beq\notag
L(\alpha;\|\g\|,\dt(\h,\Cc)) = \sqrt{\alpha^2+\sigma^2}\|\g\|-\alpha \dt(\h,\Cc).
\eeq
In Lemma \ref{lemma:pert1} we have analyzed useful properties of the function $L(x;a,b)$, which are of key importance for the purposes of this proof. This lemma focuses on perturbation analysis and investigates $L(x';a',b')-L(x;a,b)$ where $x',a',b'$ are the perturbations from the fixed values $x,a,b$. In this sense, $a',b'$ correspond to $\|\g\|,\dt(\h,\Cc)$ which are probabilistic quantities and $a,b$ correspond to $\sqrt{m},\sqrt{\DC}$, i.e. the approximate means of the former ones.

In what follows, we refer continuously to statements of Lemma \ref{lemma:pert1} and use them to complete the proof of the ``Probabilistic result'' of Lemma \ref{lemma:dev}. Let us denote the minimizer of $L(x;a,b)$ by $x^*(a,b)$.
To see how the definitions above are relevant to our setup, it follows from the first statement of Lemma \ref{lemma:pert1} that,  
%\begin{align}\label{eq:sw4}
%L(\alpha) = L(\alpha; \|g\| , \dt(\h,\Cc)),\quad \text{ for all } \alpha,
%\end{align}
%and
%\begin{align}\label{eq:sw1}
%x^*(\sqrt{m},\sqrt{\DC}) = \frac{\sigma\cdot\dt(\h,\Cc)}{\sqrt{\| \g \|^2-\dt^2(\h,\Cc)}}.
%\end{align}
%Furthermore, 
\begin{align}\label{eq:sw5}
L\left( x^*(\sqrt{m},\sqrt{\DC}); \sqrt{m},\sqrt{\DC}\right) = \sigma\sqrt{m-\DC},
\end{align}
and
\begin{align}\label{eq:sw2}
x^*(\sqrt{m},\sqrt{\DC}) = \sigma\sqrt{\frac{\DC}{m-\DC}} = C_{dev},
\end{align}

$\bullet$ {\bf{Verifying assumption \eqref{eq:ass_dev}:}} Going back to the proof, we begin by proving that assumption \eqref{eq:ass_dev} of the second statement of Lemma \ref{lemma:dev} is valid with high probability. Observe that from the definition of $S_{dev}$ and \eqref{eq:sw2}, assumption \eqref{eq:ass_dev} can be equivalently written as 
\begin{align}\label{eq:sw3}
\left| \frac{x^*(\|\g\|,\dt(\h,\Cc))}{x^*(\sqrt{m},\sqrt{\DC})} -1 \right| \leq \delta_{dev}.
\end{align} 
On the other hand, from the third statement of Lemma \ref{lemma:pert1} there exists sufficiently small constant $\eps_1>0$ such that \eqref{eq:sw3} is true
for all $\g$ and $\h$ satisfying 
\begin{align}\label{eq:trr0}
| \|\g\| - \sqrt{m}| \leq \eps_1\sqrt{m}\quad \text{ and } \quad | \dt(\h,\Cc) - \sqrt{\DC}| \leq \eps_1\sqrt{m}.
\end{align}
Furthermore, for large enough $\DC$ and from basic concentration arguments (see Lemma \ref{lemma:conc4}), $\g$ and $\h$ satisfy \eqref{eq:trr0} with probability $1-2\exp(-\frac{\eps_1^2m}{2})$. This proves that assumption \eqref{eq:ass_dev} holds with the same high probability.

%To simplify the notation in what follows denote,
% \begin{align}
%L(\alpha) = L(\alpha; \|\g\| , \dt(\h,\Cc)),\quad \text{ for all } \alpha\geq 0,\nn
%\end{align}
%just like in Section \ref{sec:detRes_dev}.

$\bullet$ {\bf{Lower bounding $\Lc_{dev}$:}} From the deterministic result of Lemma \ref{lemma:dev}, once \eqref{eq:ass_dev} is satisfied then 
\begin{align}\label{eq:opa}
\Lc_{dev}(\g,\h)=L\left((1\pm\delta_{dev})C_{dev};\|\g\| , \dt(\h,\Cc)\right).
\end{align}
Thus, to prove \eqref{eq:conc_res1} we will show that there exists $t>0$ such that
\begin{align}
L\left((1\pm\delta_{dev})C_{dev};\|\g\| , \dt(\h,\Cc)\right) \geq (1+t)\sigma\sqrt{m-\DC},\label{eq:sw4}
\end{align}
with high probability. Equivalently, using \eqref{eq:sw5}, it suffices to show that there exists a constant $t>0$ such that
\begin{align}\label{eq:sw6}
L\left( (1\pm\delta_{dev})x^*(\sqrt{m},\sqrt{\DC}); \|\g\| , \dt(\h,\Cc) \right) - L\left( x^*(\sqrt{m},\sqrt{\DC}); \sqrt{m} , \sqrt{\DC} \right) \geq t\sigma\sqrt{m},
\end{align}
with high probability.
Applying the sixth statement of Lemma \ref{lemma:pert1} with $\gamma\gets\delta_{dev}$, for any constant $\delta_{dev}>0$, there exists constants $t,\eps_2$ such that \eqref{eq:sw6} holds for all $\g$ and $\h$ satisfying 
$$
| \|\g\| - \sqrt{m}| \leq \eps_2\sqrt{m}\quad \text{ and } \quad | \dt(\h,\Cc) - \sqrt{\DC}| \leq \eps_2\sqrt{m},
$$
which holds with probability $1-2\exp(-\frac{\eps_2^2m}{2})$ for sufficiently large $\DC$. Thus, \eqref{eq:sw6} is true with the same high probability.

Union bounding over the events that \eqref{eq:sw3} and \eqref{eq:sw6} are true, we end up with the desired result. The reason is that with high probability \eqref{eq:opa} and \eqref{eq:sw6} hold, i.e.,
\beq\notag
\Lc_{dev}(\g,\h)=L\left((1\pm\delta_{dev})C_{dev};\|\g\| , \dt(\h,\Cc)\right)\geq  L\left( x^*(\sqrt{m},\sqrt{\DC}); \sqrt{m} , \sqrt{\DC} \right) + t\sigma\sqrt{m}=\sigma \sqrt{m-\DC}+t\sigma\sqrt{m}.
\eeq
% which mean and invoking the deterministic result of Lemma \ref{lemma:dev} completes the proof.

%
%Now, under assumption \eqref{eq:ass_dev} and for the particular value of $C_{dev}$ chosen, it follows from the second statement of Lemma \ref{lemma:upKey} that,
%\begin{align}
%\Lc_{dev}(\g,\h)= L\left((1\pm\eps)\sigma\sqrt{\frac{\DC}{m-\DC}}\right)
%\end{align}
%%
%%\begin{align}
%%\Lc_{dev}(\g,\h)= \frac{\sigma}{\sqrt{m-\DC}}  {\psi(\|\g\|,\dt(\h,\Cc))} .
%%\end{align}
%%where we have denoted
%%\begin{align}\label{eq:2comb0}
%%\psi(\|\g\|,\dt(\h,\Cc)) = \sqrt{m + (\eps^2\pm 2\eps) \DC }\|\g\| - (1\pm\eps)\sqrt{\DC}\dt(\h,\Cc).
%%\end{align}
%Assume $\delta_{dev}$ is chosen sufficiently small so that 
%$8(m+\delta_{dev}\dt(\h,\Cc))^2\geq \delta_{dev}^2\dt(\h,\Cc)(m-\dt(\h,\Cc))$.
%   Lemma \ref{lemma:pert2} in the Appendix shows that there exist $t>0$ and sufficiently small $\eta>0$ such that
%   \begin{align}\label{eq:2comb1}
% \psi(\|\g\|,\dt(\h,\Cc)) \geq (1+t)(m-\DC))  
%   \end{align}
% for all $\g$ and $\h$ satisfying 
%\begin{align}\label{eq:trr1}
%| \|\g\| - \sqrt{m}| \leq \eta\sqrt{m}\quad \text{ and } \quad | \dt(\h,\Cc) - \sqrt{\DC}| \leq \eta\sqrt{m}.
%\end{align} 
%Again, basic concentration arguments show that \eqref{eq:trr1} is satisfied with high probability. Combining equations \eqref{eq:2comb0} and \eqref{eq:2comb1} completes the proof.

\section{Deviation Analysis: Key Lemma}

%%%%%%%%%%%%%%%%%%%%%%%%
%%%%%%%%%%%%%%%%%%%%%%%%
\begin{lem} \label{lemma:pert1}
Consider the following function over $x\geq 0$:
\beq
\fone(x) := \fone(x;a,b) =  \sqrt{x^2+\sigma^2}a-xb\notag
\eeq
where $\sigma>0$ is constant  and $a,b$ are positive parameters satisfying 
%$a>b$ and 
$(1-\eps)a>b>\eps a$ for some constant $\eps>0$.
Denote the minimizer of $\fone(x;a,b)$ by $x^*(a,b)$. Then,
\begin{enumerate} 
\item $x^*(a,b)=\frac{\sigma b}{\sqrt{a^2-b^2}}$ and $\fone(x^*(a,b);a,b) = \sigma\sqrt{a^2 - b^2}$.
\item For fixed $a$ and $b$, $\fone(x;a,b)$ is strictly convex in $x\geq0$.
\item For any constant $\eta>0$, there exists sufficiently small constant $\eps_1>0$, such that
\begin{align}
%&|\fone(x^*(a',b'))-\fone(x^*(a,b))|\leq cC_1\eps_1a\notag\\
&\left|\frac{x^*(a',b')}{x^*(a,b)}-1\right|\leq \eta,\notag
\end{align}
for all $a',b'$ satisfying $|a'-a|<\eps_1a$ and $|b'-b|<\eps_1a$.
\item There exists positive constant $\eta>0$, such that, for sufficiently small constant $\eps_1>0$,
\begin{align}\notag
%&|\fone(x^*(a',b'))-\fone(x^*(a,b))|\leq cC_1\eps_1a\notag\\
&\left|\fone(x^*(a,b);a',b')-\fone(x^*(a,b),a,b)\right|\leq \eta \eps_1 \sigma a,
\end{align}
for all $a',b'$ satisfying $|a'-a|<\eps_1a$ and $|b'-b|<\eps_1a$.
\item For any constant $\gamma>0$, there exists a constant $\eps_2>0$ such that for sufficiently small constant $\eps_1>0$,
\begin{align}\notag
%&|\fone(x^*(a',b'))-\fone(x^*(a,b))|\leq cC_1\eps_1a\notag\\
&\fone(x;a',b')-\fone(x^*(a,b);a',b')\geq \eps_2 \sigma a,
\end{align}
for all $x,a'$ and $b'$ satisfying $|x-x^*(a,b)|>\gamma x^*(a,b)$, $|a'-a|<\eps_1a$ and $|b'-b|<\eps_1a$.

\item For any constant $\gamma>0$, there exists a constant $\eps_2>0$ such that for sufficiently small constant $\eps_1>0$,
\begin{align}\notag
%&|\fone(x^*(a',b'))-\fone(x^*(a,b))|\leq cC_1\eps_1a\notag\\
&\fone(x;a',b')-\fone(x^*(a,b);a,b)\geq \eps_2 \sigma a,
\end{align}
for all $x,a'$ and $b'$ satisfying $|x-x^*(a,b)|>\gamma x^*(a,b)$, $|a'-a|<\eps_1a$ and $|b'-b|<\eps_1a$.
\item Given $c_{low}>0$, consider the restricted optimization, $\min_{x\geq c_{low}}\fone(x;a,b)$. We have,
\beq
\lim_{c_{low}\rightarrow\infty} \min_{x\geq c_{low}}\fone(x;a,b)\rightarrow \infty
\eeq
\end{enumerate}
\end{lem}

% ------------------------------------------------------------------------------------------------------------ %
\begin{proof}
%
%\noindent
\emph{First statement:} The derivative (w.r.t. $x$) of $\fone(x;a,b)$ is:
\beq\notag
\fone'(x;a,b)=\frac{ax}{\sqrt{x^2+\sigma^2}}-b.
\eeq
Setting this to $0$, using strict convexity and solving for $x$, we obtain the first statement.

\emph{Second statement:} The second derivative is,
\beq\notag
\fone''(x;a,b)=\frac{a\sqrt{x^2+\sigma^2}-\frac{ax^2}{\sqrt{x^2+\sigma^2}}}{x^2+\sigma^2}=\frac{a\sigma^2}{(x^2+\sigma^2)^{3/2}}>0,
\eeq
for all $x\geq 0$.
Consequently, $f$ is strictly convex.

%\noindent {\bf{Third statement:}} 
%At the optimal point, we have:
%\beq\notag
%\fone(x^*(a,b),a,b)=c\sqrt{a^2-b^2}
%\eeq
%We have,
%\beq\notag
%|\fone(x^*(a',b'),a',b')-\fone(x^*(a,b),a,b)|=c|\sqrt{a'^2-b'^2}-\sqrt{a^2-b^2}|
%\eeq
%Computing the gradient of $\sqrt{u^2-v^2}$ at $(u,v)=(a,b)$, and letting $\eps_1\rightarrow 0$, we find,
%\beq\notag
%|\sqrt{a'^2-b'^2}-\sqrt{a^2-b^2}|\leq \frac{\eps_1(a^2+ab)}{\sqrt{a^2-b^2}}\leq \eps_1a\frac{2-\eps}{\sqrt{2\eps-\eps^2}}=\eps_1a\sqrt{\frac{2-\eps}{\eps}}
%\eeq
\emph{Third statement:}
We can write,
\beq\notag
\left|x^*(a',b')-x^*(a,b)\right|=\sigma\left|\frac{b'}{\sqrt{a'^2-b'^2}}-\frac{b}{\sqrt{a^2-b^2}}\right|.
\eeq
Observe that $x^*(a,b)=\frac{b}{\sqrt{a^2-b^2}}$ is decreasing in $a$ and increasing in $b$ as long as $a>b\geq 0$. Also, for sufficiently small constant $\eps_1$, we have, $a',b'>0$  for all $|a'-a|<\eps_1a,|b'-b|<\eps_1a$. Therefore,
\beq\notag
\frac{b-\eps_1a}{\sqrt{(a+\eps_1a)^2-(b-\eps_1a)^2}}\leq \frac{b'}{\sqrt{a'^2-b'^2}}\leq \frac{b+\eps_1a}{\sqrt{(a-\eps_1a)^2-(b+\eps_1a)^2}}.
\eeq
Now, for any constant $\delta>0$, we can choose $\eps_1$ sufficiently small such that both ${b-\eps_1a}$ and ${b+\eps_1a}$ lie in the interval $(1\pm \delta)b$. Similarly, $(a\pm \eps_1a)^2-(b\mp\eps_1a)^2$ can be also chosen to lie in the interval $(1\pm \delta)(a^2-b^2)$. Combining, we obtain,
\beq\notag
\left| \frac{b'}{\sqrt{a'^2-b'^2}}- \frac{b}{\sqrt{a^2-b^2}}\right|<\eta(\delta)  \frac{b}{\sqrt{a^2-b^2}},
\eeq
as desired.

\emph{Fourth statement:}  For $|a-a'|<\eps_1a$ and $|b-b'|<\eps_1a$, we have,
\beq\notag
|\fone(x^*(a,b);a',b')-\fone(x^*(a,b);a,b)|=\frac{\sigma}{\sqrt{a^2-b^2}}|(aa'-bb')-(a^2-b^2)|\leq\eps_1 \sigma \frac{|a^2+ab|}{a^2-b^2}.
\eeq
By assumption, $(1-\eps)a>b>\eps a$. Thus,  
$$\eps_1 \sigma\frac{|a^2+ab|}{a^2-b^2}\leq \eps_1\sigma\frac{2a^2}{2\eps a^2} =\frac{\eps_1\sigma}{\eps}.$$
Choosing $\eps_1$ sufficiently small, we conclude with the desired result.

\emph{Fifth statement:} We will show the statement for a sufficiently small $\gamma$. Notice that, as $\gamma$ gets larger, the set $|x-x^*(a,b)|\geq \gamma x^*(a,b)$ gets smaller hence, proof for small $\gamma$ implies the proof for larger $\gamma$.

Using the Third Statement, choose $\eps_1$ to ensure that $|x^*(a',b')-x^*(a,b)|<\gamma x^*(a,b)$ for all $|a'-a|<\eps_1a$ and $|b'-b|<\eps_1 a$. 
%This ensures that, for all $|a'-a|,|b'-b|<\eps_1 a$, $x^*(a',b')$ lies in the interval $(1\pm\gamma)x^*(a,b)$.
For each such $a',b'$, since $\fone(x,a',b')$ is a strictly convex function of $x$ and the  minimizer $x^*(a',b')$ lies between $(1\pm\gamma) x^*(a,b)$ we have,
\beq\notag
\fone(x,a',b')\geq \min\{\fone((1-\gamma)x^*(a,b),a',b'),\fone((1+\gamma)x^*(a,b),a',b')\},
\eeq
for all $|x-x^*(a,b)|>\gamma x^*(a,b)$.
In summary, we simply need to characterize the increase in the function value at the points $(1\pm \gamma)x^*(a,b)$.% to have a lower bound over all points $x$. 

We have that,
\begin{align}\label{eq:laou1}
 \fone((1\pm \gamma)x^*(a,b);a',b')=\frac{\sigma}{\sqrt{a^2-b^2}}(\sqrt{a^2+(\pm2 \gamma+\gamma^2)b^2}a'-(1\pm\gamma)bb'),
 \end{align}
 and
 \begin{align}\label{eq:laou2}
 \fone(x^*(a,b);a',b')=\frac{\sigma}{\sqrt{a^2-b^2}}(aa'-bb').
 \end{align}
%\begin{itemize}
%\item $\fone((1\pm \gamma)x^*(a,b);a',b')=\frac{c}{\sqrt{a^2-b^2}}(\sqrt{a^2+(\pm2 \gamma+\gamma^2)b^2}a'-(1\pm\gamma)bb')$
%\item $\fone(x^*(a,b);a',b')=\frac{c}{\sqrt{a^2-b^2}}(aa'-bb')$
%\end{itemize}
In the following discussion, without loss of generality, we consider only the ``$+\gamma$" case  in \eqref{eq:laou1} since the exact same argument works for the ``$-\gamma$" case as well.% when $-\gamma$ is substituted.

Subtracting \eqref{eq:laou2} from \eqref{eq:laou1} and discarding the constant in front, we will focus on the following quantity,
\begin{align}
\text{diff}( \gamma)&=(\sqrt{a^2+(2 \gamma+\gamma^2)b^2}a'-(1+\gamma)bb')-(aa'-bb')\notag\\
&=(\underbrace{\sqrt{a^2+(2 \gamma+\gamma^2)b^2}}_{:=g(\gamma)}-a)a'-\gamma bb'.\label{eq:rhss0}
\end{align}
To find a lower bound for $g(\gamma)$, write
\begin{align}
g(\gamma)&=\sqrt{a^2+(2 \gamma+\gamma^2)b^2}\notag\\
&=\sqrt{(a+\gamma\frac{b^2}{a})^2+\gamma^2(b^2-\frac{b^4}{a^2})}\notag\\
&\geq (a+\gamma \frac{b^2}{a})+\frac{\gamma^2(b^2-\frac{b^4}{a^2})}{4(a+\gamma\frac{b^2}{a})},\label{eq:rhss1}
\end{align}
where we have assumed $\gamma\leq 1$ and used the fact that $(a+\gamma \frac{b^2}{a})^2\geq a^2\geq b^2-\frac{b^4}{a^2}$. Equation \eqref{eq:rhss1} can be further lower bounded by,
\beq\notag
g(\gamma)\geq (a+\gamma \frac{b^2}{a})+\frac{\gamma^2(a^2b^2-b^4)}{8a^3}
\eeq
Combining with \eqref{eq:rhss0} , we find that,
\beq\label{eq:rhss2}
\text{diff}(\gamma)\geq\gamma (\frac{b^2}{a}a'-bb')+\gamma^2\frac{a^2b^2-b^4}{8a^3}a'.
\eeq
Consider the second term on the right hand side of the inequality in \eqref{eq:rhss2}. Choosing $\eps_1<1/2$, we ensure, $a'\geq a/2$, and thus,
\beq\label{eq:rhss3}
\gamma^2\frac{a^2b^2-b^4}{8a^3}a'\geq \gamma^2\frac{a^2b^2-b^4}{16a^2}\geq  \gamma^2\frac{\eps a^2b^2}{16a^2}=\gamma^2\eps \frac{b^2}{16}.
\eeq
Next, consider the other term in \eqref{eq:rhss2}. We have,
\beq\notag
\left(\frac{b^2}{a}a'-bb'\right)=\frac{b^2}{a}(a'-a)-b(b'-b)\geq  -\left(\left|\frac{b^2}{a}(a'-a)\right|+|b(b'-b)|\right).
\eeq
Choosing $\eps_1$ sufficiently small (depending only on $\gamma$), we can ensure that, 
 \beq\label{eq:rhss4}
\left |\frac{b^2}{a}(a'-a)\right|+|b(b'-b)|<\gamma\eps \frac{b^2}{32}.
 \eeq
 Combining \eqref{eq:rhss2}, \eqref{eq:rhss3} and \eqref{eq:rhss4}, we conclude that there exists sufficiently small constant $\eps_1>0$ such that,
 \beq
 \text{diff}(\gamma)\geq \gamma^2\eps \frac{b^2}{32}.
 \eeq
 Multiplying with $\frac{\sigma}{\sqrt{a^2-b^2}}$, we end up with the desired result since $\frac{b^2}{\sqrt{a^2-b^2}}\geq \frac{\eps^2}{\sqrt{1-\eps^2}}a$.

\emph{Sixth statement:} The last statement can be deduced from the fourth and fifth statements. Given $\gamma>0$, choose $\eps_1>0$ sufficiently small to ensure,
\beq
\fone(x;a',b')-\fone(x^*(a,b),a',b')\geq \eps_2\sigma a
\eeq
and 
\beq
|\fone(x^*(a,b);a,b)-\fone(x^*(a,b),a',b')|\geq \eta\eps_1 \sigma a
\eeq
Using the triangle inequality,
\begin{align}
\fone(x;a',b')-\fone(x^*(a,b),a,b)&\geq \fone(x;a',b')-\fone(x^*(a,b),a',b')-|\fone(x^*(a,b),a',b')-\fone(x^*(a,b),a,b)|\notag\\
&\geq (\eps_2-\eta\eps_1)\sigma a.\label{eq:rhss5}
\end{align}
Choosing $\eps_1$ to further satisfy $\eta\eps_1<\frac{\eps_2}{2}$, \eqref{eq:rhss5} is guaranteed to be larger than $\frac{\eps_2}{2}\sigma a$ which gives the desired result.%For sufficiently small $\eps_1$, $\eps_2-\eta>0$.

\emph{Seventh statement:} To show this, we may use $a>b$ and simply write,
\beq
\fone(x;a,b)\geq (a-b)x\implies \lim_{c_{low}\rightarrow\infty} \min_{x\geq c_{low}}\fone(x;a,b)\geq\lim_{c_{low}\rightarrow\infty} (a-b)c_{low}=\infty
\eeq
\end{proof}
\section{Proof of Lemma \ref{lemma:keyProp}}\label{appendixe}
Proof of the Lemma requires some work. We prove the statements in the specific order that they appear.

%%%%%%%%%%%%%%%%%%%%%%%%%%%%
\noindent{\textbf{Statement 1}:} 
We have
\begin{align*}
n = \E\left[ \|\h\|^2 \right] = \E\left[ \| \bu_\la(\h) + \h-\bu_\la(\h)  \|^2 \right] &= \E[\|\bu_\la(\h)\|^2] + \E[\|\bi_\la(\h)\|^2] + 2\E[\li\bi_\la(\h), \bu_\la(\h)  \ri] \\
&=\Plf + \Dlf + 2\Clf.
\end{align*}

%%%%%%%%%%%%%%%%%%%%%%%%%%%%
\noindent{\textbf{Statement 2}:}
We have
$\bu_0(\h) = \mathbf{0}$ and $\bi_0(\h) = \h$, and the statement follows easily.

%%%%%%%%%%%%%%%%%%%%%%%%%%%%
\noindent{\textbf{Statement 3}:}
Let $r=\inf_{\s\in\paf}\|\s\|$. Then, for any $\la\geq 0$, $\|\bu_\la(\vb)\|\geq \la\|\s\|$, which implies $\Plf\geq \la^2\|\s\|^2$. Letting $\la\rightarrow\infty$, we find $\Plf\rightarrow\infty$.

Similarly, for any $\h$, application of the triangle inequality gives
\beq\notag
\|\Pi_\la(\h)\|\geq \la r-\|\h\|\implies \|\Pi_\la(\h)\|^2\geq \la^2r^2-2\la r\|\h\|.
\eeq
Let $\h\sim\Nn(0,I)$ and take expectations in both sides of the inequality above. Recalling that $\E[\|\h\|]\leq \sqrt{n}$, and letting $\la\rightarrow\infty$, we find $\Dlf\rightarrow\infty$. 

Finally, since $\Dlf+\Plf+2\Clf=n$, $\Clf\rightarrow-\infty$ as $\la\rightarrow\infty$. This completes the proof.

%%%%%%%%%%%%%%%%%%%%%%%%%%%%
\noindent{\textbf{Statement 4}:}
Continuity of $\Dlf$ follows from Lemma $B.2$ in Amelunxen et al. \cite{McCoy}. We will now show continuity of $\Plf$ and continuity of $\Clf$ will follow from the fact that $\Clf$ is a continuous function of $\Dlf$ and $\Plf$. 

Recall that $\bu_\la(\vb)=\la\bu_1(\frac{\vb}{\la})$. Also, given $\vb_1,\vb_2$, we have,
\beq
\|\bu_\la(\vb_1)-\bu_\la(\vb_2)\|\leq \|\vb_1-\vb_2\|
\eeq
Consequently, given $\la_1,\la_2>0$,
\begin{align}
\|\bu_{\la_1}(\vb)-\bu_{\la_2}(\vb)\|&=\|\la_1\bu_1(\frac{\vb}{\la_1})-\la_2\bu_1(\frac{\vb}{\la_2})\|\\
&\leq |\la_1-\la_2|\|\bu_1(\frac{\vb}{\la_1})\|+\|\la_2(\bu_1(\frac{\vb}{\la_1})-\bu_1(\frac{\vb}{\la_2}))\|\\
&\leq |\la_1-\la_2|\|\bu_1(\frac{\vb}{\la_1})\|+\la_2\|\vb\|\frac{|\la_1-\la_2|}{\la_1\la_2}\\
&=|\la_1-\la_2|(\|\bu_1(\frac{\vb}{\la_1})\|+\frac{\|\vb\|}{\la_1})
\end{align}
Hence, setting $\la_2=\la_1+\eps$,
\beq
\|\bu_{\la_2}(\vb)\|^2\leq [\|\bu_{\la_1}(\vb)\|+\eps(\|\bu_1(\frac{\vb}{\la_1})\|+\frac{\|\vb\|}{\la_1})]^2
\eeq
which implies,
\beq
\|\bu_{\la_2}(\vb)\|^2-\|\bu_{\la_1}(\vb)\|^2\leq 2\eps(\|\bu_1(\frac{\vb}{\la_1})\|+\frac{\|\vb\|}{\la_1})\|\bu_{\la_1}(\vb)\|+\eps^2(\|\bu_1(\frac{\vb}{\la_1})\|+\frac{\|\vb\|}{\la_1})
\eeq
Similarly, using $\|\bu_{\la_2}(\vb)\|\geq \|\bu_{\la_1}(\vb)\|-\eps(\|\bu_1(\frac{\vb}{\la_1})\|+\frac{\|\vb\|}{\la_1})$, we find,
\beq
\|\bu_{\la_1}(\vb)\|^2-\|\bu_{\la_2}(\vb)\|^2\leq 2\eps(\|\bu_1(\frac{\vb}{\la_1})\|+\frac{\|\vb\|}{\la_1})\|\bu_{\la_1}(\vb)\|{\la_1})
\eeq
Combining these, we always have,
\beq
|\|\bu_{\la_2}(\vb)\|^2-\|\bu_{\la_1}(\vb)\|^2|\leq 2\eps(\|\bu_1(\frac{\vb}{\la_1})\|+\frac{\|\vb\|}{\la_1})\|\bu_{\la_1}(\vb)\|+\eps^2(\|\bu_1(\frac{\vb}{\la_1})\|+\frac{\|\vb\|}{\la_1})
\eeq
Now, letting $\vb\sim\Nn(0,I)$ and taking the expectation of both sides and letting $\eps\rightarrow 0$, we conclude with the continuity of $\Plf$ for $\la>0$. 

To show continuity at $0$, observe that, for any $\la>0$, we have, $\|\bu_\la(\vb)\|\leq R\la$ where $R=\sup_{\s\in \paf}\|\s\|$. Hence,
\beq
|\Plf-\Pf0)|=\Plf\leq R^2\la^2
\eeq
As $\la\rightarrow 0$, $\Plf=0$.

%%%%%%%%%%%%%%%%%%%%%%%%%%%%
\noindent{\textbf{Statement 5}:} 
For a proof see Lemma $B.2$ in \cite{McCoy}.

%%%%%%%%%%%%%%%%%%%%%%%%%%%%
\noindent{\textbf{Statement 6}:} 
Based on Lemma \ref{angle1}, given vector $\vb$, set $\Cc$ and scalar $1\geq c>0$, we have,
\beq
\frac{\|\bu(c\vb,\Cc)\|}{c}\geq \|\bu(\vb,\Cc)\|
\eeq
Given $\la_1>\la_2>0$, this gives,
\beq
\|\bu(\vb,\la_1\paf)\|=\la_1\|\bu(\frac{\vb}{\la_1},\paf)\|\geq \la_1\frac{\la_2}{\la_1}\|\bu(\frac{\vb}{\la_2},\paf)\|=\|\bu(\vb,\la_2\paf)\|
\eeq
Since this is true for all $\vb$, choosing $\vb\sim \Nn(0,I)$, we end up with $\Df\la_1)\geq \Df\la_2)$.

Finally, at $0$ we have $\Df0)=0$ and by definition $\Dlf\geq 0$ which implies the increase at $\la=0$.
For the rest of the discussion, given three points $A,B,C$ in $\R^n$, the angle induced by the lines $AB$ and $BC$ will be denoted by $A\hat{B}C$.
\begin{lem} \label{angle1} Let $\Cc$ be a convex and closed set in $\R^n$. Let $\z$ and $0<\alpha<1$ be arbitrary, let $\p_1=\bu(\z,\Cc)$, $\p_2=\bu(\alpha\z,\Cc)$. Then,
\beq
\|\p_1\|\leq \frac{\|\p_2\|}{\alpha}
\eeq
\end{lem}
\begin{proof}

Denote the points whose coordinates are determined by $0,\p_1,\p_2,\z$ by $O, P_1,P_2$ and $Z$ respectively. We start by reducing the problem to a two dimensional one. Obtain $\Cc'$ by projecting the set $\Cc$ to the $2D$ plane induced by the points $Z,P_1$ and $O$. Now, let $\p_2'=\bu(\alpha\z,\Cc')$. Due to the projection, we still have:
$\|\z-\p_2'\|\leq\|\z-\p_2\|$ and $\|\p_2'\|\leq\|\p_2\|$. We wish to prove that $\|\p_2'\|\geq \|\alpha\p_1\|$. Figures \ref{cases} and \ref{cases2} will help us explain our approach. 

Let the line $UP_1$ be perpendicular to $ZP_1$. Let $P'Z'$ be parallel to $P_1Z_1$. Observe that $P'$ corresponds to $\alpha \p_1$. $H$ is the intersection of $P'Z'$ and $P_1U$. Denote the point corresponding to $\p_2'$ by $P_2'$. Observe that $P_2'$ satisfies the following:
\begin{itemize}
\item $P_1$ is the closest point to $Z$ in $\Cc$ hence $P_2'$ lies on the side of $P_1U$ which doesn't include $Z$.
\item $P_2$ is the closest point to $Z'$. Hence, $Z'\hat{P_2}P_1$ is not acute angle. Otherwise, we can draw a perpendicular to $P_2P_1$ from $Z'$ and end up with a shorter distance. This would also imply that $Z'\hat{P_2'}P_1$ is not acute as well as $Z'P_1$ stays same but $|Z'P_2'|\leq |Z'P_2|$ and $|P_2'P_1|\leq |P_2P_1|$.
%\item $P_2'$ has to lie below or on the line $OP_1$ otherwise, perpendicular to $OP_1$ from $Z'$ would yield shorter distance.

\end{itemize}

We will do the proof case by case.

 \noindent{\bf{When $Z\hat{P_1}O$ is wide angle:}} Assume $Z\hat{P_1}O$ is wide angle and $UP_1$ crosses $ZO$ at $S$. 

Based on these observations, we investigate the problem in two cases illustrated by Figure \ref{cases}.
\begin{figure}

  \begin{center}
{\includegraphics[scale=0.22]{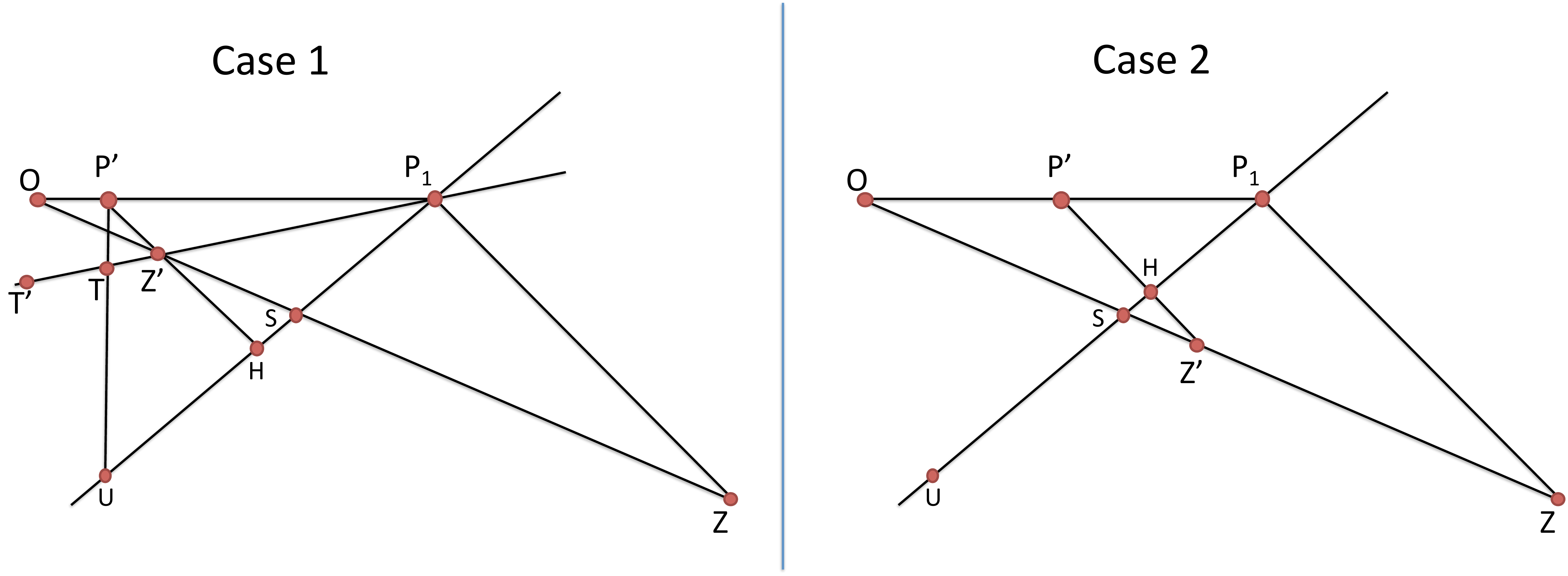}}  
  \end{center}
  \caption{Possible configurations of the points in Lemma \ref{angle1} when $Z\hat{P_1}O$ is wide angle.}
 \label{cases}
\end{figure}

\begin{figure}

  \begin{center}
{\includegraphics[scale=0.25]{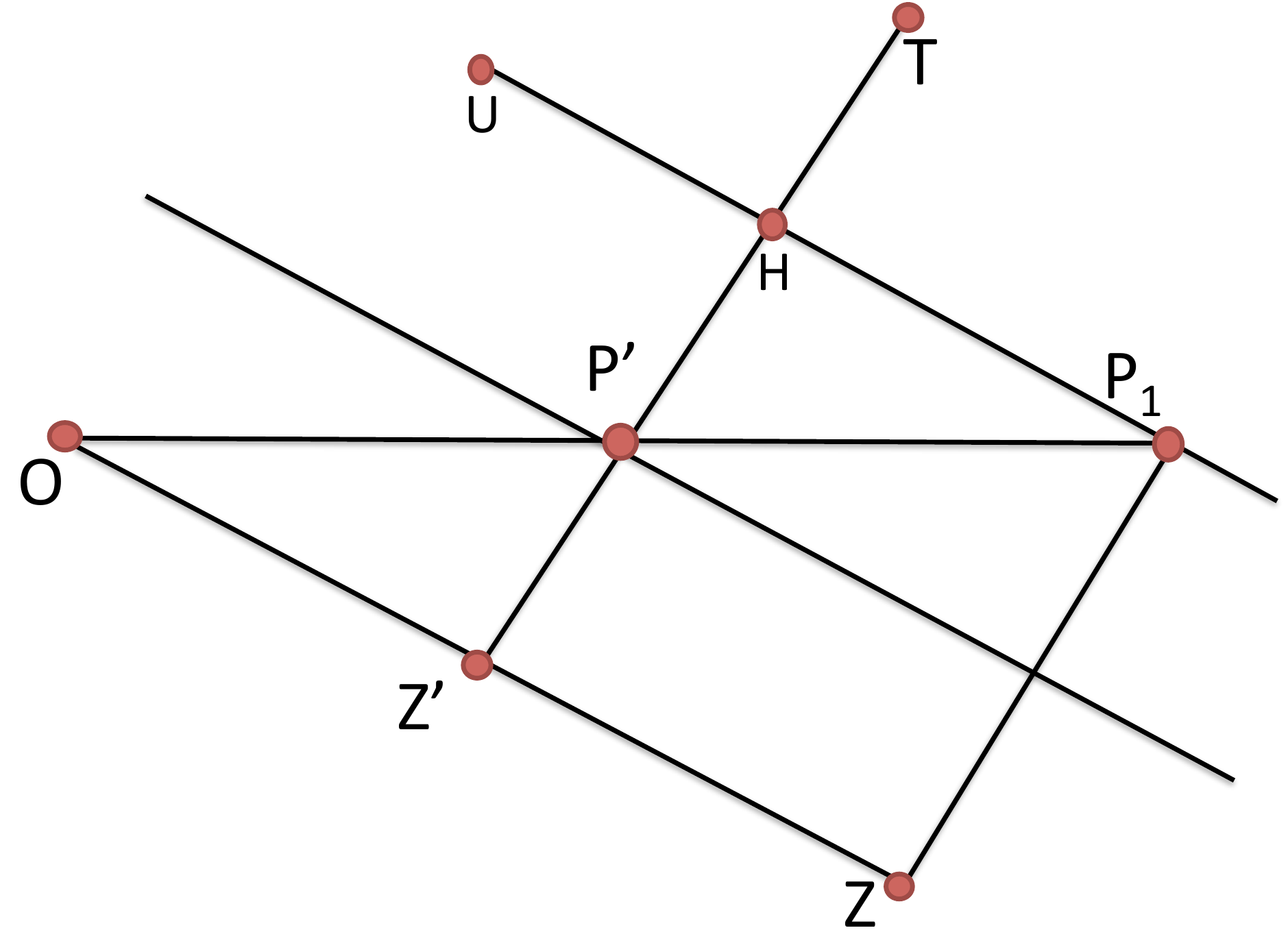}}  
  \end{center}
  \caption{Lemma \ref{angle1} when $Z\hat{P_1}O$ is acute or right angle.}
 \label{cases2}
\end{figure}

\noindent {\bf{Case 1 ($S$ lies on $Z'Z$):}} Consider the lefthand side of Figure \ref{cases}. If $P_2'$ lies on the triangle $P'P_1H$ then $O\hat{P}'P_2'>O\hat{P}'Z$ which implies $O\hat{P}'P_2'$ is wide angle and $|OP_2'|\geq|OP'|$. If $P_2'$ lies on the region induced by $OP'Z'T'$ then $P_1\hat{P}_2'Z'$ is acute angle as $P_1\hat{Z}'P_2'>P_1\hat{Z}'O$ is wide, which contradicts with $P_1\hat{P}_2'Z'$ is not acute.

Finally, let $U$ be chosen so that $P'U$ is perpendicular to $OP_1$. Then, if $P_2'$ lies on the quadrilateral $UTZ'H$ then $|OP_2'|\geq |OP'|$ as $O\hat{P}'P_2'$ is wide or right angle. If it lies on the remaining region $T'TU$, then $Z'\hat{P}_2'P_1$ is acute. The reason is, $P'_2\hat{Z}'P_1$ is wide as follows:
\beq
P'_2\hat{Z}'P_1\geq U\hat{Z}'P_1>U\hat{T}P_1>U\hat{P}'P_1=\frac{\pi}{2}
\eeq

\noindent {\bf{Case 2 ($S$ lies on $OZ'$):}} Consider the righthand side of Figure \ref{cases}. Due to location restrictions, $P_2'$ lies on either $P_1P'H$ triangle or the region induced by $OP'HU$. If it lies on $P_1P'H$ then, $O\hat{P'}P_2'>O\hat{P'}H$ which implies $|OP_2'|\geq |OP'|$ as $O\hat{P}'P_2'$ is wide angle.

If $P_2'$ lies on $OP'HU$ then, $P_1\hat{P}_2'Z'<P_1\hat{H}Z'=\frac{\pi}{2}$ hence $P_1\hat{P}'_2Z'$ is acute angle which cannot happen as it was discussed in the list of properties of $P_2'$.

 \noindent{\bf{When $Z\hat{P_1}O$ is right or acute angle:}} Consider Figure \ref{cases2}. $P_2'$ lies above $UP_1$. It cannot belong to the region induced by $UHT$ as it would imply $Z'\hat{P_2'}P_1< Z'\hat{H} P_1\leq \frac{\pi}{2}$. Then, it belongs to the region induced by $THP_1$ which implies the desired result as $O\hat{P'}P_2'$ is at least right angle.

In all cases, we end up with $|OP_2'|\geq |OP'|$ which implies $\|\p_2\|\geq \|\p_2'\|\geq\alpha\|\p_1\|$ as desired.

\end{proof}

%%%%%%%%%%%%%%%%%%%%%%%%%%%%
\noindent{\textbf{Statement 7}:} 
For a proof see Lemma $B.2$ in \cite{McCoy}.

%%%%%%%%%%%%%%%%%%%%%%%%%%%%
\noindent{\textbf{Statement 8}:} 
From Statement 7, $\Clf = -\frac{\la}{2}{\frac{d\Dlf}{d\la}}$. Also from Statement 5, $\Dlf$ is strictly convex. Thus, ${\frac{d\Dlf}{d\la}}\leq0$ for all $\la\in[0,\labb]$ which yields $\Clf\geq0$ for all $\la\in[0,\labb]$. Similarly,  ${\frac{d\Dlf}{d\la}}\geq0$ for all $\la\in[\labb,\infty)$ which yields $\Clf\leq0$ for all $\la\in[\labb,\infty)$. Finally, $\labb$ minimizes $\Dlf$. Hence ${\frac{d\Dlf}{d\la}}|_{\la=\labb}=0$ which yields $\Cf\labb)=0$.

%%%%%%%%%%%%%%%%%%%%%%%%%%%%
\noindent{\textbf{Statement 9}:} 
We prove that for any $0\leq \la_1<\la_2\leq\labb$,
\begin{align}\label{eq:o67}
\Df\la_1) + \Cf\la_1)> \Df\la_2) + \Cf\la_2).
\end{align}
From Statement 5, $\Dlf$ is strictly decreasing for $\la\in[0,\labb]$. Thus,
\begin{align}\label{eq:o68}
\Df\la_1) > \Df\la_2).
\end{align}
Furthermore, from Statement 6, $\Plf$ is an increasing function of $\la$. Thus, 
\begin{align}\label{eq:o69}
\Df\la_1) + 2\Cf\la_1)  \geq \Df\la_2) + 2\Cf\la_2).
\end{align} 
where we have used  Statement 1.
Combining \eqref{eq:o68} and \eqref{eq:o69}, we conclude with \eqref{eq:o67}, as desired.

%Combining these, we end up with the desired result.

%\input{UniformMaxFormula}
%\input{Comparison.tex}
\section{Explicit formulas for well-known functions}\label{explicit_form}

\subsection{$\ell_1$ minimization}
Let $\x_0\in\R^n$ be a $k$ sparse vector and let $\beta=\frac{k}{n}$. Then, we have the following when $f(\cdot)=\|\cdot\|_1$,
\begin{itemize}
\item $\frac{\Dlf}{n}=(1+\la^2)(1-(1-\beta)\text{erf}(\frac{\la}{\sqrt{2}}))-\sqrt{\frac{2}{\pi}}(1-\beta)\la\exp(-\frac{\la^2}{2})$
\item $\frac{\Plf}{n}=\beta\la^2+(1-\beta)[\text{erf}(\frac{\la}{\sqrt{2}})+\la^2\text{erfc}(\frac{\la}{\sqrt{2}})-\sqrt{\frac{2}{\pi}}\la\exp(-\frac{\la^2}{2})]$
\item $\frac{\Clf}{n}=-\la^2\beta+(1-\beta)[\sqrt{\frac{2}{\pi}}\la\exp(-\frac{\la^2}{2})-\la^2\text{erfc}(\frac{\la}{\sqrt{2}})]$
\end{itemize}
These are not difficult to obtain. For example, to find $\Dlf$, pick $\g\sim\Nn(0,\Iden)$ and consider the vector $\Pi(\g,\la \paf)$. The distance vector to the subdifferential of the $\ell_1$ norm takes the form of soft thresholding on the entries of $\g$. In particular,
\beq
(\Pi(\g,\la \paf))_i=\begin{cases}\g(i)-\la \cdot\text{sgn}(\x_0(i)) &\text{if}~\x_0(i)\neq 0,\\\text{shrink}_\la(\g(i))& \text{otherwise}.\end{cases}\nn
\eeq
where $\text{shrink}_\la(\g(i))$ is the soft thresholding operator defined as,
\beq
\text{shrink}_\la (x)=\begin{cases}
x-\la & \text{if}~~~x>\la,\\
0 &\text{if}~~~|x|\leq \la,\\
 x+\la &\text{if}~~~x<-\la.\end{cases}\nn
\eeq
Consequently, we obtain our formulas after taking the expectation of $\g(i)-\la\cdot\text{sgn}(\x_0(i))$ and $\text{shrink}_\la(\g(i))$. For more details on these formulas, the reader is referred to \cite{DonPhaseTrans,DonCentSym,Sto1,soft-thresh} which calculate the phase transitions of $\ell_1$ minimization.

\subsubsection{Closed form bound}
We will now find a closed form bound on $\Dlf$ for the same sparse signal $\x_0$. In particular, we will show that $\Dlf\leq (\la^2+2)k$ for $\la\geq \sqrt{2\log\frac{n}{k}}$. Following the above discussion and letting $\g\sim\Nn(0,\Iden_n)$, first observe that, $\E[(\g_i-\la\cdot\text{sgn}(\x_0(i)))^2]=\la^2+1$
\beq
\Dlf=\sum\E[(\g(i)-\la\cdot\text{sgn}(\x_0(i)))^2]+(n-k)\E[\text{shrink}_\la (\g(i))^2]\label{combine4}
\eeq
The sum on the left hand side is simply $(\la^2+1)k$. The interesting term is $\text{shrink}_\la (\g(i))$. To calculate this, we will use the following lemma.
\begin{lem} \label{useful lem inf}Let $x$ be a nonnegative random variable. Assume, there exists $c>0$ such that for all $t>0$,
\beq
\Pro(x\geq c+t)\leq \exp(-\frac{t^2}{2})
\eeq
For any $a\geq 0$, we have, \beq\E[\text{shrink}_{a+c}(x)^2]\leq \frac{2}{a^2+1}\exp(-\frac{a^2}{2}).\eeq
\end{lem}
\begin{proof} Let $Q(t)=\Pro(x\geq t)$.
% Then, using integration by parts,
%\begin{align}
%\E[\text{shrink}_{a+c}(x)^2]&=\int_{a+c}^\infty (x-a-c)^2d(-Q(x))\leq \int_{a+c}^\infty (x^2-(a+c)^2)d(-Q(x))\\
%&\leq-[Q(x)x^2]_{a+c}^\infty+\int_{a+c}^\infty Q(x)dx^2+(a+c)^2[Q(\infty)-Q(a+c)]\\
%&\leq (a+c)^2Q(a+c)+\int_{a+c}^\infty Q(x)dx^2+(a+c)^2[Q(\infty)-Q(a+c)] \\
%&\leq \int_{a+c}^\infty Q(x)dx^2 \leq \int_{a+c}^\infty \exp(-\frac{(x-c)^2}{2}) dx^2\\
%&\leq \int_{a}^\infty 2(u+c)\exp(-\frac{u^2}{2}) du
%\end{align}

\begin{align}
\E[\text{shrink}_{a+c}(x)^2]&=\int_{a+c}^\infty (x-a-c)^2d(-Q(x))\\
&\leq-[Q(x)(x-a-c)^2]_{a+c}^\infty+\int_{a+c}^\infty Q(x)d(x-a-c)^2=\int_{a+c}^\infty Q(x)d(x-a-c)^2\label{lastone2}\\
&\leq \int_{a+c}^\infty 2(x-a-c)Q(x)d(x-a-c)\leq  2\int_{a+c}^\infty (x-a-c)\exp(-\frac{(x-c)^2}{2}) d(x-a-c)\nn\\
& \leq 2 \int_{a}^\infty (u-a)\exp(-\frac{u^2}{2}) du\leq  2\exp(-\frac{a^2}{2})-2a\frac{a}{a^2+1}\exp(-\frac{a^2}{2})=\frac{2}{a^2+1}\exp(-\frac{a^2}{2})\label{lastone4}
\end{align}
\eqref{lastone2} follows from integration by parts and \eqref{lastone4} follows from the standard result on Gaussian tail bound, $ \int_{a}^\infty \exp(-\frac{u^2}{2}) du\geq \frac{a}{a^2+1}\exp(-\frac{a^2}{2})$
\end{proof}

To calculate $\E[\text{shrink}_{\la}(g)^2]$ for $g\sim\Nn(0,1)$ we make use of the standard fact about Gaussian distribution, $\Pro(|g|>t)\leq \exp(-\frac{t^2}{2})$. Applying the Lemma \ref{useful lem inf} with $c=0$ and $a=\la$ yields, $\E[|\text{shrink}_{\la}(g)|^2]\leq \frac{2}{\la^2+1}\exp(-\frac{\la^2}{2})$. Combining this with \eqref{combine4}, we find,
\beq
\Dlf\leq (\la^2+1)k+\frac{2n}{\la^2+1}\exp(-\frac{\la^2}{2})
\eeq
For $\la\geq \sqrt{2\log\frac{n}{k}}$, $\exp(-\frac{\la^2}{2})\leq \frac{k}{n}$. Hence, we obtain,
\beq
\Dlf\leq (\la^2+1)k+\frac{2k}{\la^2+1}\leq (\la^2+3)k
\eeq\begin{comment}
\subsection{$\ell_1$ minimization subject to nonnegativity}
In this case, $\x_0$ is still a sparse signal however, we additionally use the set constraint $\x\in\R^n_+$ where $\R^n_+=\{\vb\in\R^n\big|v_i\geq 0,~1\leq i\leq n\}$. The positivity constraint is active at the entries of $\x_0$ that are equal to $0$. Considering the set $\R^n_+\cup\la\|\x_0\|_1$, one can obtain the following.
\begin{itemize}
\item $\frac{\Dlf}{n}=(1+\la^2)(\beta+\frac{1-\beta}{2}\text{erfc}(\frac{\la}{\sqrt{2}}))-\frac{1-\beta}{\sqrt{2\pi}}\la\exp(-\frac{\la^2}{2})$
\item $\frac{\Plf}{n}=\beta\la^2+\frac{1-\beta}{2}[1+\text{erf}(\frac{\la}{\sqrt{2}})+\la^2\text{erfc}(\frac{\la}{\sqrt{2}})-\sqrt{\frac{2}{\pi}}\la\exp(-\frac{\la^2}{2})]$
\item $\frac{\Clf}{n}=-\la^2\beta+\frac{1-\beta}{2}[\sqrt{\frac{2}{\pi}}\la\exp(-\frac{\la^2}{2})-\la^2\text{erfc}(\frac{\la}{\sqrt{2}})]$
\end{itemize}
For these formulas, the reader is referred to \cite{DonCentSym,Sto1} which calculate the phase transitions of $\ell_1$ minimization subject to nonnegativity constraints.
\end{comment}

\subsection{Nuclear norm minimization}
Assume $\X_0$ is a $d\times d$ matrix of rank $r$ and $\x_0$ is its vector representation where $n=d^2$ and we choose nuclear norm to exploit the structure. Denote the spectral norm of a matrix by $\|\cdot\|_2$. Assume $\X_0$ has skinny singular value decomposition $\U\Sigma\V^T$ where $\Sigma\in\R^{r\times r}$. Define the ``support'' subspace of $\X_0$ as,
\beq
S_{\X_0}=\{{\bf{M}}\big|(\Iden-\U\U^T){\bf{M}}(\Iden-\V\V^T)=0\}
\eeq
The subdifferential of nuclear norm is given as,
\beq
\pa\|\X_0\|_\st=\{\Sb\in\R^{d\times d}\big|\bu(\Sb,S_{\X_0})=\U\V^T,~\text{and}~\|\bu(\Sb,\bar{S}_{\X_0})\|_2\leq 1\}
\eeq
Based on this, we wish to calculate $\dt(\Gb,\la\paf)$ when $\Gb$ has i.i.d. standard normal entries. As it has been discussed in \cite{Matesh2,Oym,Oym2}, $\Pi(\Gb,\la\paf)$ effectively behaves as singular value soft thresholding. In particular, we have,
\beq
\Pi(\Gb,\la\paf)=(\bu(\Gb,S_{\X_0})-\la\U\V^T)+\sum_{i=1}^{n-r}\text{shrink}_\la(\sigma_{\Gb,i})\ub_{\Gb,i}\vb_{\Gb,i}^T
\eeq
where $\bu(\Gb,\bar{S}_{\X_0})$ has singular value decomposition $\sum_{i=1}^{n-r}\sigma_{\Gb,i}\ub_{\Gb,i}\vb_{\Gb,i}^T$.

Based on this behavior, $\dt(\Gb,\la\paf)$ has been analyzed in various works in the linear regime where $\frac{r}{d}$ is constant. This is done by using the fact that the singular value distribution of a $d\times d$ matrix approaches to quarter circle law when singular values are normalized by $\sqrt{d}$.
\begin{align}
\psi(x)=\begin{cases}\frac{1}{\pi}\sqrt{4-x^2}~\text{if}~0\leq x\leq 2\\ 0~\text{else}\end{cases}
\end{align}
Based on $\psi$, define the quantities related to the moments of tail of $\psi$. Namely,
\begin{align}
\Psi_i(x)=\int_{x}^\infty x^i\psi(x)dx
\end{align}
We can now give the following explicit formulas for the asymptotic behavior of $\pa\|\X_0\|_\st$ where $\frac{r}{d}=\beta$ is fixed. Define,
\begin{align}
&\upsilon=\frac{\la}{2\sqrt{1-\beta}}
\end{align}
\begin{itemize}
\item $\frac{\Df\la\sqrt{d})}{n}=[2\beta-\beta^2+\beta\la^2]+[(1-\beta)\la^2\Psi_0(\upsilon)+(1-\beta)^2\Psi_2(\upsilon)-2(1-\beta)^{3/2}\la\Psi_1(\upsilon)]$
\item $\frac{\Pf\la\sqrt{d})}{n}=\beta\la^2+(1-\beta)\la^2\Psi_0(\upsilon)+(1-\beta)^2(1-\Psi_2(\upsilon))$
\item $\frac{\Cf\la\sqrt{d})}{n}=-\la^2\beta-(1-\beta)\la^2\Psi_0(\upsilon)+(1-\beta)^{3/2}\la\Psi_1(\upsilon)$
\end{itemize}

\subsubsection{Closed form bounds}
Our approach will exactly follow the proof of Proposition 3.11 in \cite{Cha}. Given $\Gb$ with i.i.d. standard normal entries, the spectral norm of the off-support term $\bu(\Gb,\bar{S}_{\X_0})$ satisfies,
\beq
\Pro(\|\bu(\Gb,\bar{S}_{\X_0})\|_2\geq 2\sqrt{d-r}+t)\leq \exp(-\frac{t^2}{2})
\eeq
It follows that all singular values of $\bu(\Gb,\bar{S}_{\X_0})$ satisfies the same inequality as well. Consequently, for any singular value and for $\la\geq 2\sqrt{d-r}$, applying Lemma \ref{useful lem inf}, we may write,
\beq
\E[\text{shrink}_\la(\sigma_{\Gb,i})^2]\leq \frac{2}{(\la-2\sqrt{d-r})^2+1}\exp(-\frac{(\la-2\sqrt{d-r})^2}{2})\leq 2
\eeq
It follows that,
\beq
\sum_{i=1}^{d-r}\E[\text{shrink}_\la(\sigma_{\Gb,i})^2]\leq 2(d-r)
\eeq
To estimate the in-support terms, we need to consider $\bu(\Gb,S_{\X_0})-\la\U\V^T$. Since $\la\U\V^T$ and $\bu(\Gb,S_{\X_0})$ are independent, we have,
\beq
\|\bu(\Gb,S_{\X_0})-\la\U\V^T\|_F^2=\la^2r+|S_{\X_0}|=\la^2r+2dr-r^2
\eeq
Combining, we find,
\beq
\Dlf\leq \la^2r+2dr-r^2+2d-2r\leq(\la^2+2d)r+2d
\eeq
%When $\la\geq 2\sqrt{d}$, observe that, $(\la-2\sqrt{d-r})^2\geq 4r$
%\beq
%8d-4r-8\sqrt{d^2-dr}\geq \frac{r^2}{d}
%\eeq
\subsection{Block sparse signals}
Let $n=t\times b$ and assume entries of $\x_0\in\R^n$ can be partitioned into $t$ blocks of size $b$ so that only $k$ of these $t$ blocks are nonzero. To induce the structure, use the $\ell_{1,2}$ norm which sums up the $\ell_2$ norms of the blocks, \cite{StojBlock,Yonina,RechtBlock}. In particular, denoting the subvector corresponding to $i$'th block of $\x$ by $\x_i$
\beq
\|\x\|_{1,2}=\sum_{i=1}^t\|\x_i\|
\eeq
To calculate $\Dlf,\Clf,\Plf$ with $f(\cdot)=\|\cdot\|_{1,2}$, pick $\g\sim\Nn(0,\Iden_n)$ and consider $\Pi(\g,\la\partial \|\x_0\|_{1,2})$ and $\bu(\g,\la\partial \|\x_0\|_{1,2})$. Similar to $\ell_1$ norm and the nuclear norm, distance to subdifferential will correspond to a ``soft-thresholding''. In particular, $\Pi(\g,\la\partial \|\x_0\|_{1,2})$ has been studied in \cite{StojBlock,RechtBlock} and is given as,
\beq
\Pi(\g,\la\partial \|\x_0\|_{1,2})=\begin{cases}\g_i-\la\frac{\x_{0,i}}{\|\x_{0,i}\|}~~~\text{if}~~~\x_{0,i}\neq 0\\\text{vshrink}_{\la}(\g_i)~~~\text{else}\end{cases}
\eeq
where the vector shrinkage $\text{vshrink}_{\la}$ is defined as,
\beq
\text{vshrink}_{\la}(\vb)=\begin{cases}\vb(1-\frac{\la}{\|\vb\|})~~~\text{if}~~~\|\vb\|> \la\\0~~~\text{if}~~~\|\vb\|\leq \la\end{cases}
\eeq
When $\x_{0,i}\neq 0$ and $\g_i$ is i.i.d. standard normal, $\E[\|\g_i-\la\frac{\x_{0,i}}{\|\x_{0,i}\|^2}\|^2]=\E[\|\g_i\|^2]+\la^2=b+\la^2$. Calculation of $\text{vshrink}_{\la}(\g_i)$ and has to do with the tails of $\chi^2$-distribution with $b$ degrees of freedom (see Section $3$ of \cite{RechtBlock}). Similar to previous section, define the tail function of a $\chi^2$-distribution with $b$ degrees of freedom as,
\beq
\Psi_i(x)=\int_{x}^{\infty}x^i \frac{1}{2^{\frac{k}{2}}\Gamma(\frac{k}{2})}x^{\frac{k}{2}-1}\exp(-\frac{x}{2})dx
\eeq
Then, $\E[\|\text{vshrink}_{\la}(\g_i)\|^2]=\Psi_1(\la^2)+\Psi_0(\la^2)\la^2-2\Psi_\frac{1}{2}(\la^2)\la$. Based on this, we calculate $\Dlf,\Plf$ and $\Clf$ as follows.

\begin{itemize}
\item $\Dlf=k(b+\la^2)+[\Psi_1(\la^2)+\Psi_0(\la^2)\la^2-2\Psi_\frac{1}{2}(\la^2)\la](t-k)$
\item $\Plf=\la^2k+[(\Psi_1(0)-\Psi_1(\la^2))+\la^2\Psi_0(\la^2)](t-k)$
\item $\Clf=-\la^2k+[\la\Psi_\frac{1}{2}(\la^2)-\la^2\Psi_0(\la^2)](t-k)$
\end{itemize}
\subsubsection{Closed form bound}
Similar to Proposition 3 of \cite{Foygel}, we will make use of the following bound for a $x$ distributed with $\chi^2$-distribution with $b$ degrees of freedom.
\beq
\Pro(\sqrt{x}\geq \sqrt{b}+t)\leq \exp(-\frac{t^2}{2})~~~\text{for all}~t>0\label{tail bound}
\eeq
Now, the total contribution of nonzero blocks to $\Dlf$ is simply $(\la^2+b)k$ as $\E[\|\g_i-\la\frac{\x_{0,i}}{\|\x_{0,i}\|}\|^2]=\la^2+b$. For the remaining, we need to estimate $\E[\|\text{vshrink}_{\la}(\g_i)\|^2]$ for an i.i.d. standard normal $\g_i\in\R^d$. Using Lemma \ref{useful lem inf}, with $c=\sqrt{b}$ and $a=\la-\sqrt{b}$ and using the tail bound \eqref{tail bound}, we obtain,
\beq
\E[\|\text{vshrink}_{\la}(\g_i)\|^2]\leq \frac{2}{(\la-\sqrt{b})^2+1}\exp(-\frac{(\la-\sqrt{b})^2}{2})
\eeq
Combining everything,
\beq
\Dlf\leq k(\la^2+b)+ \frac{2t}{(\la-\sqrt{b})^2+1}\exp(-\frac{(\la-\sqrt{b})^2}{2})
\eeq
Setting $\la\geq\sqrt{b}+\sqrt{2\log{\frac{t}{k}}}$, we ensure, $\exp(-\frac{(\la-\sqrt{b})^2}{2})\leq \frac{k}{t}$, hence,
\beq
\Dlf\leq k(\la^2+b)+ \frac{2k}{(\la-\sqrt{b})^2+1}\leq k(\la^2+b+2)
\eeq
\section{Gaussian Width of the Widened Tangent Cone}\label{widewidth}
The results in this appendix will be useful to show the stability of $\ell_2^2$-LASSO for all $\tau>0$. To state the results, we will first define the Gaussian width which has been the topic of closely related papers \cite{Gor,Oymak,Cha,McCoy}. 
\begin{defn}
Let $S\subseteq\R^n$. Let $\g\sim\Nn(0,\Iden_n)$. Then, the Gaussian width of $S$ is given as,
\beq
\omega(S)=\E[\sup_{\vb\in S} \li\vb,\g\ri]
\eeq
\end{defn}
Let us also state a standard result on the Gaussian width and cones that can be found in \cite{Foygel,McCoy}.
\begin{propo}
\end{propo}
The following lemma provides a Gaussian width characterization of ``widening of a tangent cone''.
\begin{lem} \label{widen cone}Assume $f(\cdot)$ is a convex function and $\x_0$ is not a minimizer of $f(\cdot)$. Given $\eps_0>0$, consider the $\eps_0$-widened tangent cone defined as,
\beq
\Tc_f(\x_0,\eps_0)=\text{Cl}(\{\alpha\cdot\w\big|f(\x_0+\w)\leq f(\x_0)+\eps_0\|\w\|,~\alpha\geq 0\})\label{wide cone}
\eeq
Let $R_{min}=\min_{\s\in\paf}\|\s\|$ and $\Bc^{n-1}$ be the unit $\ell_2$-ball in $\R^n$. Then,
\beq
\omega(\Tc_f(\x_0,\eps_0)\cap \Bc^{n-1})\leq \omega(\Tc_f(\x_0)\cap \Bc^{n-1})+\frac{\eps_0\sqrt{n}}{R_{min}}
\eeq
\end{lem}
\begin{proof} Let $\w\in \Tc_f(\x_0,\eps_0)$. Write $\w=\w_1+\w_2$ via Moreau's decomposition theorem (Fact \ref{more}) where $\w_1\in \Tc_f(\x_0)$ and $\w_2\in \text{cone}(\paf)$ and $\w_1^T\w_2=0$. Here we used the fact that $\x_0$ is not a minimizer and $\Tc_f(\x_0)^*=\text{cone}(\paf)$. To find a bound on $\Tc_f(\x_0,\eps_0)$ in terms of $\Tc_f(\x_0)$, our intention will be to find a reasonable bound on $\w_2$ and to argue $\w$ cannot be far away from its projection on the tangent cone.

To do this, we will make use of the followings.
\begin{itemize}
\item If $\w_2\neq 0$, since $\w_1^T\w_2=0$, $\max_{\s\in\paf} \w_1^T\s=0$.
\item Assume $\w_2\neq 0$. Then $\w_2=\alpha\s(\w_2)$ for some $\alpha>0$ and $\s(\w_2)\in\paf$.
\end{itemize}
From convexity, for any $1>\eps>0$, $\eps\eps_0\|\w\|\geq f(\eps\w+\x_0)-f(\x_0)$. Now, using Proposition \ref{prop:F3T} with $\delta\rightarrow0$, we obtain,
\begin{align}
\eps_0\|\w\|\geq\lim_{\eps\rightarrow 0} \frac{f(\eps\w+\x_0)-f(\x_0)}{\eps}&= \sup_{\s\in\paf}\w^T\s\nn\\
&\geq  \w^T\s(\w_2)=\w_1^T\s(\w_2)+\w_2^T\s(\w_2)\nn\\
&=\|\w_2\|\|\s(\w_2)\|\geq \|\w_2\|R_{min}
\end{align}
This gives, $\frac{\|\w_2\|}{\|\w\|}\leq\frac{\eps_0}{R_{min}}$. Equivalently, for a unit size $\w$, $\|\w_2\|\leq \frac{\eps_0}{R_{min}}$.

What remains is to estimate the Gaussian width of $\Tc_f(\x_0,\eps_0)\cap \Bc^{n-1}$. Let $\g\sim\Nn(0,\Iden_n)$. $\w_1,\w_2$ still denote the projection of $\w$ onto $\Tc_f(\x_0)$ and $\text{cone}(\paf)$ respectively.
\begin{align}
\omega(\Tc_f(\x_0,\eps_0)\cap \Bc^{n-1})&=\E[\sup_{\w\in \Tc_f(\x_0,\eps_0)\cap \Bc^{n-1} }\w^T\g]\\
&\leq \E[\sup_{\w\in \Tc_f(\x_0,\eps_0)\cap \Bc^{n-1} }\w_1^T\g]+\E[\sup_{\w\in \Tc_f(\x_0,\eps_0)\cap \Bc^{n-1}}\w_2^T\g]
\end{align}
Observe that, for $\w\in\Tc_f(\x_0,\eps_0)\cap\Bc^{n-1}$, $\|\w_2\|\leq \frac{\eps_0}{R_{min}}$,
\beq
\E[\sup_{\w\in \Tc_f(\x_0,\eps_0)\cap \Bc^{n-1}}\w_2^T\g]\leq \E[\sup_{\w\in \Tc_f(\x_0,\eps_0)\cap \Bc^{n-1}}\|\w_2\|\|\g\|]\leq \frac{\eps_0}{R_{min}}\E\|\g\|\leq \frac{\eps_0\sqrt{n}}{R_{min}}
\eeq
For $\w_1$, we have $\w_1\in\Tc_f(\x_0)$ and $\|\w_1\|\leq \|\w\|\leq 1$ which gives,
\beq
 \E[\sup_{\w\in \Tc_f(\x_0,\eps_0)\cap \Bc^{n-1} }\w_1^T\g]\leq  \E[\sup_{\w'\in \Tc_f(\x_0)\cap \Bc^{n-1} }{\w'}^T\g]=\omega(\Tc_f(\x_0)\cap \Bc^{n-1})
\eeq
Combining these individual bounds, we find,
\beq
\omega(\Tc_f(\x_0,\eps_0)\cap \Bc^{n-1})\leq \omega(\Tc_f(\x_0)\cap \Bc^{n-1})+\frac{\eps_0\sqrt{n}}{R_{min}}
\eeq

\end{proof}

\begin{lem} \label{wide lem}Let $\Tc_f(\x_0,\eps_0)$ denote the widened cone defined in \eqref{wide cone} and consider the exact same setup in Lemma \ref{widen cone}. Fix $\eps_1>0$. Let $\A\in\R^{m\times n}$ have i.i.d. standard normal entries. Then, whenever,
\beq
\gamma(m,f,\eps_0,\eps_1):=\sqrt{m-1}-\sqrt{\Delxf}-\frac{\eps_0\sqrt{n}}{R_{min}}-\eps_1>0
\eeq
we have,
\beq
\Pro(\min_{\vb\in \Tc_f(\x_0,\eps_0)\cap \Bc^{n-1}}\|\A\vb\|\geq \eps_1)\geq 1-2\exp(-\frac{1}{2}\gamma(m,f,\eps_0,\eps_1)^2)
\eeq

\end{lem}
\begin{proof} Our proof will follow the same lines as the proof of Corollary 3.3 of Chandrasekaran et al. \cite{Cha}. For this proof, we will make use of the following lemma of Gordon \cite{Gor} (Corollary 1.2).
\begin{propo} Let $\Cc\in\R^n$ be a closed and convex subset of $\Bc^{n-1}$. Then,
\beq
\E[\min_{\vb\in\Cc}\|\A\vb\|]\geq\sqrt{m-1}-\omega(\Cc)
\eeq
\end{propo}
Pick $\Cc=\Tc_f(\x_0,\eps_0)\cap \Bc^{n-1}$ in the above proposition. Combined with Lemma \ref{widen cone}, this gives,
\beq
\E[\min_{\vb\in \Tc_f(\x_0,\eps_0)\cap \Bc^{n-1}}\|\A\vb\|]\geq\sqrt{m-1}-\omega(\Tc_f(\x_0))-\frac{\eps_0\sqrt{n}}{R_{min}}\label{eqeps1}
\eeq
Following \cite{Cha}, the function $\min_{\vb\in \Tc_f(\x_0,\eps_0)\cap \Bc^{n-1}}\|\A\vb\|$ is $1$-Lipschitz function of $\A$ in Frobenius norm. Using Lemma \ref{fact:lipIneq}, for $\eps_1$ smaller than the right hand side of \eqref{eqeps1}, we find,
\beq
\Pro(\min_{\vb\in \Tc_f(\x_0,\eps_0)\cap \Bc^{n-1}}\|\A\vb\|\geq \eps_1)\geq 1-2\exp(-\frac{1}{2}(\sqrt{m-1}-\omega(\Tc_f(\x_0)\cap \Bc^{n-1})-\frac{\eps_0\sqrt{n}}{R_{min}}-\eps_1)^2)
\eeq

To conclude, we will use $\omega(\Tc_f(\x_0)\cap \Bc^{n-1})\leq \Delxf$. To see this, applying Moreau's decomposition theorem (Fact \ref{more}), observe that for a closed and convex cone $\Kc$ and an arbitrary vector $\g$,
\beq
\|\bu(\g,\Kc)\|=\sup_{\vb\in \Kc\cap \Bc^{n-1}}\vb^T\g
\eeq
Picking $\Kc=\Tc_f(\x_0)$ and $\g\sim\Nn(0,\Iden_n)$,
\beq
\omega(\Tc_f(\x_0)\cap \Bc^{n-1})=\E[\|\bu(\g,\Tc_f(\x_0))\|]\leq\sqrt{ \E[\|\bu(\g,\Tc_f(\x_0))\|^2]}=\Delxf
\eeq
%From the lemma above, we have a deterministic bound on the statistical dimension of . Now, applying Gordon's Lemma, we find that, whenever $m\geq \delta()$, \eqref{simpsimpeq} holds with the desired probability.
\end{proof}

\end{document}